\documentclass[11pt, a4paper, twoside]{book}
\usepackage{a4wide}

\usepackage[Sonny]{fncychap}
\usepackage[final]{graphicx}
\usepackage[dutch,english]{babel}
\usepackage{bbold}
\usepackage{url}
\usepackage[small,bf,flushleft]{caption}
\usepackage{xspace}
\usepackage{tocbibind}
\usepackage{amsmath}
\usepackage[final]{pdfpages}
\usepackage{amsthm}
\usepackage{rotating}
\usepackage{emptypage}

\usepackage{pdfpages}

\usepackage{algorithm}
\usepackage{algpseudocode}

\usepackage{multirow}

\usepackage{makeidx}
\makeindex

\usepackage{color}

\allowdisplaybreaks[1] 

\newcommand{\bra}[1]{\langle #1|}
\newcommand{\ket}[1]{|#1\rangle}
\newcommand{\braket}[2]{\langle #1|#2\rangle}

\newtheorem{theorem}{Theorem}

\begin{document}

%
%
%

\thispagestyle{empty}
\begin{center}
\includegraphics[width=\textwidth]{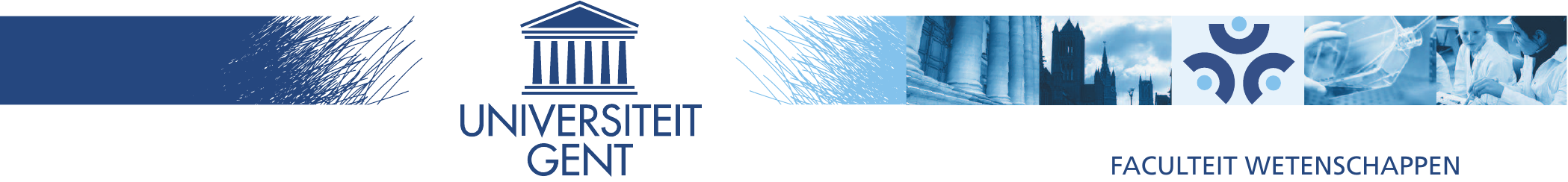}
\end{center}
\begin{flushright}
\vspace*{1.2cm}

{\Large \bf \textsf{Variational determination of the two-particle density matrix as a quantum many-body technique\\}}
\vspace{0.2cm}
\rule{\textwidth}{1mm}\\
\vspace{0.4cm}
{\large  \textsf{Brecht Verstichel}\\}
\end{flushright}

\vspace*{\fill}
\par
\noindent
{\normalsize \textsf{Promotor: Prof. dr. Dimitri Van Neck}}\\
{\normalsize \textsf{Co-promotor: Prof. dr. Patrick Bultinck}}\\
\vspace{0.2cm}
\\
{\normalsize \textsf{Proefschrift ingediend tot het behalen van de academische graad van\\
Doctor in de Wetenschappen: Fysica}}\\
\vspace{-0.2cm}\\
\noindent
{\normalsize \textsf{Universiteit Gent}}\\
{\normalsize \textsf{Faculteit Wetenschappen}}\\
{\normalsize \textsf{Vakgroep Fysica en Sterrenkunde}}\\
{\normalsize \textsf{Academiejaar 2011-2012}}\\

\clearpage{\pagestyle{empty}\cleardoublepage}

%

\frontmatter

\tableofcontents

\mainmatter

\chapter{Introduction}
In the introductory Section, one is supposed to set the stage for the coming Chapters. This stage is the quantum many-body problem. Its importance lies in the fact that it provides the fundamental description of processes in fields as varied as atomic, molecular, solid state and nuclear physics. Apart from the thrill of exploring physical phenomena at the quantum level there is also the realization that predicting and manipulating such microscopic processes has powered much of the 20th century technology, and is likely to lead to further breakthroughs in our 21st century.

In fact, the quantum mechanical description of many interacting particles is a problem that has been around since the dawn of quantum mechanics. Already in 1929 Dirac wrote \cite{dirac}:
\begin{quote}
The underlying physical laws necessary for the mathematical theory of a large part of physics and the whole of chemistry are thus completely known, and the difficulty is only that the exact application of these laws leads to equations much too complicated to be soluble.
\end{quote}
What Dirac meant by this is that, in principle, the formalism to treat many-electron problems arising in the study of atoms, molecules and solids is completely known. Assuming that one can neglect relativity and that electrons interact solely through the Coulomb interaction, quantum mechanics provides the recipe by which the problem can be tackled. This recipe consists of solving the Schr\"odinger equation:
\begin{equation}
\hat{H}\Psi = E\Psi~.
\label{schrodinger}
\end{equation}
Mathematically, this is nothing but an ordinary eigenvalue equation in a linear space endowed with an inproduct, called Hilbert space in the context of quantum mechanics. The $N$-particle wave function $\Psi$ is a vector in Hilbert space. The linear and Hermitian operator $\hat{H}$ is called the Hamiltonian, and contains all the information about the interparticle interactions. The eigenvalue $E$ denotes the possible energies of the system.
When identical particles are considered, an extra permutation symmetry of the wave function $\Psi$ has to be imposed, better known as the Pauli-principle for electrons.

The problem is that the direct application of this recipe is necessarily limited to small systems, because of the exponential scaling of Hilbert space with the number of particles involved. This makes the solution of the eigenvalue problem in Eq.~(\ref{schrodinger}) problematic as it scales as $O(n^3)$ where $n$ is the dimension of Hilbert space. A linear spin chain, consisting of sites with spin-$\frac{1}{2}$ degrees of freedom, provides a simple example that illustrates this dramatic scaling of Hilbert space. A single spin-$\frac{1}{2}$ site is described by a wave function consisting of two components, a spin-up and a spin-down term:
\begin{equation}
\ket{\Psi} = C_\uparrow \ket{\uparrow} + C_\downarrow\ket{\downarrow}~.
\end{equation}
Only two complex numbers $C_\sigma$ are needed to completely characterize this state. For describing two sites one requires:
\begin{equation}
\ket{\Psi} = C_{\uparrow\uparrow} \ket{\uparrow\uparrow} 
+ C_{\downarrow\uparrow}\ket{\downarrow\uparrow}
+ C_{\uparrow\downarrow} \ket{\uparrow\downarrow}
+ C_{\downarrow\downarrow}\ket{\downarrow\downarrow}
~,
\end{equation}
so four numbers are needed. It is easy to appreciate that for an $L$-spin state the general wavefunction is:
\begin{equation}
\ket{\Psi} = \sum_{\sigma_i \in \{\uparrow,\downarrow\}} C_{\sigma_1\ldots\sigma_L}\ket{\sigma_1,\ldots,\sigma_L}~,
\end{equation}
and one requires $2^L$ numbers to completely describe the state. This observation led Dirac to remark:
\begin{quote}
It therefore becomes desirable that approximate practical methods of applying quantum mechanics should be developed, which can lead to an explanation of the main features of complex atomic systems without too much computation.
\end{quote}
Over the last eighty years a great variety of such approximate methods have been developed \cite{fetter_walecka,bijbel,QMC,schollwock,verstraete}: from perturbation theory and cluster expansions over self-consistent field and variational methods to renormalization group methods and stochastical techniques like Quantum Monte Carlo. Most of these approximate methods try to somehow capture the relevant information, present in the wave function, in a reduced object. 

If the quantum many-body problem is the stage, the main protagonist in this thesis is the reduced density matrix. The reduced density matrix method discussed in this thesis is an approximate method that tries to replace the wave function, with its exponentially scaling number of variables, with the two-particle density matrix (2DM), for which only a quartically scaling number of variables are needed. This is a very efficient reduction since these are just the degrees of freedom needed for the exact evaluation of the energy.

The reduced density matrix makes its first appearance in the work of Dirac, in which the single-particle density matrix (1DM) is used in the description of Hartree-Fock theory \cite{dirac_dm}. Husimi \cite{husimi} was the first to note that, for a system of identical particles interacting only in a pairwise manner, the energy can be expressed exactly as a function of the 2DM. This becomes very clear in second-quantized notation (see {\it e.g.} \cite{fetter_walecka,bijbel}), where a system of identical particles interacting pairwise is described by the general Hamiltonian:
\begin{equation}
\hat{H} = \sum_{\alpha\beta}t_{\alpha\beta}a^\dagger_\alpha a_\beta + \frac{1}{4}\sum_{\alpha\beta\gamma\delta}V_{\alpha\beta;\gamma\delta}a^\dagger_\alpha a^\dagger_\beta a_\delta a_\gamma~.
\label{2pham}
\end{equation}
The expectation value for the energy of any ensemble of $N$-particle wave functions $\ket{\Psi_i^N}$ with positive weights $w_i$, can then be expressed as a function of the 2DM alone:
\begin{equation}
\sum_iw_i\bra{\Psi^N_i}\hat{H}\ket{\Psi^N_i}= \mathrm{Tr}~\Gamma H^{(2)} = \frac{1}{4}\sum_{\alpha\beta\gamma\delta}\Gamma_{\alpha\beta;\gamma\delta}H^{(2)}_{\alpha\beta;\gamma\delta}~,
\label{ener_func}
\end{equation}
in which we have introduced the 2DM:
\begin{equation}
\Gamma_{\alpha\beta;\gamma\delta} = \sum_iw_i\bra{\Psi^N_i}a^\dagger_\alpha a^\dagger_\beta a_\delta a_\gamma \ket{\Psi^N_i}~,\qquad\text{with}\qquad\sum_iw_i=1~,
\label{2DM_intro}
\end{equation}
and the reduced two-particle Hamiltonian,
\begin{equation}
H^{(2)}_{\alpha\beta;\gamma\delta} = \frac{1}{N-1}\left(\delta_{\alpha\gamma}t_{\beta\delta} - \delta_{\alpha\delta}t_{\beta\gamma} - \delta_{\beta\gamma}t_{\alpha\delta} + \delta_{\beta\delta}t_{\alpha\gamma}\right) + V_{\alpha\beta;\gamma\delta}~.
\label{reduced_ham}
\end{equation}
The idea to use the 2DM as a variable in a variational scheme was first published in literature by L\"owdin in his groundbreaking article \cite{lowdin}, but even earlier, in 1951, John Coleman tried a practical variational calculation on Lithium. To his surprise, the energy he obtained was far too low, after which he realized the variation was performed over too large a class of 2DM's \cite{rdm_book}. Independently and unaware of the work by L\"owdin and Coleman, Joseph Mayer \cite{mayer} used the 2DM in a study of the electron gas. In a reply to Mayer's paper, Tredgold \cite{tredgold} pointed out the unphysical nature of the results, and suggested that additional conditions on the density matrix are needed to improve on them. 

These results led Coleman, in his seminal review paper \cite{coleman}, to formulate the $N$-rep\-re\-sen\-ta\-bi\-li\-ty problem. This is the problem of finding the necessary and sufficient conditions which a reduced density matrix has to fulfil to be derivable from a statistical ensemble of physical wave functions, {\it i.e.} expressible as in Eq.~(\ref{2DM_intro}). In this paper he also derived the necessary and sufficient conditions for ensemble $N$-representability of the 1DM (see also Section~\ref{n_rep_1dm}), and some bounds on the eigenvalues of the 2DM. A big step forward was the derivation of the $\mathcal{Q}$ and $\mathcal{G}$ matrix positivity conditions by Garrod and Percus \cite{garrod}. These were practical constraints, which allowed for a computational treatment of the problem. The first numerical calculation using these conditions on the Beryllium atom \cite{fusco,garrod_mih_ros} was very encouraging, as the results were highly accurate. It turned out, however, that Beryllium, due to its simple electronic stucture, is a special case where these conditions perform very well. A subsequent study showed that these conditions do not work well at all for nuclei \cite{mihailovic,rosina}. This disappointing result, together with the computational complexity of the problem, caused activity in the field to diminish for the next 25 years. The change came with the development of a new numerical technique, called semidefinite programming, which turned out to be very suited for the determination of the 2DM under matrix positivity constraints. Maho Nakata {\it et al.} \cite{nakata_first} were the first to use a standard semidefinite programming package to calculate the ground-state energies of some atoms and molecules, and obtained quite accurate results. He was quickly followed by Mazziotti \cite{mazziotti}. These results reinvigorated interest in the method, and sparked of a lot of developments. New $N$-representability conditions were introduced, {\it e.g.} the three-index $\mathcal{T}$ conditions (see Section~\ref{three_index}), set forth by Zhao {\it et al.} \cite{zhao}, which led to mHartree accuracy \cite{hammond,nakata_last,mazz_T_con,Gido_T_con,mazz_book,braams_book} for molecules near equilibrium geometries. 

In recent years interest in the method has been growing, as the variational determination of the 2DM results in a lower bound, which is highly complementary to the upper bound obtained in variational approaches based on the wave function. In addition, the method is essentially non-perturbative in nature, and has a completely different structure unrelated to other many-body techniques. A lot of activity has been devoted to the search for new $N$-representability conditions, which improve the result in a computationally cheap way \cite{dimi,qsep}. There have also been efforts to improve the semidefinite programming algorithms by adapting them to the specific problem of density matrix optimization \cite{maz_prl,maz_bp,primal_dual}, allowing the study larger systems.

In this thesis we discuss the framework of density matrix optimization, its applications to physical systems, and the contributions we made to the field. In Chapter~\ref{n_rep} the $N$-representability problem is introduced. We start by discussing its formal definition, and how useful necessary constraints can be derived. A non-exhaustive overview of what is known about $N$-representability is provided, and some of the non-standard approaches we developed are discussed. The next Chapter is devoted to the semidefinite formulation of density matrix optimization. We present a detailed analysis of the three different semidefinite programming algorithms that were developed during the PhD. A comparison is made of their performance using the one-dimensional Hubbard model as a benchmark. Chapter~\ref{symmetry} is rather technical, and deals with how the semidefinite programming algorithms can be made more performant by exploiting symmetries. In Chapter~\ref{applications} we present the results that are obtained when applying this method to a variety of physical systems, and how new $N$-representability constraints are derived when the results are not satisfactory. In the application of the method to the one-dimensional Hubbard model a drastic failure of the standard two-index $N$-representability constraints in the strong-correlation limit is observed. An analysis is made of why these constraints fail, and what the relevant correlations are that need to be included in order to fix the problem. In Chapter~\ref{v2.5DM} an approach is presented which includes these correlations, without becoming computionally too expensive. We present results which show that the strong-correlation limit is well described using this approach, and that the quality of the results is improved for the whole phase diagram. Finally, in Chapter~\ref{conclusions} we draw conclusions about the method, discuss what is needed to make it a competitive electronic structure method, and provide ideas for future research.

\chapter{\label{n_rep}$N$-representability}

In the introduction we introduced the $N$-representability problem, and how it was discovered historically. In this Chapter we try to give an overview of what is known about exact and approximate $N$-representability of reduced density matrices. In Section \ref{def_n_rep} the $N$-representability problem is placed in the broader context of what is known in mathematics as marginal problems, and different definitions of $N$-representability are introduced. In Section \ref{standard_n_rep} we proceed by deriving the necessary and sufficient $N$-representability conditions for the one-particle density matrix, and some necessary conditions on the two-particle density matrix that are well known and frequently used. The last part of this Chapter deals with non-standard $N$-representability constraints, such as subsystem constraints or diagonal conditions. For more information on some of the mathematical concepts used in this Chapter we refer to Appendix \ref{math_notes}.

\section{\label{def_n_rep}Definitions of $N$-representability}

The $N$-representability problem, as introduced by Coleman \cite{coleman}, is actually a special case of a set of problems known in mathematics as marginal problems \cite{classical_marg}. Given a probability distribution with $N$ variables, $p(x_1,\ldots,x_N)$, a $k$-marginal distribution is defined as:
\begin{equation}
^k_Np(x_{i_1},\ldots,x_{i_k}) = \sum_{x_1,\ldots,x_{i_1 - 1},x_{i_1 + 1},\ldots,x_{i_k -1},x_{i_k + 1},\ldots,x_N}p(x_1,\ldots,x_N)~.
\label{clas_marg}
\end{equation}
The classical marginal problem can be formulated as the question what the conditions are that a set of $\left(\begin{matrix}N\\k\end{matrix}\right)$ $k$-marginal distributions has to fulfil to be consistently derivable from an $N$-variable probability distribution as in Eq.~(\ref{clas_marg}). The logical quantum extension to this problem is to replace the probability distribution with a wave function:
\begin{equation}
p(x_1,\ldots,x_N) \rightarrow \psi^*(x_1,\ldots,x_N) \psi(x_1,\ldots,x_N)~.
\label{translate_cm2qm}
\end{equation}
The right-hand side of Eq.~(\ref{translate_cm2qm}) is just the $N$-particle density matrix ($N$DM) expressed in some basis $\ket{x_1\ldots x_N}$ as introduced by Johann Von Neumann \cite{von_neumann}: 
\begin{equation}
{}^ND= \sum_i w_i \ket{\Psi^N_i}\bra{\Psi^N_i} \qquad\text{with weights}\qquad w_i \geq 0 \qquad\text{and}\qquad \sum_i w_i = 1~.
\label{von_neumann_DM}
\end{equation}
He considers the $N$DM to be an object better suited for quantum mechanics, because it describes a system as a statistical ensemble of pure states, as opposed to the wave function framework which can only handle pure states. This is an advantage, because in reality a system is entangled with its environment, and a measurement on the system is described by a statistical ensemble of states, and not by a pure state. Because of the description of an $N$DM as a statistical ensemble of pure-state density matrices, the set of $N$DM's is convex (see Appendix \ref{math_notes}). From its definition in Eq.~(\ref{von_neumann_DM}) it is clear that an $N$DM must be Hermitian, positive semidefinite and have unit trace. 
As with classical probability distributions, we can define marginal or reduced $p$-density matrices ($p$DM), where $(N-p)$ particles are traced out:
\begin{equation}
^p_ND_{\alpha_{i_1}\ldots\alpha_{i_p};\beta_{i_1}\ldots\beta_{i_p}} = \sum_{\lambda_1\ldots\lambda_{i_1 - 1}\lambda_{i_1+1}\ldots} {}^ND_{\lambda_1\ldots\lambda_{i_1 -1}\alpha_{i_1}\lambda_{i_1 + 1}\ldots;\lambda_1\ldots\lambda_{i_1 -1}\beta_{i_1}\lambda_{i_1 + 1}\ldots}~.
\label{pDM}
\end{equation}
The quantum marginal problem can be defined as establishing the conditions that a set of $\left(\begin{matrix}N\\p\end{matrix}\right)$ $p$DM's must fulfil to be consistently derivable from an $N$DM as in Eq.~(\ref{pDM}) \cite{klyachko}. For systems of identical particles, fermions or bosons, every $p$DM derivable by (\ref{pDM}) has to be the same, because they can be mapped onto each other by a simple permutation of the indices. The $N$-representability problem is just the quantum marginal problem limited to systems of identical particles. Rewritten in the language of second quantization, a $p$DM is $N$-representable if it can be derived from an ensemble of $N$-particle wave functions:
\begin{equation}
^p_N\Gamma_{\alpha_1\ldots\alpha_p;\beta_1\ldots\beta_p} = \sum_i w_i~ \bra{\Psi^N_i}a^\dagger_{\alpha_1}\ldots a^\dagger_{\alpha_p}a_{\beta_{p}}\ldots a_{\beta_1}\ket{\Psi^N_i}~.
\label{pDM_sq}
\end{equation}
When looking at Eq.~(\ref{pDM_sq}) we see that the $p$DM is Hermitian and positive semidefinite, the trace is $\frac{N!}{(N-p)!}$, but it is not at all clear which additional conditions a physical $p$DM has to satisfy. In physical terms, the question we are asking is: what are the constraints that are put on $p$-particle properties of a quantum system when it is part of a larger system of $N$ identical particles. 
\subsection{\label{dual_N_rep}Dual definition of $N$-representability}
Using the variational principle, and the fact that the set of $N$-representable $p$DM's is convex, we can derive an alternative definition of $N$-representability, which turns out to be very useful to derive necessary conditions. Let's first define a real symmetric $p$-particle operator:
\begin{equation}
\hat{H}^{(p)} = \sum_{\alpha_1\ldots\alpha_p}\sum_{\beta_1\ldots\beta_p} H^{(p)}_{\alpha_1\ldots\alpha_p;\beta_1\ldots\beta_p}a^\dagger_{\alpha_1}\ldots a^\dagger_{\alpha_p}a_{\beta_p}\ldots a_{\beta_1}~.
\label{p-ham}
\end{equation}
The expectation value of this operator in any statistical ensemble of $N$-particle states can be expressed using only the $p$DM derived from this state:
\begin{eqnarray}
\nonumber\sum_iw_i\bra{\Psi^N_i}\hat{H}^{(p)}\ket{\Psi^N_i} &=& \sum_{\alpha_1\ldots\alpha_p}\sum_{\beta_1\ldots\beta_p}{}^p_N\Gamma_{\alpha_1\ldots\alpha_p;\beta_1\ldots\beta_p} H^{(p)}_{\alpha_1\ldots\alpha_p;\beta_1\ldots\beta_p}\\
&=& \mathrm{Tr}~\left[{}^p_N\Gamma~H^{(p)}\right]~.
\label{p-exp}
\end{eqnarray}
Being an expectation value, the trace on the right of Eq.~(\ref{p-exp}) cannot be lower than the lowest eigenvalue $E^N_0\left(H^{(p)}\right)$. A lower value therefore implies that the $^p_N\Gamma$ used cannot be physical. As such the first part of our proof can be stated: if a $p$DM is $N$-representable, then for all possible $p$-particle operators $H^{(p)}_\nu$:
\begin{equation}
\mathrm{Tr}~\left[{}^p_N\Gamma~H^{(p)}_\nu\right] \geq E^N_0\left(H^{(p)}_\nu\right)~.
\end{equation}
Suppose that we are given a $p$DM, $^p_N\Gamma^*$, which is not $N$-representable. From the separating hyperplane theorem for convex sets \cite{nagy} it follows that a such a $p$DM can always be separated from the $N$-representable set by a hyperplane. A hyperplane $\mathcal{C}$ in $p$DM space (see Appendix \ref{math_notes}) is defined by a number $x$ and a matrix $O^{(p)}$ as:
\begin{equation}
^p_N\Gamma \in \mathcal{C}\qquad \text{if} \qquad\mathrm{Tr}~\left[{}^p_N\Gamma~O^{(p)}\right] = x~.
\label{hyperplane}
\end{equation}
The fact that a non $N$-representable $p$DM is always separated from the $N$-representable set by a hyperplane means that one can always find a $p$-Hamiltonian for which:
\begin{equation}
\mathrm{Tr}~\left[{}^p_N\Gamma^*~H^{(p)}\right] < x~,
\end{equation}
whereas for all the $^p_N\Gamma$ in the $N$-representable set:
\begin{equation}
\mathrm{Tr}~\left[{}^p_N\Gamma~H^{(p)}\right] > x~.
\end{equation}
This means that:
\begin{equation}
\mathrm{Tr}~\left[{}^p_N\Gamma^*~H^{(p)}_\nu\right] < x \leq E^N_0\left(H^{(p)}\right) \leq \mathrm{Tr}~\left[{}^p_N\Gamma~H^{(p)}_\nu\right]~.
\end{equation}
This implies that for a non $N$-representable $p$DM, we can always find a $p$-Hamiltonian for which the expectation value of the energy is lower than the ground-state energy of the Hamiltonian.
The dual definition of $N$-representability can now be stated as:
\begin{theorem}
\label{theorem_dual_n_rep}
A $p$DM is $N$-representable if and only if 
\begin{equation}
\mathrm{Tr}~\left[{}^p_N\Gamma~H^{(p)}_\nu\right] \geq E^N_0\left(H^{(p)}_\nu\right)~,
\label{eq_dual_n_rep}
\end{equation}
for all $p$-particle Hamiltonians $H^{(p)}_\nu$.
\end{theorem}
In Fig. \ref{n_rep_fig} an artist impression of this theorem is shown. The image is trying to convey the idea that the convex $N$-representable set is created by an infinite number of hyperplanes, which are defined as in Eq.~(\ref{hyperplane}) by all possible $p$-Hamiltonians and their ground-state energies.
\begin{figure}
\centering
\includegraphics[scale=0.5]{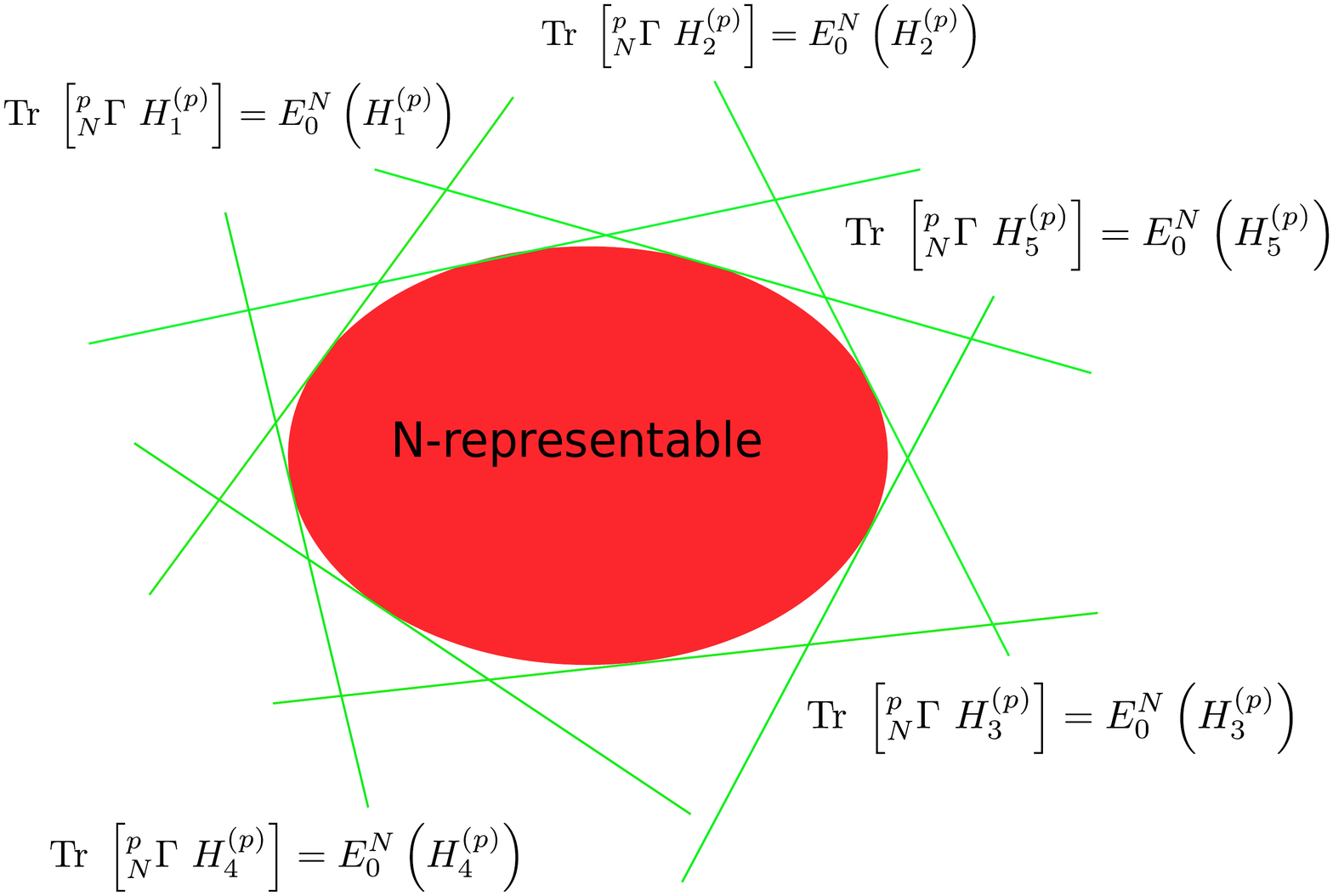}
\caption{\label{n_rep_fig} Schematic representation of Theorem \ref{theorem_dual_n_rep}, it shows the $N$-representable region demarcated by an infinite number of hyperplanes defined by $p$-Hamiltonians and their ground-state energies.}
\end{figure}

\section{\label{standard_n_rep}Standard $N$-representability conditions}
In the previous Section we discussed the origin of the $N$-representability problem, and provided two equivalent definitions. In the first, Eq.~(\ref{pDM_sq}), the $N$-particle wave function is referenced. It would of course be senseless to replace the wave function by the density matrix in a variational approach, and then reintroduce the wave function in order to impose $N$-representability. The second definition, in Eq.~(\ref{eq_dual_n_rep}), is no better in a practical sense, as it requires the knowledge of the ground-state energy of all $p$-particle Hamiltonians, which is of course exactly what we are trying to find. 

In this Section we derive some useful necessary $N$-representability constraints which are used in practical calculations. A necessary constraint is valid for all $N$-representable $p$DM's, but it could be valid for non $N$-representable $p$DM's too. {\it E.g.} all $N$-representable $p$DM's are antisymmetrical in their single-particle indices, but not all $p$-particle matrices that are antisymmetric in their single-particle indices are $N$-representable. It follows that if we optimize a $p$DM over a limited set of necessary constraints the resulting energy is a lower bound to the exact energy, because we have varied over a set that is too large, as can be seen in Figure \ref{nec_n_rep}.
\begin{figure}
\centering
\includegraphics[scale=0.5]{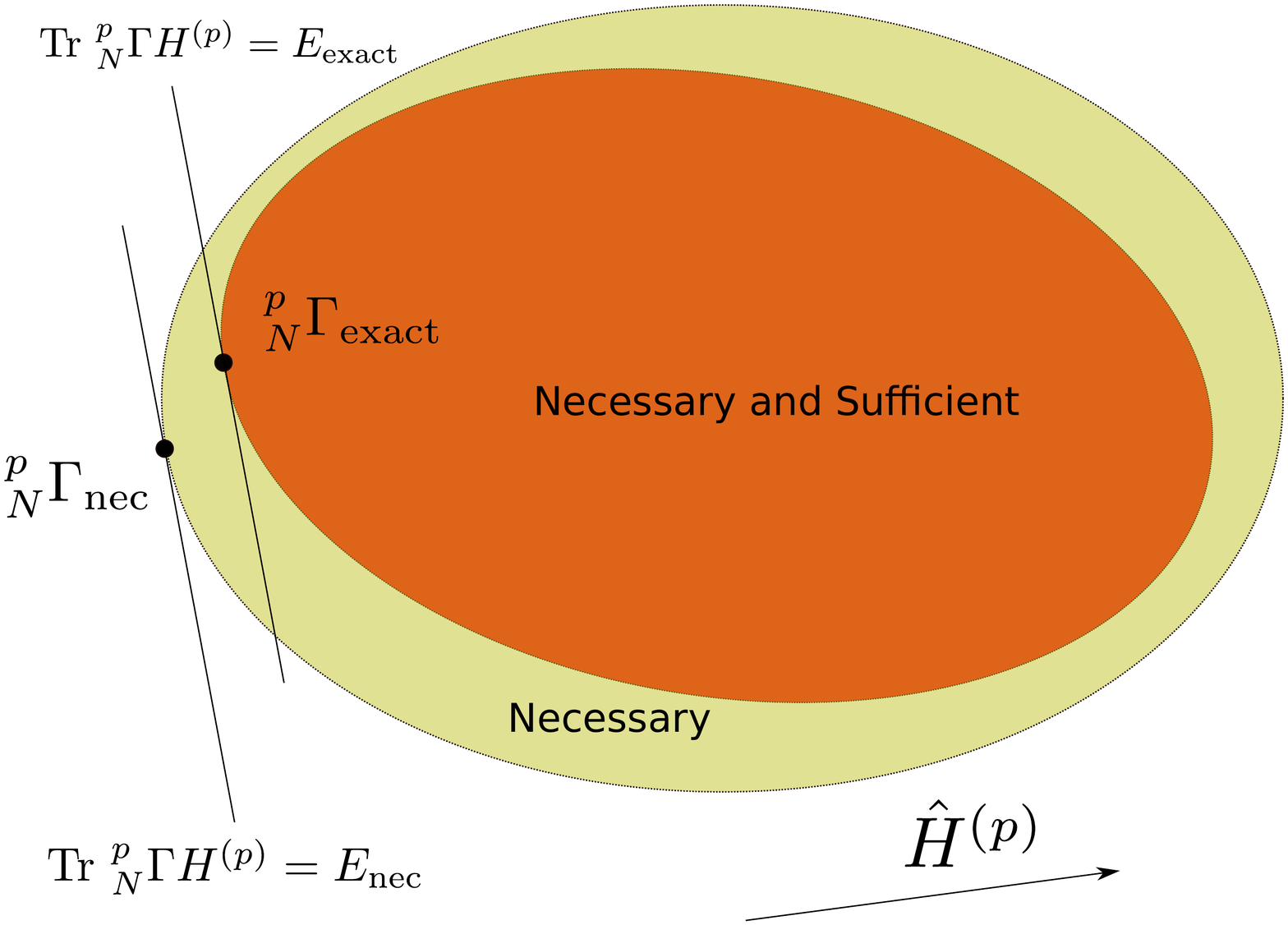}
\caption{\label{nec_n_rep} The set of $p$DM's which satisfies the necessary $N$-representability conditions is larger than the exact set, which means that the optimization of a 2DM under these constraints, for a certain Hamiltonian $\hat{H}^{(p)}$, leads to a lower bound in the energy: $E_{\text{nec}} < E_{\text{exact}}$.}
\end{figure}
\subsection{\label{n_rep_1dm}$N$-representability of the 1DM}
The first practical $N$-representability conditions were derived for the 1DM \cite{coleman}. This object is used so frequently we introduce a special symbol $\rho$ for it:
\begin{equation}
\rho_{\alpha\beta} = ~^1_N\Gamma_{\alpha\beta} = \sum_i w_i \bra{\Psi^N_i}a^\dagger_\alpha a_\beta\ket{\Psi^N_i}~.
\label{1DM}
\end{equation}
At this point it is interesting to note that $N$-representability is invariant under unitary transformations of the single-particle basis, as is seen from the definition of the reduced $p$DM in Eq.~(\ref{pDM}). This means that when we have an $N$-representable $p$DM, we can rotate the single-particle basis at will, and stay inside the $N$-representable set. A consequence for the 1DM is that the necessary and sufficient $N$-representability conditions must be expressible as a function of its eigenvalues alone!

From Eq.~(\ref{1DM}) we can derive some obvious necessary conditions on the 1DM, {\it i.e.}:
\begin{eqnarray}
\mathrm{Tr}~\rho &=& \sum_{\alpha} \rho_{\alpha\alpha} = N~,\\
\rho_{\alpha\beta} &=& \rho_{\beta\alpha}~,\\
\label{1DM_pos}\rho &\succeq& 0~.
\end{eqnarray}
A more systematic way of finding necessary conditions is through the dual definition of $N$-representability in Eq.~(\ref{eq_dual_n_rep}). The idea is to use a class of Hamiltonians for which a lower bound to the groundstate energy is known, and use this to obtain necessary constraints on the $p$DM. One such class consists of the manifestly positive Hamiltonians, for which we know the ground-state energy is larger than zero:
\begin{equation}
\hat{H}_\text{pos} = \sum_i \hat{B}^\dagger_i \hat{B}_i~,
\label{pos_ham}
\end{equation}
in which $\hat{B}^\dagger$ is a $p$-particle operator. For the 1DM there are two independent classes of Hamiltonians that can be constructed:
\begin{enumerate}
\item $\hat{B}^\dagger = \sum_\alpha p_\alpha a^\dagger_\alpha$ leads to the condition:
\begin{equation}
\sum_{\alpha\beta} p_\alpha \bra{\Psi^N}a^\dagger_\alpha a_\beta \ket{\Psi^N} p_\beta \geq 0~,
\end{equation}
which is equivalent to the already stated positivity condition on the 1DM (\ref{1DM_pos}).
\item $\hat{B}^\dagger = \sum_\alpha q_\alpha a_\alpha$ leads to the condition:
\begin{equation}
\sum_{\alpha\beta} q_\alpha \bra{\Psi^N}a_\alpha a^\dagger_\beta \ket{\Psi^N} q_\beta \geq 0~.
\label{q_index}
\end{equation}
This is a new matrix positivity condition, which can be expressed as a linear map of the 1DM through the use of anticommutation relations:
\begin{equation}
q \succeq 0 \qquad\text{where}\qquad q(\rho)_{\alpha\beta} = \delta_{\alpha\beta} - \rho_{\alpha\beta}~.
\label{q_con}
\end{equation}
The interpretation of this condition is that the probability of finding a hole in a particular single-particle state has to be positive. Equivalently, the number of particles occupying an orbital can't be larger than one.
\end{enumerate}
These two conditions combined lead to the result that the eigenvalues of a fermionic 1DM have to lie between 0 and 1, which is just an expression of the Pauli exclusion principle.

We now have derived necessary conditions for the 1DM. Are these sufficient to ensure $N$-representability? To prove this it is instructive to look at the structure of a 1DM derived from a Slater determinant. A Slater determinant has $N$ occupied single-particle orbitals, so the 1DM has $N$ eigenvalues equal to 1, and the rest is zero. An arbitrary 1DM that satifisies (\ref{1DM_pos}) and (\ref{q_con}) has eigenvalues between 0 and 1. It has been proved that one can always find an ensemble of Slater 1DM's which has the same eigenvalues and eigenvectors \cite{coleman}. From this fact it follows that the conditions are also sufficient, because every 1DM that satisfies them can be derived from an ensemble of Slater 1DM's.

Finding the ground-state energy of an arbitrary single-particle Hamiltonian $H^{(1)}_\nu$ can now be reformulated as a variational problem:
\begin{eqnarray}
\label{1DM_energy}E^N_0\left(H^{(1)}_\nu\right) &=& \min_{\rho} \mathrm{Tr}~\left[\rho H^{(1)}_\nu\right]~,\\
\text{u.c.t.}&&\left\{
   \begin{array}{l}
   \rho \succeq 0~,\\
   q(\rho) \succeq 0~,\\
   \mathrm{Tr}~\rho = N~.
   \end{array}
   \right.
   \label{1DM_constraints}
\end{eqnarray}
In the ground-state 1DM, the particles are distributed over the $N$ lowest eigenstates of the Hamiltonian $\hat{H}^{(1)}_\nu$, because of the conditions (\ref{1DM_constraints}). The energy expectation value of two-particle Hamiltonians cannot be expressed using the 1DM alone, one also needs the 2DM. We can however construct an \emph{uncorrelated} 2DM out of a 1DM:
\begin{equation}
^2_N\Gamma_{\alpha\beta;\gamma\delta} = \left(\rho\wedge\rho\right)_{\alpha\beta;\gamma\delta} = \rho_{\alpha\gamma}\rho_{\beta\delta} - \rho_{\alpha\delta}\rho_{\beta\gamma}~.
\label{uncor_2DM}
\end{equation}
Using this object in a variational optimization under the constraints (\ref{1DM_constraints}) leads to exactly the same result as with the Hartree-Fock approximation \cite{bijbel}.
\subsection{Necessary conditions for the 2DM}
In the previous Section we have seen that is relatively easy to derive the necessary and sufficient conditions for the 1DM. This should give us hope that it won't be that hard to derive conditions for the 2DM. Unfortunately it is almost impossible to derive necessary and sufficient conditions for the 2DM. In fact, the $N$-representability problem for the 2DM has been proven to belong to the Quantum Merlin-Arthur (QMA) complete complexity class \cite{qma}. This is not surprising, as we have seen in Theorem \ref{theorem_dual_n_rep} that the $N$-representability problem is equivalent to determining the ground-state energy for all possible two-particle Hamiltonians. Finding the ground-state energy for only one two-particle Hamiltonian already scales exponentially with system size, which gives an indication of the complexity of the $N$-representability problem. In this Section we show how people try to get around this and derive some frequently used necessary conditions on the 2DM. From now on we drop the indices on $\Gamma$ for the 2DM and define it as:
\begin{equation}
\Gamma_{\alpha\beta;\gamma\delta} = ~^2_N\Gamma_{\alpha\beta;\gamma\delta} = \sum_i w_i \bra{\Psi^N_i}a^\dagger_\alpha a^\dagger_\beta a_\delta a_\gamma\ket{\Psi^N_i}~.
\label{2DM}
\end{equation}
Immediately one sees some necessary conditions from Eq.~(\ref{2DM}):
\begin{eqnarray}
\mathrm{Tr}~\Gamma &=& \frac{1}{2}\sum_{\alpha\beta}\Gamma_{\alpha\beta;\alpha\beta} = \frac{N(N - 1)}{2}~,\\
\Gamma_{\alpha\beta;\gamma\delta} &=& \Gamma_{\gamma\delta;\alpha\beta}~,\\
\Gamma_{\alpha\beta;\gamma\delta} &=& -\Gamma_{\beta\alpha;\gamma\delta} = -\Gamma_{\alpha\beta;\delta\gamma} = \Gamma_{\beta\alpha;\delta\gamma}~,\\
\Gamma &\succeq& 0~.
\end{eqnarray}
These are, however, not at all sufficient to give a good approximation to the ground-state energy. For deriving further necessary conditions we use the same method as for the 1DM, using positive Hamiltonians with the structure of (\ref{pos_ham}).
\subsubsection{Two-index conditions}
Combining two creation and/or annihilation operators yields four different forms of the $\hat{B}^\dagger$ operator.
\paragraph{The $\mathcal{I}$ condition:}
the Hamiltonian constructed from $\hat{B}^\dagger = \sum_{\alpha\beta} p_{\alpha\beta} a^\dagger_\alpha a^\dagger_\beta$ has to have an expectation value larger than zero:
\begin{equation}
\sum_{\alpha\beta\gamma\delta} p_{\alpha\beta}\bra{\Psi^N}a^\dagger_\alpha a^\dagger_\beta a_\delta a_\gamma\ket{\Psi^N}p_{\gamma\delta} \geq 0~.
\end{equation}
This is just an expression of the fact that $\Gamma$ has to be positive semidefinite, which we already knew. Physically speaking, it is the expression of the fact that the probability of finding a two-particle pair is larger than zero. We refer to this condition as the $\mathcal{I}$ condition, for the identity condition.
\paragraph{The $\mathcal{Q}$ condition:}
in exactly the same way $\hat{B}^\dagger = \sum_{\alpha\beta}q_{\alpha\beta}a_\alpha a_\beta$ leads to:
\begin{equation}
\sum_{\alpha\beta\gamma\delta} q_{\alpha\beta}\bra{\Psi^N}a_\alpha a_\beta a^\dagger_\delta a^\dagger_\gamma\ket{\Psi^N}q_{\gamma\delta} \geq 0~,
\end{equation}
expressing the positive semidefiniteness of the two-hole matrix, {\it i.e.} the probability of finding a two-hole pair has to be larger than zero:
\begin{equation}
\mathcal{Q} \succeq  0\qquad\text{where}\qquad \mathcal{Q}_{\alpha\beta;\gamma\delta} = \sum_i w_i \bra{\Psi^N_i}a_\alpha a_\beta a^\dagger_\delta a^\dagger_\gamma\ket{\Psi^N_i}~,
\end{equation}
as first derived by C. Garrod and J. K. Percus in \cite{garrod}.
This matrix can be rewritten as a function of the 2DM and the 1DM by reordening the creation and annihilation operators using anticommutation relations, leading to the following linear matrix map:
\begin{equation}
\mathcal{Q}(\Gamma)_{\alpha\beta;\gamma\delta} = \delta_{\alpha\gamma}\delta_{\beta\delta} - \delta_{\beta\gamma}\delta_{\alpha\delta} +\Gamma_{\alpha\beta;\gamma\delta} - \left(\delta_{\alpha\gamma}\rho_{\beta\delta} - \delta_{\alpha\delta}\rho_{\beta\gamma} - \delta_{\beta\gamma}\rho_{\alpha\delta} + \delta_{\beta\delta}\rho_{\alpha\gamma}\right)~.
\label{Q_2DM}
\end{equation}
One can see that the combination of the $\mathcal{I}$ and $\mathcal{Q}$ conditions already insures that the Pauli principle is obeyed, as the $\rho$ and $q$ condition can be derived from them by tracing one index:
\begin{eqnarray}
\rho_{\alpha\gamma} &=& \frac{1}{N - 1}\sum_{\beta}\Gamma_{\alpha\beta;\gamma\beta}~,\\
q_{\alpha\gamma} &=& \frac{1}{M - N - 1}\sum_{\beta}\mathcal{Q}(\Gamma)_{\alpha\beta;\gamma\beta}~,
\end{eqnarray}
in which $M$ is the dimension of single-particle space. It is seen that $\rho$ and $q$ are automatically positive semidefinite if $\mathcal{I}$ and $\mathcal{Q}$ are. This, however, does not mean that enforcing the $\mathcal{I}$ and $\mathcal{Q}$ conditions leads us to a better result than Hartree-Fock. The set over which we are varying is too large, so we get a lower bound to the ground-state energy. Applying only $\mathcal{I}$ and $\mathcal{Q}$ can lead to considerable overestimation of the correlation energy, defined as the difference between the exact and the Hartree-Fock energy. 

Two more constraints can be defined in the same way, using $aa^\dagger$ or $a^\dagger a$ operators. They were first derived in \cite{garrod}, but in a different form (see Section~\ref{primes}).
\paragraph{The $\mathcal{G}_1$ condition:}
the particle-hole operator $\hat{B}^\dagger = \sum_{\alpha\beta}g^1_{\alpha\beta}a^\dagger_{\alpha}a_\beta$ leads to the following positivity expression:
\begin{equation}
\sum_{\alpha\beta\gamma\delta}g^1_{\alpha\beta}\bra{\Psi^N}a^\dagger_\alpha a_\beta a^\dagger_\delta a_\gamma \ket{\Psi^N}g^1_{\gamma\delta} \geq 0~,
\end{equation}
which is a mathematical translation of the fact that the probability of finding a particle-hole pair must be larger than zero. We call this condition the $\mathcal{G}_1$ condition:
\begin{equation}
\mathcal{G}_1 \succeq 0 \qquad\text{where}\qquad \left(\mathcal{G}_1\right)_{\alpha\beta;\gamma\delta} = \sum_iw_i\bra{\Psi^N_i}a^\dagger_\alpha a_\beta a^\dagger_\delta a_\gamma \ket{\Psi^N_i}~,
\label{G_con}
\end{equation}
which can again be rewritten as a matrix map of the 2DM using anticommutation relations:
\begin{equation}
\mathcal{G}_1\left(\Gamma\right)_{\alpha\beta;\gamma\delta} = \delta_{\beta\delta}\rho_{\alpha\gamma} - \Gamma_{\alpha\delta;\gamma\beta}~.
\label{G1_2DM}
\end{equation}
Note that there is no symmetry between the single-particle orbitals here. This condition turns out to be much more stringent than the $\mathcal{I}$ or $\mathcal{Q}$ conditions. The combined conditions $\mathcal{IQG}_1$ are known as the standard two-index conditions. Applying these can already lead to very good approximations for some systems, {\it e.g.} the Beryllium atom. For other systems, however, additional constraints are needed to obtain decent results.
\paragraph{The $\mathcal{G}_2$ condition:} the only remaining combination of creation an annihilation operators is $\hat{B}^\dagger = \sum_{\alpha\beta}g^2_{\alpha\beta}a_\alpha a^\dagger_\beta$, leading to:
\begin{equation}
\sum_{\alpha\beta\gamma\delta}g^2_{\alpha\beta}\bra{\Psi^N}a_\alpha a^\dagger_\beta a_\delta a^\dagger_\gamma \ket{\Psi^N}g^2_{\gamma\delta} \geq 0~,
\end{equation}
{\it i.e.} the probability of finding a hole-particle pair has to be positive. We call it the $\mathcal{G}_2$ condition:
\begin{equation}
\mathcal{G}_2 \succeq 0 \qquad\text{where}\qquad \left(\mathcal{G}_2\right)_{\alpha\beta;\gamma\delta} = \sum_iw_i\bra{\Psi^N_i}a_\alpha a^\dagger_\beta a_\delta a^\dagger_\gamma \ket{\Psi^N_i}~.
\end{equation}
Expressed as a function of the 2DM this becomes:
\begin{equation}
\mathcal{G}_2\left(\Gamma\right)_{\alpha\beta;\gamma\delta} = \delta_{\alpha\beta}\delta_{\gamma\delta} - \delta_{\alpha\beta}\rho_{\gamma\delta} - \delta_{\gamma\delta}\rho_{\alpha\beta} + \delta_{\alpha\gamma}\rho_{\beta\delta} - \Gamma_{\alpha\delta;\gamma\beta}~.
\label{G2_2DM}
\end{equation}
As it turns out this is not an independent condition, as is shown in Section~\ref{primes}.
\subsubsection{\label{three_index}Three-index conditions}
As mentioned in the previous paragraph, the two-index conditions sometimes give a very good result, but often stronger conditions are needed describe systems correctly. One example of stronger conditions are the so-called three-index conditions. They are also derived from positive Hamiltonians of the type (\ref{pos_ham}), but slightly different from those used in the previous paragraph. For these conditions, use is made of the fact that taking the anticommutator of fermionic three-particle operators lowers the rank by one, and the expectation value is therefore expressible as a function of the 2DM alone. The general form of the positive Hamiltonian used here is:
\begin{equation}
\hat{H} = \hat{B}^\dagger \hat{B} + \hat{B}\hat{B}^\dagger~,
\label{pos_ham_3i}
\end{equation}
in which the $\hat{B}^\dagger$'s are combinations of three creation and/or annihilation operators. There are three independent conditions that can be derived in this way. These conditions were first alluded to in the work of R.M. Erdahl \cite{erdahl}, and explicitly used in \cite{zhao,hammond}.
\paragraph{The $\mathcal{T}_1$ condition}
If $\hat{B}^\dagger = \sum_{\alpha\beta\gamma}t^1_{\alpha\beta\gamma}a^\dagger_\alpha a^\dagger_\beta a^\dagger_\gamma$ the positivity of the Hamiltonian leads to the following expression:
\begin{equation}
\sum_{\alpha\beta\gamma}\sum_{\delta\epsilon\zeta}t^1_{\alpha\beta\gamma}\left(\bra{\Psi^N}a^\dagger_\alpha a^\dagger_\beta a^\dagger_\gamma a_\zeta a_\epsilon a_\delta \ket{\Psi^N} + \bra{\Psi^N} a_\zeta a_\epsilon a_\delta a^\dagger_\alpha a^\dagger_\beta a^\dagger_\gamma\ket{\Psi^N}\right)t^1_{\delta\epsilon\zeta}\geq 0~.
\end{equation}
This expresses the positivity of the $\mathcal{T}_1$ matrix defined as:
\begin{equation}
\mathcal{T}_1(\Gamma)\succeq 0
\end{equation}
with
\begin{equation}
\left(\mathcal{T}_1\right)_{\alpha\beta\gamma;\delta\epsilon\zeta} = \sum_i w_i\left(\bra{\Psi^N_i}a^\dagger_\alpha a^\dagger_\beta a^\dagger_\gamma a_\zeta a_\epsilon a_\delta \ket{\Psi^N_i} + \bra{\Psi^N_i} a_\zeta a_\epsilon a_\delta a^\dagger_\alpha a^\dagger_\beta a^\dagger_\gamma\ket{\Psi^N_i}\right)~.
\label{T1_con}
\end{equation}
Due to the anticommutation of the $\hat{B}^\dagger$ and $\hat{B}$ this expression can still be written as a function of the 2DM when we anticommute the creation and annihilation operators:
\begin{eqnarray}
\nonumber\mathcal{T}_1\left(\Gamma\right)_{\alpha\beta\gamma;\delta\epsilon\zeta} &=&\delta_{\gamma\zeta}\delta_{\beta\epsilon}\delta_{\alpha\delta} - \delta_{\gamma\epsilon}\delta_{\alpha\delta}\delta_{\beta\zeta} + \delta_{\alpha\zeta}\delta_{\gamma\epsilon}\delta_{\beta\delta} - \delta_{\gamma\zeta}\delta_{\alpha\epsilon}\delta_{\beta\delta} + \delta_{\beta\zeta}\delta_{\alpha\epsilon}\delta_{\gamma\delta} -\delta_{\alpha\zeta}\delta_{\beta\epsilon}\delta_{\gamma\delta}\\
\nonumber&& -\left(\delta_{\gamma\zeta}\delta_{\beta\epsilon} - \delta_{\beta\zeta}\delta_{\gamma\epsilon}\right)\rho_{\alpha\delta} + \left(\delta_{\gamma\zeta}\delta_{\alpha\epsilon} - \delta_{\alpha\zeta}\delta_{\gamma\epsilon}\right)\rho_{\beta\delta} - \left(\delta_{\beta\zeta}\delta_{\alpha\epsilon}- \delta_{\alpha\zeta}\delta_{\beta\epsilon}\right)\rho_{\gamma\delta}\\
\nonumber&& + \left(\delta_{\gamma\zeta}\delta_{\beta\delta} - \delta_{\beta\zeta}\delta_{\gamma\delta}\right)\rho_{\alpha\epsilon} - \left(\delta_{\gamma\zeta}\delta_{\alpha\delta} - \delta_{\alpha\zeta}\delta_{\gamma\delta}\right)\rho_{\epsilon\beta} + \left(\delta_{\beta\zeta}\delta_{\alpha\delta} - \delta_{\alpha\zeta}\delta_{\beta\delta}\right)\rho_{\gamma\epsilon}\\
\nonumber&& - \left(\delta_{\beta\delta}\delta_{\gamma\epsilon} - \delta_{\beta\epsilon}\delta_{\gamma\delta}\right)\rho_{\alpha\zeta} + \left(\delta_{\gamma\epsilon}\delta_{\alpha\delta} - \delta_{\alpha\epsilon}\delta_{\gamma\delta}\right)\rho_{\beta\zeta} - \left(\delta_{\beta\epsilon}\delta_{\alpha\delta} - \delta_{\alpha\epsilon}\delta_{\beta\delta}\right)\rho_{\gamma\zeta}\\
\nonumber&&+ \delta_{\gamma\zeta}\Gamma_{\alpha\beta;\delta\epsilon} - \delta_{\beta\zeta}\Gamma_{\alpha\gamma;\delta\epsilon} + \delta_{\alpha\zeta}\Gamma_{\beta\gamma;\delta\epsilon} - \delta_{\gamma\epsilon}\Gamma_{\alpha\beta;\delta\zeta} + \delta_{\beta\epsilon}\Gamma_{\alpha\gamma;\delta\zeta} - \delta_{\alpha\epsilon}\Gamma_{\beta\gamma;\delta\zeta}\\
&& + \delta_{\gamma\delta}\Gamma_{\alpha\beta;\epsilon\zeta} - \delta_{\beta\delta}\Gamma_{\alpha\gamma;\epsilon\zeta} + \delta_{\alpha\delta}\Gamma_{\beta\gamma;\epsilon\zeta}~.
\label{T1}
\end{eqnarray}
Note that this is a matrix on three-particle space, and as such will be computationally much heavier to use. Using this condition does not make the $\mathcal{I}$ or $\mathcal{Q}$ condition redundant, since $\mathcal{T}_1$ implies positiveness of the sum $BB^\dagger + B^\dagger B$ and not of the individual terms.
\paragraph{The $\mathcal{T}_2$ condition}
The next three-index condition is derived trough the positivity of (\ref{pos_ham_3i}) using $\hat{B}^\dagger = \sum_{\alpha\beta\gamma}t^2_{\alpha\beta\gamma}a^\dagger_\alpha a^\dagger_\beta a_\gamma$, producing the inequality:
\begin{equation}
\sum_{\alpha\beta\gamma}\sum_{\delta\epsilon\zeta}t^2_{\alpha\beta\gamma}\left(\bra{\Psi^N}a^\dagger_\alpha a^\dagger_\beta a_\gamma a^\dagger_\zeta a_\epsilon a_\delta \ket{\Psi^N} + \bra{\Psi^N} a^\dagger_\zeta a_\epsilon a_\delta a^\dagger_\alpha a^\dagger_\beta a_\gamma\ket{\Psi^N}\right)t^2_{\delta\epsilon\zeta}\geq 0~.
\end{equation}
This expresses the positivity of the sum of the particle-particle-hole and the hole-hole-particle matrix $\mathcal{T}_2$:
\begin{equation}
\left(\mathcal{T}_2\right)_{\alpha\beta\gamma;\delta\epsilon\zeta} = \sum_i w_i \left(\bra{\Psi^N_i}a^\dagger_\alpha a^\dagger_\beta a_\gamma a^\dagger_\zeta a_\epsilon a_\delta \ket{\Psi^N_i} + \bra{\Psi^N_i} a^\dagger_\zeta a_\epsilon a_\delta a^\dagger_\alpha a^\dagger_\beta a_\gamma\ket{\Psi^N_i}\right)~,
\label{T2_con}
\end{equation}
which, following the standard recipe, becomes an expression of the 2DM alone:
\begin{eqnarray}
\nonumber\mathcal{T}_2(\Gamma)_{\alpha\beta\gamma;\delta\epsilon\zeta} &=& \left(\delta_{\alpha\delta}\delta_{\beta\epsilon} - \delta_{\alpha\epsilon}\delta_{\beta\delta}\right)\rho_{\gamma\zeta} + \delta_{\gamma\zeta}\Gamma_{\alpha\beta;\delta\epsilon}\\
\label{T2}&&- \delta_{\alpha\delta}\Gamma_{\gamma\epsilon;\zeta\beta} + \delta_{\beta\delta}\Gamma_{\gamma\epsilon;\zeta\alpha} + \delta_{\alpha\epsilon}\Gamma_{\gamma\delta;\zeta\beta} - \delta_{\beta\epsilon}\Gamma_{\gamma\delta;\zeta\alpha}~.
\end{eqnarray}
This is a more compact expression than the $\mathcal{T}_1$, but it turns out to be a stronger condition, much like the $\mathcal{G}$ condition was more stringent than $\mathcal{I}$ and $\mathcal{Q}$.
\paragraph{The $\mathcal{T}_3$ condition}
The last independent combination of creation/annihilation operators is $\hat{B}^\dagger = \sum_{\alpha\beta\gamma}t^3_{\alpha\beta\gamma}a^\dagger_\alpha a_\beta a^\dagger_\gamma$, which gives the following positivity expression:
\begin{equation}
\sum_{\alpha\beta\gamma}\sum_{\delta\epsilon\zeta}t^3_{\alpha\beta\gamma}\left(\bra{\Psi^N}a^\dagger_\alpha a_\beta a^\dagger_\gamma a_\zeta a^\dagger_\epsilon a_\delta \ket{\Psi^N} + \bra{\Psi^N} a_\zeta a^\dagger_\epsilon a_\delta a^\dagger_\alpha a_\beta a^\dagger_\gamma\ket{\Psi^N}\right)t^3_{\delta\epsilon\zeta}\geq 0~.
\end{equation}
The above equation expresses the positivity of the sum of the particle-hole-particle and the hole-particle-hole matrix $\mathcal{T}_3$:
\begin{equation}
\left(\mathcal{T}_3\right)_{\alpha\beta\gamma;\delta\epsilon\zeta} = \sum_i w_i \left(\bra{\Psi^N_i}a^\dagger_\alpha a_\beta a^\dagger_\gamma a_\zeta a^\dagger_\epsilon a_\delta \ket{\Psi^N_i} + \bra{\Psi^N_i} a_\zeta a^\dagger_\epsilon a_\delta a^\dagger_\alpha a_\beta a^\dagger_\gamma\ket{\Psi^N_i}\right)~.
\label{T3_con}
\end{equation}
Once more, the anticommutation of the creation and annihilation operators enables us to write this as a function of the 2DM alone:
\begin{eqnarray}
\nonumber\mathcal{T}_3(\Gamma)_{\alpha\beta\gamma;\delta\epsilon\zeta} &=& \delta_{\alpha\delta}\delta_{\beta\gamma}\delta_{\epsilon\zeta}-\delta_{\alpha\delta} \Gamma_{\epsilon\gamma;\beta\zeta} - \delta_{\alpha\zeta}\Gamma_{\epsilon\gamma;\delta\beta} - \delta_{\gamma\delta}\Gamma_{\alpha\epsilon;\beta\zeta} + \delta_{\beta\epsilon}\Gamma_{\alpha\gamma;\delta\zeta}\\
&&\nonumber- \delta_{\gamma\zeta}\Gamma_{\alpha\epsilon;\delta\beta} -\delta_{\alpha\delta}\delta_{\epsilon\zeta}\rho_{\beta\gamma} - \delta_{\alpha\delta}\delta_{\gamma\beta}\rho_{\epsilon\zeta} + \delta_{\alpha\zeta}\delta_{\gamma\beta}\rho_{\epsilon\delta} + \delta_{\gamma\delta}\delta_{\epsilon\zeta} \rho_{\alpha\beta}\\
\label{T3}&& + (\delta_{\alpha\delta}\delta_{\gamma\zeta} - \delta_{\alpha\zeta}\delta_{\gamma\delta})\rho_{\beta\epsilon}~.
\end{eqnarray}
In the next Section, however, we show that this condition becomes redundant the $\mathcal{T}_2$ condition is generalized.
\subsubsection{\label{primes}The primed conditions}
Up to now, we have derived necessary constraints using positive Hamiltonians, in which the positive Hamiltonians were built with some combination of creation and annihilation operators. In some cases it is possible to derive more stringent conditions, called \emph{primed} conditions, involving linear combinations of such operators. The first condition of this type to be derived was by Garrod and Percus \cite{garrod}, in what was also the first derivation of a $\mathcal{G}$-like condition.
\paragraph{The $\mathcal{G}'$ condition}
The idea of primed conditions is to generalize the operator $\hat{B}^\dagger$. For the $\mathcal{G}'$ condition it becomes:
\begin{equation}
\hat{B}^\dagger = \sum_{\alpha\beta}g_{\alpha\beta}'a^\dagger_\alpha a_\beta + C~,
\label{g_prime_op}
\end{equation}
in which $C$ is a constant operator. The positivity of the Hamiltonian constructed with this $\hat{B}^\dagger$ leads to the following expression:
\begin{eqnarray}
\nonumber\sum_i w_i \bra{\Psi^N_i}\hat{B}^\dagger \hat{B} \ket{\Psi^N_i} &=& \sum_{\alpha\beta\gamma\delta}g_{\alpha\beta}'~\left(\mathcal{G}_1\right)_{\alpha\beta;\gamma\delta}~g_{\gamma\delta}' + C \sum_{\alpha\beta}g_{\alpha\beta}'~\rho_{\alpha\beta} \\
&&\qquad + C \sum_{\gamma\delta}g_{\gamma\delta}'~\rho_{\gamma\delta} + C^2 \geq 0~,
\end{eqnarray}
which can be rewritten as:
\begin{equation}
\sum_{\alpha\beta\gamma\delta}g_{\alpha\beta}'~\left(\mathcal{G}_1\right)_{\alpha\beta;\gamma\delta}~g_{\gamma\delta}' + \left(C + \sum_{\alpha\beta}g_{\alpha\beta}'~\rho_{\alpha\beta}\right)^2 - \sum_{\alpha\beta\gamma\delta}g_{\alpha\beta}'~\rho_{\alpha\beta}\rho_{\gamma\delta}~g_{\gamma\delta}' \geq 0~.
\label{deriv_gpc}
\end{equation}
It is clear that we get the most stringent condition if we choose:
\begin{equation}
C = - \sum_{\alpha\beta}g_{\alpha\beta}'~\rho_{\alpha\beta}~,
\end{equation}
since this removes a strictly positive contribution to Eq.~(\ref{deriv_gpc}).
The remaining terms in Eq.~(\ref{deriv_gpc}) then express the positive semidefiniteness of the matrix:
\begin{equation}
\mathcal{G}'\left(\Gamma\right)_{\alpha\beta;\gamma\delta} = \delta_{\beta\delta}\rho_{\alpha\gamma} - \Gamma_{\alpha\delta;\gamma\beta} - \rho_{\alpha\beta}\rho_{\gamma\delta}~.
\label{G_prime}
\end{equation}
From the definition (\ref{g_prime_op}) it should be clear that when the $\mathcal{G}'$ is fulfilled, both $\mathcal{G}_1$ and $\mathcal{G}_2$ will be fulfilled too, since these are just special cases of $\mathcal{G}'$ with a specific value for the number $C$. The $\mathcal{G}_1$ condition is $\mathcal{G}'$ when $C = 0$. For the $\mathcal{G}_2$ condition it can be seen through anticommutation that:
\begin{eqnarray}
\nonumber \hat{B}^\dagger &=& \sum_{\alpha\beta}g^2_{\alpha\beta}a_\alpha a^\dagger_\beta\\
\nonumber&=& \sum_{\alpha\beta}g^2_{\alpha\beta}\delta_{\alpha\beta} - \sum_{\alpha\beta}g^2_{\alpha\beta}a^\dagger_\alpha a_\beta\\
&=& \sum_{\alpha\beta}g'_{\alpha\beta}~a^\dagger_\alpha a_\beta + C
\end{eqnarray}
with
\begin{equation}
g'_{\alpha\beta} = -g^2_{\alpha\beta} \qquad\text{and}\qquad C = \sum_{\alpha}g^2_{\alpha\alpha}~.
\end{equation}
Unfortunately $\mathcal{G'}$ is not linear in $\Gamma$ which hinders its use in a standard semidefinite program. However, one can show that the domain in 2DM space for which $\mathcal{G}'$ is positive is exactly the same as that where $\mathcal{G}_1$ is positive, {\it i.e.}
\begin{theorem}
\label{G_prime_eq_G}
The nullspaces of the $\mathcal{G}'$ and $\mathcal{G}_1$ matrix map coincide.
\end{theorem}
\begin{proof}
In the first part of the proof we show that when, for some $\Gamma$, $\mathcal{G}_1(\Gamma)$ has a zero eigenvalue, then $\mathcal{G}'(\Gamma)$ has a zero eigenvalue too. Suppose $\mathcal{G}_1(\Gamma)$ has a zero eigenvalue with eigenvector $v$:
\begin{equation}
\left(\mathcal{G}_1(\Gamma)v\right)_{\alpha\beta} = \sum_{\gamma\delta}\left(\delta_{\beta\delta}\rho_{\alpha\gamma}-\Gamma_{\alpha\delta;\gamma\beta}\right)~v_{\gamma\delta} = \sum_\gamma \rho_{\alpha\gamma}v_{\gamma\beta} - \sum_{\gamma\delta}\Gamma_{\alpha\delta;\gamma\beta}~v_{\gamma\delta} = 0~.
\label{proof_2_eq_1}
\end{equation}
If we trace Eq.~(\ref{proof_2_eq_1}), we get
\begin{equation}
N\sum_{\alpha\gamma}\rho_{\alpha\gamma}v_{\gamma\alpha} = 0~.
\end{equation}
This means that $\mathcal{G}'(\Gamma)$ also has a zero eigenvalue with the same eigenvector:
\begin{equation}
\left(\mathcal{G}'(\Gamma)v\right)_{\alpha\beta} = \sum_{\gamma\delta} \mathcal{G}_1(\Gamma)_{\alpha\beta;\gamma\delta}~v_{\gamma\delta} - \rho_{\alpha\beta}\sum_{\gamma\delta}\rho_{\gamma\delta}~v_{\gamma\delta} = 0~.
\end{equation}
In the second part we prove the inverse. Suppose $\mathcal{G}'(\Gamma)$ has a zero eigenvalue with eigenvector $v$, then $\mathcal{G}_1(\Gamma)$ also has a zero eigenvalue, with eigenvector $v'$:
\begin{equation}
v'_{\alpha\beta} = v_{\alpha\beta} - \left(\frac{\mathrm{Tr}~\rho v}{N}\right)\delta_{\alpha\beta}~.
\end{equation}
First we note that $v'$ is also an eigenvector of $\mathcal{G}'$ because:
\begin{equation}
\sum_{\gamma}\mathcal{G}'(\Gamma)_{\alpha\beta;\gamma\gamma} = 0~.
\end{equation}
Next, we multiply $\mathcal{G}_1$ with $v'$:
\begin{eqnarray}
\nonumber\sum_{\gamma\delta}\mathcal{G}_1(\Gamma)_{\alpha\beta;\gamma\delta}~v'_{\gamma\delta} &=& \sum_{\gamma\delta}\left(\mathcal{G'}(\Gamma)_{\alpha\beta\gamma\delta} + \rho_{\alpha\beta}\rho_{\gamma\delta}\right)v'_{\gamma\delta}\\
&=& \sum_{\gamma\delta}\rho_{\alpha\beta}\rho_{\gamma\delta}\left[v_{\gamma\delta} - \left(\frac{\mathrm{Tr}~\rho v}{N}\right)\delta_{\gamma\delta}\right]=0~.
\end{eqnarray}
\end{proof}
From Theorem \ref{G_prime_eq_G} we learn that using $\mathcal{G}'$ or $\mathcal{G}_1$ during the optimization yields the same results, since the edges of the convex sets for which $\mathcal{G}_1 \succeq 0$ and $\mathcal{G}'\succeq 0$ are the same! The $\mathcal{G}_1$ condition is to be preferred, because it is linear in $\Gamma$ and therefore easier to use in semidefinite programs. At the same time one sees that $\mathcal{G}_2$ is redundant, because it is included by $\mathcal{G}'$, and following from the above theorem, also by $\mathcal{G}_1$. So from now on we refer to $\mathcal{G}_1$ as the $\mathcal{G}$ condition and never use $\mathcal{G}_2$ again.  
\paragraph{The $\mathcal{T}_2'$ condition}
A primed condition can also be derived for the three-index case, through a generalization of the $\hat{B}^\dagger$ for the $\mathcal{T}_2$ condition,(see \cite{braams_book,maz_T2_prime}). The positive Hamiltonian has the form:
\begin{equation}
\hat{H} = \hat{\tilde{B}}^\dagger \hat{\tilde{B}} + \hat{B}\hat{B}^\dagger~,
\end{equation}
in which $\hat{\tilde{B}}^\dagger$ is the generalized $\hat{B}^\dagger$ of the $\mathcal{T}_2$:
\begin{equation}
\hat{\tilde{B}}^\dagger = \sum_{\alpha\beta\gamma}t^2_{\alpha\beta\gamma}a^\dagger_\alpha a^\dagger_\beta a_\gamma + \sum_\mu x_\mu a^\dagger_\mu~.
\end{equation}
To it is added the regular second part of the $\mathcal{T}_2$ Hamiltonian, to cancel out the three-particle terms appearing in the first part. The expression of positivity then leads to:
\begin{eqnarray}
\nonumber\sum_i w_i \bra{\Psi^N_i}\hat{H}\ket{\Psi^N_i} &=&\sum_{\alpha\beta\gamma} \sum_{\delta\epsilon\zeta} t^2_{\alpha\beta\gamma}\mathcal{T}_2(\Gamma)_{\alpha\beta\gamma;\delta\epsilon\zeta}~t^2_{\delta\epsilon\zeta} + \sum_{\alpha\beta\gamma\nu} t^2_{\alpha\beta\gamma}\omega_{\alpha\beta\gamma;\nu}x_\nu\\
&&\qquad\qquad + \sum_{\mu\delta\epsilon\zeta}x_\mu \omega^\dagger_{\mu;\delta\epsilon\zeta}t^2_{\delta\epsilon\zeta} + \sum_{\mu\nu}x_\mu\rho_{\mu\nu}x_\nu\geq 0~,
\label{positivity_T2P}
\end{eqnarray}
where
\begin{equation}
\omega_{\alpha\beta\gamma;\nu} = \sum_i w_i \bra{\Psi^N_i}a^\dagger_\alpha a^\dagger_\beta a_\gamma a_\nu\ket{\Psi^N_i} = \Gamma_{\alpha\beta;\nu\gamma}~.
\end{equation}
The positivity expressed in Eq.~(\ref{positivity_T2P}) can be translated to a matrix positivity condition:
\begin{equation}
\mathcal{T}_2'
=
\left(
\begin{matrix}
\left(\mathcal{T}_2\right)_{\alpha\beta\gamma;\delta\epsilon\zeta} & \omega_{\alpha\beta\gamma;\nu}\\
\omega^\dagger_{\mu;\delta\epsilon\zeta} & \rho_{\mu\nu}
\end{matrix}
\right)
\succeq 0~.
\label{T2P}
\end{equation}
One can see from Eq.~(\ref{T2P}) that this matrix has a slightly larger dimension than the regular $\mathcal{T}_2$ condition, but this is negligible compared to the size of the $\mathcal{T}_2$ matrix. It is also obvious that $\mathcal{T}_2'$ includes $\mathcal{T}_2$, as diagonal blocks of a positive definite matrix are also positive. What is more, the $\mathcal{T}_3$ can be shown to be a part of the $\mathcal{T}_2'$ class too, since the $\hat{B}^\dagger$ that creates the $\mathcal{T}_3$ condition is a special case of the general $\mathcal{T}_2'$ $\hat{\tilde{B}}^\dagger$:
\begin{eqnarray}
\nonumber \hat{B}^\dagger &=& \sum_{\alpha\beta\gamma}t^3_{\alpha\beta\gamma}a^\dagger_\alpha a_\beta a^\dagger_\gamma\\
&=& \sum_{\alpha\beta\gamma}\delta_{\beta\gamma}t^3_{\alpha\beta\gamma}a^\dagger_\alpha - \sum_{\alpha\beta\gamma}t^3_{\alpha\beta\gamma}a^\dagger_{\alpha}a^\dagger_\gamma a_{\beta}\\
&=& \sum_{\alpha\beta\gamma}t^2_{\alpha\beta\gamma}a^\dagger_\alpha a^\dagger_\beta a_\gamma + \sum_\mu x_\mu a^\dagger_\mu~,
\end{eqnarray}
in which
\begin{equation}
t^2_{\alpha\beta\gamma} = -t^3_{\alpha\gamma\beta} \qquad\text{and}\qquad x_\mu = \sum_{\beta}t^3_{\mu\beta\beta}~.
\end{equation}
As reported in \cite{braams_book,maz_T2_prime}, slightly better results are obtained when adding the $\mathcal{T}_2'$ condition compared to only applying the $\mathcal{T}_1$ and $\mathcal{T}_2$ conditions.
\section{Non-standard $N$-representability conditions}
Up to now we have focussed on constraints which can be expressed as matrix positivity constraints of linear maps of the 2DM. These constraints are commonly used in 2DM optimization. In some situations the two-index conditions suffice to obtain a good approximation of the ground-state properties, but often the results are not good enough. In a standard approach one would use the three-index conditions to increase accuracy, but this limits the size of the systems that can be studied as the computational cost increases dramatically. In this Section we introduce some non-standard $N$-representability conditions that can be used to increase accuracy without making the optimization computationally unfeasible.
\subsection{\label{subsystem}Subsystem constraints}
The first example of non-standard constraints we introduce are the subsystem constraints \cite{qsep}. These constraints were developed for the special case of diatomic dissociation (See section \ref{diatomic}), but are more generally applicable. In this Section we derive these constraints without a specific application in mind. The idea of subsystem constraints is to impose constraints on the 2DM restricted to some subspace of single-particle Hilbert space. To derive them however we first need to introduce the concept of $N$-representability for systems with a fractional number of particles.
\subsubsection{Fractional $N$-representability}
Until now we have only discussed systems with an integer number of particles. When we talk about a fractional particle number we don't mean that the particles are actually split up, but that the expectation value of the `number of particles'-operator is fractional. This means that the system is open, it can exchange particles with its environment, so the number of particles in the system fluctuates. The natural way to deal with these kinds of systems is by taking an ensemble of von Neumann density matrices with different particle numbers:
\begin{equation}
^{\bar{N}}\Gamma = \sum_{Ni} x^i_N \ket{\Psi^N_i}\bra{\Psi^N_i}\qquad\text{where}\qquad x^i_N \geq 0 \text{ , }\sum_{Ni} x^i_{N} = 1 \text{  and  }\sum_{Ni}x^i_{N} N = \bar{N}~.
\label{frac_NDM}
\end{equation}
The obvious generalization of Eq.~(\ref{frac_NDM}) to reduced density matrices is:
\begin{equation}
^p_{\bar{N}}\Gamma_{\alpha_1\ldots\alpha_p;\beta_1\ldots\beta_p} = \sum_{Ni} x^i_N \bra{\Psi^N_i}a^\dagger_{\alpha_1}\ldots a^\dagger_{\alpha_p}a_{\beta_p}\ldots a_{\beta_1}\ket{\Psi^N_i}~,
\label{frac_pDM}
\end{equation}
in which the $x^i_N$'s satisfy the same relations as for the $\bar{N}$DM. It is important to note that it is not possible to derive the 1DM from the 2DM as with integer $N$ density matrices. When a fractional-$N$ system is described by a Hamiltonian containing both one- and two-body interaction terms, one needs both the 1DM and the 2DM to express the energy. For these systems it is therefore more natural to consider the pair $(\rho,\Gamma)$ as fractional-$N$ representable, if they can both be derived by an ensemble as in Eq.~(\ref{frac_pDM}) with the same weights ${x_N^i}$.

The set of fractional-$N$ representable $(\rho,\Gamma)$'s is obviously convex, and the whole argumentation of Section~\ref{dual_N_rep} is applicable here as well. The dual definition of fractional $N$-representability can therefore be immediately stated. The pair $(\rho,\Gamma)$ is fractional-$N$ representable if, and only if, for all:
\begin{equation}
\hat{H} = \sum_{\alpha\beta} t_{\alpha\beta}a^\dagger_\alpha a_\beta + \frac{1}{4}\sum_{\alpha\beta\gamma\delta}V_{\alpha\beta;\gamma\delta}a^\dagger_\alpha a^\dagger_\beta a_\delta a_\gamma~,
\label{ham_sub}
\end{equation}
the energy expressed in terms of $(\rho,\Gamma)$ is larger than the ground-state energy of this Hamiltonian for an ensemble with average number of particles $\bar{N}$:
\begin{equation}
\mathrm{Tr}~t\rho + \mathrm{Tr}~V\Gamma \geq E^{\bar{N}}_0\left(\hat{H}\right)~.
\label{dual_N_rep_frac}
\end{equation}
\subsubsection{Subsystem 2DM's are fractional-$N$ representable}
Let us now take an arbitrary subspace, $\mathcal{V}$, of single-particle Hilbert space. We use Latin letters ($a,b,\ldots$) for single-particle orbitals that are in the subspace, as opposed to Greek letters ($\alpha,\beta,\ldots$) for general single-particle orbitals. In this Section we prove that if $_N\Gamma$ is an integer $N$-representable 2DM, then the pair $(\rho^{\mathcal{V}},\Gamma^{\mathcal{V}})$, defined as
\begin{eqnarray}
\label{2DM_sub}\Gamma^\mathcal{V}_{ab;cd} &=& _N\Gamma_{ab;cd}\\
\label{1DM_sub}\rho^\mathcal{V}_{ac} &=& \frac{1}{N-1}\sum_{\beta} {}_N\Gamma_{a\beta;c\beta}
\end{eqnarray}
is fractional-$N$ representable on the subspace $\mathcal{V}$, with $\bar{N} = \sum_{a}\rho^{\mathcal{V}}_{aa}$~. As $_N\Gamma$ is integer $N$-representable, we can write:
\begin{equation}
\Gamma^\mathcal{V}_{ab;cd} = \sum_i x_i \bra{\Psi^N_i}a^\dagger_a a^\dagger_b a_d a_c\ket{\Psi^N_i}~.
\label{N_rep_sub_2DM}
\end{equation}
We can expand each $|\Psi^N_i\rangle$ in Slater determinants, classified according to the number of subsystem orbitals they contain:
\begin{equation}
\ket{\Psi^N_i} = \sum_{j=0}^N \sum_{s_j \bar{s}_{N-j}} ~{}^iC^{j}_{s_j\bar{s}_{N-j}} \ket{s_j}\ket{\bar{s}_{N-j}}~,
\end{equation}
in which $s_j$ represents a set of $j$ subsystem orbitals, and $\bar{s}_{N-j}$ a set of $N-j$ orbitals not in the subsystem. Using the fact that the string of subsystem-type creation/annihilation operators in Eq.~(\ref{N_rep_sub_2DM}) does not change the number of subsystem orbitals, that it leaves the non-subsystem part of the Slater determinant unchanged, and using orthonormality of the $\bar{s}_{N-j}$ states, we see that 
\begin{equation}
\sum_i x_i \langle\Psi^N_i|a^\dagger_a a^\dagger_b a_d a_c  |\Psi^N_i\rangle = \sum_i x_i \sum_{j \bar{s}_{N-j}} \langle\Psi^{j}_{i\bar{s}_{N-j}}|a^\dagger_a a^\dagger_b a_d a_c  |\Psi^{j}_{i\bar{s}_{N-j}}\rangle~,
\end{equation}  
where 
\begin{equation}
\ket{\Psi^{j}_{i\bar{s}_{N-j}}}  = \sum_{s_j} {}^iC^{j}_{s_j \bar{s}_{N-j}} \ket{s_j} ~,
\end{equation}
is a state with $j$ particles in the Fock space generated by the subsystem orbitals. These states are not normalized, their norm is given by 
\begin{equation}
\braket{\Psi^{j}_{i\bar{s}_{N-j}}}{\Psi^{j}_{i\bar{s}_{N-j}}} = \sum_{s_j} |{}^iC^{j}_{s_j \bar{s}_{N-j}}|^2 = w^j_{i\bar{s}_{N-j}} ~.
\end{equation}
If we replace them by normalized states
\begin{equation}
|\tilde{\Psi}^{j}_{i\bar{s}_{N-j}}\rangle  = [w^j_{i\bar{s}_{N-j}}]^{-1/2} |\Psi^{j}_{i\bar{s}_{N-j}}\rangle ~,
\end{equation}
it follows that 
\begin{equation}
\Gamma^\mathcal{V}_{ab;cd} = \sum_{j;i\bar{s}_{N-j}} x_i w^j_{i\bar{s}_{N-j}} \langle\tilde{\Psi}^{j}_{i\bar{s}_{N-j}}|a^\dagger_a a^\dagger_b a_d a_c  |\tilde{\Psi}^{j}_{i\bar{s}_{N-j}}\rangle ~,
\end{equation}
where 
\begin{equation}
\sum_{j;i\bar{s}_{N-j}} x_i w^j_{i\bar{s}_{N-j}}=1~,
\end{equation}
because of the normalization of the original $N$-particle states. 
In an analogous way one shows that the subsystem 1DM $\rho^{\mathcal{V}}$ can be written as
\begin{equation}
\rho^\mathcal{V}_{ac} = \sum_{j;i\bar{s}_{N-j}} x_i w^j_{i\bar{s}_{N-j}} \langle\tilde{\Psi}^{j}_{i\bar{s}_{N-j}}|a^\dagger_a a_c  |\tilde{\Psi}^{j}_{i\bar{s}_{N-j}}\rangle  ~.
\end{equation}
This proves that $\Gamma^\mathcal{V}$ and $\rho^\mathcal{V}$ can be derived from the same ensemble of wave functions containing only orbitals in the subsystem. 
This ensemble has a fractional number of particles (in the subsystem space) given by: 
\begin{equation}
\bar{N}= \sum_a \rho^\mathcal{V}_{aa}= \sum_{j;i\bar{s}_{N-j}}j x_i w^j_{i\bar{s}_{N-j}}~.
\end{equation}

We can again harness the power of the dual formulation of $N$-representability (\ref{dual_N_rep_frac}) to introduce the subsystem constraints:
\begin{theorem}
\label{theo_sub_constr}
If $_N\Gamma$ is integer $N$-representable, then for an arbitrary subspace of single-particle Hilbert space $\mathcal{V}$, the pair $(\rho^\mathcal{V},\Gamma^{\mathcal{V}})$ as defined in Eqs.~(\ref{2DM_sub}) and (\ref{1DM_sub}) must obey the inequality:
\begin{equation}
\mathrm{Tr}~\rho^\mathcal{V} t^\mathcal{V} + \mathrm{Tr}~\Gamma^\mathcal{V} V^\mathcal{V} \geq E^{\bar{N}}_0\left(\hat{H}^\mathcal{V}\right)~,
\label{eq_sub_constr}
\end{equation}
with $\bar{N} = \mathrm{Tr}\rho^\mathcal{V}$, and for every Hamiltonian $\hat{H}^\mathcal{V}$ of the form (\ref{ham_sub}), defined in the subspace $\mathcal{V}$.
\end{theorem}
\subsubsection{\label{applicability_of_ssc}Applicability of subsystem constraints}
Having established the explicit form of the subsystem constraints, one could ask the question: why would these constraints be important? One answer to that question is that Theorem~\ref{theo_sub_constr} does not hold for \emph{approximate} $N$-representability, using the matrix positivity conditions introduced in the previous Sections. In other words, the fact that the global 2DM of the $N$-particle system obeys necessary $N$-representability conditions does not imply that the subsystem pairs $(\rho^{\mathcal{V}},\Gamma^{\mathcal{V}})$ obey the same fractional $N$-representability conditions. In analogy to Eq.~(\ref{frac_pDM}) one could say that the pair $({}_{\bar{N}}\rho,{}_{\bar{N}}\Gamma)$ is \emph{approximate} fractional-$N$ representable under a set of necessary matrix positivity conditions $\mathcal{L}$, if and only if there exists a set of weights $x_N$ in which ${}_{\bar{N}}\rho$ and ${}_{\bar{N}}\Gamma$ can be expanded as:
\begin{align}
&{}_{\bar{N}}\rho = \sum_N x_N~ {}_N\rho~, \\
&{}_{\bar{N}}\Gamma = \sum_N x_N~ {}_N\Gamma~,
\label{2DM_frac_expansion}
\end{align}
where for every $_N\Gamma$ in the expansion holds that
\begin{equation}
\mathcal{L}(_N\Gamma) \succeq 0~.
\label{2DM_app_frac}
\end{equation}
If we have an approximate integer-$N$ representable $\Gamma$, {\it i.e.} for which $\mathcal{L}(\Gamma)\succeq 0$ for a set of $\mathcal{L}$'s. The constraints imposed by this on a subspace 2DM $\Gamma^\mathcal{V}$ are not sufficient to ensure approximate fractional-$N$ representability. The subspace $\mathcal{L}(\Gamma)^\mathcal{V}$'s are diagonal blocks of the full system $\mathcal{L}(\Gamma)$'s, and will therefore be positive semidefinite:
\begin{equation}
\mathcal{L}(\Gamma)^\mathcal{V}\succeq 0~.
\label{no_frac_nrep}
\end{equation}
This is, however, not sufficient to ensure that there exists an expansion of the form (\ref{2DM_frac_expansion}), in which, for every $_N\Gamma$, Eq.~(\ref{2DM_app_frac}) holds, which is required for it to be approximate fractional-$N$ representable. This means that 2DM optimization under matrix positivity conditions doesn't treat the full system and the subsystem on equal footing. This has the important consequence that the method is not size consistent (more on that subject in Section~\ref{diatomic}). Imposing additional constraints on the subsystem can fix this, as was shown in \cite{qsep} for diatomic molecules. For more general systems there are a plethora of subspaces and Hamiltonians to impose (\ref{eq_sub_constr}), so it is crucial to choose these carefully. A stochastical method in which these are optimized was proposed in \cite{shenvi}. Once the subspace and Hamiltonian are chosen one can calculate the right hand side of Eq.~(\ref{eq_sub_constr}) using exact diagonalization if the subsystem is small enough, or a density matrix optimization if it is larger. 

\subsection{The sharp conditions}
The sharp conditions are another type of condition that makes use of the dual definition of $N$-representability. Here, once again the positive Hamiltonians $\hat{B}^\dagger \hat{B}$ are used, but in a completely different way: now sharper bounds are imposed on the upper and lower eigenvalues of these Hamiltonians \cite{dimi}.
\subsubsection{Sharp bounds on $\mathcal{I}(\Gamma)$}
\paragraph{Canonical transformation}
We first show that a two-fermion creation operator:
\begin{equation}
\hat{B}^\dagger = \sum_{\alpha\beta}B_{\alpha\beta}a^\dagger_\alpha a^\dagger_\beta~,
\label{2p_creation}
\end{equation}
with $B$ a skew-symmetric matrix, can always be written as,
\begin{equation}
\hat{B}^\dagger = \frac{1}{\sqrt{2}}\sum_{\alpha} x_\alpha a^\dagger_{\alpha} a^\dagger_{\bar{\alpha}}~,
\label{pairing_operator}
\end{equation}
through a canonical transformation of the single-particle basis. Each single-particle state $\alpha$ has its paired state $\bar{\alpha}$. For the sake of completeness, a short proof follows.

The 1DM of the two-particle state constructed by $\hat{B}^\dagger$ has the form:
\begin{equation}
\rho = 2B^*B^T~.
\end{equation}
If we transform to the basis of natural orbitals, $\rho$ is diagonal and real:
\begin{equation}
\rho = \rho^*\qquad\text{which means that}\qquad BB^* = B^*B~.
\end{equation}
In this basis the commutator of $\rho$ and $B$ is:
\begin{eqnarray}
[\rho,B] &=& 2[B^*B^T,B] = -2[B^*B,B]\\
&=& -2(B^*BB - BB^*B)\\
&=& -2[B^*,B]B = 0~.
\end{eqnarray}
It follows that the matrix $B$ is block diagonal in the natural basis, with $(2\times2)$-blocks since the eigenvalues of $\rho$ are at least doubly degenerate \cite{coleman}:
\begin{equation}
[\rho,B]_{\alpha\beta} = B_{\alpha\beta}(\lambda_\alpha - \lambda_\beta) = 0~.
\end{equation}
Let us now define a new blockdiagonal matrix $D$ as:
\begin{equation}
D^{(\alpha)} = \sqrt{\frac{2}{\lambda_\alpha}}B^{(\alpha)}~,
\end{equation}
in which $B^{(\alpha)}$ is the $B$-block corresponding to the eigenvalue $\lambda_\alpha$. $D$ is skew-symmetric, but also unitary, which means that if we diagonalize it the eigenvalues are purely imaginary phases, $i$ and $-i$, and the two eigenvectors are complex conjugate, $X$ and $X^*$. If one more transformation is performed:
\begin{eqnarray}
X_- &=& \frac{1}{\sqrt{2}i}\left(X - X^*\right)~,\\
X_+ &=& \frac{1}{\sqrt{2}}\left(X + X^*\right)~,
\end{eqnarray}
every block in $D$ becomes:
\begin{equation}
\left(\begin{matrix}i & 0\\0 & -i\end{matrix}\right)\qquad\rightarrow\qquad\left(\begin{matrix}0 & 1\\-1 & 0\end{matrix}\right)~.
\end{equation}
The same transformation turns $B$ into:
\begin{equation}
B^{(\alpha)} = \left(\begin{matrix}0&\sqrt{\frac{\lambda_\alpha}{2}}\\-\sqrt{\frac{\lambda_\alpha}{2}}&0\end{matrix}\right)~,
\end{equation}
which means we can write $\hat{B}^\dagger$ as (\ref{pairing_operator}) with
\begin{equation}
x_\alpha = \sqrt{\frac{\lambda_\alpha}{2}} \qquad\text{and}\qquad x_{\bar{\alpha}} = -x_\alpha~.
\end{equation}
\paragraph{Pairing Hamiltonian}
It turns out that the Hamiltonian $\hat{B}^\dagger \hat{B}$, with $\hat{B}^\dagger$ a pairing operator (\ref{pairing_operator}), is exactly solvable through a Bethe-ansatz approach \cite{richardson,richardson3,gaudin}. The eigenstates can be classified according to the presence of unpaired single-particle states. As the Hamiltonian doesn't connect states with a different number of unpaired states, we can restrict the discussion to the fully-paired eigenstates; in other subspaces the unpaired states are blocked, and simply removed from the available single-particle space. The eigenstates are labeled by $n-1$ numbers $y$, in which $n$ is the number of pairs present in the eigenstate:
\begin{equation}
\ket{\Psi^N_{\{y\}}} = \phi^\dagger(0)\phi^\dagger(y_1)\ldots\phi^\dagger(y_{n-1})\ket{0}~,
\label{eig_pair}
\end{equation}
and a new pair-creation operator $\phi^\dagger$ is defined as:
\begin{equation}
\phi^\dagger(y) = \frac{1}{\sqrt{2}}\sum_\alpha \frac{x_\alpha}{1-yx_\alpha^2}a^\dagger_\alpha a^\dagger_{\bar{\alpha}}~.
\label{phi}
\end{equation}
One can show (see {\it e.g.} in \cite{bijbel}) that the action of the pairing Hamiltonian on the states (\ref{eig_pair}) is:
\begin{equation}
\hat{B}^\dagger \hat{B} \ket{\Psi^N_{\{y\}}} = E\left(\{y\}\right)\ket{\Psi^N_{\{y\}}} + \sum_{i = 1}^{n-1}V_i\left(\{y\}\right)\hat{B}^\dagger \hat{B}^\dagger\left(\prod_{j (\neq i) = 1}^{n - 1}\phi^{\dagger}(y_j)\right)\ket{0}~,
\label{BBPsi}
\end{equation}
in which
\begin{eqnarray}
V_i\left(\{y\}\right) &=& \sum_{\alpha}\frac{x_\alpha^2}{1 - y_ix_\alpha^2} + 4\left(\frac{1}{y_i} + \sum_{j ( \neq i) = 1}^{n-1}\frac{1}{y_i - y_j}\right)~,\\
E\left(\{y\}\right) &=& \sum_\alpha x_\alpha^2 -4\sum_{k = 1}^{n-1}\frac{1}{y_k}~.
\label{rich_ener}
\end{eqnarray}
From Eq.~(\ref{BBPsi}) it is clear that (\ref{eig_pair}) will only be an eigenstate if the $\{y\}$ satisfy the equations:
\begin{equation}
V_i\left(\{y\}\right) = 0~, \qquad\forall i~.
\label{richard_eq}
\end{equation}
The energy of the eigenstate is then given by $(\ref{rich_ener})$~. 

In general a system of non-linear equations can be hard to solve, but in this particular case, Newton's method can be used to find the largest eigenvalue.  The different solutions are separated by singularities \cite{rombouts,stijn}, so we have to find an initial point in the right compartment of $y$ space. For the highest eigenvalue, a good guess is that all $y$'s will be negative and close to zero. The following initial point:
\begin{equation}
y^0_{\alpha} = -\frac{\alpha}{10}\qquad\text{for}\qquad \alpha = 1,\ldots,n-1~,
\end{equation}
was always observed to converge to the highest eigenvalue.

It is now straightforward to introduce the new constraint, using the dual definition of $N$-representability for the pairing Hamiltonian:
\begin{equation}
-\mathrm{Tr}~\Gamma B^\dagger B \geq E^N_0(-\hat{B}^\dagger \hat{B}) = -\lambda_N^{\text{max}}\left[B\right]~.
\end{equation}
Note that the eigenvalue spectrum of the Hamiltonian $B^\dagger B$ only depends on the singular values $x$ as introduced in Eq.~(\ref{pairing_operator}).
We can rephrase this condition as follows: for every two-particle state, its occupation in the 2DM has to be smaller than the maximally allowed occupation in an $N$-particle system calculated through the Richardson equation:
\begin{equation}
\label{sharp_P}
\mathrm{Tr}~\Gamma B^\dagger B \leq \lambda_N^{\text{max}}\left[x(B)\right]~.
\end{equation}
It is interesting to see that for a structureless pairing operator:
\begin{equation}
\hat{X}^\dagger = \sqrt{\frac{1}{2M}}\sum_{\alpha}\sigma_\alpha a^\dagger_{\alpha}a^\dagger_{\bar{\alpha}}~,\qquad\text{with}\qquad \sigma_\alpha = -\sigma_{\bar{\alpha}} = \pm 1~,
\end{equation}
which has the highest maximal eigenvalue of all possible $\hat{B}^\dagger$'s, the maximal eigenvalue is 
\begin{equation}
\lambda^{\text{max}}_N\left[x(X)\right] = n\left(1 - \frac{2(n-1)}{M}\right)~.
\end{equation}
This is the sharpest upper bound that can be put on the eigenvalues $\Gamma$ without knowledge of the eigenvectors, which was derived by F. Sasaki \cite{sasaki}.
\paragraph{Finding the most stringent condition}
The constraint put forward in Eq.~(\ref{sharp_P}) is not practical yet, because we don't know how to choose $\hat{B}^\dagger$, and we can't impose the inequality for all $\hat{B}^\dagger$'s. In this Section we propose a method to find the $\hat{B}^\dagger$ that violates Eq.~(\ref{sharp_P}) the most. Given an arbitrary 2DM $\Gamma$, we introduce a cost function $F$:
\begin{equation}
F(B) = \frac{1}{\mathrm{Tr}~B^\dagger B}\left[\lambda^{\text{max}}_N\left[x(B)\right] - \mathrm{Tr}~\Gamma B^\dagger B\right]~,
\label{cost_richardson}
\end{equation}
which we minimize as a function of the matrix $B$. To optimize this function we need the gradient of $F$ with respect to the pair amplitudes $B_{\mu\nu}$. The most difficult term is the first one as we need an analytical formula which tells how the maximal eigenvalue varies with $B$. Using the chain rule we find:
\begin{equation}
\frac{\partial \lambda^{\text{max}}_N\left[x(B)\right]}{B_{\mu\nu}} = \sum_{\kappa>0} \left(\frac{\partial \lambda^{\text{max}}_N(x)}{\partial x_\kappa}\right)\left( \frac{\partial x_\kappa}{\partial B_{\mu\nu}}\right)~,
\label{deriv_lambda}
\end{equation}
where $\kappa > 0$ means we only sum over independent $x_\kappa$'s, not over the $x_{\bar{\kappa}}$.
The relationship between $B$ and $x$ is a unitary transformation:
\begin{equation}
B_{\alpha\beta} = \sqrt{2}\sum_{\kappa}U_{\kappa\alpha} x_\kappa U_{\bar{\kappa}\beta}\qquad\text{or inversely}\qquad x_\kappa = \frac{1}{\sqrt{2}}\sum_{\alpha\beta}U_{\kappa\alpha}B_{\alpha\beta}U_{\bar{\kappa}\beta}~.
\end{equation}
This means that the second term in Eq.~(\ref{deriv_lambda}) is:
\begin{equation}
\frac{\partial x_\kappa}{\partial B_{\mu\nu}} = \frac{1}{\sqrt{2}}\left[U_{\kappa\mu}U_{\bar{\kappa}\nu} - U_{\kappa\nu}U_{\bar{\kappa}\mu}\right]~.  \end{equation}
The first term can be derived from Eq.~(\ref{rich_ener}) by differentiation:
\begin{equation}
\frac{\partial \lambda^{\text{max}}_N}{\partial x_\mu}(x) = 4x_\mu + 4 \sum_{k = 1}^{n-1}\frac{1}{y_k^2(x)}\frac{\partial y_k}{\partial x_\mu}(x)~,
\end{equation}
which is written in function of the $y$'s. It is left to determine how the $y$'s change as a function of the structure coefficients $x$. The infinitesimally changed $y$'s have to remain a solution to the Richardson equations, which leads to the expression:
\begin{eqnarray}
\nonumber\frac{ \partial V_i}{\partial x_\mu}\left(y(x)\right) &=& \frac{4x_\mu}{\left(1 - y_i x_\mu^2\right)^2} +\left[2\sum_{\alpha=1}^M \frac{x_\alpha^4}{\left(1 - y_ix_\alpha^2\right)^2} - 4\left(\frac{1}{y_i^2} + \sum_{j (\neq i) = 1}^{n-1}\frac{1}{\left(y_i - y_j\right)^2}\right)\right]\frac{\partial y_i}{\partial x_\mu}\\
\label{lin_sys_deriv}&& \qquad+~4\sum_{j (\neq i) = 1}^{n-1} \frac{1}{(y_i - y_j)^2}\frac{\partial y_j}{\partial x_\mu} = 0~,
\end{eqnarray}
which yields $n$ \emph{linear} equations, that can easily be solved to obtain the derivatives. In summary, the gradient of the Eq.~(\ref{cost_richardson}) can be written as:
\begin{eqnarray}
\nonumber\frac{\partial F}{\partial B_{\mu\nu}} &=& \frac{1}{\mathrm{Tr}~B^\dagger B} \left[\frac{1}{\sqrt{2}}\sum_{\kappa>0}\frac{\partial\lambda^{\text{max}}_N}{\partial x_\kappa}\left[U_{\kappa\mu}U_{\bar{\kappa}\nu} - U_{\kappa\nu}U_{\bar{\kappa}\mu}\right] - \sum_{\gamma\delta}\Gamma_{\mu\nu;\gamma\delta}B_{\gamma\delta}\right]\\
\label{grad_sharp_P}&& -\frac{2B_{\mu\nu}}{\left(\mathrm{Tr}~B^\dagger B\right)^2}\left[\lambda^{\text{max}}_N\left[x(B)\right] - \mathrm{Tr}~\Gamma B^\dagger B\right]~.
\end{eqnarray}
After a 2DM optimization, the most violated condition can now be found using a non-linear conjugate gradient algorithm \cite{jrs}, and can be added as a constraint in a subsequent optimization. This can be done iteratively, until there are no violated constraints left.
\subsubsection{Sharp bounds on $\mathcal{Q}(\Gamma)$}
Having established a sharp bound on $\mathcal{I}(\Gamma)$, we can ask ourselves if the same thing can be done for the other constraint matrices. If an upper bound can be found for the eigenvalues of $\mathcal{Q}$, the following constraint can be formulated:
\begin{equation}
\mathrm{Tr}~\mathcal{Q}(\Gamma)B B^\dagger \leq E_{\text{max}}^N(\hat{B}\hat{B}^\dagger)~.
\label{sharp_Q}
\end{equation}
It turns out that the eigenvalues of $\hat{H} = \hat{B}\hat{B}^\dagger$ can also be found using a Richardson-like approach discussed in the previous paragraph. Suppose we have found an eigenvector of $\hat{B}\hat{B}^\dagger$ with eigenvalue $\lambda$:
\begin{equation}
\hat{B}\hat{B}^\dagger\ket{\Psi^N} = \lambda \ket{\Psi^N}~.
\end{equation}
If we let another $\hat{B}^\dagger$ operator act on the left and right side of this equation we have:
\begin{equation}
\hat{B}^\dagger\hat{B}\ket{\Psi^{N+2}} = \lambda \ket{\Psi^{N+2}}~,\qquad\text{with}\qquad\ket{\Psi^{N+2}} = \hat{B}^\dagger\ket{\Psi^N}~,
\end{equation}
which means that every non-zero eigenvalue of $\hat{B}\hat{B}^\dagger$ on the $N$-particle system is also an eigenvalue of $\hat{B}^\dagger \hat{B}$ on the $N+2$-particle system. In an analogous way one can show that the inverse is also true, which means that the maximal eigenvalue of $\hat{B}\hat{B}^\dagger$ on an $N$-particle system can be found by solving the Richardson equations (\ref{richard_eq}) for $N+2$ particles\,! Given an arbitrary 2DM $\Gamma$, the most violating two-hole operator $\hat{B}$ can be found by optimizing the cost function:
\begin{equation}
F_{Q}(B) = \frac{1}{\mathrm{Tr}~B^\dagger B}\left(\lambda^{\text{max}}_{N+2}\left(x(B)\right) - \mathrm{Tr}~\left[\mathcal{Q}\left(\Gamma\right)BB^\dagger \right]\right)~,
\end{equation}
for which the gradient can be calculated in the same way as for (\ref{grad_sharp_P}).

Using the approach introduced in the previous two paragraphs we constructed additional constraints to increase the accuracy of the standard two-index conditions. In practice, however, we found that for molecular and atomic systems, the sharp-$\mathcal{I}$ (\ref{sharp_P}) and sharp-$\mathcal{Q}$ (\ref{sharp_Q}) conditions are never violated if the $\mathcal{G}$ condition is active. So in general these constraints did not deliver the increased accuracy hoped for. As would be expected, the constraints did help for pairing type Hamiltonians:
\begin{equation}
\hat{H} = \sum_{\alpha}\epsilon_\alpha a^\dagger_\alpha a_\alpha - \frac{g}{2} \sum_{\alpha\beta} x_\alpha x_{\beta} a^\dagger_{\alpha}a^\dagger_{\bar{\alpha}}a_{\bar{\beta}}a_\beta~,
\end{equation}
when $g$ is large enough. They insure that for $g\rightarrow\infty$ the correct limit is found.
\subsubsection{Sharp bounds on $\mathcal{G}(\Gamma)$}
For the $\mathcal{G}$ condition, it is not possible to transform the Hamiltonian to a type solvable by the Richardson equations. For a general $\mathcal{G}$-type $\hat{B}^\dagger$ operator:
\begin{equation}
\hat{B}^\dagger = \sum_{\alpha\beta}B_{\alpha\beta}a^\dagger_\alpha a_\beta~,
\end{equation}
the problem seems to be very hard, and no solution has thus far been found. For a restricted class of $B^\dagger$ operators, when $B$ is a Hermitian matrix, it is possible to derive some constraints of the form:
\begin{equation}
\mathrm{Tr}~\left[\mathcal{G}(\Gamma)B^\dagger B\right] \leq E_{\text{max}}^N(\hat{B}^\dagger \hat{B})~,
\qquad\text{and}\qquad
\mathrm{Tr}~\left[\mathcal{G}(\Gamma)B^\dagger B\right] \geq E_{\text{min}}^N(\hat{B}^\dagger \hat{B})~.
\label{sharp_G}
\end{equation}
Compared to the solution for sharp-$\mathcal{I}$ and sharp-$\mathcal{Q}$, this derivation is relatively easy. Consider a Hermitian operator $\hat{B}^\dagger = \hat{B}$, which in its canonical basis looks like:
\begin{equation}
\hat{B}^\dagger = \sum_\alpha \epsilon_\alpha a^\dagger_\alpha a_\alpha~.
\end{equation}
Its eigenvectors on $N$-particle space are Slater determinants with as eigenvalues the sum of the energies of the $N$ occupied orbitals. Because $\hat{B}^\dagger$ is Hermitian, $\hat{B}^\dagger \hat{B} $ is just $\hat{B}^2$, which has the same eigenvectors as $\hat{B}^\dagger$, but with the eigenvalues squared. Now that we have diagonalized the Hamiltonian, we just have to find out for which set of $N$ single-particle energies $\mathcal{S}$, the eigenvalue:
\begin{equation}
E_\mathcal{S} = \left(\sum_{\alpha\in\mathcal{S}}\epsilon_\alpha\right)^2~,
\label{G_sharp_ener}
\end{equation}
is extremal. 
\paragraph{Maximal eigenvalue}
For the maximal eigenvalue there are only two possibilities, either the $N$ highest or the $N$ lowest single-particle energies sum up to the highest absolute value.
\paragraph{Minimal eigenvalue}
Because the different terms in the sum (\ref{G_sharp_ener}) can cancel each other, the lowest eigenvalue is much more difficult to obtain. Basically any combination of single-particle energies can have the lowest absolute value, which means it is a combinatorial optimization problem. 
This constraint is probably very important, as the $\mathcal{G}$ map seems most closely related to the structure of physical Hamiltonians studied in physics and chemistry. One example of where this constraint would surely improve results is in the case of a dispersion Hamiltonian:
\begin{equation}
\hat{D}(\hat{B}^\dagger,\lambda) = \left(\hat{B}^\dagger - \lambda\hat{\mathbb{1}}\right)^2~,
\label{disp_ham}
\end{equation}
where it has been shown that the standard matrix positivity conditions fail drastically \cite{dispersion}.
\subsection{\label{diag_constr}Diagonal constraints}
The diagonal conditions are a hierarchic set of linear inequalities that only involve the diagonal part of the 2DM. The first conditions of this type were derived by Weinhold and Wilson \cite{weinhold}, and later generalized and expanded by Davidson \cite{davidson_1,davidson_2}. Davidson derived these constraints using the dual definition of $N$-representability with the following Hamiltonian:
\begin{equation}
\hat{H} = A + \sum_{\alpha}B_\alpha a^\dagger_\alpha a_\alpha + \sum_{\alpha\beta}C_{\alpha\beta} a^\dagger_\alpha a^\dagger_\beta a_\beta a_\alpha~.
\label{ham_diag_constr}
\end{equation}
The expression of the energy expectation value of this type of Hamiltonian as a function of the 2DM only involves diagonal elements of the 2DM. The eigenstates of this Hamiltonian are Slater determinants, with energy eigenvalue:
\begin{equation}
E_\mathcal{S} = A + \sum_{\alpha\in\mathcal{S}}B_\alpha + \sum_{\alpha,\beta\in\mathcal{S}}C_{\alpha\beta}~. 
\end{equation}
A 2DM will satisfy the $N$-representability condition (\ref{eq_dual_n_rep}) for every Hamiltonian of the type (\ref{ham_diag_constr}) if it satisfies (\ref{eq_dual_n_rep}) for the extreme points of the convex set of Hamiltonians (\ref{ham_diag_constr}). These extreme points generate the diagonal inequalities, and in \cite{davidson_1,davidson_2} a complicated general algorithm is used to derive some of these. In \cite{davidson_3} however, a much simpler way of deriving these conditions is introduced, which involves the positivy of polynomials of the type:
\begin{equation}
\bra{\Psi^N}\sum_{\alpha\beta\gamma\ldots}y_{\alpha\beta\gamma\ldots}\hat{n}_\alpha \hat{n}_\beta\hat{n}_\gamma\ldots\ket{\Psi^N} \geq 0~,
\end{equation}
in which $\hat{n}_\alpha$ is the number operator $a^\dagger_\alpha a_\alpha$. The so-called $(2,2)$-conditions are given by:
\[
\begin{matrix}
\langle \hat{n}_\alpha \hat{n}_\beta \rangle &\geq& 0~&,\qquad &\langle (1 - \hat{n}_\alpha) \hat{n}_\beta \rangle &\geq& 0~,\\
\langle \hat{n}_\alpha (1 - \hat{n}_\beta) \rangle &\geq& 0~&\qquad\text{and}\qquad&\langle (1-\hat{n}_\alpha) (1-\hat{n}_\beta) \rangle &\geq& 0~,
\end{matrix}
\]
which are just the diagonal elements of the two-index conditions. In the same way one can derive $(3,3)$-conditions on the 3DM, in which the three-particle term sometimes has a positive and sometimes a negative prefactor. Adding these constraints together one can construct new constraints only containing diagonal elements of the 2DM, which are called the $(3,2)$-conditions, {\it e.g.} the sum:
\begin{equation}
\langle \hat{n}_\alpha \hat{n}_\beta \hat{n}_\gamma\rangle + \langle(1-\hat{n}_\alpha) (1-\hat{n}_\beta) (1-\hat{n}_\gamma)\rangle = \langle1-\hat{n}_\alpha -\hat{n}_\beta - \hat{n}_\gamma + \hat{n}_\alpha \hat{n}_\beta + \hat{n}_\alpha \hat{n}_\gamma + \hat{n_\beta}\hat{n}_\gamma\rangle~,
\end{equation}
which is the diagonal $\mathcal{T}_1$. It is possible to continue the construction of higher-order $(r,r)$-conditions in this way. By taking the appropriate combination of these $(r,r)$-conditions one derives higher-order constraints that can be expressed as a function of the diagonal terms of the 2DM alone, called $(r,2)$-conditions.

What is nice about these constraints is that they are computationally much cheaper to impose than the full matrix-positivity conditions. Imposing all fourth-order constraints would scale in the same way as imposing the full $\mathcal{G}$-matrix condition. A disadvantage of this technique is that the results are dependent on the choice of the single-particle basis used, {\it i.e.} we lose unitary invariance of our approximate $N$-representability. One can envision two ways to try to cure this, which have not been examined up to now (but should be in the future). The first way is by imposing that the optimized energy is stationary to infinitesimal unitary transformations. A unitary transformation $U$ is defined by its anti-hermitian generator $\epsilon$:
\begin{eqnarray}
U = e^{\lambda \epsilon}~.
\end{eqnarray}
An infinitesimal unitary transformation can therefore be expressed as $U = 1 + \lambda\epsilon$. The change in the energy expression under influence of this infinitesimal transformation is, to first order in $\lambda$:
\begin{eqnarray}
\nonumber E(U) &=& \mathrm{Tr}~\Gamma H + \frac{\lambda}{4}\sum_{\alpha\beta\gamma\delta}H_{\alpha\beta;\gamma\delta}\left(
\sum_{\alpha'}\epsilon^*_{\alpha'\alpha}\Gamma_{\alpha'\beta;\gamma\delta}\right.\\
\nonumber&&\left.\qquad\qquad\qquad+\sum_{\beta'}\epsilon^*_{\beta'\beta}\Gamma_{\alpha\beta';\gamma\delta}
+\sum_{\gamma'}\epsilon_{\gamma'\gamma}\Gamma_{\alpha\beta;\gamma'\delta}
+\sum_{\delta'}\epsilon_{\delta'\delta}\Gamma_{\alpha\beta;\gamma\delta'} \right)\\
\nonumber&=& \mathrm{Tr}~\Gamma H + \frac{\lambda}{2}\left(\sum_{\gamma\gamma'}\epsilon_{\gamma'\gamma}\sum_{\delta}\left(H\Gamma\right)_{\gamma\delta;\gamma'\delta} - \sum_{\alpha\alpha'}\epsilon_{\alpha\alpha'}\sum_\beta \left(H\Gamma\right)_{\alpha\beta;\alpha'\beta}\right)\\
\label{stationarity}&=& \mathrm{Tr}~\Gamma H + \frac{\lambda}{2}\sum_{\alpha\beta}\epsilon_{\alpha\beta}~\left(\overline{[\Gamma,H]}\right)_{\alpha\beta}~,
\end{eqnarray}
in which a bar over a matrix indicates that we trace over one pair of indices:
\begin{equation}
\bar{\Gamma}_{\alpha\gamma} = \sum_{\beta}\Gamma_{\alpha\beta;\gamma\beta}~.
\end{equation}
From Eq.~(\ref{stationarity}) it is clear that the energy will be stable to infinitesimal unitary transformations if we impose the constraint:
\begin{equation}
\overline{\left[\Gamma,H\right]} = 0~.
\end{equation}

A second way to improve the result obtained by diagonal conditions is by optimizing the single-particle basis. Suppose we have performed an optimization using only two-index conditions, and in the next iteration we want to impose diagonal three-index conditions of the $\mathcal{T}_2$ type:
\begin{equation}
\mathcal{D}^3_2\left(\Gamma\right)_{\alpha\beta\gamma} = \rho_{\gamma\gamma} +\Gamma_{\alpha\beta;\alpha\beta} - \Gamma_{\beta\gamma;\beta\gamma} - \Gamma_{\alpha\gamma;\alpha\gamma} \geq 0~.
\end{equation}
How to decide what single-particle basis to use in which to impose these? We could search for the most violating single-particle basis by minimizing the following functional:
\begin{eqnarray}
F_{\alpha\beta\gamma}(U) &=&  \sum_{\gamma'}U^*_{\gamma'\gamma}\rho_{\gamma'\gamma'}U_{\gamma'\gamma} + \sum_{\alpha'\beta'}U^*_{\alpha'\alpha}U^*_{\beta'\beta}\Gamma_{\alpha'\beta';\alpha'\beta'}U_{\alpha'\alpha}U_{\beta'\beta} \\
\nonumber&& - \sum_{\beta'\gamma'}U^*_{\beta'\beta}U^*_{\gamma'\gamma}\Gamma_{\beta'\gamma';\beta'\gamma'}U_{\beta'\beta}U_{\gamma'\gamma} - \sum_{\alpha'\gamma'}U^*_{\alpha'\alpha}U^*_{\gamma'\gamma}\Gamma_{\alpha'\gamma';\alpha'\gamma'}U_{\alpha'\alpha}U_{\gamma'\gamma}~,
\end{eqnarray}
with respect to the unitary matrix $U$.

\chapter{\label{SDP}Semidefinite programming}
The variational determination of the 2DM of a quantum system, henceforth called the v2DM technique, can be formulated mathematically as a semidefinite program (SDP). A semidefinite program is not a program, but a type of constrained optimization, in which a cost function is optimized under the constraint that a matrix remains positive semidefinite. There is a vast literature on this subject, and a nice duality theory has been established. In this Chapter we first give a short introduction to some important results from the literature, and show how to formulate density matrix optimization as an SDP. After this a number of standard algorithms, which have been tailored to the specific form of v2DM, are explained in some detail.
\section{\label{ha2DM}Hermitian adjoint maps}
For the following it is useful to introduce the Hermitian adjoints of matrix maps introduced in Section~\ref{standard_n_rep}. The Hermitian adjoint maps are defined through:
\begin{equation}
\label{gen_herm}
\mathrm{Tr}~\mathcal{L}_i(\Gamma) A = \mathrm{Tr}~\mathcal{L}_i^\dagger(A)\Gamma~,
\end{equation}
in which $A$ is a matrix of the same dimension as the image of the map $\mathcal{L}_i$ in question (\emph{e.g.} a three-particle matrix for a $\mathcal{T}_1$ map, \emph{etc.}), and the trace is a sum over the appropriate indices. 

It is important to note that from now on slightly different matrix maps are used compared to the ones introduced in Section \ref{standard_n_rep}. It is mathematically and computationally more pleasing to deal with maps that are homogeneous in $\Gamma$, so we add a term $\frac{2\mathrm{Tr}~\Gamma}{N(N-1)}$ to the non-homogeneous terms.

The $\mathcal{I}$ and $\mathcal{Q}$ maps are Hermitian, so they are identical to their Hermitian adjoints. For the other maps however this is not the case. Using Eq.~(\ref{gen_herm}) the Hermitian adjoint of the $\mathcal{G}$ map can be shown to have the form:
\begin{eqnarray}
\label{G_down}
\mathcal{G}^\dagger\left(A\right)_{\alpha\beta;\gamma\delta} &=& \frac{1}{N-1}\left[\delta_{\beta\delta}\overline{A}_{\alpha\gamma} - \delta_{\alpha\delta}\overline{A}_{\beta\gamma} - \delta_{\beta\gamma}\overline{A}_{\alpha\delta} + \delta_{\alpha\gamma}\overline{A}_{\beta\delta}\right]\\
\nonumber&&\qquad\qquad - A_{\alpha\delta;\gamma\beta} + A_{\beta\delta;\gamma\alpha} + A_{\alpha\gamma;\delta\beta} - A_{\beta\gamma;\delta\alpha}~,
\end{eqnarray}
in which a particle-hole matrix $A$ is mapped on two-particle matrix space and
\begin{equation}
\overline{A}_{\alpha\gamma} = \sum_\lambda A_{\alpha\lambda;\gamma\lambda}~.
\end{equation}
The $\mathcal{T}_1$ operator maps a two-particle matrix on a three-particle matrix, so its Hermitian adjoint has to map a three-particle matrix $A$ on two-particle matrix space. Solving Eq.~(\ref{gen_herm}) with $\mathcal{L} = \mathcal{T}_1$ one finds that:
\begin{eqnarray}
\label{T1_down}
\mathcal{T}^\dagger_1\left(A\right)_{\alpha\beta;\gamma\delta} &=& \frac{2}{N(N-1)}\left(\delta_{\alpha\gamma}\delta_{\beta\delta} - \delta_{\alpha\delta}\delta_{\beta\gamma}\right)\mathrm{Tr}~A + \overline{A}_{\alpha\beta;\gamma\delta}\\
\nonumber&&-\frac{1}{2(N-1)}\left[\delta_{\beta\delta}\overline{\overline{A}}_{\alpha\gamma} - \delta_{\alpha\delta}\overline{\overline{A}}_{\beta\gamma} - \delta_{\beta\gamma}\overline{\overline{A}}_{\alpha\delta} + \delta_{\alpha\gamma}\overline{\overline{A}}_{\beta\delta}\right]~,
\end{eqnarray}
with
\begin{eqnarray}
\overline{A}_{\alpha\beta;\gamma\delta} &=& \sum_\lambda A_{\alpha\beta\lambda;\gamma\delta\lambda}~,\\
\overline{\overline{A}}_{\alpha\gamma} &=& \sum_{\lambda\kappa} A_{\alpha\lambda\kappa;\gamma\lambda\kappa}~.
\end{eqnarray}
In the same way one can derive for $\mathcal{L}=\mathcal{T}_2$ that
\begin{eqnarray}
\label{T2_down}
\mathcal{T}^\dagger_2(A)_{\alpha\beta;\gamma\delta} &=& \frac{1}{2(N-1)}\left[\delta_{\beta\delta}\tilde{\tilde{A}}_{\alpha\gamma} - \delta_{\alpha\delta}\tilde{\tilde{A}}_{\beta\gamma} - \delta_{\beta\gamma}\tilde{\tilde{A}}_{\alpha\delta} + \delta_{\alpha\gamma}\tilde{\tilde{A}}_{\beta\delta}\right] + \overline{A}_{\alpha\beta;\gamma\delta}\\
\nonumber&&-\left[\tilde{A}_{\delta\alpha;\beta\gamma} - \tilde{A}_{\delta\beta;\alpha\gamma} - \tilde{A}_{\gamma\alpha;\beta\delta} + \tilde{A}_{\gamma\beta;\alpha\delta}\right]~,
\end{eqnarray}
where $A$ is a matrix on two-particle-one-hole space and
\begin{eqnarray}
\tilde{\tilde{A}}_{\alpha\gamma} &=& \sum_{\lambda\kappa}A_{\lambda\kappa\alpha;\lambda\kappa\gamma}~,\\
\overline{A}_{\alpha\beta;\gamma\delta} &=& \sum_{\lambda}A_{\alpha\beta\lambda;\gamma\delta\lambda}~,\\
\tilde{A}_{\alpha\beta;\gamma\delta} &=& \sum_{\lambda}A_{\lambda\alpha\beta;\lambda\gamma\delta}~.
\end{eqnarray}
The Hermitian adjoint of $\mathcal{T}_2'$ is slightly more involved. We construct the adjoint by demanding that:
\begin{equation}
\mathrm{Tr}~\mathcal{T}_2'(\Gamma) A =
\mathrm{Tr}~\left[
    \begin{pmatrix}
    \mathcal{T}_2 \left(\Gamma\right) & \omega \\
       \omega^\dagger & \rho
       \end{pmatrix}
       \begin{pmatrix}
       A_\mathcal{T} & A_\omega \\
          A_\omega^\dagger & A_\rho
   \end{pmatrix}\right]
   = \mathrm{Tr}~\mathcal{T}_2'^\dagger(A)\Gamma~,
\end{equation}
is fulfilled. This leads to the following expression for the $\mathcal{T}_2'^\dagger$ map:
\begin{eqnarray}
     \mathcal{T}_2'^\dagger \left( A \right)_{\alpha\beta;\gamma\delta} &=& 
\mathcal{T}_2^\dagger(A_\mathcal{T}) +
    \left( A_\omega \right)_{\alpha\beta\delta;\gamma} + \left( A_\omega \right)_{\gamma\delta\beta;\alpha} -  \left( A_\omega \right)_{\alpha\beta\gamma;\delta} - \left( A_\omega \right)_{\gamma\delta\alpha;\beta} \nonumber \\
 & & + \frac{1}{N-1}  \left( \delta_{\beta\delta}\left(A_\rho\right)_{\gamma\alpha} - \delta_{\alpha\delta}\left(A_\rho\right)_{\gamma\beta} - \delta_{\beta\gamma}\left(A_\rho\right)_{\delta\alpha}  + \delta_{\alpha\gamma}\left(A_\rho\right)_{\delta\beta}  \right).
    \label{T2P_down}
\end{eqnarray}
\section{Primal and dual semidefinite programs}
The general form of a semidefinite program \cite{vandenberghe}, in its primal formulation, is given by:
\begin{equation}
\max_X -\mathrm{Tr}~X~u^0 \qquad\text{u.c.t.}\qquad X\succeq0\qquad\text{and}\qquad \mathrm{Tr}~Xu^i = h^i~,
\label{primal_sdp}
\end{equation}
in which a matrix $X$ is varied to optimize a linear function $(-\mathrm{Tr}~Xu^0)$ under the constraint that it remains positive semidefinite ($X \succeq 0$), and under a set of equality constraints ($ \mathrm{Tr}~Xu^i = h^i$). This problem is completely defined by the set of matrices $u^\alpha = \{u^0,u^i\}$, and a vector $h$. As with many constrained optimization problems, one can define a dual optimization problem using the Lagrangian of the original problem \cite{boyd}:
\begin{equation}
\min_\gamma \sum_i \gamma_i h^i \qquad\text{u.c.t.}\qquad Z = u^0 + \sum_i\gamma_i~u^i\succeq0~.
\label{dual_sdp}
\end{equation}
In this case a matrix function of the dual variable $\gamma$ has to remain positive definite. An important feature of this duality is that the primal objective function always bounds the dual from below (and vice-versa), when both primal and dual variables satisfy all equality and inequality constraints. This can be inferred from:
\begin{equation}
\sum_i \gamma_ih^i + \mathrm{Tr}~Xu^0  = \sum_i \gamma_i \mathrm{Tr}~Xu^i + \mathrm{Tr}Xu^0 = \mathrm{Tr}~XZ \geq 0~,
\end{equation}
because both $Z$ and $X$ are positive semidefinite. When the primal and dual problems are getting closer to their optimal values, they will move towards each other. We call the difference between the primal and dual optimal value the primal-dual gap:
\begin{equation}
\eta = \mathrm{Tr}~XZ~.
\label{pd_gap}
\end{equation}
It has been shown that (when both primal and dual feasible regions\footnote{The feasible region consists of the set of variables for which all the equality and inequality constraints are satisfied.} aren't empty) the primal-dual gap vanishes at the optimal value of $X$ and $\gamma$ \cite{nesterov}. Because $X$ and $Z$ are also positive semidefinite at their optimum, it follows that the much stronger \emph{complementary slackness} condition holds when both $X$ and $Z$ are optimal:
\begin{equation}
XZ = 0 ~.
\label{compl_slack}
\end{equation}
\section{Formulation of v2DM as a semidefinite program}
In v2DM we want to optimize the energy by varying a matrix, the 2DM, under the constraints that it has the right particle number, and that some linear matrix maps of the 2DM are positive semidefinite, {\it i.e.}
\begin{eqnarray}
\label{v2DM}E^N_{\text{SDP}}\left(H_\nu\right) &=& \min_{\Gamma} \mathrm{Tr}~\left[\Gamma H^{(2)}_\nu\right]\\
\text{u.c.t.}&&\left\{
   \begin{array}{l}
   \mathrm{Tr}~\Gamma = \frac{N(N-1)}{2}~,\\
   \mathcal{L}_j\left(\Gamma\right) \succeq 0~.
   \end{array}
   \right.
   \label{2DM_constraints}
\end{eqnarray}
Here the $\mathcal{L}_j$ are a collection of constraints as defined in Section~\ref{standard_n_rep}. This is obviously a very similar problem to the SDP introduced in the previous Section. The explicit connection is worked out in this Section for both the primal and dual formulation.
\subsection{Primal formulation}
In a primal SDP one varies over a matrix $X$, which has to be positive semidefinite. It is therefore reasonable to have $X$ as a blockmatrix, with the different constraint matrices as blocks:
\begin{equation}
X = \bigoplus_j X_{\mathcal{L}_j} \succeq 0~.
\end{equation}
From now on we drop the index $j$ on the constraints $\mathcal{L}_j$ when talking about a general constraint different from $\mathcal{I}$.
The matrix $X$ contains free-ranging variables, so we need to put linear constraints on these variables to make sure that, for all matrix inequalities $\mathcal{L}$:
\begin{equation}
X_{\mathcal{L}} = \mathcal{L}\left(X_\mathcal{I}\right)~.
\label{primal_equalities}
\end{equation}
This can be done using the linear constraints of the primal SDP. If we introduce a complete and orthonormal basis of the constraint-matrix spaces, $\{g_\mathcal{L}^i\}$, we can reformulate the linear equalities that have to be fulfilled as (where the trace is on $\mathcal{L}$-space):
\begin{equation}
\mathrm{Tr}~X_{\mathcal{L}}g_\mathcal{L}^i = \mathrm{Tr}~\mathcal{L}\left(X_\mathcal{I}\right)~g_\mathcal{L}^i~,
\end{equation}
or, using the Hermitian adjoint maps defined in the previous Section:
\begin{equation}
\mathrm{Tr}~X_{\mathcal{L}}g_\mathcal{L}^i = \mathrm{Tr}~X_\mathcal{I}~\mathcal{L}^\dagger\left(g_\mathcal{L}^i\right)~.
\end{equation}
The equalities (\ref{primal_equalities}) can now be imposed in a standard primal SDP fashion:
\begin{equation}
\mathrm{Tr}~Xu_\mathcal{L}^i = 0~,
\end{equation}
in which $u_\mathcal{L}^i$ is a blockmatrix with two non-zero blocks, the $\mathcal{I}$ and the $\mathcal{L}$ block:
\begin{eqnarray}
\left(u^i_\mathcal{L}\right)_\mathcal{I} &=& -\mathcal{L}^\dagger\left(g^i_\mathcal{L}\right)~,\\
\left(u^i_\mathcal{L}\right)_\mathcal{L} &=& g^i_\mathcal{L}~.
\end{eqnarray}
As an example, when the active constraints are $\mathcal{IQG}$, the $u_\mathcal{G}^i$ matrices which impose the linear constraints on the $X_\mathcal{G}$ matrix are:
\begin{equation}
u^i_\mathcal{G} = 
\left(
\begin{matrix}
-\mathcal{G}^\dagger\left(g^i_\mathcal{G}\right) & 0 & 0\\
0&0&0\\
0&0&g^i_\mathcal{G}
\end{matrix}
\right)~.
\end{equation}
The particle number constraint can easily be written in the standard form:
\begin{equation}
\mathrm{Tr}~Xu^{\mathrm{Tr}}=\frac{N(N-1)}{2}~,
\end{equation}
where only the $\mathcal{I}$ block of $u^\mathrm{Tr}$ is non-zero, and equal to the unit matrix on two-particle space:
\begin{equation}
[(u^\mathrm{Tr})_\mathcal{I}]_{\alpha\beta;\gamma\delta} = \mathbb{1}_{\alpha\beta;\gamma\delta} = \delta_{\alpha\gamma}\delta_{\beta\delta} - \delta_{\alpha\delta}\delta_{\beta\gamma}~.
\label{unit}
\end{equation}
For the energy to be optimized we only have to set the $\mathcal{I}$ block of $u^0$ equal to the Hamiltonian and the rest equal to zero. Now we have everything to express v2DM as a primal SDP, and use standard SDP algorithms to solve it. This primal formulation was used by M. Nakata \cite{nakata_first} in his pioneering article on variational density matrix optimization.
\subsection{Dual formulation}
To express v2DM as a dual SDP it is useful to introduce a complete, orthonormal basis of traceless two-particle matrix space $\{f^i\}$, which satisfies the following relationships:
\begin{eqnarray}
\mathrm{Tr}~f^i &=& 0~,\\
\mathrm{Tr}~f^if^j &=& \delta^{ij}~,\\
f^i_{\alpha\beta;\gamma\delta} = -f^i_{\beta\alpha;\gamma\delta} &=& -f^i_{\delta\gamma;\alpha\beta} = f^i_{\delta\gamma;\beta\alpha}~.
\end{eqnarray}
In this basis every 2DM can be decomposed as:
\begin{equation}
\Gamma = \left(\frac{2\mathrm{Tr}~\Gamma}{M(M-1)}\right)\mathbb{1} + \sum_i \mathrm{Tr}~\left[\Gamma f^i\right] f^i~.
\label{f_decomp}
\end{equation}
If we define:
\begin{equation}
\gamma_i = \mathrm{Tr}~\Gamma f^i \qquad\text{and}\qquad h^i = \mathrm{Tr}~H^{(2)}f^i~,
\end{equation}
the energy expression can be written as:
\begin{equation}
E = \mathrm{Tr}~\Gamma H^{(2)} = \frac{2\left(\mathrm{Tr}~\Gamma\right)\left(\mathrm{Tr}~H^{(2)}\right)}{M(M-1)} + \sum_i \gamma_i h^i~.
\end{equation}
When the particle number is fixed, one can see that the only freedom one has to optimize the energy is in the traceless part of the matrix. This means that optimizing over the matrix $\Gamma$ under the particle number constraint is equivalent to:
\begin{equation}
\min_\gamma \sum_i \gamma_i h^i~.
\end{equation}
Since the matrix maps $\mathcal{L}$ are linear and homogeneous in $\Gamma$, one can write:
\begin{equation}
\mathcal{L}(\Gamma) = \frac{N(N-1)}{M(M-1)}\mathcal{L}(\mathbb{1}) + \sum_i \gamma_i \mathcal{L}(f^i)~.
\end{equation}
Defining:
\begin{equation}
u^0 = \frac{N(N-1)}{M(M-1)}\bigoplus_j \mathcal{L}_j\left(\mathbb{1}\right)\qquad\text{and}\qquad u^i = \bigoplus_j \mathcal{L}_j\left(f^i\right)~,
\label{u_def}
\end{equation}
it is clear that the constraint
\begin{equation}
Z = u^0 + \sum_i \gamma_i u^i \succeq 0~,
\label{dual_constr}
\end{equation}
is equivalent to constraining all the matrix maps to be positive definite. One can see that it is far more natural to write the v2DM problem as a dual than as a primal SDP. There are considerably less variables to be optimized, and the constraints are much easier to apply. In the rest of the thesis we therefore use the dual formulation.
\section{\label{int_point}Interior point methods}
The first class of algorithms to discuss are interior point methods \cite{nesterov}. These methods try to optimize the cost function while staying inside of the feasible region. Interior point methods are generally very stable and have nice convergence properties, but are computationally demanding because at every iteration a large linear system has to be solved. An important concept for interior point methods is the central path.
\subsection{The central path}
\begin{figure}
\centering
\includegraphics[scale=0.5]{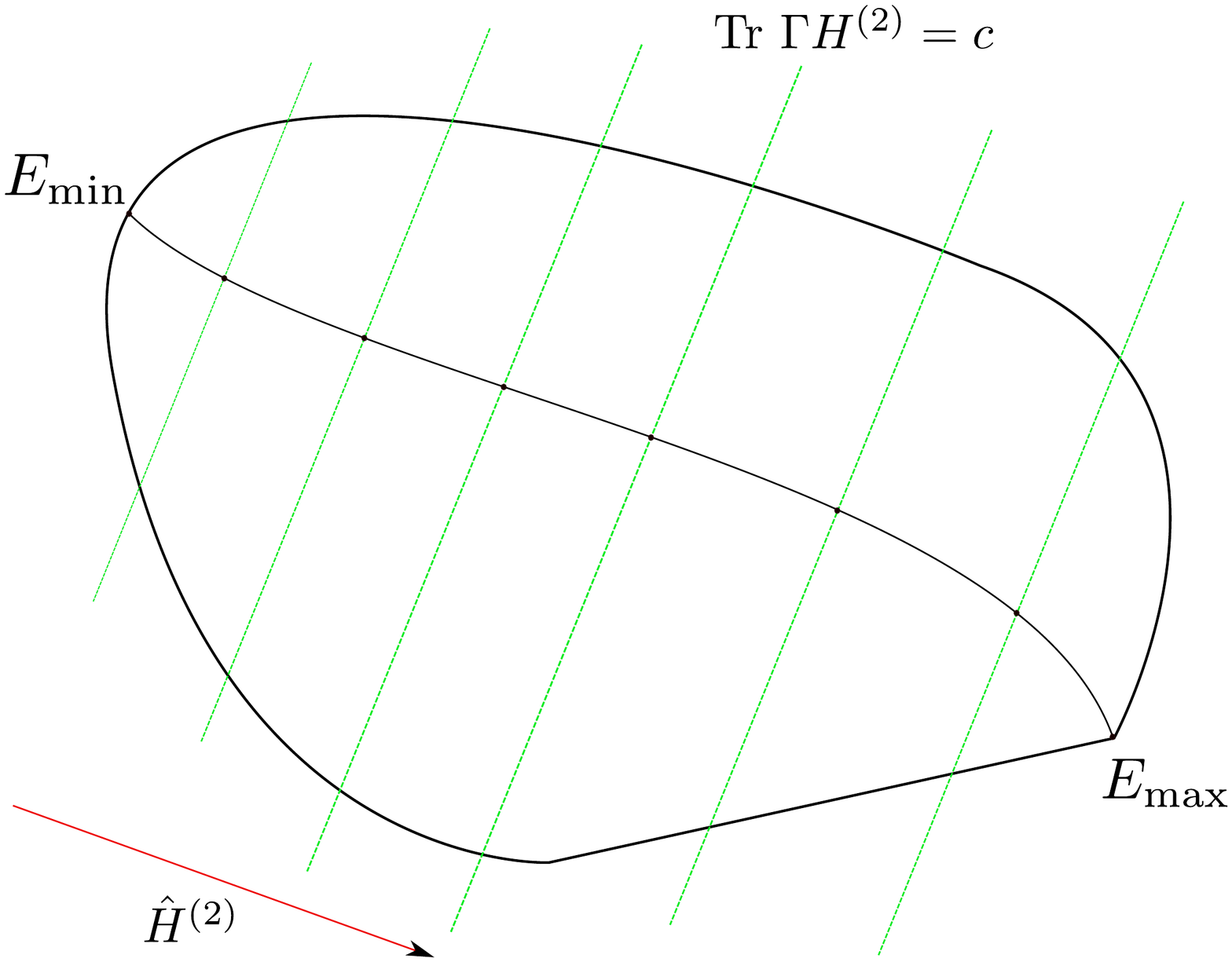}
\caption{\label{central_path} The central path is defined by the minima of the potential barrier $\phi$, in all Hamiltonian planes with allowed energies. One can see that both the maximal and minimal energy lie on the central path.}
\end{figure}
To define the central path we introduce a barrier function, which has the property of self-concordance\footnote{An $\mathbb{R}\rightarrow\mathbb{R}$ function $f(x)$ is self-concordant, if $|f'''(x)|\leq 2f''(x)^{\frac{3}{2}}$. A multi dimensional $\mathbb{R}^n\rightarrow \mathbb{R}$ function $g$ is self-concordant if the restriction of $g$ to any line segment is self-concordant. This property is primarily important in the convergence analysis of various algorithms.} \cite{self_concord}, defined on the feasible region:
\begin{equation}
\phi(\gamma) = -\ln{\det{Z(\gamma)}}~.
\label{barrier}
\end{equation}
This function is rather flat on the inside of the feasible region, becomes larger and larger when coming closer to the boundary and is $+\infty$ on the edges of the feasible region. The point that minimizes this potential is called the analytical center of the feasible region. For the following set of optimization problems:
\begin{equation}
\min_\gamma \phi(\gamma)\qquad\text{u.c.t.}\qquad  \sum_i \gamma_i h^i = e~,
\label{dual_optim_const_ener}
\end{equation}
the optima, for all allowed energies $E_{\text{min}} \leq e \leq E_{\text{max}}$, define the central path, as shown in Figure \ref{central_path}. Both the maximal and minimal energy obtainable lie on the central path. This is important because interior point methods try to follow the central path in their search for the optimum, since the barrier function becomes singular on the edge of the feasible region, which should be avoided as long a possible.

There are some nice properties associated with the central path of the primal and the dual problem. To see this, consider the optimality conditions for (\ref{dual_optim_const_ener}), which is that the gradient of Lagrangian of the problem has to vanish, {\it i.e.}:
\begin{equation}
\mathrm{Tr}~Z^{-1}u^i = \lambda h^i~,
\label{optimality_ce}
\end{equation}
with $\lambda$ some Lagrange multiplier. The fact that Eq.~(\ref{optimality_ce}) shows a remarkable similarity to the primal feasibility condition is no coincidence. The matrix $Z^{-1}/\lambda$ is not only primal feasible, but can be shown to be the solution to the following optimization problem:
\begin{equation}
\min_X -\ln\det X \qquad\text{u.c.t.}\qquad\mathrm{Tr}~Xu^i = h^i\qquad\text{and}\qquad -\mathrm{Tr}~Xu^0 = \eta - \frac{n}{\lambda}~,
\end{equation}
with $n$ the dimension of the matrix $X$. This means there is a pairing between points on the central path of the primal and the dual problem. For every point on the dual central path $Z_c$, there is one on the primal central path $X_c$, and up to a scale factor they are each others inverse:
\begin{equation}
X_c Z_c = \frac{\eta}{n}\mathbb{1}~.
\end{equation}
When we have reached to optimum, the primal-dual gap vanishes and we retain the complentary slackness condition:
\begin{equation}
X_c^* Z_c^* = 0~.
\end{equation}
\subsection{\label{pr_sdp}Dual-only potential reduction method}
This is the most transparent interior point algorithm, and it only considers the dual problem. It is a conceptually simple algorithm that is very flexible and easy to adapt to the specific structure of the problem at hand. The idea is to use a potential barrier function to impose the positive semidefiniteness of $Z$. If we minimize the following potential:
\begin{equation}
\phi(\gamma) = \sum_i \gamma_i h^i - t\ln\det Z(\gamma)~,
\label{dual_only_pot}
\end{equation}
for a certain value of $t$, the optimal value lies where the two competing terms in Eq.~(\ref{dual_only_pot}) are in balance. The energy term tries to move $\gamma$ in the negative direction of $h$, but as it gets closer to the edge of the feasible region, the potential gets steeper. The parameter $t$ decides where this balance occurs, and the smaller it becomes, the closer to the edge the optimum is. As a function of $t$, the optimum $\gamma^*(t)$ lies on the central path. For large $t$, $\gamma^*(t)$ lies around the analytic center, for $t\rightarrow0$, $\gamma^*(t)$ lies on the edge, which is exactly the solution to the dual SDP. The idea of the potential reduction method is to solve Eq.~(\ref{dual_only_pot}) for a large enough value of $t$, and use the solution as a starting point for a subsequent optimization with a smaller value of $t$. This is repeated until $t$ is small enough, and we have converged to a point close to the edge.
\subsubsection{Solution for fixed penalty}
To solve the optimization problem Eq.~(\ref{dual_only_pot}) we use the Newton-Raphson method. This method approximates the non-linear potential by its second-order Taylor expansion:
\begin{equation}
\phi(\gamma_0 + \delta\gamma) \approx \phi(\gamma_0) + \sum_i \delta\gamma_i \underbrace{\left(\frac{\partial \phi}{\partial \gamma_i}\right)}_{:= \nabla\phi^i} + \frac{1}{2}\delta\gamma_i\underbrace{\left(\frac{\partial^2 \phi}{\partial \gamma_i \partial \gamma_j}\right)}_{:=\mathcal{H}^{ij}}\delta\gamma_j~,
\label{newton}
\end{equation}
and then searches for the step $\delta\gamma$ that minimizes this expansion. The optimality conditions for Eq.~(\ref{newton}) lead to the following system of linear equations to determine $\delta\gamma$:
\begin{equation}
\sum_j \mathcal{H}^{ij} \delta\gamma_j = -\nabla\phi^i~.
\label{newt_eq}
\end{equation}
When sufficiently close to the optimum, the Newton method is known to converge quadratically. The only problem left is the solution of the linear equations (\ref{newt_eq}). 

\paragraph{The gradient:} the gradient of Eq.~(\ref{dual_only_pot}) reads:
\begin{equation}
\nabla\phi^i = h^i - t\mathrm{Tr}~Z^{-1}u^i~,
\label{formal_gradient}
\end{equation}
and can be rewritten as:
\begin{eqnarray}
\nonumber \nabla\phi^i &=& h^i - t\sum_j\mathrm{Tr}~\left[\mathcal{L}_j(\Gamma)^{-1}\mathcal{L}_j(f^i)\right]\\
&=& \mathrm{Tr}~\left[\left(H^{(2)} - t\sum_j \mathcal{L}_j^\dagger\left[\mathcal{L}_j\left(\Gamma\right)^{-1}\right]\right)f^i\right]~.
\end{eqnarray}
This means we can express the gradient as a matrix, without any reference to the basis $\{f^i\}$:
\begin{equation}
\nabla\phi = \hat{P}_{\mathrm{Tr}}\left(H^{(2)} - t\sum_j \mathcal{L}_j^\dagger\left[\mathcal{L}_j\left(\Gamma\right)^{-1}\right]\right)~,
\label{gradient}
\end{equation}
where $\hat{P}_{\mathrm{Tr}}$ is the projection on traceless 2DM-space:
\begin{equation}
\hat{P}_\mathrm{Tr}(A) = A - \frac{2\mathrm{Tr}~A}{M(M-1)}\mathbb{1}~.
\label{proj_Tr}
\end{equation}
\paragraph{The Hessian:}
the Hessian is the matrix formed by the second derivatives of Eq.~(\ref{dual_only_pot}):
\begin{equation}
\mathcal{H}^{ij} = t\mathrm{Tr}~\left[Z^{-1}u^i Z^{-1}u^j\right] = t \sum_k\mathrm{Tr}~\left[\mathcal{L}_k(\Gamma)^{-1}\mathcal{L}_k(f^i)\mathcal{L}_k(\Gamma)^{-1}\mathcal{L}_k(f^j)\right] ~.
\label{hessian_structure}
\end{equation}
The action of the Hessian on any traceless matrix $\Delta = \sum_i \delta_i f^i$ can be expressed as:
\begin{eqnarray}
\nonumber\sum_j\mathcal{H}^{ij}\delta_j &=& t \sum_k\mathrm{Tr}~\left[\mathcal{L}_k(\Gamma)^{-1}\left(\sum_j\mathcal{L}_k(f^j)\delta_j\right)\mathcal{L}_k(\Gamma)^{-1}\mathcal{L}_k(f^i)\right] \\
&=&t\sum_k \mathrm{Tr}~\left[\mathcal{L}^\dagger_k\left(\mathcal{L}_k(\Gamma)^{-1}\mathcal{L}_k(\Delta)\mathcal{L}_k(\Gamma)^{-1}\right)f^i\right]~.
\end{eqnarray}
It follows that the action of the Hessian on a traceless matrix $\Delta$ can be written as another matrix:
\begin{equation}
\mathcal{H}\Delta = t \hat{P}_{\mathrm{Tr}}\left[\sum_k \mathcal{L}^\dagger_k\left(\mathcal{L}_k(\Gamma)^{-1}\mathcal{L}_k(\Delta)\mathcal{L}_k(\Gamma)^{-1}\right)\right]~.
\label{hessian}
\end{equation}
It is important to notice that Eq.~(\ref{hessian}) constitutes a very efficient matrix-vector product for the linear system Eq.~(\ref{newt_eq}). If the dimension of single-particle space is $M$, a Hessian matrix-vector product would scale normally as $M^8$. The special structure of the Hessian for physical problems implies that the matrix-vector product in Eq.~(\ref{hessian}) scales as only $M^6$. In addition, the Hessian is positive definite and symmetric, and this makes the linear conjugate gradient method an attractive way to solve the linear system iteratively \cite{jrs}. In this way we do not need to construct and store the Hessian, but just use it through its action on matrices. We also have no need to choose an explicit basis $\{f^i\}$ because everything can be expressed in terms of matrices.
\paragraph{The line search:}
once we have found the direction $\Delta$ that minimizes the quadratic approximation of the potential (\ref{dual_only_pot}), we can speed up the convergence of the Newton method by minimizing (\ref{dual_only_pot}) as a function of the steplength $\alpha$ taken in this direction, {\it i.e.} find the $\alpha$ for which:
\begin{equation}
\nabla_\alpha\phi(\alpha) = \frac{ \partial }{\partial \alpha}\phi(\gamma + \alpha\delta\gamma)=0~.
\label{optim_ls}
\end{equation}
This function can be calculated analytically:
\begin{equation}
\nabla_\alpha\phi(\alpha) = \mathrm{Tr}~\Delta H^{(2)} - t \sum_j\mathrm{Tr}~\left[\mathcal{L}_j\left(\Gamma + \alpha\Delta\right)^{-1}\mathcal{L}_j(\Delta)\right]~,
\label{optim_ls_matrix}
\end{equation}
and evaluated for any $\alpha$ (given a 2DM $\Gamma$ and a direction $\Delta$). Every evaluation of this function involves inverting a matrix, which means this is quite slow. There is a way to simplify Eq.~(\ref{optim_ls_matrix}) by solving the generalized eigenvalue problem:
\begin{equation}
\mathcal{L}(\Delta) w = \lambda^\mathcal{L} \mathcal{L}(\Gamma) w~.
\label{gen_eig_ls}
\end{equation}
This can be transformed to a normal symmetric eigenvalue problem with real eigenvalues and orthogonal eigenvectors:
\begin{equation}
\left(\mathcal{L}(\Gamma)^{-\frac{1}{2}}\mathcal{L}(\Delta)\mathcal{L}(\Gamma)^{-\frac{1}{2}}\right) v = \lambda^\mathcal{L} v~\qquad\text{with}\qquad v=\mathcal{L}(\Gamma)^{\frac{1}{2}}w~.
\end{equation}
The completeness of the eigenvectors $v$ of the above problem, implies that:
\begin{eqnarray}
\label{L_Gamma_ls}\mathcal{L}(\Gamma) &=& \mathcal{L}(\Gamma)^{\frac{1}{2}}\left(\sum_i v_i v_i^T\right)\mathcal{L}(\Gamma)^{\frac{1}{2}}~,\\
\label{L_Delta_ls}\mathcal{L}(\Delta) &=& \mathcal{L}(\Gamma)^{\frac{1}{2}}\left(\sum_i \lambda^\mathcal{L}_i v_i v_i^T\right)\mathcal{L}(\Gamma)^{\frac{1}{2}}~.
\end{eqnarray}
Using Eqs.~(\ref{L_Gamma_ls}) and (\ref{L_Delta_ls}) it follows that:
\begin{equation}
\mathcal{L}\left(\Gamma + \alpha\Delta\right) = \mathcal{L}(\Gamma)^{\frac{1}{2}}\left[\sum_i (1 + \alpha\lambda^\mathcal{L}_i)v_iv_i^T \right]\mathcal{L}(\Gamma)^{\frac{1}{2}}~,
\end{equation}
and we can express the matrixproduct term in Eq.~(\ref{optim_ls_matrix}) as:
\begin{eqnarray}
\nonumber\mathcal{L}\left(\Gamma + \alpha\Delta\right)^{-1}\mathcal{L}(\Delta) &=& \mathcal{L}(\Gamma)^{-\frac{1}{2}}\left[\sum_i \left(\frac{1}{1 +\alpha\lambda^\mathcal{L}_i}\right)v_i v_i^T\right]\left[\sum_j \lambda^\mathcal{L}_j v_j v_j^T\right]\mathcal{L}(\Gamma)^{\frac{1}{2}}\\
\label{matrixprod_ls}&=& \mathcal{L}(\Gamma)^{-\frac{1}{2}}\left[\sum_i\left(\frac{\lambda^\mathcal{L}_i}{1 +\alpha\lambda^\mathcal{L}_i}\right)v_iv_i^T\right]\mathcal{L}(\Gamma)^{\frac{1}{2}}~,
\end{eqnarray}
in which the orthonormality of the eigenvectors has been used. The trace of Eq.~(\ref{matrixprod_ls}) can therefore be expressed purely in terms of the eigenvalues $\lambda^\mathcal{L}$:
\begin{equation}
\mathrm{Tr}~\left[\mathcal{L}\left(\Gamma + \alpha\Delta\right)^{-1}\mathcal{L}(\Delta)\right] = \sum_i\frac{\lambda^\mathcal{L}_i}{1 +\alpha\lambda^\mathcal{L}_i}~.
\end{equation}
It is now clear that Eq.~(\ref{optim_ls_matrix}) can be expressed in terms of the eigenvalues of Eq.~(\ref{gen_eig_ls}):
\begin{equation}
\nabla_\alpha\phi(\alpha) = \mathrm{Tr}~\Delta H^{(2)} - t \sum_j\left(\sum_i\frac{\lambda_i^{\mathcal{L}_j}}{1 + \alpha\lambda^{\mathcal{L}_j}_i}\right)~.
\label{lsfunc}
\end{equation}
Once the generalized eigenvalue problem Eq.~(\ref{gen_eig_ls}) is solved, one can evaluate (\ref{lsfunc}) as a \emph{scalar} function for any value of $\alpha$, and the bisection method to solve Eq.~(\ref{optim_ls}) becomes very efficient.
\paragraph{Optimal value for fixed $t$:}
the optimal value of the optimization problem for fixed $t$ is located where the gradient is zero:
\begin{equation}
\nabla\phi^i = h^i - t\mathrm{Tr}~\left[Z^{-1}u^i\right] = 0~.
\end{equation}
This point lies on the central path, and we see that the matrix:
\begin{equation}
X = tZ^{-1}~,
\end{equation}
is primal feasible. Based on Eq.~(\ref{optimality_ce}), we can estimate that the primal-dual gap, when we have optimized for a certain value of $t$, is given by:
\begin{equation}
\eta = \mathrm{Tr}~ZX = nt~,
\label{pd_gap_dual_only}
\end{equation}
and use this number as a convergence criterion in the program.
\subsubsection{Outline of the algorithm}
To summarize the different aspects of the method exposed in the previous paragraphs a schematic outline of the algorithm is given below:
\begin{algorithm}
\caption{\label{pr_algo} The dual-only potential reduction algorithm}
\begin{algorithmic}
\State Choose $\epsilon > 0~$;$~\epsilon_{\text{new}} > 0$;$~0 < \beta < 1$ 
\State $\Gamma = \frac{N(N-1)}{M(M-1)}\mathbb{1}$;~$t = 1$
\While{$nt > \epsilon$}
\While{$\delta_{\text{new}} > \epsilon_{\text{new}}$}\Comment{Newton-Raphson loop}
\State $\nabla\phi = \hat{P}_{\mathrm{Tr}}\left(H^{(2)} - t\sum_j \mathcal{L}_j^\dagger\left[\mathcal{L}_j\left(\Gamma\right)^{-1}\right]\right)~$
\State Solve $\mathcal{H}\Delta = -\nabla\phi$ for $\Delta$~\Comment{Linear Conjugate gradient method}
\State Diagonalize $\mathcal{L}_i(\Gamma)^{-\frac{1}{2}}\mathcal{L}_i(\Delta)\mathcal{L}_i(\Gamma)^{-\frac{1}{2}}\rightarrow \lambda^{\mathcal{L}_i}_j$
\State Solve $\mathrm{Tr}~\Delta H^{(2)} - t \sum_j\left(\sum_i\frac{\lambda_i^{\mathcal{L}_j}}{1 + \alpha\lambda^{\mathcal{L}_j}_i}\right) = 0$ for $\alpha$
\State $\Gamma \gets \Gamma + \alpha\Delta$
\State $\delta = \alpha \|\Delta\|$
\EndWhile
\State $t \gets \beta t$
\EndWhile
\end{algorithmic}
\end{algorithm}
\subsubsection{\label{linineq_do}Adding linear inequality constraints}
Suppose we want to add a number of additional constraints of the form:
\begin{equation}
\mathrm{Tr}~\Gamma C_i \geq c_i \qquad\text{for}\qquad i = 1,\ldots,m~.
\label{stand_linineq}
\end{equation}
How can we fit these constraints into the formalism derived in the previous Sections? For the dual-only potential reduction algorithm it turns out to be quite straightforward. We can rewrite Eq.~(\ref{stand_linineq})
\begin{equation}
L(\Gamma)_{ii} = \mathrm{Tr}~\Gamma C^0_i = \mathrm{Tr}~\left[\Gamma\left(C_i - \frac{2c_i}{N(N-1)}\mathbb{1}\right)\right]\geq 0~,
\label{linineq}
\end{equation}
as a diagonal matrix of dimension $m$ that has to be positive semidefinite. The standard formulation of the dual SDP can still be used, and we just have to change the matrices $\{u^0,u^i\}$ in Eq.~(\ref{u_def}) to include the new constraints:
\begin{equation}
u^0 = \frac{N(N-1)}{M(M-1)}\left[\bigoplus_j \mathcal{L}_j\left(\mathbb{1}\right)\right]\bigoplus L(\mathbb{1})~,\qquad u^i = \left[\bigoplus_j \mathcal{L}_j\left(f^i\right)\right]\bigoplus L(f^i)~.
\label{u_li}
\end{equation}
The changes in the algorithm pertain to the calculation of the gradient, the Hessian map and the line search function.
\paragraph{The gradient:}
the gradient is still of the form Eq.~(\ref{formal_gradient}), but is now extended as:
\begin{equation}
\nabla\phi^i = h^i - t\sum_j\mathrm{Tr}~\left[\mathcal{L}_j(\Gamma)^{-1}\mathcal{L}_j(f^i)\right] - t\sum_jL(\Gamma)_{jj}^{-1}L(f^i)_{jj}~.
\end{equation}
The last term can be written as
\begin{equation}
L_{jj}(f^i) = \mathrm{Tr}~\left[C^0_j f^i\right]~,
\end{equation}
so the matrix representation of the gradient reads:
\begin{equation}
\nabla\phi = \hat{P}_{\mathrm{Tr}}\left(H^{(2)} - t\sum_j \mathcal{L}_j^\dagger\left[\mathcal{L}_j\left(\Gamma\right)^{-1}\right] - t\sum_j L(\Gamma)^{-1}_{jj}C^0_j\right)~.
\end{equation}
\paragraph{The Hessian:}
when linear constraints are added, the Hessian matrix gets the form:
\begin{equation}
\mathcal{H}^{ij} = t \sum_k\mathrm{Tr}~\left[\mathcal{L}_k(\Gamma)^{-1}\mathcal{L}_k(f^i)\mathcal{L}_k(\Gamma)^{-1}\mathcal{L}_k(f^j)\right] + t\sum_kL(\Gamma)^{-2}_{kk} L(f^i)_{kk}L(f^j)_{kk}~.
\end{equation}
The action of the Hessian on a traceless matrix is now:
\begin{equation}
\sum_j\mathcal{H}^{ij}\delta_j = t \sum_k\mathrm{Tr}~\left[\mathcal{L}_k(\Gamma)^{-1}\mathcal{L}_k(\Delta)\mathcal{L}_k(\Gamma)^{-1}\mathcal{L}_k(f^i)\right] + t\sum_k\left(\frac{L(\Delta)_{kk}}{L(\Gamma)_{kk}^2}\right)L(f^i)_{kk}~,
\end{equation}
which can be written as a matrix:
\begin{equation}
\mathcal{H}\Delta = t \hat{P}_{\mathrm{Tr}}\left[\sum_k \mathcal{L}^\dagger_k\left(\mathcal{L}_k(\Gamma)^{-1}\mathcal{L}_k(\Delta)\mathcal{L}_k(\Gamma)^{-1}\right) + \sum_k \left(\frac{L(\Delta)_{kk}}{L(\Gamma)^2_{kk}}\right)C^0_k\right]~.
\label{hessian_li}
\end{equation}
\paragraph{The line search:}
the terms that are added to the line search function (\ref{optim_ls}) are already scalar, and if we define:
\begin{equation}
\lambda^L_i = \frac{L(\Delta)_{ii}}{L(\Gamma)_{ii}}~,
\end{equation}
the added terms are of exactly the same form as the original ones, and the new line search function becomes:
\begin{equation}
\nabla_\alpha\phi(\alpha) = \mathrm{Tr}~\Delta H^{(2)} - t \sum_j\left(\sum_i\frac{\lambda_i^{\mathcal{L}_j}}{1 + \alpha\lambda^{\mathcal{L}_j}_i}\right) -  t\sum_i \frac{\lambda^L_{i}}{1+\alpha\lambda^L_{i}}~.
\end{equation}
\subsubsection{\label{pr_sdp_le}Adding linear equality constraints}
Adding linear equality constraints of the form:
\begin{equation}
\mathrm{Tr}~\Gamma E_n = e_n~,
\label{equality_constr}
\end{equation}
can be necessary when we want to restrict the optimization to a certain spin expectation value, or if we want to impose some symmetry on the 2DM. It is relatively easy to include this in the algorithm: we just have change the $u^0$ and $u^i$ matrices slightly, and define a new orthogonal basis $\{f^i\}$ which is not only traceless, but also orthogonal to all the $\{E_n\}$'s. For this we first introduce a new orthogonal set of matrices $\{\tilde{E}_n\}$, which is a basis for the space spanned by the $\{E_n'\}=\{\mathbb{1},E_n\}$:
\begin{equation}
\tilde{E}_n = \sum_{nm}(S^{-\frac{1}{2}})_{nm}E'_m\qquad\text{with}\qquad S_{nm} = \mathrm{Tr}~E'_nE'_m~.
\end{equation}
Applying the same transformation to the corresponding scalars $\{e'_n\}=\{\frac{N(N-1)}{2},e_n\}$:
\begin{equation}
\tilde{e}_n = \sum_{nm}(S^{-\frac{1}{2}})_{nm}e'_m~,
\end{equation}
we can expand every 2DM satisfying the equality constraints (\ref{equality_constr}) as:
\begin{equation}
\Gamma = \sum_n \tilde{e}_n \tilde{E}_n + \sum_i \gamma_i f^i~,
\end{equation}
in which the set of matrices $\{f^i\}$ form an orthonormal basis that spans the orthogonal complement of $\{\tilde{E}_n\}$. This can be represented as a dual SDP, in which the $u^i$ matrices are still defined as in Eq.~(\ref{u_def}), but with the new definition for the $f^i$'s. The $u^0$ matrix can be written as:
\begin{equation}
u^0 = \bigoplus_i \mathcal{L}_i(f^0)~,
\end{equation}
where $f^0$ is defined as:
\begin{equation}
f^0 = \sum_n \tilde{e}_n \tilde{E}_n~.
\end{equation}
The algorithm itself remains basically unaltered. The only things that change are the gradient and the Hessian map. As $\{f^i\}$ is now not only traceless but also orthogonal to the $\tilde{E}_n$'s, we have to replace the projection $\hat{P}_\mathrm{Tr}$ in Eqs.~(\ref{gradient}) and (\ref{hessian}) with a new projection operator:
\begin{equation}
\hat{P}_f(A) = A - \sum_n \mathrm{Tr}~\left[A\tilde{E}_n\right]\tilde{E}_n~.
\label{P_f}
\end{equation}
\subsection{\label{primal_dual_algo}Primal-dual predictor-corrector method}
The next algorithm solves both the dual and the primal problem at the same time \cite{primal_dual}. It is a much more complex algorithm than the dual-only method, and conceptually not so simple to describe. At first sight, it would seem superfluous to solve both the primal and the dual version of the problem, as the solution to the dual alone gives us everything we need. But as we will show, knowledge of the primal problem can speed up convergence for the dual, and vice-versa. By exploiting the specific structure of the linear matrix maps, we are able to solve the primal and the dual at almost the same cost as solving the dual alone.
The method we use is a path-following method, which means we won't optimize some potential, but try to follow the central path which leads us to the optimal value. Starting from a feasible primal-dual starting point $(X,Z)$, we try to find a new primal-dual point:
\begin{equation}
(Z+\Delta_Z)(X+\Delta_X) = \nu\frac{\eta}{n}\mathbb{1}~,\qquad\text{with}\qquad 0 \leq \nu \leq 1~,
\label{least_square_cp}
\end{equation}
that has a reduced primal-dual gap and is located on the central path. The problem with Eq.~(\ref{least_square_cp}) is that the left-hand side is not a symmetric matrix. There are several ways to symmetrize this equation, which lead to systems of linear equations that determine the primal-dual direction $(\Delta_X,\Delta_Z)$. We follow the approach introduced by J.~F.~Sturm \cite{sturm}.
\subsubsection{Equations of motion}
To derive the equations of motion for $(\Delta_X,\Delta_Z)$ we first consider how a semidefinite program transforms under general linear transformations $L$. If the primal matrix $X$ transforms as (where $L^{-T} = (L^{-1})^T$):
\begin{equation}
\overline{X} = L^TXL~,
\end{equation}
obviously the conditions
\begin{equation}
\overline{X} \succeq 0 \qquad\text{and}\qquad X \succeq 0~,
\end{equation}
are equivalent. If the objective function and the linear constraints are to remain the same, {\it i.e.}
\begin{equation}
\mathrm{Tr}~\overline{X}\overline{u}^0 = \mathrm{Tr}~{X}{u}^0~,\qquad\text{and}\qquad\mathrm{Tr}~\overline{X}\overline{u}^i =\mathrm{Tr}~Xu^i = h^i~,
\end{equation}
the matrices $\{u^\alpha\} = \{u^0,u^i\}$ have to transform as:
\begin{equation}
\overline{u}^\alpha = L^{-1}u^\alpha L^{-T}~.
\end{equation}
This means that the dual matrix transforms in the same way:
\begin{equation}
\overline{Z} = \overline{u}^0+\sum_i\gamma_i\overline{u}^i = L^{-1} Z L^{-T}~.
\end{equation}
As a result, the primal-dual gap remains invariant,
\begin{equation}
\mathrm{Tr}~\overline{X}\overline{Z} = \mathrm{Tr}~L^T{X}LL^{-1}{Z}L^{-T} = \mathrm{Tr}~XZ~,
\end{equation}
and the central path is mapped on itself:
\begin{equation}
\overline{X}\overline{Z} = L^T{X}LL^{-1}{Z}L^{-T} = \frac{\eta}{n}\mathbb{1}~.
\end{equation}
Now consider the transformation $L_d$ that maps $X$ and $Z$ to the same matrix $V$:
\begin{equation}
L_d^TXL_d = L_d^{-1}ZL_d^{-T} = V~.
\label{transfo_to_V}
\end{equation}
For an $L_d$ satisfying (\ref{transfo_to_V}) the following also holds:
\begin{equation}
L_dL_d^TXL_dL_d^T = Z~,
\end{equation}
from which it follows that:
\begin{equation}
\left(X^{\frac{1}{2}}L_dL_d^TX^{\frac{1}{2}}\right)^{2} = X^{\frac{1}{2}}ZX^{\frac{1}{2}}~.
\end{equation}
All of this implies that $L_d$ satisfies the following equation:
\begin{equation}
D(X,Z) :=  X^{-\frac{1}{2}}\left(X^{\frac{1}{2}}ZX^{\frac{1}{2}}\right)^{\frac{1}{2}}X^{-\frac{1}{2}}=  L_dL_d^T ~,
\label{D}
\end{equation}
and $L_d$ can be constructed by taking a Cholesky decomposition of $D(X,Z)$. Even more, when we right-multiply $L_d$ with an arbitrary orthogonal transformation the product still satisfies (\ref{D}). It follows from Eq.~(\ref{transfo_to_V}) that this freedom can be exploited to transform $V$ into a diagonal matrix. 
For a feasible primal-dual point $(X,Z)$, we can use an $L_d$ of the form (\ref{D}) to transform it into $(\overline{X},\overline{Z}) = (V,V)$. We can now ask what step $(D_X,D_Z)$ in this transformed space we have to take in order to satisfy:
\begin{equation}
(V+D_X)(V+D_Z) = \frac{\nu\eta}{n}\mathbb{1}~,\qquad\text{u.c.t.}\qquad \mathrm{Tr}~D_X\overline{u}^i = 0\quad\text{and}\quad D_Z = \sum_i d^Z_i \overline{u}^i~.
\end{equation}
Ignoring the quadratic term, this reduces to:
\begin{equation}
V^2+V(D_X+D_Z) = \frac{\nu\eta}{n}\mathbb{1},~\qquad\text{or}\qquad D_X+D_Z = \frac{\nu\eta}{n}V^{-1} - V~,
\label{transf_eom}
\end{equation}
which expresses the primal-dual step in the transformed space as a function of $V$. If we transform this back to the original space, we obtain two equivalent equations. By left-multiplying Eq.~(\ref{transf_eom}) with $L_d^{-T}$ and right-multiplying with $L_d^{-1}$ one obtains what we call the dual equation:
\begin{equation}
\Delta_X + D(X,Z)^{-1}\Delta_Z D(X,Z)^{-1} = \frac{\nu\eta}{n}Z^{-1} - X~.
\label{dual_eom}
\end{equation}
Multiplying to the left with $L_d$ and to the right with $L_d^T$ yields the primal equation:
\begin{equation}
\Delta_Z + D(X,Z)\Delta_X D(X,Z) = \frac{\nu\eta}{n}X^{-1} - Z~.
\label{primal_eom}
\end{equation}
In both equations the primal and dual step have to satisfy the equalities:
\begin{equation}
\mathrm{Tr}~\Delta_X u^i = 0\qquad\text{and}\qquad \Delta_Z = \sum_i \delta\gamma_iu^i~,
\label{pd_step_lineq}
\end{equation}
from which it follows that they are orthogonal to each other:
\begin{equation}
\mathrm{Tr}~\Delta_X \Delta_Z = 0~.
\label{pd_step_ortho}
\end{equation}
These equations can also been derived in a primal-dual potential-reduction approach, as shown by Nesterov and Todd in \cite{nest_todd}.
\subsubsection{\label{overlapmatrix}The overlap matrix}
The non-orthogonal matrices $\{u^\alpha\} = \{u^0,u^i\}$ as defined in Eq.~(\ref{u_def}), are linearly independent and they span a subspace of the total space the matrices $X$ and $Z$ inhabit. If we want to project a general matrix on the subspace spanned by the $u^\alpha$'s, we need an expression for the overlap matrix, which is defined as:
\begin{equation}
\mathcal{S}_{\alpha\beta} = \mathrm{Tr}~u^\alpha u^\beta~.
\label{def_overlap}
\end{equation}
Using the Hermitian adjoints of the linear maps $\mathcal{L}$ we can rewrite Eq.~(\ref{def_overlap}) as:
\begin{equation}
\mathcal{S}_{\alpha\beta} = \sum_k \mathrm{Tr}~\left[\mathcal{L}_k^\dagger\left(\mathcal{L}_k\left(f^\alpha\right)\right) f^\beta\right]~,
\label{expanded_overlap}
\end{equation}
in which $\{f^\alpha\}$ is the orthogonal basis of traceless 2DM space $\{f^i\}$, expanded with the normalized unity matrix $f^0$. From Eq.(\ref{expanded_overlap}) it follows that the overlap matrix can be seen as a linear map from two-particle space on itself, whose action on a two-particle matrix $\Gamma$ is:
\begin{equation}
\mathcal{S}\left(\Gamma\right) = \sum_k \mathcal{L}_k^\dagger\left(\mathcal{L}_k\left(\Gamma\right)\right)~.
\end{equation}
It turns out that this map can be written as a generalized $\mathcal{Q}$ map, which is defined as:
\begin{eqnarray}
\nonumber\mathcal{Q}(a,b,c)\left(\Gamma\right)_{\alpha\beta;\gamma\delta} &=& a\Gamma_{\alpha\beta;\gamma\delta} + b\left(\delta_{\alpha\gamma}\delta_{\beta\delta} - \delta_{\alpha\delta}\delta_{\beta\gamma}\right)\overline{\overline{\Gamma}}\\
&& - c\left(\delta_{\alpha\gamma}\overline{\Gamma}_{\beta\delta} - \delta_{\beta\gamma}\overline{\Gamma}_{\alpha\delta} - \delta_{\alpha\delta}\overline{\Gamma}_{\beta\gamma} + \delta_{\beta\delta}\overline{\Gamma}_{\alpha\gamma}\right)~.
\label{Q_like}
\end{eqnarray}
This is like the $\mathcal{Q}$ map introduced in (\ref{Q_2DM}) but with general coefficients $(a,b,c)$. The proof is somewhat tedious and proceeds by considering every $\mathcal{L}_k$ separately.
\paragraph{(1) $\mathcal{L}^\dagger_k\mathcal{L}_k = \mathcal{I}^2$:}
it is trivial to see that $\mathcal{I}^2(\Gamma) = \Gamma$ and that this is a generalized $\mathcal{Q}$ map with coefficients 
\begin{equation}
a = 1\qquad b = 0\qquad c = 0~.
\end{equation}
\paragraph{(2) $\mathcal{L}^\dagger_k\mathcal{L}_k = \mathcal{Q}^2$:}
to re-express $\mathcal{Q}^2$ we first calculate the various pieces,
\begin{eqnarray}
\overline{\mathcal{Q}}(\Gamma)_{\alpha\gamma} &=& \left[\frac{M-N-1}{N(N-1)}\right]\delta_{\alpha\gamma}\overline{\overline{\Gamma}} - \left[\frac{M-N-1}{N-1}\right]\overline{\Gamma}_{\alpha\gamma}~,\\
\overline{\overline{\mathcal{Q}}}(\Gamma) &=& \left[\frac{(M-N)(M-N-1)}{N(N-1)}\right]\overline{\overline{\Gamma}}~.
\end{eqnarray}
Substitution into Eq.~(\ref{Q_2DM}) leads once again to a generalized $\mathcal{Q}$ map with coefficients:
\begin{equation}
a = 1\qquad b = \frac{4N^2 + 2N - 4NM + M^2 -M}{N^2(N-1)^2}\qquad c = \frac{2N-M}{(N-1)^2}~.
\end{equation}
\paragraph{(3) $\mathcal{L}^\dagger_k\mathcal{L}_k = \mathcal{G}^\dagger\mathcal{G}$:}
using the same strategy one finds on the basis of Eq.~(\ref{G1_2DM}):
\begin{equation}
\overline{\mathcal{G}}(\Gamma)_{\alpha\gamma} = \frac{M-1}{N-1}\overline{\Gamma}_{\alpha\gamma}~.
\end{equation}
Substituting this into (\ref{G_down}) leads to another generalized $\mathcal{Q}$ map with coefficients:
\begin{equation}
a = 4\qquad b = 0 \qquad c = \frac{2N - M - 2}{(N-1)^2}~.
\end{equation}
\paragraph{(4) $\mathcal{L}^\dagger_k\mathcal{L}_k = \mathcal{T}_1^\dagger\mathcal{T}_1$:}
the terms needed by (\ref{T1_down}) are found by tracing (\ref{T1}):
\begin{eqnarray}
\overline{\mathcal{T}}_1\left(\Gamma\right)_{\alpha\beta;\gamma\delta} &=& (M-4)\Gamma_{\alpha\beta;\gamma\delta} + \left[\frac{M-N-2}{N(N-1)}\right]~(\delta_{\alpha\gamma}\delta_{\beta\delta} - \delta_{\alpha\delta}\delta_{\beta\gamma})\overline{\overline{\Gamma}}~\\
&&- \left[\frac{M-N-2}{N-1}\right]\hat{A}\left[\delta_{\alpha\gamma}\overline{\Gamma}_{\beta\delta} - \delta_{\beta\gamma}\overline{\Gamma}_{\alpha\delta}-\delta_{\alpha\delta}\overline{\Gamma}_{\beta\gamma} + \delta_{\beta\delta}\overline{\Gamma}_{\alpha\gamma}\right]~,\\
\overline{\overline{\mathcal{T}}}_1\left(\Gamma\right)_{\alpha\gamma} &=& \left[\frac{(M-N-2)(M-N-1)}{N(N-1)}\right]\delta_{\alpha\gamma}\overline{\overline{\Gamma}} - \left[\frac{(M-3)(M-2N)}{N-1}\right]~\overline{\Gamma}_{\alpha\gamma}~,\\
\overline{\overline{\overline{\mathcal{T}}}}_1\left(\Gamma\right) &=& \left[\frac{(M-2)(M(M-1) - 3N(M-N))}{N(N-1)}\right]\overline{\overline{\Gamma}}~,
\end{eqnarray}
and substitution into Eq.~(\ref{T1_down}) leads to the coefficients:
\begin{eqnarray*}
a &=& M-4~,\\
b &=& \frac{M^3-6M^2N-3M^2+12MN^2+12MN+2M-18N^2-6N^3}{3N^2(N-1)^2}~,\\
c &=& -\frac{M^2 + 2N^2 - 4MN - M + 8N - 4}{2(N-1)^2}~.
\end{eqnarray*}
\paragraph{(5) $\mathcal{L}^\dagger_k\mathcal{L}_k = \mathcal{T}_2^\dagger\mathcal{T}_2$:}
for the contribution of the $\mathcal{T}_2$ condition to the overlap matrix we need the terms derived from Eq.~(\ref{T2}):
\begin{eqnarray}
\overline{\mathcal{T}}_2\left(\Gamma\right)_{\alpha\beta;\gamma\delta} &=& \frac{\overline{\overline{\Gamma}}}{N - 1}(\delta_{\alpha\gamma}\delta_{\beta\delta} - \delta_{\alpha\delta}\delta_{\beta\gamma}) + M~ \Gamma~,\\
&&- \left[\delta_{\alpha\gamma}\overline{\Gamma}_{\beta\delta} - \delta_{\beta\gamma}\overline{\Gamma}_{\alpha\delta} - \delta_{\alpha\delta}\overline{\Gamma}_{\beta\gamma} + \delta_{\beta\delta}\overline{\Gamma}_{\alpha\gamma}\right]~,\\
\tilde{\mathcal{T}}_2\left(\Gamma\right)_{\alpha\beta;\gamma\delta} &=& \frac{M - N}{N - 1}\overline{\Gamma}_{\beta\delta}\delta_{\alpha\gamma} + \delta_{\beta\delta}\overline{\Gamma}_{\alpha\gamma} - (M - 2)\Gamma_{\alpha\delta;\gamma\beta}~,\\
\tilde{\tilde{\mathcal{T}}}_2\left(\Gamma\right)_{\alpha\gamma}&=& \left[\frac{M(M - N) - (N - 1)(M - 2)}{N - 1}\right]\overline{\Gamma}_{\alpha\gamma} + \delta_{\alpha\gamma}\overline{\overline{\Gamma}}~,
\end{eqnarray}
which, when substituted into Eq.~(\ref{T2_down}) give the following coefficients:
\begin{equation}
a = 5M - 8\qquad b = \frac{2}{N - 1}\qquad c = \frac{2N^2 + (M - 2)(4N-3) - M^2}{2(N - 1)^2}~.
\end{equation}
\paragraph{(6) $\mathcal{L}^\dagger_k\mathcal{L}_k = {\mathcal{T}_2}'^\dagger\mathcal{T}_2'$:}
finally, for the $\mathcal{T}_2'$ condition, the contribution to the overlap matrix is almost the same as for the regular $\mathcal{T}_2$. The extra terms in (\ref{T2P_down}) just add some simple terms to $a$ and $c$:
\begin{equation}
a = 5M - 4\qquad b = \frac{2}{N - 1}\qquad c = \frac{2N^2 + (M - 2)(4N-3) - M^2 - 2}{2(N - 1)^2}~.
\end{equation}

The overlap-matrix map is just the sum of the various terms obtained, and hence itself a generalized $\mathcal{Q}$ map, with rather complex coefficients.
\paragraph{Inverse map of a generalized $\mathcal{Q}$ map}
The inverse of a generalized $\mathcal{Q}$ map can be shown to be another generalized $\mathcal{Q}$ map. Consider for brevity the notation: 
\begin{equation}
\mathcal{Q}(a,b,c)(\Gamma) = Q~,
\end{equation}
then applying partial trace operations on Eq.~(\ref{Q_like}) leads to:
\begin{eqnarray}
\overline{\overline{\Gamma}} &=& \frac{\overline{\overline{Q}}}{a + M(M - 1)b - 2(M - 1)c}~,\\
\overline{\Gamma}_{\alpha\gamma} &=& \frac{1}{a - c(M - 2)}\left[\overline{Q}_{\alpha\gamma} - \frac{b(M - 1) - c}{a + M(M - 1)b - 2(M - 1)c}\delta_{\alpha\gamma}\overline{\overline{Q}}\right]~.
\end{eqnarray}
Upon substitution into Eq.~(\ref{Q_like}) and solving for $\Gamma$ one obtains,
\begin{equation}
\Gamma = \mathcal{Q}^{-1}(a,b,c)(Q) = \mathcal{Q}(a',b',c')(Q)~,
\end{equation}
where
\begin{eqnarray}
a' &=& \frac{1}{a}~,\\
b' &=& \frac{ba + bcM -2c^2}{a\left[c(M - 2) - a\right]\left[a + bM(M - 1) - 2c(M - 1)\right]}~,\\
c' &=& \frac{c}{a\left[c(M - 2) - a\right]}~.
\end{eqnarray}
These are important relations since they allow to evaluate the action of the inverse overlap matrix on a two-particle matrix as fast as a $\mathcal{Q}$ map. \emph{i.e.} at a computational cost which is negligible compared to the other matrix manipulations.
\subsubsection{Solution to the equations of motion}
In this Section we show how to solve Eqs.~(\ref{dual_eom}) and (\ref{primal_eom}). This is done by decoupling the primal from the dual variables, using the fact that the primal and dual step are orthogonal to each other (\ref{pd_step_ortho}). We start by deriving the dual step $\Delta_Z$ from the dual equation.
\paragraph{Solution to the dual equation}
If we project the dual equation (\ref{dual_eom}) on the space spanned by the non-orthogonal basis $\{u^i\}$ (which we will call $\mathcal{U}$-space) we obtain:
\begin{equation}
\label{dual_cg}
\sum_j \underbrace{\left(\mathrm{Tr}~D^{-1} u^j D^{-1} u^i\right)}_{\mathcal{H}^D_{ij}} \delta\gamma_j = \mathrm{Tr}~B u^i~,
\end{equation}
in which the right-hand side of Eq.~(\ref{dual_eom}) is denoted by $B$, and the linear equalities for the primal and dual step (\ref{pd_step_lineq}) have been used. The form of the equation can be seen to be identical to the Newton equations obtained in the dual-only program (\ref{newt_eq}), which means it can be solved in exactly the same way, using the linear conjugate gradient method, without explicit construction of the dual Hessian matrix $\mathcal{H}^D$ or any reference to the non-orthogonal basis set $\{u^i\}$. 
The right-hand side of Eq.~(\ref{dual_cg}) can be transformed to a two-particle matrix:
\begin{equation}
\sum_i \mathrm{Tr}~\left[Bu^i\right] f^i = \sum_k \sum_i\mathrm{Tr}~\left[B_k \mathcal{L}(f^i)\right]f^i = \hat{P}_{\mathrm{Tr}}\left[\sum_k \mathcal{L}_k^\dagger(B_k) \right]~,
\label{collapse}
\end{equation}
where $B_k$ are the different blocks of $B$ corresponding to the various matrix maps. In Eq.~(\ref{collapse}) a matrix $B$, block diagonal in the different constraint spaces, is mapped on a traceless two-particle matrix, this is called the collapse map from now on.
The left-hand side of Eq.~(\ref{dual_cg}) has exactly the same structure as the Hessian in the dual only program (\ref{hessian_structure}). As a result, one can express the action of the dual Hessian matrix $\mathcal{H}^D$ on an arbitrary traceless two-particle matrix $\Delta$ as:
\begin{equation}
\mathcal{H}^D\Delta = \hat{P}_{\text{Tr}}\left[\sum_k\mathcal{L}^\dagger_k\left(D_k^{-1}\mathcal{L}_k\left(\Delta\right) D_k^{-1}\right)\right]~,
\label{dual_hessian}
\end{equation}
in which the $D_k$ are again the blocks of the $D$ matrix corresponding to the different constraints $\mathcal{L}_k$. Once the solution, $\Delta_{\text{sol}}$, to this system of equations (\ref{dual_cg}) is found we can construct the dual step as:
\begin{equation}
\Delta_Z = \bigoplus_i \mathcal{L}_i\left(\Delta_{\text{sol}}\right)~.
\end{equation}
It is interesting to see that the right-hand side of the dual equation (\ref{dual_cg}) is exactly the same as the negative gradient in the dual-only potential reduction program:
\begin{equation}
\mathrm{Tr}~\left[u^i\left(\frac{\nu\eta}{n}Z^{-1} - X\right)\right] = \frac{\nu\eta}{n}\mathrm{Tr}~\left[Z^{-1}u^i\right] - h^i~,
\end{equation}
with $t = \frac{\nu\eta}{n}$. The Hessian however, is only identical when $D = Z$, which is the case when $X \approx Z^{-1}$, {\it i.e.} when $(X,Z)$ is on the central path. When we deviate from the central path, this Hessian will use information of the primal SDP to find a better direction toward the central path and the optimum.
\paragraph{Solution to the primal equation}
The solution to the primal equation (\ref{primal_eom}) is obtained in a similar manner, by projecting this equation on $\mathcal{C}$-space spanned by matrices $\{c^i\}$, which is the orthogonal complement of $\mathcal{U}$-space. With $B$ again denoting the right-hand side of Eq.~(\ref{primal_eom}) and making use of Eq.~(\ref{pd_step_lineq}) one gets:
\begin{equation}
\label{primal_cg}
\sum_j \underbrace{\left(\mathrm{Tr}~D~c^j~D~c^i\right)}_{\mathcal{H}^P_{ij}} \delta x_j = \mathrm{Tr}~B c^i~,
\end{equation}
where we have used
\begin{equation}
\Delta_X = \sum_i \delta x_i~c^i~.
\end{equation}
This is again a symmetrical positive-definite system of linear equations that can be solved iteratively using the linear conjugate gradient method, without explicit construction of the Hessian matrix $\mathcal{H}^P$, or any reference to the basisset $\{c^i\}$. Indeed, $\mathcal{H}^P$ can be seen as a map from $\mathcal{C}$-space on itself, since for an arbitrary matrix in $\mathcal{C}$-space:
\begin{equation}
\epsilon = \sum_i \epsilon_i c^i~,
\end{equation}
the image of $\epsilon$ under the primal Hessian map is
\begin{equation}
\mathcal{H}^P\epsilon = \hat{P}_{\mathcal{C}}\left[D\epsilon D\right]~,
\end{equation}
in which $\hat{P}_\mathcal{C}$ is the projection on $\mathcal{C}$-space. This projection can be executed quickly by using the inverse of the overlap matrix of the $\mathcal{U}$-space basis vectors. Suppose we have an arbitrary block matrix $A$ of the same dimension as $X$ and $Z$. First we project it on the space spanned by the basis $\{u^0,u^i\} = \{u^\alpha\}$. The projected matrix $A'$ reads as:
\begin{equation}
A' = \sum_{\alpha\beta} \mathrm{Tr}~\left[Au^\alpha\right] \left(\mathcal{S}^{-1}\right)_{\alpha\beta}u^\beta~,
\label{proj_U}
\end{equation}
where the overlap matrix $\mathcal{S}$ appears because of the non-orthogonality of the basis. As shown in Section \ref{overlapmatrix} the inverse overlap matrix can also be considered as a map from two-particle space on itself. The projected matrix $A'$ can now be written in block-matrix form as:
\begin{equation}
A' = \bigoplus_j \mathcal{L}_j\left[\mathcal{S}^{-1}\left(\sum_k \mathcal{L}^\dagger_k\left(A_k\right)\right)\right]~.
\end{equation}
To project $A$ on $\mathcal{U}$-space we still have to remove the component along the $u^0$-matrix:
\begin{equation}
\hat{P}_{\mathcal{U}}(A) = A' - \left(\frac{\mathrm{Tr}~u^0 A'}{\mathrm{Tr}~u^0u^0}\right)u^0~.
\label{proj_u_0}
\end{equation}
Since $\mathcal{C}$-space is the orthogonal complement of the $\mathcal{U}$-space, the desired projection of $A$ on the $\mathcal{C}$-space is simply given by
\begin{equation}
\hat{P}_{\mathcal{C}}(A) = A - \hat{P}_{\mathcal{U}}(A)~.
\label{proj_C}
\end{equation}
Note that the right-hand side of the primal equation is the negative gradient of the potential:
\begin{equation}
\phi_P(X) = \mathrm{Tr}~Xu^0 - t\ln{\det{X}}~,
\end{equation}
with $t=\frac{\nu\eta}{n}$, and the Hessian is the same when $(X,Z)$ is on the central path.
\subsubsection{Outline of the algorithm}
In this Section a short outline of the algorithm is presented. The first step is to initialize the primal-dual variables, after which they are directed towards the central path. Then the actual minimization of the primal-dual gap takes place, which is done in a predictor-corrector loop.
\paragraph{Initialization}
We need a feasible primal-dual starting point. An initial feasible dual point $Z^{(0)}$, \emph{i.e.} a matrix that satisfies the inequality (\ref{dual_constr}), is easily found by setting
\begin{equation}
Z^{(0)} = u^0~,
\end{equation}
which corresponds to setting al the $\gamma_i$'s equal to zero. To construct a feasible primal starting point we take a completely random matrix $X$ and project it on a matrix $X'$ which satisfies the primal equality constraint:
\begin{equation}
\label{projham}
\mathrm{Tr}~X'u^i = h^i~.
\end{equation}
This is achieved using the inverse overlap matrix of the $\{u^\alpha\}$ basis,
\begin{equation}
X' = X - \underbrace{\sum_{\alpha\beta}\left(\mathrm{Tr}~Xu^\alpha - h^\alpha\right)\mathcal{S}^{-1}_{\alpha\beta}u^\beta}_{X^\perp}~.
\end{equation}
The last term on the right-hand side can be computed as:
\begin{equation}
X^\perp = \bigoplus_j \mathcal{L}_j\left[\mathcal{S}^{-1}\left(\sum_k \mathcal{L}^\dagger_k\left(X_k\right) - H^{(2)}\right)\right]~.
\end{equation}
At this point, $X'$ satifies the equality constraint (\ref{projham}), and one just has to add $u^{0}$, with a positive scaling factor that is large enough to ensure positive semidefiniteness:
\begin{equation}
X^{(0)} = X' + \alpha u^{0}\succeq 0~.
\end{equation}
\paragraph{Centering run}
Before the actual program can be started, a couple of centering steps have to be taken, which is done by solving the equations (\ref{dual_eom}) and (\ref{primal_eom}) with $\nu = 1$. The purpose is to go sufficiently near the central path, without bothering about the primal-dual gap. In a first step Eq.~(\ref{dual_cg}), which has the smallest dimension, is solved using the conjugate gradient method, and the dual solution $\Delta_Z$ is obtained. The primal solution $\Delta_X$ then follows from the dual equation (\ref{dual_eom}) by substitution. For these initial centering steps, both linear systems are so well conditioned that hardly any iterations are needed for convergence. As a measure for the distance from the center we use the potential \cite{vandenberghe}:
\begin{equation}
\Phi(X,Z) = -\ln\det X - \ln \det Z~,
\end{equation}
which is obviously minimal (for points with the same primal-dual gap $\eta = \mathrm{Tr}~XZ$) on the central path. The potential then has the value:
\begin{equation}
\Phi(X^c,Z^c) = -n \ln{\frac{\eta}{n}}~.
\end{equation}
When the potential difference (which is always positive):
\begin{eqnarray}
\nonumber\Psi(X,Z) &=& \Phi(X,Z) - \Phi(X^c,Z^c)\\
\label{logpot}&=& {n}\ln \mathrm{Tr}~XZ -{n}\ln{n}-\ln\det X - \ln \det Z~,
\end{eqnarray}
is sufficiently small, the centering run is stopped.
\paragraph{Predictor-corrector run}
In this part of the program the primal-dual gap is minimized by alternating predictor and corrector steps. A predictor step tries to reduce the primal-dual gap by solving the equations (\ref{dual_eom}) and (\ref{primal_eom}) with $\nu = 0$. This is done in exactly the same way as for the centering run, by first solving (\ref{dual_cg}) for $\Delta_Z$, then substituting into (\ref{dual_eom}) to obtain an approximate primal step $\Delta_X$. The final primal step $\Delta_X$ is obtained by solving (\ref{primal_cg}) using the conjugate gradient method with the approximate $\Delta_X$ as a starting point. Note that when the primal-dual gap decreases, the condition number of the primal and dual Hessian matrices increases and more iterations are needed before convergence is reached. One can adjust the convergence criteria of the primal and dual conjugate gradient loops, in order to minimize the combined number of iterations.

At this point we have a predictor direction $(\Delta_X,\Delta_Z)$. The logarithmic potential
\begin{equation}
\theta(\alpha) = \Psi(X + \alpha \Delta_X,Z + \alpha \Delta_Z)~,
\label{pd_ls}
\end{equation}
in the predictor direction (see Eq.~(\ref{logpot})) can be simply evaluated (analogously to the line search function in the dual-only program) for any value of $\alpha$ by solving two eigenvalue equations:
\begin{equation}
\left(X^{-\frac{1}{2}}\Delta_X X^{-\frac{1}{2}}\right)v_X = \lambda^X~v_X \qquad\text{and}\qquad\left(Z^{-\frac{1}{2}}\Delta_Z Z^{-\frac{1}{2}}\right)v_Z = \lambda^Z~v_Z~.
\end{equation}
Using these eigenvalues one can write:
\begin{eqnarray}
\nonumber\ln{\det{(X+\alpha\Delta_X)}} &=& \ln{\det{\left[X^{\frac{1}{2}}\left(\mathbb{1}+\alpha X^{-\frac{1}{2}}\Delta_X X^{-\frac{1}{2}}\right)X^{\frac{1}{2}}\right]}}\\
&=&\ln{\det{X}} + \sum_i \ln{(1 + \alpha\lambda^X_i)}~,
\end{eqnarray}
and thus simply rewrite Eq.~(\ref{pd_ls}) as:
\begin{eqnarray}
\nonumber\theta(\alpha) &=& \Psi(X,Z) +  \ln\left[1 + \alpha (c_X + c_Z)\right]\\
&&\qquad- \sum_i \ln (1 + \alpha \lambda^X_i) - \sum_i \ln(1 + \alpha \lambda^Z_i)~,
\label{pd_ls_expanded}
\end{eqnarray}
where
\begin{equation}
c_Z =  \frac{1}{\eta}\mathrm{Tr}~X\Delta_Z \qquad\text{and}\qquad c_X =  \frac{1}{\eta}\mathrm{Tr}~Z\Delta_X~.
\end{equation}
With a standard bisection method one can now compute the stepsize $\alpha$ corresponding to the maximal deviation from the central path we want to allow.

After the predictor step, a corrector step is taken, which is equivalent to the centering step described previously. The alternation of predictor and corrector steps continues until the primal-dual gap is smaller than the desired value.
\subsubsection{Adding linear inequality constraints}
Adding linear inequality constraints of the form (\ref{linineq}) to the primal-dual algorithm is much more involved than for the dual-only algorithm. As described in Section \ref{linineq_do}, the only change in the primal and dual SDP is that an extra term is added to the matrices $\{u^\alpha\}$, as in Eq.~(\ref{u_li}). This means, however, that the overlap matrix of the set $\{u^\alpha\}$ is modified, and the formulas derived in Section \ref{overlapmatrix} cannot be used. The extra terms in the overlap matrix are simply:
\begin{equation}
\mathcal{S}^L_{\alpha\beta} = \sum_i L(f^\alpha)_{ii} L(f^\beta)_{ii}~,
\end{equation}
and as a result the action of the overlap matrix on a general two-particle matrix can be written as an extended generalized $\mathcal{Q}$-like map:
\begin{eqnarray}
\mathcal{Q}_L(a,b,c)\left(\Gamma\right)_{\alpha\beta;\gamma\delta} &=& a\Gamma_{\alpha\beta;\gamma\delta} + b\left(\delta_{\alpha\gamma}\delta_{\beta\delta} - \delta_{\alpha\delta}\delta_{\beta\gamma}\right)\overline{\overline{\Gamma}} \label{overlap_linineq}\\
\nonumber&& - c\left(\delta_{\alpha\gamma}\overline{\Gamma}_{\beta\delta} - \delta_{\beta\gamma}\overline{\Gamma}_{\alpha\delta} - \delta_{\alpha\delta}\overline{\Gamma}_{\beta\gamma} + \delta_{\beta\delta}\overline{\Gamma}_{\alpha\gamma}\right)~
+ \frac{1}{4}\sum_i\left(\overline{\overline{\Gamma C^0_i}}\right)C^0_i~.
\end{eqnarray}
We now show that the inverse map can still be constructed, at the cost of solving a system of linear equations with dimension $2m + 1$ (with $m$ the number of inequality constraints). As before we invert the equation:
\begin{equation}
\mathcal{Q}_L(a,b,c)(\Gamma) = Q~,
\end{equation}
by calculating the partial traces:
\begin{eqnarray}
\overline{Q}_{\alpha\gamma} &=& \left[a - c(M-2)\right]\overline{\Gamma}_{\alpha\gamma} + \delta_{\alpha\gamma}\overline{\overline{\Gamma}}\left[b(M-1) - c\right] + \frac{1}{4}\sum_i \left(\overline{\overline{\Gamma C^0_i}}\right)\overline{C^0_i}~,\\
\overline{\overline{Q}} &=& \left[bM(M-1)-2c(M-1)+a\right]\overline{\overline{\Gamma}}  + \frac{1}{4}\sum_i \left(\overline{\overline{\Gamma C^0_i}}\right)\overline{\overline{C^0_i}}~.
\end{eqnarray}
We cannot solve this system, because we have three equations but $m + 3$ unknowns, so we need more equations. We can derive $m$ extra equations by:
\begin{equation}
\overline{\overline{QC^0_i}} = \sum_j \left[\frac{1}{4}\left(\overline{\overline{C^0_i C^0_j}}\right) + a\delta_{ij}\right]\left(\overline{\overline{\Gamma C^0_j}}\right) - 4\overline{\overline{\Gamma}\overline{C^0_i}} + 2\overline{\overline{C^0_i}}\overline{\overline{\Gamma}}~,
\end{equation}
but we get $m$ new unknowns:
\begin{equation}
\overline{\overline{\Gamma}\overline{C^0_i}} = \sum_{\alpha\gamma}\overline{\Gamma}_{\alpha\gamma}\left(\overline{C^0_i}\right)_{\gamma\alpha}~.
\end{equation}
If we add $m$ more equations of the form:
\begin{equation}
\overline{\overline{Q}\overline{C^0_i}} = \left[a - c(M-2)\right]\overline{\overline{\Gamma}\overline{C^0_i}} + \overline{\overline{C^0_i}}\left[b(M-1) - c\right]\overline{\overline{\Gamma}} + \frac{1}{4}\sum_j \left(\overline{\overline{\Gamma C^0_j}}\right)\overline{\overline{C^0_i}~\overline{C^0_j}}~,
\end{equation}
we finally have a closed system of $2m + 1$ equations and $2m+1$ unknowns:
\begin{equation}
\left(\begin{matrix}\overline{\overline{Q}}\\\overline{\overline{QC^0_i}}\\\overline{\overline{Q}\overline{C^i_0}}\end{matrix}\right)
=
\left(
\begin{matrix}
\lambda &\frac{1}{4}\overline{\overline{C^0_j}}&0\\
2\overline{\overline{C^0_i}}&\frac{1}{4}\overline{\overline{C^0_i C^0_j}}+a\delta_{ij} & -4\\
\mu \overline{\overline{C^0_i}} & \frac{1}{4}\overline{\overline{C^0_i}~\overline{C^0_j}}&\delta_{ij}\kappa
\end{matrix}
\right)
\left(\begin{matrix}\overline{\overline{\Gamma}}\\\overline{\overline{\Gamma C^0_j}}\\\overline{\overline{\Gamma}\overline{C^j_0}}\end{matrix}\right)~,
\end{equation}
where
\begin{eqnarray}
\lambda &=& bM(M-1) - 2c(M-1) -a~,\\
\kappa&=& a - c(M-2)~,\\
\mu &=& b(M-1)-c~.
\end{eqnarray}
This system can be solved numerically to obtain the relations:
\begin{eqnarray}
\alpha(Q) &=& \alpha_1~\overline{\overline{Q}} + \sum_i \alpha_2^{i}~\overline{\overline{QC^0_i}} + \sum_i \alpha_3^{i}~\overline{\overline{Q}\overline{C^0_i}} = \overline{\overline{\Gamma}}~,\\
\beta^j(Q) &=& \beta^j_1~\overline{\overline{Q}} + \sum_i \beta_2^{ji}~\overline{\overline{QC^0_i}} + \sum_i \beta_3^{ji}~\overline{\overline{Q}\overline{C^0_i}} = \overline{\overline{\Gamma C^0_j}}~,\\
\gamma^j(Q) &=& \gamma^j_1~\overline{\overline{Q}} + \sum_i \gamma_2^{ji}~\overline{\overline{QC^0_i}} + \sum_i \gamma_3^{ji}~\overline{\overline{Q}\overline{C^0_i}} = \overline{\overline{\Gamma}\overline{C^0_j}}~.
\end{eqnarray}
The calculated numbers can be substituted:
\begin{equation}
\overline{\Gamma}_{\alpha\gamma} = \frac{1}{\kappa}\overline{Q}_{\alpha\gamma}-\frac{\mu}{\kappa}\delta_{\alpha\gamma}\alpha(Q) -\frac{1}{4\kappa}\sum_i \beta^i(Q)\overline{C^0_i}~,
\end{equation}
and finally yield the inverse of the extended generalized $\mathcal{Q}$ map:
\begin{eqnarray}
\nonumber\mathcal{Q}^{-1}_L(a,b,c)(Q)_{\alpha\beta;\gamma\delta} &=& \frac{1}{a}Q_{\alpha\beta;\gamma\delta} - \frac{1}{a}\left(b+\frac{2c\mu}{\kappa}\right)\left(\delta_{\alpha\gamma}\delta_{\beta\delta}-\delta_{\alpha\delta}\delta_{\beta\gamma}\right)\alpha(Q)\\
\nonumber&&+\frac{c}{\kappa a}\left(\delta_{\beta\delta}\overline{Q}_{\alpha\gamma} -\delta_{\alpha\delta}\overline{Q}_{\beta\gamma}-\delta_{\beta\gamma}\overline{Q}_{\alpha\delta}+\delta_{\alpha\gamma}\overline{Q}_{\beta\delta}\right)\\
\nonumber&&+\frac{c}{4\kappa a}\sum_i \beta^i(Q)\left(\delta_{\beta\delta}\overline{C^0_i}_{\alpha\gamma} -\delta_{\alpha\delta}\overline{C^0_i}_{\beta\gamma}-\delta_{\beta\gamma}\overline{C^0_i}_{\alpha\delta}+\delta_{\alpha\gamma}\overline{C^0_i}_{\beta\delta}\right)\\
\label{inverse_overlap_linineq}&& -\frac{1}{4a}\sum_i\beta^i(Q)C^0_i = \Gamma_{\alpha\beta;\gamma\delta}~.
\end{eqnarray}
This inverse overlap-matrix map is needed if we want to project on $\mathcal{U}$-space, defined by the $\{u^i\}$'s with added terms for linear inequalities. The form of the projection is exactly the same as before (\ref{proj_U}), but with the old overlap matrix replaced by the new one (\ref{inverse_overlap_linineq}). The solution to the primal system now changes, in that the projection operator on $\mathcal{C}$ space (\ref{proj_C}), appearing in the definition of the primal Hessian, uses the changed projection operator on $\mathcal{U}$-space.

For the solution to the dual equation, the changes are exactly the same as those mentioned in Section \ref{linineq_do}. There are extra terms added to the collapse map due to the linear inequalities:
\begin{equation}
\sum_i \mathrm{Tr}~\left[Bu^i\right] f^i = \hat{P}_{\mathrm{Tr}}\left[\sum_k \mathcal{L}_k^\dagger(B_k) + \sum_j B^L_j C^0_j\right]~,
\end{equation}
in which the $B^L_i$ are numbers on the diagonal corresponding to the $i$'th linear constraint. The modifications to the dual Hessian map are exactly the same as those in Eq.~(\ref{hessian_li}):
\begin{equation}
\mathcal{H}^D\Delta = \hat{P}_{\mathrm{Tr}}\left[\sum_k \mathcal{L}^\dagger_k\left(D_k^{-1}\mathcal{L}_k(\Delta)D_k^{-1}\right) + \sum_k \left(\frac{L(\Delta)_{kk}}{{D^L_{kk}}^2}\right)C^0_k\right]~.
\end{equation}
\subsubsection{Adding linear equality constraints}
Adding linear equality constraints to the primal-dual algorithm is quite straightforward. For the solution of the dual problem the changes are exactly the same as with the dual-only program, {\it i.e.} the $\hat{P}_{\mathrm{Tr}}$ appearing in Eqs.~(\ref{collapse}) and (\ref{dual_hessian}) is replaced by a $\hat{P}_f$ as defined in Eq.~(\ref{P_f}). For the solution to the primal algorithm the only modificiation is the projection on $\mathcal{U}$-space, more specifically the removal of the component along $u^0$ (\ref{proj_u_0}). For this we have to introduce several matrices:
\begin{equation}
u^0_n = \bigoplus_i \mathcal{L}_i(\tilde{E}_n)~,
\end{equation}
and orthogonalize them:
\begin{equation}
\tilde{u}^0_n = \sum_{m} \left(S^{-\frac{1}{2}}_u\right)_{nm} u^0_m~,\qquad\text{with}\qquad (S_u)_{nm} = \mathrm{Tr}~\left[u^0_m u^0_n\right]~.
\end{equation}
The new projection is then defined as:
\begin{equation}
A' = A - \sum_n \mathrm{Tr}~\left[\tilde{u}^0_n A\right]\tilde{u}^0_n~.
\end{equation}
\section{\label{bp_sdp}Boundary point method}
The interior point methods discussed in the previous Section~\ref{int_point} are very stable and always converge to the desired accuracy. A disadvantage, however, is that at every point a Newton-like system of equations has to be solved. In the algorithms presented earlier we have avoided solving this system explicitly by using iterative methods and a fast matrix-vector product. The number of iterations needed to solve these equations, however, explodes when we get close to the optimum, because the system becomes ill-conditioned. This occurs because, on approaching the edge of the feasible region, the matrices involved have near zero eigenvalues, causing the condition number of the Hessians to diverge. This limits the size of the systems that can be studied.
Another, completely orthogonal, approach is the boundary point method, developed in \cite{rendl,malick}, where one remains on the hypersurface defined by the complementary slackness condition (\ref{compl_slack}), and moves towards the feasible region. This method was developed for problems where the number of variables is so large that standard SDP algorithms don't work anymore. An implementation of this algorithm for variational density matrix optimization was presented in \cite{maz_bp}.
\subsection{The augmented Lagrangian}
The boundary point method is actually an instance of the more general class of augmented Lagrangian approaches for solving convex optimization problems. The standard Lagrangian for the SDP:
\begin{equation}
\min_\gamma \sum_i \gamma_ih^i \qquad\text{u.c.t.}\qquad Z = u^0 + \sum_i \gamma_i u^i~,
\label{dual_eq}
\end{equation}
introduces the Lagrange multiplier matrix $X$:
\begin{equation}
L(\gamma,Z;X) = \gamma^T h + \mathrm{Tr}~\left[X\left(Z - u^0 - \sum_i \gamma_i u^i\right)\right]~.
\end{equation}
The \emph{augmented} Lagrangian for Eq~(\ref{dual_eq}) adds a quadratic penalty for infeasibility:
\begin{equation}
L_\sigma(\gamma,Z;X) = \gamma^T h + \mathrm{Tr}~\left[X\left(Z - u^0 - \sum_i \gamma_i u^i\right)\right] + \frac{\sigma}{2}\|Z - u^0 - \sum_i \gamma_i u^i\|^2~,
\label{augmented_lagrangian}
\end{equation}
where a parameter $\sigma>0$ determines how strong the penalty is. If we introduce the new matrix variable:
\begin{equation}
W(\gamma) = u^0 + \sum_i \gamma_i u^i - \frac{1}{\sigma}X~,
\end{equation}
we can rewrite the augmented Lagrangian as:
\begin{equation}
L_\sigma(\gamma,Z;X) = f(\gamma,Z) - \frac{1}{2\sigma}\| X\|^2~,
\end{equation}
with
\begin{equation}
f(\gamma,Z) = \gamma^T h + \frac{\sigma}{2}\| Z-W(\gamma)\|^2~.
\label{f_aug_lag}
\end{equation}
The idea of augmented Lagrangian methods for solving SDP's is to perform the optimization in two stages: first minimize $f(\gamma,Z)$ under the constraint $Z\succeq 0$, keeping $X$ constant; then update $X$ in some way until convergence is reached.
\subsection{Solution to the inner problem}
The inner problem is a quadratic SDP:
\begin{equation}
\min_{\gamma,Z} f(\gamma,Z) \qquad\text{u.c.t.}\qquad Z\succeq 0~,
\end{equation}
for which the Lagrangian reads as:
\begin{equation}
L(\gamma,Z;V) = f(\gamma,Z) - \mathrm{Tr}~VZ~,
\end{equation}
with Lagrange multiplier $V$ for the constraint $Z\succeq 0$. The optimal point has to satisfy the conditions:
\begin{eqnarray}
\label{grad_of_aug}
\frac{\partial L}{\partial \gamma_i} &=& h^i -\sigma\mathrm{Tr}~\left[u^i\left(Z-u^0-\sum_j\gamma_ju^j + \frac{1}{\sigma}X\right)\right]=0~,\\
\nabla_ZL &=& \sigma\left[Z-W(\gamma)\right] -V = 0~,\\
Z&\succeq&0~,\qquad V\succeq0~,\qquad ZV =0~.
\end{eqnarray}
Keeping $Z$ constant, the optimality condition for $\gamma$ (\ref{grad_of_aug}) can be fulfilled by solving the linear system:
\begin{equation}
\sum_j \mathrm{Tr}~\left[u^iu^j\right]\gamma_j = -\frac{h^i}{\sigma} + \mathrm{Tr}~\left[u^i\left(Z - u^0 + \frac{1}{\sigma}X\right)\right]~.
\label{lin_sys_bp}
\end{equation}
It is clear that (\ref{grad_of_aug}) holds when the sum of the primal and the dual equality constraints are satisfied:
\begin{equation}
\mathrm{Tr}~\left[Xu^i\right] + Z = h^i + u^0+\sum_i\gamma_iu^i~,
\end{equation}
so the solution of the linear system (\ref{lin_sys_bp}) can be seen as a projection on the feasible plane.
If, on the other hand, one keeps $\gamma$ constant, the optimization of (\ref{f_aug_lag}) under the condition that $Z\succeq0$ is just a projection on the cone of semidefinite matrices. Decomposing $W(\gamma)$ into a positive and negative part\footnote{The positive part of a symmetric matrix with spectral decomposition $A = \sum_i \lambda_i X_i X_i^T$ is obtained by restraining the sum to positive eigenvalues $\lambda_i$ only, and vice-versa for the negative part.} as:
\begin{equation}
W(\gamma) = W(\gamma)_+ + W(\gamma)_-~,
\end{equation}
the optimal $Z$ with constant $\gamma$ is:
\begin{equation}
Z = W(\gamma)_+~.
\end{equation}
For $V$ to satisfy the optimality conditions, it must be equal to:
\begin{equation}
V = -\sigma W(\gamma)_-~,
\end{equation}
for which it can be seen that the conditions
\begin{equation}
V = \sigma\left[Z - W(\gamma)\right]~,\qquad VZ = 0\qquad\text{and}\qquad V\succeq 0~,
\label{optim_af_prim_loop}
\end{equation}
are automatically fulfilled.
Projecting on the feasible plane, followed by a projection on the cone of positive definite matrices, brings us closer to the optimum, as shown in Fig. \ref{fig_bp}. We can repeat this procedure until convergence is reached, {\it i.e.} until the linear equations (\ref{lin_sys_bp}) are satisfied \emph{after} projection on the positive cone, to within the prescribed accuracy $\epsilon$. This convergence criterion:
\begin{equation}
h^i - \sigma\mathrm{Tr}~\left[u^i\left(Z - W(\gamma)\right)\right] \leq \epsilon~,
\end{equation}
can be rewritten using Eq.~(\ref{optim_af_prim_loop}) as:
\begin{equation}
h^i - \sigma\mathrm{Tr}~\left[u^iV\right] \leq \epsilon~.
\end{equation}
As a result, the Lagrange multiplier $V$ is approximately primal feasible after the inner loop. The obvious update for $X$ in the outer loop is:
\begin{equation}
X \leftarrow V~,
\end{equation}
and one can consider the inner loop as a projection on primal feasibility. The program is finished when both the primal and the dual are feasible; it follows that the convergence criterion for the outer loop is dual feasibility within the preset accuracy.
\begin{figure}
\centering
\includegraphics[scale=0.5]{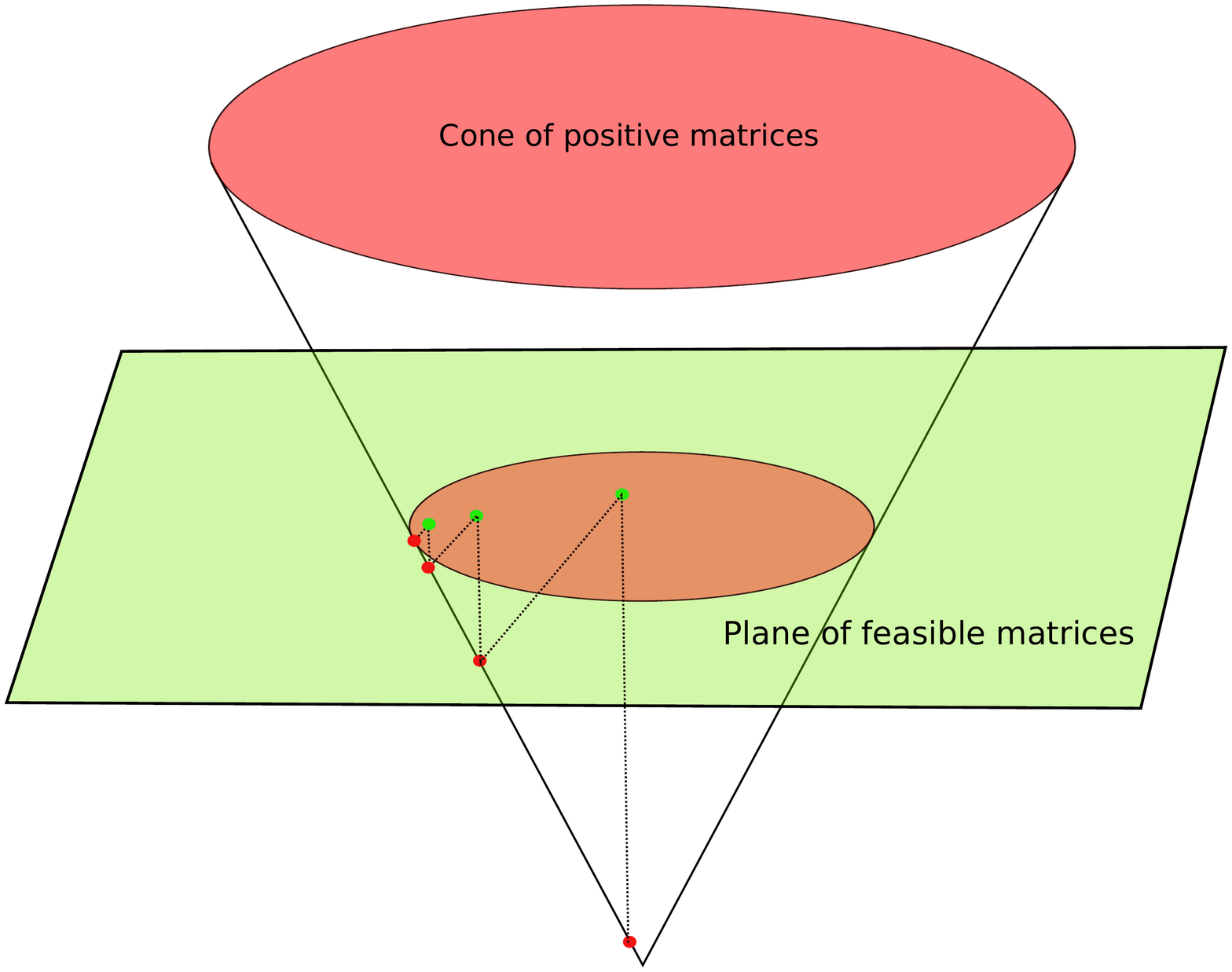}
\caption{\label{fig_bp}A schematic representation of how the optimal point is reached by subsequent projections on the feasible plane and the semidefinite cone.}
\end{figure}
\subsection{Outline of the algorithm}
It is important to realize that throughout the program the primal and dual variables are not feasible, {\it i.e.} that $\gamma$ and $Z$ are independent. The algorithm itself is short and simple, and is given in pseudocode in Algorithm~\ref{bp_algo}. At first sight the computationally most demanding step in the algorithm is the solution of the linear system Eq.~(\ref{lin_sys_bp}). However, it can be seen that when rewriting the system in traceless two-particle matrix space, with $\gamma = \sum_i \gamma_if^i$:
\begin{equation}
\mathcal{S}\gamma = \hat{P}_{\mathrm{Tr}}\left[-\frac{1}{\sigma}H^{(2)} + \sum_i \mathcal{L}_i^\dagger\left(Z_i - W(\gamma)_i\right)\right]~,
\label{lin_sys_bp_mat}
\end{equation}
where $\mathcal{S}$ is the overlap matrix defined in Section \ref{overlapmatrix} limited to traceless space. The generalized $\mathcal{Q}$ map transforms traceless matrices into traceless matrices. This implies that the inverse of the overlap matrix limited to traceless space is exactly the same as the inverse of the overlap matrix on full space. The system can thus be solved easily by applying the inverse overlap matrix, as defined in Section \ref{overlapmatrix} to the right-hand side of Eq.~(\ref{lin_sys_bp_mat}). The most expensive step in the algorithm is the separation of the matrix $W$ into a positive and negative part, for which an eigenvalue decomposition is necessary which scales as $O(n^3)$ in the dimension of the matrix. If only two-positivity conditions are applied this leads to a scaling of $M^6$ per iteration, for the three-index conditions the scaling per iteration is $M^9$.
\begin{algorithm}
\caption{\label{bp_algo} The boundary point algorithm}
\begin{algorithmic}
\State Choose $\epsilon_{\text{in}} > 0~$,$~\epsilon_{\text{out}} > 0$ and $\sigma>0$ 
\State $X^0 = 0$;~$Z^0 = 0$;~$k = 0$
\While{$\delta_{\text{out}} > \epsilon_{\text{out}}$}
\While{$\delta_{\text{in}} > \epsilon_{\text{in}}$}
\State Solve $\mathrm{Tr}~\left[u^iu^j\right]\gamma^{(k)}_j = -\frac{h^i}{\sigma} + \mathrm{Tr}~\left[u^i\left(Z - u^0 + \frac{1}{\sigma}X\right)\right]$ for $\gamma^{(k)}$
\State $W(\gamma) = u^0 + \sum_i \gamma_i u^i - \frac{1}{\sigma}X$
\State $W(\gamma) \rightarrow W(\gamma)_+ + W(\gamma)_-$
\State $Z^{(k)} = W(\gamma)_+$ ; $V^{(k)} = -\sigma W(\gamma)_-$
\State $\delta_{\text{in}} = \sum_i \sqrt{(\mathrm{Tr}~[u^iV^{(k)}] - h^i)^2}$\Comment{primal infeasibility}
\EndWhile
\State $X^{(k+1)} = V^{(k)}$
\State $k \gets k + 1$
\State $\delta_{\text{out}} = \|Z^{(k)} - u^0 - \sum_i \gamma_i u^i\|$\Comment{dual infeasibility}
\EndWhile
\end{algorithmic}
\end{algorithm}
\subsection{Adding linear equality and inequality constraints}
The addition of linear inequality constraints to the boundary point method is quite analogous to the primal-dual algorithm. The only modification is that the overlap matrix takes on the form (\ref{overlap_linineq}). For the solution of the linear system (\ref{lin_sys_bp_mat}) we now let the inverse (\ref{inverse_overlap_linineq}) act on the right-hand side of Eq.~(\ref{lin_sys_bp_mat}), in which evidently the constraint matrices ${C}^0_i$ are replaced by their projection on traceless matrix space.

For the equality constraints this trick no longer works, and we will have to solve the linear system (\ref{lin_sys_bp_mat}) using the conjugate gradient method. The overlap-matrix map which repeatedly acts on the right-hand side of Eq.~(\ref{lin_sys_bp_mat}) is defined as:
\begin{eqnarray}
\nonumber\mathcal{S}\left(\Gamma\right)_{\alpha\beta;\gamma\delta} &=& \hat{P}_f\left[a\Gamma_{\alpha\beta;\gamma\delta} + b\left(\delta_{\alpha\gamma}\delta_{\beta\delta} - \delta_{\alpha\delta}\delta_{\beta\gamma}\right)\overline{\overline{\Gamma}}\right.\\
&&\left. - c\left(\delta_{\alpha\gamma}\overline{\Gamma}_{\beta\delta} - \delta_{\beta\gamma}\overline{\Gamma}_{\alpha\delta} - \delta_{\alpha\delta}\overline{\Gamma}_{\beta\gamma} + \delta_{\beta\delta}\overline{\Gamma}_{\alpha\gamma}\right)\right]~,
\label{Q_like_lineq}
\end{eqnarray}
in which the factors $(a,b,c)$ are determined by the applied constraints, and $\hat{P}_f$ is given by Eq.~(\ref{P_f}).
\section{Computational performance of the methods}
A general comparison between different algorithms is difficult, as the performance of the algorithms can depend on the specific system being studied. There are also ways in which an algorithm can be tweaked to converge faster, by making adjustments which are sometimes logical, and sometimes don't make sense. A comparison of different algorithms is therefore inherently biased towards the algorithm one spent most time optimizing. This being said, this Section contains an honest comparison of the various algorithms developed and tested during the course of my PhD, and an attempt to explain in some detail what the bottlenecks are in the different methods.
\subsection{Interior point methods}
As mentioned earlier, for interior point methods the bottleneck is the solution to the Newton-like Hessian equation that has to be solved at every iteration:
\begin{equation}
\mathcal{H}^{ij}\delta_j = b^i~.
\label{hessian_in_comp_perf}
\end{equation}
The dimension of the system is the number of free variables in the optimization, which scales as the number of elements in a 2DM, being $M^4$. General methods to solve linear systems ({\it e.g.} using LU or Cholesky decomposition) scale as $O(n^3)$, which means an algorithm based on a direct solution of Eq.~(\ref{hessian_in_comp_perf}) scales as $M^{12}$. An explicit construction of the Hessian matrix would be even more demanding, as for every element some elementary $M^6$ matrix calculation has to be done, resulting in a scaling of $M^{14}$. This huge computational cost can be avoided by referencing the Hessian matrix only through its action on a 2DM and resorting to iterative methods, like the linear conjugate gradient method, to solve the system (\ref{hessian_in_comp_perf}). At every iteration of the conjugate gradient method, a matrix-vector product has to be performed, which scales as $M^6$ for two-index conditions and $M^9$ for three-index conditions. The drawback to this iterative approach is that the efficiency of the solution, {\it i.e.} the number of iterations needed to converge, depends on the condition number of the Hessian. For interior point methods, the condition number diverges near the edge of the feasible region, while a point near to the edge is reached quite rapidly, gaining those few extra digits in precision becomes very demanding.
In Figure \ref{iter_vs_pdg} we have plotted the number of iterations needed for the linear conjugate gradient method to converge to the desired accuracy as a function of the primal-dual gap, for a 10-site one-dimensional Hubbard model with 10 particles. One can see that, for both the primal-dual and the dual-only method, the number of iterations increases drastically when the primal-dual gap decreases. For the dual-only method, a free parameter that can be tweaked to optimize the total number of iterations, is the factor by which the potential parameter $t$ (see Eq.~(\ref{dual_only_pot})) is scaled down after every Newton loop. If we choose this factor to be large, many more Newton steps are needed to converge to the optimum of the next $t$ value, and the number of iterations in Figure \ref{iter_vs_pdg} is much higher. In the primal-dual method, how much the primal-dual gap is reduced per Newton step is determined by how far one can step in the predictor direction without exceeding the maximal deviation from the central path (See Eq.~(\ref{pd_ls_expanded})). One can see in Figure \ref{iter_vs_pdg} that there are more green points than red points, which means that there are more Newton steps needed in the dual-only method than in the primal-dual. In general the primal-dual method will require less Newton steps to converge than the dual-only approach, due to the fact that information of both the primal and the dual problem are used to determine the predictor direction. 
\begin{figure}
\centering
\includegraphics[scale=0.8]{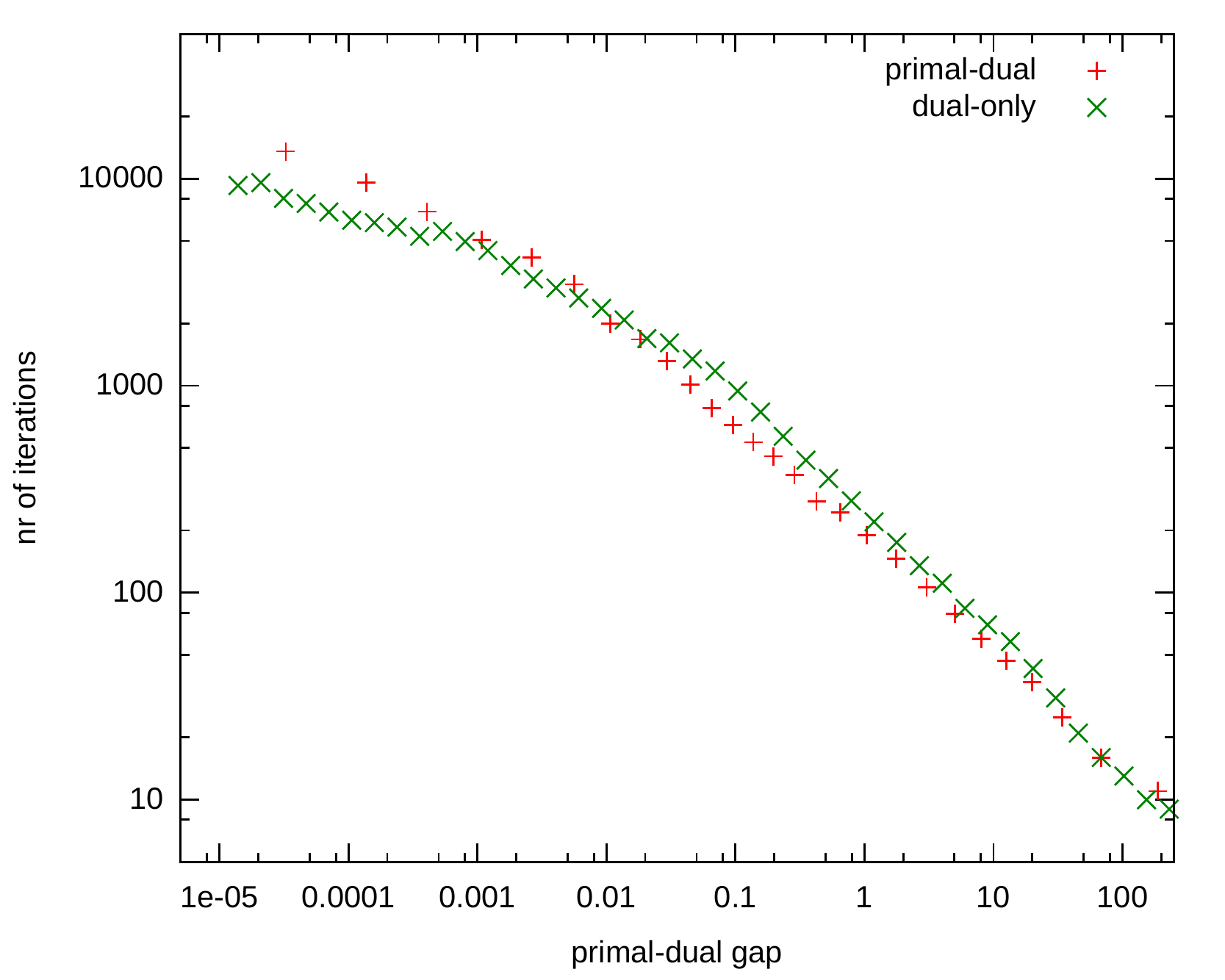}
\caption{\label{iter_vs_pdg}The number of iterations needed by the linear conjugate gradient method to solve the inner Newton problem as a function of the primal-dual gap, for both interior point methods, applied to a 10 site 1D-Hubbard model with 10 particles.}
\end{figure}
\subsection{Boundary point method}
The boundary point method is in its approach completely orthogonal to the interior point methods, and it therefore has very different numerical properties and difficulties. A key advantage is the absence of a divergence of the condition number of some linear system. A disadvantage is that the method is much less stable, and its convergence properties are more dependent on the systems studied than for interior point methods. An important free parameter in this algorithm is the prefactor $\sigma$ to the dual infeasibility penalty in the augmented Lagrangian (See Eq.~(\ref{augmented_lagrangian})). This parameter has to be chosen carefully and updated during the algorithm to ensure balance between primal and dual infeasibility. In theory one lets the inner loop converge to a certain precision $\epsilon_{\text{in}}$. In practice the algorithm turns out to work best when only one inner iteration is performed. After every outer iteration the primal and dual infeasibility are checked. If the dual infeasibility is larger than the primal infeasibility, $\sigma$ is multiplied by some factor $\tau > 1$. If the dual infeasibility is smaller then the primal, $\sigma$ is divided by $\tau$. In Figure \ref{bp_conv} the convergence behaviour for the boundary point method is shown. It is seen that there are two different regimes during convergence. First the primal and dual infeasibility are not balanced and they do not decrease monotonically, but rather oscillate in a chaotic manner. Once they have become balanced, the convergence towards the feasible area proceeds in a more stable, monotonic way.
\begin{figure}
\centering
\includegraphics[scale=0.8]{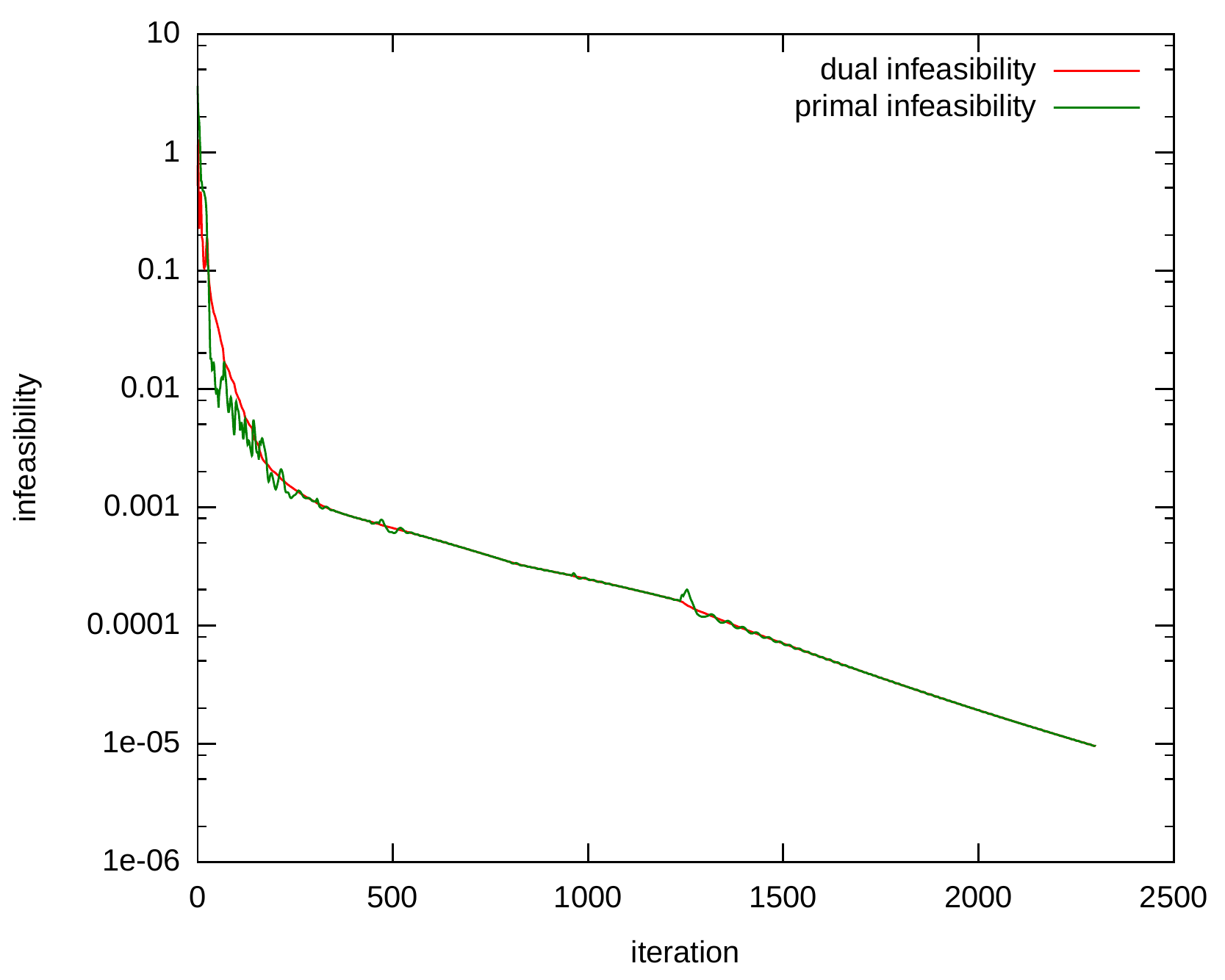}
\caption{\label{bp_conv} The evolution of primal and dual infeasibility during the convergence of the boundary point method.}
\end{figure}
\subsection{Comparison of scaling properties}
To conlude the Chapter,  some scaling properties of the different methods are compared, when applied to the one-dimensional Hubbard model. All the methods have the same basic computational scaling behaviour, being $O(M^6)$ required for multiplying or diagonalizing a matrix. In Figure~\ref{scaling} the number of these operations needed to converge to the optimum is plotted as a function of lattice size. The interior point methods both have to solve a linear system of size $M^4$, so it is not surprising that the scaling, on top of the basic matrix computations, of these methods is $M^4$. More surprising is that there seems to be no, or a very limited, scaling for the boundary point method. The number of iterations required remains around 3000, irrespective of the size of the system. It must be stressed that this is a result limited to the one-dimensional hubbard model, and cannot be extrapolated to other systems, like molecules, where the convergence properties of the boundary point method can be completely different. One reason for the succes of the boundary point method applied to the Hubbard model is the amount of symmetry present in the system. The boundary point method is designed for problems with a huge amount of dual variables or primal constraints. For most physical systems the dimensions of the matrices involved are already unfeasibly large before the boundary point method would becomes advantageous. The Hubbard model, however, contains many symmetries (see Chapter \ref{symmetry}), implying that the matrix dimensions are considerably reduced, and the number of dual variables can get very large before the matrix computations involved become unfeasible. In this case, the domain where the boundary point method is advantageous is actually reached.
\begin{figure}
\centering
\includegraphics[scale=0.8]{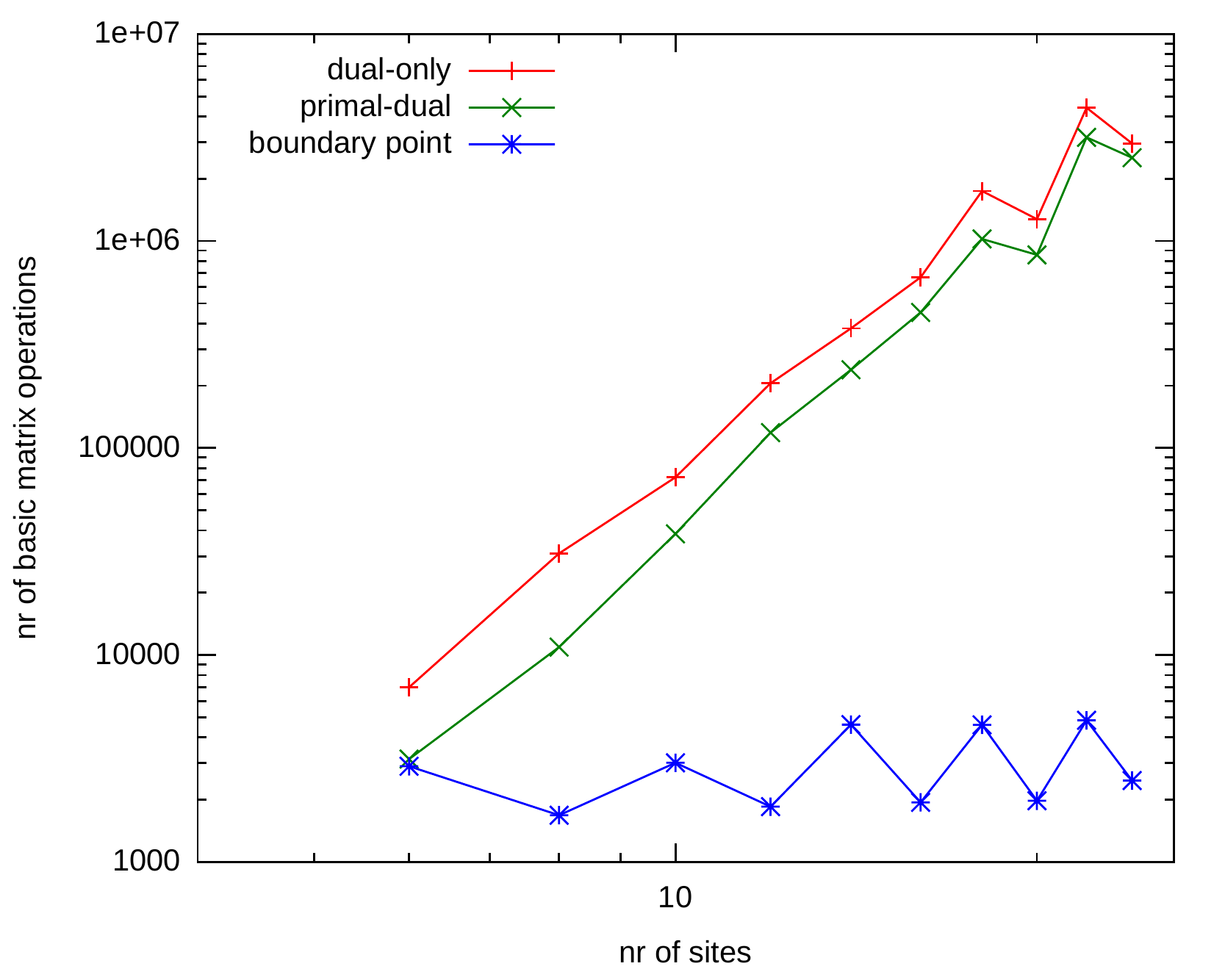}
\caption{\label{scaling} The number of basic matrix operations needed to converge for different lattice sizes of the one-dimensional Hubbard model.}
\end{figure}

\chapter{\label{symmetry}Symmetry adaptation of the 2DM}
In the previous Chapter we have shown how the variational determination of the density matrix can be formulated as a semidefinite program. The current implementations of semidefinite programming algorithms, however, are not yet competitive with other many-body methods of comparable accuracy (like the coupled-cluster framework with single and double excitations). In this Chapter it is shown that including symmetry in the v2DM algorithm is a straightforward, though sometimes tedious, task, and leads to a considerable speedup of all computations. A system is said to have a certain symmetry when it is invariant under the action of some operator $\hat{O}$. In quantum mechanics, this is expressed mathematically by the fact that the Hamiltonian operator describing the system commutes with the operator $\hat{O}$:
\begin{equation}
[\hat{H},\hat{O}] = 0~.
\end{equation}
As we will see, this implies a partial (block)-diagonalization of all matrix quantities involved.
In this Chapter an overview is given of the different symmetries occurring in the systems under study, and how these were exploited. The first Section deals with the inclusion of spin symmetry, which is present in all the atomic, molecular and lattice systems that have been considered. Then additional symmetries are discussed that are more specific to the different systems, for instance the symmetries present in atomic systems (rotational symmetry and parity). The last Section deals with the symmetries in the one-dimensional Hubbard model, {\it i.e.} translational invariance combined with spin and parity. For a more complete account of symmetries and group theory in physics we refer to \cite{cornwell,hamermesh}; for a basic overview of the angular momentum algebra used here, see Appendix~\ref{angular_momentum}.
\section{\label{spinsym}Spin symmetry}
Spin symmetry is one instance of the much larger class of symmetries arising from the invariance of a system under rotations, either in real space or in some abstract spin space.
As is known from basic group theory \cite{cornwell}, the generators of a rotational group are angular momentum operators $(\mathcal{J}_x,\mathcal{J}_y,\mathcal{J}_z)$ which satisfy the Lie algebra:
\begin{equation}
[\mathcal{J}_x,\mathcal{J}_y] = i \mathcal{J}_z~,\qquad [\mathcal{J}_y,\mathcal{J}_z] = i \mathcal{J}_x~,\qquad [\mathcal{J}_z,\mathcal{J}_x] = i \mathcal{J}_y~.
\end{equation}
The invariance of a Hamiltonian under a rotation then translates into the conservation of angular momentum:
\begin{equation}
[\hat{H},\mathcal{J}_z] = 0~,\qquad[\hat{H},\mathcal{J}^2] = 0~,
\label{cons_am}
\end{equation}
where
\begin{equation}
\mathcal{J}^2 = \mathcal{J}_x^2 + \mathcal{J}_y^2+ \mathcal{J}_z^2~.
\end{equation}
From Eq.~(\ref{cons_am}) and the fact that $[\mathcal{J}_z,\mathcal{J}^2] = 0$ it follows that the eigenstates of the Hamiltonian can be labeled by the eigenvalues of $\mathcal{J}^2$ and $\mathcal{J}_z$:
\begin{equation}
\hat{H}\ket{\Psi^N_{\mathcal{JM}}} = E^N_{\mathcal{JM}}\ket{\Psi^N_{\mathcal{JM}}}~,
\end{equation}
with
\begin{eqnarray}
\mathcal{J}^2\ket{\Psi^N_\mathcal{JM}} &=& \mathcal{J}(\mathcal{J}+1)\ket{\Psi^N_\mathcal{JM}}~,\\
\mathcal{J}_z\ket{\Psi^N_\mathcal{JM}} &=& \mathcal{M}\ket{\Psi^N_\mathcal{JM}}~.
\end{eqnarray}
We now expose how this symmetry can be used to speed up the algorithm by transforming the 2DM and its matrix maps into block-diagonal matrices.

The most frequently encountered symmetry of this kind is the symmetry under rotations in spin space. Physically this can be interpreted as the absence of a preferred direction for the axis of spin quantization. All atomic, molecular and lattice systems that were studied in this thesis have this symmetry. This symmetry would be broken if {\it e.g.} an external magnetic field is added to the Hamiltonian. We rewrite the single-particle basis to explicitly introduce the electron spin:
\begin{equation}
\ket{\alpha} \rightarrow \ket{a\sigma_a}~,
\end{equation}
in which $a$ is a spatial orbital index and $\sigma_a = \pm \frac{1}{2}$ is the spin projection.
To exploit this symmetry we transform the antisymmetric two-fermion basis, in which the 2DM is expressed, to a spin-coupled basis (See Appendix \ref{angular_momentum} or \cite{angmom}):
\begin{equation}
\ket{ab;SM_S} = \sum_{\sigma_a \sigma_b}\braket{\frac{1}{2}\sigma_a\frac{1}{2}\sigma_b}{SM_S}\ket{a\sigma_ab\sigma_b}~.
\label{tp_spin}
\end{equation}
Two spin-$\frac{1}{2}$ particles can couple to a spin $S = 0$ singlet and a $S=1$ triplet:
\begin{eqnarray}
\ket{ab;00} &=& \frac{1}{\sqrt{2}}\left[\ket{a\uparrow b\downarrow} - \ket{a\downarrow b\uparrow}\right]~,\\
\ket{ab;1-1} &=& \ket{a\uparrow b\uparrow}~,\\
\ket{ab;10} &=& \frac{1}{\sqrt{2}}\left[\ket{a\uparrow b\downarrow} + \ket{a\downarrow b\uparrow}\right]~,\\
\ket{ab;1+1} &=& \ket{a\downarrow b\downarrow}~,
\end{eqnarray}
which are respectively antisymmetric and symmetric in the spin-coordinates. Combining this with the permutation symmetry of the electrons the norm of Eq.~(\ref{tp_spin}) is not unity but:
\begin{equation}
\braket{ab;SM_S}{cd;S'M_S'} = \delta_{SS'}\delta_{M_SM_S'}\left(\delta_{ac}\delta_{bd}+(-1)^S\delta_{ad}\delta_{bc}\right)~,
\end{equation}
which means we have to define the spin-coupled (see Appendix~\ref{angular_momentum}) two-particle creation operators $B^\dagger$ as:
\begin{equation}
{B^\dagger}^{SM_S}_{ab} = \frac{1}{\sqrt{(ab)}}\left[a^\dagger_a \otimes a^\dagger_b\right]^S_M = \frac{1}{\sqrt{(ab)}}\sum_{\sigma_a\sigma_b}\braket{\frac{1}{2}\sigma_a\frac{1}{2}\sigma_b}{SM_S}a^\dagger_{a\sigma_a}a^\dagger_{b\sigma_b}~,
\label{tp_op_sc}
\end{equation}
where $(ab)$ is shorthand for $(1 + \delta_{ab})$ (see Appendix~\ref{math_notes}).
In this spin-coupled basis, the 2DM for an ensemble of states with spin $\mathcal{S}$ and spin projection $\mathcal{S}_z = \mathcal{M}$ is given by:
\begin{equation}
{}^{\mathcal{SM}}\Gamma^{SM_S;S'M_S'}_{ab;cd} = \sum_i w_i \bra{\Psi^N_{\mathcal{SM},i}}{B^\dagger}^{SM_S}_{ab}{B}^{S'M_S'}_{cd}\ket{\Psi^N_{\mathcal{SM},i}}~.
\label{2DM_sc_ensav}
\end{equation}
For the following it is convenient to rewrite Eq.~(\ref{2DM_sc_ensav}) in a form where the direct product of the operators  $B^\dagger$ and $B$ is coupled to a spherical tensor operator with total spin $S_T$ (see Appendix \ref{angular_momentum}):
\begin{eqnarray}
\nonumber{}^{\mathcal{SM}}\Gamma^{SM_S;S'M_S'}_{ab;cd} &=& (-1)^{S'-M_S'}\sum_i w_i\sum_{S_TM_T}\braket{SM_SS'-M_S'}{S_TM_T}\\
&&\qquad\qquad\qquad\qquad\bra{\Psi^N_{\mathcal{SM},i}}\left[{B^\dagger}^{S}_{ab}\otimes{\tilde{B}}^{S'}_{cd}\right]^{S_T}_{M_T}\ket{\Psi^N_{\mathcal{SM},i}}~.
\end{eqnarray}
As the bra and ket wave functions in the ensemble have the same $z$-component of total spin $\mathcal{M}$, the matrixelement in the above equations can only be non-zero if $M_T = 0$, which means that $M_S = M_S'$:
\begin{eqnarray}
\nonumber{}^{\mathcal{SM}}\Gamma^{SM_S;S'M_S'}_{ab;cd} &=& \delta_{M_SM_S'}(-1)^{S'-M_S}\sum_i w_i\sum_{S_T}\braket{SM_SS'-M_S}{S_T0}\\
&&\qquad\qquad\qquad\qquad\bra{\Psi^N_{\mathcal{SM},i}}\left[{B^\dagger}^{S}_{ab}\otimes{\tilde{B}}^{S'}_{cd}\right]^{S_T}_{0}\ket{\Psi^N_{\mathcal{SM},i}}~.
\label{2DM_sc_red}
\end{eqnarray}
For a further reduction of the 2DM we first consider the case where the total spin of the wave function is zero.
\subsection{Singlet ground state}
When the ground state is a spin singlet ($\mathcal{S} =\mathcal{M} = 0$), the operator in Eq.~(\ref{2DM_sc_red}) has to be a singlet too ($S_T = 0$), for the matrixelement to be non-zero. It follows that $S=S'$ and:
\begin{eqnarray}
\nonumber{}^{{00}}\Gamma^{SM_S;S'M_S'}_{ab;cd} &=& \delta_{M_SM_S'}\delta_{SS'}(-1)^{S-M_S}\sum_i w_i\braket{SM_SS-M_S}{00}\\
&&\qquad\qquad\qquad\qquad\qquad\qquad\qquad\bra{\Psi^N_{{00},i}}\left[{B^\dagger}^{S}_{ab}\otimes{\tilde{B}}^{S}_{cd}\right]^0_{0}\ket{\Psi^N_{{00},i}}~.
\label{2DM_sc_singlet_mdep}
\end{eqnarray}
The Clebsch-Gordan coefficient
\begin{equation}
\braket{SM_SS-M_S}{00} = \frac{(-1)^{S-M_S}}{[S]}~,
\end{equation}
implies that Eq.~(\ref{2DM_sc_singlet_mdep}) can be written as:
\begin{equation}
{}^{{00}}\Gamma^{SM_S;S'M_S'}_{ab;cd} = \frac{\delta_{M_SM_S'}\delta_{SS'}}{[S]}\sum_i w_i \bra{\Psi^N_{{00},i}}\left[{B^\dagger}^{S}_{ab}\otimes{\tilde{B}}^{S}_{cd}\right]^0_{0}\ket{\Psi^N_{{00},i}}~,
\end{equation}
which is independent of $M_S$. As a result there is a decomposition of the global 2DM into four diagonal blocks, one block with $S=0$, and three with $S = 1$. Those with $S=1$ are identical, so all matrix manipulations can be restriced to one copy. It follows that the basic object that has to be stored can be written as:
\begin{equation}
{}^{{00}}\Gamma^{SM_S;S'M_S'}_{ab;cd} \longrightarrow \Gamma^S_{ab;cd}~,
\end{equation}
as is illustrated in Figure~\ref{symmetry_fig}.
\subsection{\label{spin_averaged_ensemble}Higher-spin ground state}
It turns out that such an economical block decomposition is also possible for higher-spin ground states, provided we consider spin-averaged ensembles:
\begin{eqnarray}
\nonumber{}^{\mathcal{S}}\Gamma^{SM_S;S'M_S'}_{ab;cd} &=& \delta_{M_SM_S'}(-1)^{S'-M_S}\sum_i w_i\sum_{S_T}\braket{SM_SS'-M_S}{S_T0}\\
&&\qquad\qquad\qquad\frac{1}{[\mathcal{S}]^2}\sum_{\mathcal{M}}\bra{\Psi^N_{\mathcal{SM},i}}\left[{B^\dagger}^{S}_{ab}\otimes{\tilde{B}}^{S'}_{cd}\right]^{S_T}_{0}\ket{\Psi^N_{\mathcal{SM},i}}~,
\label{2DM_sc_hs}
\end{eqnarray}
in which an equal weight ensemble is taken over the different members of a multiplet $\ket{\Psi^N_{\mathcal{SM},i}}$. We are still able to reach the lowest energy by considering such ensembles, because the different members of a multiplet are degenerate when the Hamiltonian has spin symmetry. Using the Wigner-Eckart theorem (\ref{wigner_eckart}) we can rewrite the coupled matrixelement in Eq.~(\ref{2DM_sc_hs}) as:
\begin{eqnarray}
\nonumber\bra{\Psi^N_{\mathcal{SM},i}}\left[{B^\dagger}^{S}_{ab}\otimes{\tilde{B}}^{S'}_{cd}\right]^{S_T}_{0}\ket{\Psi^N_{\mathcal{SM},i}} &=& (-1)^{\mathcal{S}-\mathcal{M}}
\left(
\begin{matrix}
\mathcal{S}&\mathcal{S} & S_T\\
\mathcal{M}&-\mathcal{M}&0
\end{matrix}
\right)\\
&&\qquad\qquad\langle \Psi^N_{\mathcal{S},i}\|\left[{B^\dagger}^{S}_{ab}\otimes{\tilde{B}}^{S'}_{cd}\right]^{S_T}\|\Psi^N_{\mathcal{S},i}\rangle~.
\end{eqnarray}
Substituting
\begin{equation}
\frac{(-1)^{\mathcal{S}-\mathcal{M}}}{[\mathcal{S}]}
=
\left(
\begin{matrix}
\mathcal{S}&\mathcal{S} & 0\\
\mathcal{M}&-\mathcal{M}&0
\end{matrix}
\right)~,\\
\end{equation}
the sum over $\mathcal{M}$ can be used in Eq.~(\ref{2DM_sc_hs}), together with the orthogonality relation:
\begin{equation}
\sum_{\mathcal{M}}
\left(
\begin{matrix}
\mathcal{S}&\mathcal{S} & 0\\
\mathcal{M}&-\mathcal{M}&0
\end{matrix}
\right)
\left(
\begin{matrix}
\mathcal{S}&\mathcal{S} & S_T\\
\mathcal{M}&-\mathcal{M}&0
\end{matrix}
\right)
=\delta_{S_T0}~,
\end{equation}
to obtain:
\begin{equation}
{}^{\mathcal{S}}\Gamma^{SM_S;S'M_S'}_{ab;cd} = \delta_{M_SM_S'}\delta_{SS'}\frac{1}{[S][\mathcal{S}]}\sum_i w_i~\langle \Psi^N_{\mathcal{S},i}\|\left[{B^\dagger}^{S}_{ab}\otimes{\tilde{B}}^{S}_{cd}\right]^0\|\Psi^N_{\mathcal{S},i}\rangle~.
\label{2DM_ensemble_averaged}
\end{equation}
Eq.~(\ref{2DM_ensemble_averaged}) implies that by taking the spin-averaged ensemble, an expression for the 2DM is obtained which is diagonal in the two-particle spin $S$, and independent of the third component $M$. It follows that the 2DM has the same block reduction for higher-spin states as for the singlet, illustrated in Figure~\ref{symmetry_fig}.
\begin{figure}
\centering
\includegraphics[scale=0.5]{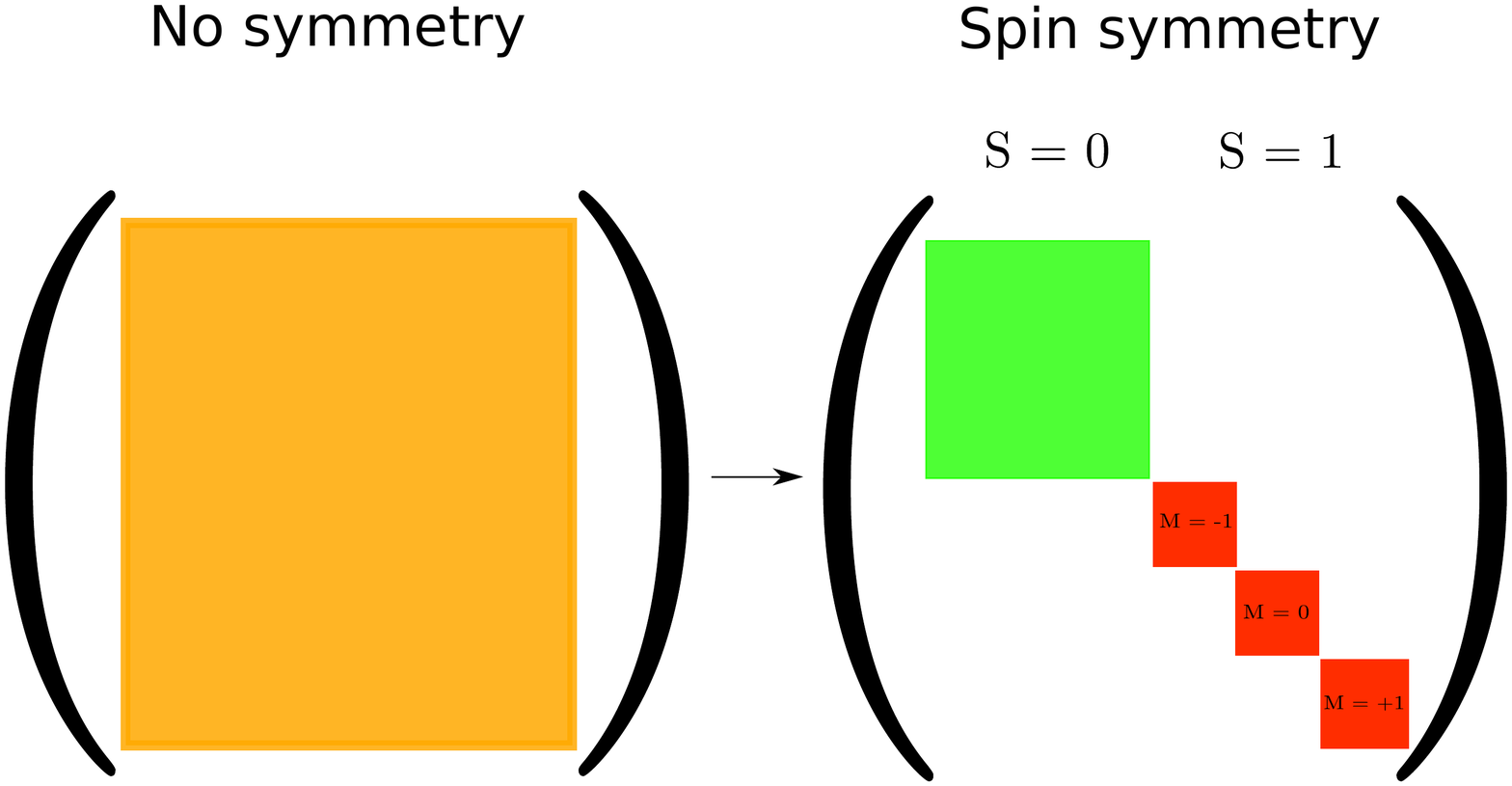}
\caption{\label{symmetry_fig} Illustration of how spin symmetry of the Hamiltonian induces a block decomposition of the 2DM.}
\end{figure}

To summarize, the spin-coupled 2DM, for all spin states $\mathcal{S}$, is related to the uncoupled 2DM as:
\begin{equation}
\Gamma^S_{ab;cd} = \frac{1}{\sqrt{(ab)(cd)}}\sum_{\sigma_a\sigma_b}\sum_{\sigma_c\sigma_d}\braket{\frac{1}{2}\sigma_a\frac{1}{2}\sigma_b}{SM}\braket{\frac{1}{2}\sigma_c\frac{1}{2}\sigma_d}{SM}\Gamma_{(a\sigma_a)(b\sigma_b);(c\sigma_c)(d\sigma_d)}~,
\label{2DM_sc}
\end{equation}
with the inverse transformation:
\begin{equation}
\Gamma_{(a\sigma_a)(b\sigma_b);(c\sigma_c)(d\sigma_d)} = \sqrt{(ab)(cd)}\sum_{SM}\braket{\frac{1}{2}\sigma_a\frac{1}{2}\sigma_b}{SM}\braket{\frac{1}{2}\sigma_c\frac{1}{2}\sigma_d}{SM}\Gamma^S_{ab;cd}~.
\label{2DM_inverse_sc}
\end{equation}
Note that the symmetry under the exchange of single-particle indices in the 2DM is now:
\begin{equation}
\Gamma^{S}_{ab;cd} = (-1)^S~\Gamma^S_{ba;cd}  = (-1)^S~\Gamma^S_{ab;dc} = \Gamma^S_{ba;dc}~,
\end{equation}
because of the extra phase $(-1)^{1 - S}$ that arises when the spins are exchanged in a Clebsch-Gordan coefficient. The $S=0$ block has dimension $\frac{M}{4}(\frac{M}{2}+1)$ as it is symmetrical in the spatial single-particle indices. The three identical $S=1$ blocks are antisymmetrical in spatial single-particle indices, and therefore have dimension $\frac{M}{4}(\frac{M}{2}-1)$. This sums up nicely to the dimension of the full 2DM:
\begin{equation}
\frac{M(M-1)}{2} = \frac{M}{4}\left(\frac{M}{2}+1\right)+3\left[\frac{M}{4}\left(\frac{M}{2}-1\right)\right]~.
\end{equation}
\subsection{Decomposition of the two-index constraints}
All linear matrix maps defined in Section \ref{standard_n_rep} have a similar decomposition in block-diagonal form. In this Section the spin-coupled form of the matrix maps is introduced. The 1DM for spin-symmetrical systems is defined as:
\begin{equation}
\rho_{(a\sigma_a)(b\sigma_b)} = \delta_{\sigma_a\sigma_b}\rho^{\sigma_a}_{ab} = \delta_{\sigma_a\sigma_b}\sum_i w_i \frac{1}{[\mathcal{S}]^2} \sum_\mathcal{M}\bra{\Psi^N_{\mathcal{SM},i}}a^\dagger_{a\sigma_a}a_{b\sigma_a} \ket{\Psi^N_{\mathcal{SM},i}}~,
\label{1DM_sc}
\end{equation}
from which it is clear that the 1DM is reduced to two blocks $\rho^\sigma_{ab}$: one for spin-up and one for spin-down particles. In the same way as for the 2DM it can be shown that by taking the spin-averaged ensemble the depence on $\sigma$ vanishes, which means the up and down spin blocks are identical.
The 1DM is derived from the 2DM as:
\begin{equation}
\rho^\sigma_{ac} = \frac{1}{N-1}\sum_\beta \Gamma_{(a\sigma)\beta;(c\sigma)\beta}~.
\end{equation}
By substituting the inverse transformation (\ref{2DM_inverse_sc}) and performing elementary angular momentum recoupling the 1DM can be derived from the spin-coupled 2DM as:
\begin{equation}
\rho_{ac} = \frac{1}{N-1}\sum_S \frac{[S]^2}{2}\sum_b\sqrt{(ab)(cb)}~\Gamma^S_{ab;cb}~.
\label{1DM_of_2DM_sc}
\end{equation}
\paragraph{The $\mathcal{Q}$ condition:}
the spin-coupled $\mathcal{Q}$ condition is defined as:
\begin{eqnarray}
\mathcal{Q}(\Gamma)^{S}_{ab;cd} &=& \sum_i w_i \frac{1}{[\mathcal{S}]^2}\sum_\mathcal{M} \bra{\Psi^N_{\mathcal{SM},i}}B^S_{ab}~{B^\dagger}^S_{cd}\ket{\Psi^N_{\mathcal{SM},i}}~,
\end{eqnarray}
in which the tensor operator $B^\dagger$ is again given by Eq.~(\ref{tp_op_sc}). In exactly the same way as for the 2DM the $\mathcal{Q}$ map decomposes into a singlet block and three identical triplet blocks. To derive the spin-coupled form of the $\mathcal{Q}$ map we have to substitute the uncoupled form Eq.~(\ref{Q_2DM}) into:
\begin{equation}
\mathcal{Q}(\Gamma)^{S}_{ab;cd} = \frac{1}{\sqrt{(ab)(cd)}}\sum_{\sigma_a\sigma_b}\sum_{\sigma_c\sigma_d}\braket{\frac{1}{2}\sigma_a\frac{1}{2}\sigma_b}{SM}\braket{\frac{1}{2}\sigma_c\frac{1}{2}\sigma_d}{SM}\mathcal{Q}(\Gamma)_{(a\sigma_a)(b\sigma_b);(c\sigma_c)(d\sigma_d)}~,
\end{equation}
and carry out straightforward angular momentum recoupling. Following this procedure leads us to the spin-coupled expression of the $\mathcal{Q}$ map:
\begin{eqnarray}
\nonumber\mathcal{Q}(\Gamma)^S_{ab;cd} &=& \frac{1}{\sqrt{(ab)(cd)}}\left(\delta_{ac}\delta_{bd} + (-1)^S \delta_{ad}\delta_{bc}\right)\frac{2\mathrm{Tr}~\Gamma}{N(N-1)} + \Gamma^S_{ab;cd}\\
&&-\frac{1}{\sqrt{(ab)(cd)}}\left(\delta_{ac}\rho_{bd} + (-1)^S\delta_{ad}\rho_{bc}+(-1)^S\delta_{bc}\rho_{ad} + \delta_{bd}\rho_{ac}\right)~.
\label{Q_sc}
\end{eqnarray}
\paragraph{The $\mathcal{G}$ condition:}
the spin recoupling of the $\mathcal{G}$ condition is a bit more involved. As explained in Appendix \ref{angular_momentum}, the correct spherical tensor form of a particle-hole operator is:
\begin{equation}
{A^\dagger}^S_{ab} = \left[a^\dagger_a \otimes \tilde{a}_b\right]^S \qquad\text{with}\qquad \tilde{a}_{b\sigma_b} = (-1)^{\frac{1}{2}+\sigma_b}a_{b-\sigma_b}~.
\label{ph_op_sc}
\end{equation}
With this operator, the spin-coupled $\mathcal{G}$ map can be defined through a spin-averaged ensemble:
\begin{equation}
\mathcal{G}(\Gamma)^S_{ab;cd} = \sum_i w_i \frac{1}{[\mathcal{S}]^2}\sum_\mathcal{M} \bra{\Psi^N_{\mathcal{SM},i}}{A^\dagger}^S_{ab}~{A}^S_{cd}\ket{\Psi^N_{\mathcal{SM},i}}~,
\label{spin_averaged_G_sc}
\end{equation}
which again decomposes into four blocks, one singlet and three identical triplet blocks. The expression of Eq.~(\ref{spin_averaged_G_sc}) in terms of the spin-coupled 2DM is found by substituting Eq.~(\ref{G1_2DM}) into:
\begin{eqnarray}
\nonumber\mathcal{G}(\Gamma)^{S}_{ab;cd} &=& \sum_{\sigma_a\sigma_b}\sum_{\sigma_c\sigma_d}(-1)^{1-\sigma_a-\sigma_b}\braket{\frac{1}{2}\sigma_a\frac{1}{2}-\sigma_b}{SM}\\
&&\qquad\qquad\qquad\braket{\frac{1}{2}\sigma_c\frac{1}{2}-\sigma_d}{SM}~\mathcal{G}(\Gamma)_{(a\sigma_a)(b\sigma_b);(c\sigma_c)(d\sigma_d)}~.
\end{eqnarray}
The first term of Eq.~(\ref{G1_2DM}), containing the 1DM, is easy to recouple. For the second term, however, the recoupling formula Eq.~(\ref{extremely_useful}) is needed, which introduces a Wigner-$6j$ symbol into the coupled form of $\mathcal{G}$ map:
\begin{equation}
\mathcal{G}(\Gamma)^S_{ab;cd} =\delta_{bd}\rho_{ac} -\sqrt{(ad)(cb)}\sum_{S'}[S']^2
\left\{
\begin{matrix}
\frac{1}{2}&\frac{1}{2}&S\\
\frac{1}{2}&\frac{1}{2}&S'
\end{matrix}
\right\}
\Gamma^{S'}_{ad;cb}~.
\label{G_sc}
\end{equation}
As there is no symmetry related to the spatial single-particle indices of the $\mathcal{G}$ matrix the four blocks all have equal dimension $\frac{M^2}{4}$, summing up to the total dimension of the uncoupled $\mathcal{G}$ matrix, {\it i.e.} $M^2$.
\subsection{Decomposition of the three-index constraints}
The spin coupling of the three-index conditions is much more complicated than for the two-index conditions. There are several ways to couple three spin-$\frac{1}{2}$ particles to a good total spin, as explained in Appendix \ref{angular_momentum}, depending on which two particles are coupled to good intermediate spin. We have chosen to first couple $a\sigma_a$ and $b\sigma_b$ to intermediate spin $S_{ab} = 0/1$, and then couple $c\sigma_c$ with $S_{ab}$ to total spin $S = \frac{1}{2}/\frac{3}{2}$:
\begin{equation}
\ket{abc;(S_{ab})SM} = \sum_{\sigma_a\sigma_b}\sum_{M_{ab}\sigma_c}\braket{\frac{1}{2} \sigma_a \frac{1}{2} \sigma_b}{S_{ab}M_{ab}}\braket{S_{ab}M_{ab}\frac{1}{2}\sigma_c}{SM}\ket{a\sigma_ab\sigma_bc\sigma_c}~.
\end{equation}
The choice for a complete orthogonal basisset of three-particle space is more subtle than before, and shown in Table~\ref{dp_dim}. That this is a correct choice is seen from the fact that the sum of the dimensions of the blocks, multiplied by their degeneracy, equals the total dimension of uncoupled three-particle space $\frac{M(M-1)(M-2)}{6}$.
\begin{table}
\caption{\label{dp_dim} A complete spin-coupled basis of three-particle space.}
\centering
\begin{tabular}{|ccccc|}
\hline
$S$&$S_{ab}$&orbitals&dimension&degeneracy\\
\hline
$\frac{1}{2}$ & 0 & $a=b\neq c$ & $\frac{M}{2}(\frac{M}{2}-1)$&2\\[1pt]
$\frac{1}{2}$ & 0 & $a < b < c$ & $\frac{M}{12}(\frac{M}{2}-1)(\frac{M}{2}-2)$&2\\[1pt]
$\frac{1}{2}$ & 1 & $a < b < c$ & $\frac{M}{12}(\frac{M}{2}-1)(\frac{M}{2}-2)$&2\\[1pt]
$\frac{3}{2}$ & 1 & $a < b < c$ & $\frac{M}{12}(\frac{M}{2}-1)(\frac{M}{2}-2)$&4\\[1pt]
\hline
\end{tabular}
\end{table}
Once again a norm $\sqrt{(ab)}$ has to be added because of the symmetry of the Clebsch-Gordan coefficients, which means a normalized three-particle operator is defined as:
\begin{equation}
{B^\dagger}^{S}_{ab(S_{ab})c} = \frac{1}{\sqrt{(ab)}} \sum_{\sigma_a\sigma_b}\sum_{M_{ab}\sigma_c}\braket{\frac{1}{2} \sigma_a \frac{1}{2} \sigma_b}{S_{ab}M_{ab}}\braket{S_{ab}M_{ab}\frac{1}{2}\sigma_c}{SM}a^\dagger_{a\sigma_a}a^\dagger_{b\sigma_b}a^\dagger_{c\sigma_c}~.
\label{dp_sc_op}
\end{equation}
There is another (technical) complication that arises from the specific choice of the basis as in Table \ref{dp_dim}. For a 2DM the elements are stored in an ordered basis $a \leq b$. If we want to access a matrix element where the order is reversed, we can use the simple relation:
\begin{equation}
\Gamma^S_{ab;cd} = (-1)^S\Gamma^S_{ba;cd}~.
\end{equation}
For a three-particle basis, however, this is much more complicated. An exchange of the first two indices is still simple:
\begin{equation}
\ket{abc;(S_{ab})SM} = (-1)^{S_{ab}}\ket{bac;(S_{ab})SM}~,
\end{equation}
but the symmetry between the first two and the third index is more involved. Suppose access is required to an element $abc$ with $c < b < a$, then the recoupling formula Eq.~(\ref{recoupling_6j}) has to be used:
\begin{equation}
\ket{abc;(S_{ab})SM}= \sqrt{\frac{(cb)}{(ab)}}[S_{ab}]\sum_{S_{cb}}[S_{cb}]\left\{\begin{matrix}S&\frac{1}{2}&S_{cb}\\\frac{1}{2}&\frac{1}{2}&S_{ab}\end{matrix}\right\}\ket{cba;(S_{cb})SM}~,
\label{recoupling_T1}
\end{equation}
to express it as a function of basis states that are stored. 
\paragraph{The $\mathcal{T}_1$ condition:} 
the spin-coupled $\mathcal{T}_1$ condition is defined as:
\begin{eqnarray}
\nonumber\mathcal{T}_1(\Gamma)^{S(S_{ab};S_{de})}_{abc;dez} &=& \sum_i w_i \frac{1}{[\mathcal{S}]^2}\sum_\mathcal{M}\left( \bra{\Psi^N_{\mathcal{SM},i}}{B^\dagger}^S_{ab(S_{ab})c}~{B}^S_{de(S_{de})z}\ket{\Psi^N_{\mathcal{SM},i}}\right.\\
&&\left.\qquad\qquad\qquad\qquad +\bra{\Psi^N_{\mathcal{SM},i}}{B}^S_{de(S_{de})z}~{B^\dagger}^S_{ab(S_{ab})c}\ket{\Psi^N_{\mathcal{SM},i}}\right)~,
\label{T1_sc_ensemble}
\end{eqnarray}
with $B^\dagger$ given by Eq.~(\ref{dp_sc_op}).
In an analogous way as for the two-index constraints, it can be shown that this expression is diagonal in $SM$ and independent of $M$. It follows that the $\mathcal{T}_1$ matrix is reduced to two identical $S=\frac{1}{2}$ blocks and four identical $S=\frac{3}{2}$ blocks. One can express Eq.~(\ref{T1_sc_ensemble}) as a function of the spin-coupled 2DM by substituting (\ref{T1}) in:
\begin{eqnarray}
\nonumber\mathcal{T}_1(\Gamma)^{S(S_{ab};S_{de})}_{abc;dez} &=& \frac{1}{\sqrt{(ab)(de)}}\sum_{\sigma_a\sigma_b}\sum_{\sigma_cM_{ab}} \braket{\frac{1}{2}\sigma_a\frac{1}{2}\sigma_b}{S_{ab}M_{ab}}\braket{S_{ab}M_{ab}\frac{1}{2}\sigma_c}{SM}\\
\nonumber&&\qquad\qquad\qquad\sum_{\sigma_d\sigma_e}\sum_{M_{de}\sigma_z}\braket{\frac{1}{2}\sigma_d\frac{1}{2}\sigma_e}{S_{de}M_{de}}\braket{S_{de}M_{de}\frac{1}{2}\sigma_z}{SM}\\
&&\qquad\qquad\qquad\qquad\mathcal{T}_1(\Gamma)_{(a\sigma_a)(b\sigma_b)(c\sigma_c);(d\sigma_d)(e\sigma_e)(z\sigma_z)}~.
\label{T1_formal_sc}
\end{eqnarray}
This leads to some tedious but straightforward angular momentum recoupling which results in the final spin-coupled expression:
\begin{align}
\nonumber\mathcal{T}_1(\Gamma)^{S(S_{ab};S_{de})}_{abc;dez} &=   \left[\frac{2\mathrm{Tr}~\Gamma}{N(N -1)}\right]\frac{\delta_{cz}\delta_{S_{ab}S_{de}}}{\sqrt{(ab)(de)}}\left(\delta_{ad}\delta_{be} + (-1)^{S_{ab}}\delta_{ae}\delta_{bd}\right) + \delta_{cz}\delta_{S_{ab}S_{de}}\Gamma^{S_{ab}}_{ab;de}\\
\nonumber& - \frac{\delta_{S_{ab}S_{de}}}{\sqrt{(ab)(de)}}\left[\delta_{cz}\left(\delta_{be}\rho_{ad}+(-1)^{S_{ab}}\delta_{ae}\rho_{bd} + (-1)^{S_{de}}\delta_{bd}\rho_{ae} + \delta_{ad}\rho_{be}\right)\right.\\
\nonumber&\left.\qquad\qquad\qquad\qquad\qquad\qquad\qquad\qquad\qquad + \left(\delta_{ad}\delta_{be} + (-1)^{S_{ab}}\delta_{bd}\delta_{ae}\right)\rho_{cz}\right]\\
\nonumber&-\frac{[S_{ab}][S_{de}]}{\sqrt{(ab)(de)}} \left\{\begin{matrix}S&\frac{1}{2}&S_{ab}\\\frac{1}{2}&\frac{1}{2}&S_{de}\end{matrix}\right\}\left[\delta_{az}\delta_{cd}\rho_{be} + (-1)^{S_{de}}\delta_{az}\delta_{ec}\rho_{bd} +(-1)^{S_{ab}}\delta_{bz}\delta_{cd}\rho_{ae}\right.\\
\nonumber\nonumber&\left.\qquad\qquad\qquad\qquad\qquad+(-1)^{S_{ab}+S_{de}}\delta_{bz}\delta_{ec}\rho_{ad} + \left(\delta_{be}\delta_{cd} +(-1)^{S_{de}}\delta_{bd}\delta_{ce}\right)\rho_{az}\right.\\
\nonumber&\left.\qquad\qquad + (-1)^{S_{ab}}\left(\delta_{ae}\delta_{cd}+(-1)^{S_{de}}\delta_{ce}\delta_{ad}\right)\rho_{bz} + \left(\delta_{az}\delta_{be}+(-1)^{S_{ab}}\delta_{bz}\delta_{ae}\right)\rho_{cd}\right.\\
\nonumber&\left.\qquad\qquad\qquad\qquad\qquad\qquad\qquad\qquad + (-1)^{S_{de}}(\delta_{az}\delta_{bd}+(-1)^{S_{ab}}\delta_{bz}\delta_{ad})\rho_{ce}\right]\\
\nonumber& + [S_{ab}][S_{de}]\left\{\begin{matrix}S&\frac{1}{2}&S_{ab}\\\frac{1}{2}&\frac{1}{2}&S_{de}\end{matrix}\right\}\left[\delta_{az}\sqrt{\frac{(bz)}{(ab)}}\Gamma^{S_{de}}_{cb;de} + \delta_{bz}(-1)^{S_{ab}+S_{de}}\sqrt{\frac{(ac)}{(ab)}}\Gamma^{S_{de}}_{ac;de}\right.\\
\nonumber&\left.\qquad\qquad\qquad\qquad\qquad+ \delta_{ce}(-1)^{S_{ab}+S_{de}}\sqrt{\frac{(dz)}{(de)}}\Gamma^{S_{ab}}_{ab;dz}+\delta_{cd}\sqrt{\frac{(ez)}{(de)}}\Gamma^{S_{ab}}_{ab;ze}\right]\\
\nonumber&+\frac{[S_{ab}][S_{de}]}{\sqrt{(ab)(de)}}\sum_{S'} [S']^2\left\{\begin{matrix}S&\frac{1}{2}&S'\\\frac{1}{2}&\frac{1}{2}&S_{ab}\end{matrix}\right\}\left\{\begin{matrix}S&\frac{1}{2}&S'\\\frac{1}{2}&\frac{1}{2}&S_{de}\end{matrix}\right\}\left[\delta_{ad}\sqrt{(bc)(ez)}\Gamma^{S'}_{bc;ez}\right.\\
\nonumber&\left.\qquad\qquad + \delta_{bd}(-1)^{S_{ab}}\sqrt{(ac)(ez)}\Gamma^{S'}_{ac;ez} + \delta_{ae}(-1)^{S_{de}}\sqrt{(bc)(dz)}\Gamma^{S'}_{bc;dz}\right.\\
&\left.\qquad\qquad\qquad\qquad\qquad +\delta_{be}(-1)^{S_{ab}+S_{de}}\sqrt{(ac)(dz)}\Gamma^{S'}_{ac;dz}\right]~.
\end{align}
\paragraph{The $\mathcal{T}_2$ condition:}
for the spin-coupled $\mathcal{T}_2$ condition a good spherical tensor operator is needed which creates two particles and one hole, 
\begin{equation}
{B^\dagger}^{S}_{ab(S_{ab})c} =\frac{1}{\sqrt{(ab)}} \sum_{\sigma_a\sigma_b}\sum_{M_{ab}\sigma_c}\braket{\frac{1}{2}\sigma_a\frac{1}{2}\sigma_b}{S_{ab}M_{ab}}\braket{S_{ab}M_{ab}\frac{1}{2}\sigma_c}{SM}a^\dagger_{a\sigma_a}a^\dagger_{b\sigma_b}\tilde{a}_{c\sigma_c}~.
\label{pph_sc_op}
\end{equation}
Note the presence of the $\tilde{a}$ operator, similar to Eq.~(\ref{ph_op_sc}).
A choice for the orbital ordering which generates a complete basisset for spin-coupled two-particle-one-hole space is shown in Table~\ref{pph_dim}. In contrast to the three-particle case, there is no recoupling needed to access elements of the form $c<b<a$ since there is no symmetry involving the third index. A simple addition of the dimensions multiplied with their respective degeneracies shows that this basis spans the same dimension as the uncoupled particle-hole basis, {\it i.e.} $\frac{M^2(M-1)}{2}$. 
\begin{table}
\caption{\label{pph_dim} A complete spin-coupled basis of two-particle-one hole space.}
\centering
\begin{tabular}{|ccccc|}
\hline
$S$&$S_{ab}$&sp-coords&dimension&degeneracy\\
\hline
$\frac{1}{2}$ & 0 & $a \leq b , c$ & $\frac{M^2}{8}(\frac{M}{2}+1)$ & 2\\[1pt]
$\frac{1}{2}$ & 1 & $a < b , c$ & $\frac{M^2}{8}(\frac{M}{2}-1)$ & 2\\[1pt]
$\frac{3}{2}$ & 1 & $a < b , c$ & $\frac{M^2}{8}(\frac{M}{2}-1)$ & 4\\[1pt]
\hline
\end{tabular}
\end{table}
Using the operator defined in Eq.~(\ref{pph_sc_op}) one can proceed to define the spin-coupled $\mathcal{T}_2$ condition as:
\begin{eqnarray}
\nonumber\mathcal{T}_2(\Gamma)^{S(S_{ab};S_{de})}_{abc;dez} &=& \sum_i w_i \frac{1}{[\mathcal{S}]^2}\sum_\mathcal{M}\left( \bra{\Psi^N_{\mathcal{SM},i}}{B^\dagger}^S_{ab(S_{ab})c}~{B}^S_{de(S_{de})z}\ket{\Psi^N_{\mathcal{SM},i}}\right.\\
&&\left.\qquad\qquad\qquad\qquad +\bra{\Psi^N_{\mathcal{SM},i}}{B}^S_{de(S_{de})z}~{B^\dagger}^S_{ab(S_{ab})c}\ket{\Psi^N_{\mathcal{SM},i}}\right)~.
\label{T2_sc_ensemble}
\end{eqnarray}
One obtains the spin-coupled form of the $\mathcal{T}_2$ map by substituting Eq.~(\ref{T2}) into:
\begin{align}
\nonumber\mathcal{T}_2(\Gamma)^{S(S_{ab};S_{de})}_{abc;dez} &= \frac{1}{\sqrt{(ab)(de)}}\sum_{\sigma_a\sigma_b}\sum_{M_{ab}\sigma_c}(-1)^{\frac{1}{2} -\sigma_c} \braket{\frac{1}{2}\sigma_a\frac{1}{2}\sigma_b}{S_{ab}M_{ab}}\braket{S_{ab}M_{ab}\frac{1}{2}-\sigma_c}{SM}\\
\nonumber&\qquad\qquad\qquad\sum_{\sigma_d\sigma_e}\sum_{M_{de}\sigma_z}(-1)^{\frac{1}{2} - \sigma_z}\braket{\frac{1}{2}\sigma_d\frac{1}{2}\sigma_e}{S_{de}M_{de}}\braket{S_{de}M_{de}\frac{1}{2}-\sigma_z}{SM}\\
& \qquad\qquad\qquad\qquad\qquad\mathcal{T}_2(\Gamma)_{(a\sigma_a)(b\sigma_b)(c\sigma_c);(d\sigma_d)(e\sigma_e)(z\sigma_z)}~.
\label{pph_coupling}
\end{align}
Notice the extra phases and minus signs arising from the use of $\tilde{a}$.
In the same way as for the $\mathcal{T}_1$ matrix, the $\mathcal{T}_2$ matrix falls apart in six diagonal blocks, two identical spin-$\frac{1}{2}$ and four identical spin-$\frac{3}{2}$ blocks.
The first two terms in Eq.~(\ref{T2}) are of a type that we encountered in previous calculations. The last four terms (containing a 2DM) are a bit more involved. First we have to use Eq.~(\ref{extremely_useful}) three times, generating three Wigner-6$j$ symbols which can be reduced to one Wigner-$9j$  symbol using Eq.~(\ref{6j_to_9j}). This leads to the following final spin-coupled expression for the $\mathcal{T}_2$ map:
\begin{align}
\nonumber\mathcal{T}_2(\Gamma)^{S(S_{ab};S_{de})}_{abc;dez} &= \frac{\delta_{S_{ab}S_{de}}}{\sqrt{(ab)(de)}}\left(\delta_{ad}\delta_{be}+(-1)^{S_{ab}}\delta_{ae}\delta_{bd}\right)\rho_{cz} + \delta_{cz}\delta_{S_{ab}S_{de}}\Gamma^{S_{ab}}_{ab;de}\\*
\nonumber&-\frac{[S_{ab}][S_{de}]}{\sqrt{(ab)(de)}}\sum_{S'} [S']^2
\left\{
\begin{matrix}
S&\frac{1}{2}&S_{de}\\
\frac{1}{2}&S'&\frac{1}{2}\\
S_{ab}&\frac{1}{2}&\frac{1}{2}
\end{matrix}
\right\}
\left(\delta_{ad}\sqrt{(ce)(bz)}\Gamma^{S'}_{ce;zb}\right.\\*
\nonumber&\left.+ (-1)^{S_{ab}}\delta_{bd}\sqrt{(ce)(za)}\Gamma^{S'}_{ce;za} + (-1)^{S_{de}}\delta_{ae}\sqrt{(cd)(zb)}\Gamma^{S'}_{cd;zb}\right.\\*
&\left. \qquad\qquad\qquad\qquad\qquad\qquad+ (-1)^{S_{ab}+S_{de}}\delta_{be}\sqrt{(cd)(za)}\Gamma^{S'}_{cd;za}\right)~.
\label{sc_T2}
\end{align}
\paragraph{The $\mathcal{T}_2'$ condition:} 
the spin-coupled $\mathcal{T}_2'$ condition is different from the other constraints, in that the sum of two spins appears in the operator that generates the constraint:
\begin{equation}
{B^\dagger}^S_{ab(S_{ab})c,m} = \sqrt{\frac{1}{(ab)}}\left[\left[a^\dagger_a\otimes a^\dagger_b\right]^{S_{ab}}\otimes\tilde{a}_c\right]^S_M + a^\dagger_{m\sigma_m}~.
\end{equation}
The extra one-particle term only contributes to the result when the two-particle-one-hole operator couples to $S=\frac{1}{2}$; otherwise the spin-averaged ensemble is zero because of spin conservation. This means that only the $S=\frac{1}{2}$ part of the $\mathcal{T}_2'$ matrix is different from the regular $\mathcal{T}_2$. Performing the angular momentum recoupling yields the spin-coupled form of the $\mathcal{T}_2'$ condition:
\begin{equation}
\mathcal{T}_2'(\Gamma)^{\frac{1}{2}(S_{ab};S_{de})}_{abc,m;dez,n} = 
\left(
\begin{matrix}
\mathcal{T}_2(\Gamma)^{\frac{1}{2}(S_{ab};S_{de})}_{abc;dez} & \omega^{(S_{ab})}_{abc;n}\\
{\omega^\dagger}^{(S_{de})}_{m;dez} & \rho_{mn}
\end{matrix}
\right)~,
\end{equation}
with $\omega$ defined as:
\begin{align}
\nonumber\omega^{(S_{ab})}_{abc;n} =& \frac{1}{\sqrt{(ab)}}\sum_{\sigma_a\sigma_b}\sum_{M_{ab}\sigma_c}(-1)^{\frac{1}{2}-\sigma_c}\braket{\frac{1}{2}\sigma_a\frac{1}{2}\sigma_b}{S_{ab}M_{ab}}\braket{S_{ab}M_{ab}\frac{1}{2}-\sigma_c}{\frac{1}{2}\sigma_n}\omega_{(a\sigma_a)(b\sigma_b)(c\sigma_c);(n\sigma_n)}\\
=&\frac{[S_{ab}]}{\sqrt{2}}(-1)^{1+S_{ab}}\sqrt{(nc)}~\Gamma^{S_{ab}}_{ab;nc}~.
\end{align}
The $S=\frac{3}{2}$ block of $\mathcal{T}_2'$ is identical to the that in $\mathcal{T}_2$:
\begin{equation}
\mathcal{T}_2'(\Gamma)^{\frac{3}{2}(1;1)}_{abc;dez} = \mathcal{T}_2(\Gamma)^{\frac{3}{2}(1;1)}_{abc;dez}~.
\end{equation}
\section{\label{sym_atom}Symmetry in atomic systems}
The atomic Hamiltonian in the Born-Oppenheimer approximation consists of a single-particle term, including the kinetic energy of the electrons and their attractive interaction with the positive central charge $Z$, and a two-particle term describing the Coulombic repulsion between the electrons:
\begin{equation}
\hat{H} = -\sum_{i = 1}^N\left(\frac{1}{2}\nabla_i^2 + \frac{Z}{r_i}\right) + \sum_{i<j}^N\frac{1}{|\mathbf{r_i} - \mathbf{r_j}|}~,
\label{ham_atom}
\end{equation}
in which atomic units have been used.
Atomic systems are highly symmetric: the spin symmetry explained in the previous Section is present, but also symmetry under spatial rotations and under a reflection through the origin. The combined symmetries can be exploited to obtain a very large reduction of the dimensions of the blocks in all the matrices involved in the SDP, which leads to a huge speedup of the program.
\subsection{Symmetry under spatial rotations}
Symmetry under spatial rotations is of the same type as symmetry under spin rotations, and leads to the introduction of (orbital) angular momentum. The single-particle basis is transformed to explicity introduce the extra angular momentum quantum number:
\begin{equation}
\ket{\alpha} \rightarrow \ket{am_a\sigma_a}\rightarrow\ket{(n_al_a)m_a\sigma_a}~.
\end{equation}
The spatial orbitals $a$ are labeled by a main quantum number $n$, and an angular momentum $l$. In the discussion of spin symmetry for spin-$\frac{1}{2}$ electrons, the total two-particle spin was limited to $S = 0$  and $S = 1$. For angular momentum, depending on the basis used, the orbital angular momentum, $l$, takes on different \emph{integer} values, corresponding to the $s,p,d,f,g,\ldots$ orbitals. This implies that the two-particle angular momentum $L$ can have many different values. A two-particle state in spin and angular momentum coupled form is defined through the relation:
\begin{equation}
\ket{ab;LM_LSM_S} = \sum_{\sigma_a\sigma_b}\sum_{m_am_b}\braket{\frac{1}{2}\sigma_a\frac{1}{2}\sigma_b}{SM_S}\braket{l_am_al_bm_b}{LM_L}\ket{am_a\sigma_a bm_b\sigma_b}~.
\end{equation}
The norm of such a state is not unity:
\begin{align}
\nonumber\braket{ab;LM_LS_MS}{cd;L'M_L'S'M_S'} = \delta_{SS'}\delta_{LL'}&\delta_{M_LM_L'}\delta_{M_SM_S'}\\
&\left(\delta_{ac}\delta_{bd}+(-1)^{L+S+l_a+l_b}\delta_{ad}\delta_{bc}\right)~,
\label{tp_state_sclc}
\end{align}
implying that the normalized two-particle creation operators, which form the basis for expressing the 2DM, are defined as:
\begin{equation}
{B^\dagger}^{LS}_{ab} = \frac{1}{\sqrt{(ab)}}\sum_{\sigma_a\sigma_b}\sum_{m_am_b}\braket{\frac{1}{2}\sigma_a\frac{1}{2}\sigma_b}{SM_S}\braket{l_am_al_bm_b}{LM_L} a^\dagger_{am_a\sigma_a}a^\dagger_{bm_b\sigma_b}~.
\label{tp_op_scac}
\end{equation}
Note that the symmetry of a two-particle state under the exchange of the two spatial indices $a$ and $b$, depends on $L$, $S$ \emph{and} the single-particle angular momenta $l_a$ and $l_b$. This is a problem because, in a block with two-particle quantum numbers $LS$, the single-particle indices can have a different permutation symmetry. This is solved by the introduction of the parity quantum number in the next Section.
\subsection{Reflection symmetry and parity}
The eigenstates of a Hamiltonian symmetric under spatial inversion ($\left[\hat{H},\hat{P}\right] = 0$ where $\hat{P}$ changes $\mathbf{r}$ in $-\mathbf{r}$) can be chosen as simultaneous eigenstates of the inversion operator. Since $\hat{P}^2 = \mathbb{1}$, $\hat{P}$ can have only two distinct eigenvalues:
\begin{equation}
\hat{P}\ket{\Psi^\pi} = \pi\ket{\Psi^\pi}~,
\end{equation}
where $\pi = \pm1$ is called the parity of the state $\ket{\Psi^\pi}$.

The single-particle states in a spherical basis can be decomposed in a radial part $f_{nl}(r)$, and an angular dependend part, described by a spherical harmonic $Y_{lm}(\theta,\phi)$. Using the parity property of spherical harmonics, we have that:
\begin{equation}
\hat{P}Y_{lm}(\theta,\phi) = (-1)^l Y_{lm}(\theta,\phi)~.
\end{equation}
It follows that the spin and angular momentum coupled two-particle states introduced in Eq.~(\ref{tp_state_sclc}) have good parity:
\begin{equation}
\hat{P}\ket{ab;LM_LSM_S} = (-1)^{l_a+l_b}\ket{ab;LM_LSM_S}~.
\end{equation}
This solves the problem with permutation symmetry mentioned in the previous Section, since the exchange of the spatial orbitals can now be written in terms of the two-particle block coordinates $L,S$ and $\pi$:
\begin{equation}
\ket{ab;LM_LSM_S\pi} = \pi(-1)^{L+S}\ket{ba;LM_LSM_S\pi}~.
\end{equation}
\subsection{The 2DM for atomic systems}
All symmetries present in an atomic system can be exploited by considering a spin and angular momentum averaged ensemble:
\begin{align}
\Gamma^{L^\pi S}_{ab;cd} =& \sum_iw_i~\frac{1}{[\mathcal{S}]^2[\mathcal{L}]^2}\sum_{\mathcal{M_SM_L}}
\bra{\Psi^{N\Pi}_{\mathcal{SM_SLM_L},i}}
{B^\dagger}^{L^\pi S}_{ab}~B^{L^\pi S}_{cd}
\ket{\Psi^{N\Pi}_{\mathcal{SM_SLM_L},i}}~,
\end{align}
in which the two-particle creation operator $B^\dagger$ is given by Eq.~(\ref{tp_op_scac}). To a block $L^\pi S$, only those pairs $ab$ contribute for which the angular momenta $l_a$ and $l_b$ satisfy the triangle relation:
\begin{equation}
|l_a - l_b| \leq L \leq l_a + l_b~,
\end{equation}
and have the correct parity, {\it i.e.}
\begin{equation}
(-1)^{l_a+l_b} = \pi~.
\end{equation}
The spin and angular momentum coupled 2DM is related to the uncoupled 2DM through a unitary transformation involving four Clebsch-Gordan coefficients:
\begin{align}
\nonumber\Gamma^{L^\pi S}_{ab;cd} = \frac{1}{\sqrt{(ab)(cd)}}&
\sum_{\sigma_a\sigma_b}\sum_{m_am_b}\braket{\frac{1}{2}\sigma_a\frac{1}{2}\sigma_b}{SM_S}\braket{l_am_al_bm_b}{LM_L}\\
\nonumber&\sum_{\sigma_c\sigma_d}\sum_{m_cm_d}\braket{\frac{1}{2}\sigma_c\frac{1}{2}\sigma_d}{SM_S}\braket{l_cm_cl_dm_d}{LM_L}\\
&\qquad\qquad\qquad\Gamma_{(am_a\sigma_a)(bm_b\sigma_b);(cm_c\sigma_c)(dm_d\sigma_d)}~.
\end{align}
To appreciate what an enormous reduction is obtained in the dimensions of the 2DM blocks, the dimensions and degeneracies of the different blocks are shown in Table \ref{table_2DM_reduction_LSp} for the cc-pVDZ, cc-pVTZ and cc-pVQZ atomic basis sets \cite{dunning} for first row atoms (from $\mathrm{Li}$ to $\mathrm{Ne}$). A cc-pVDZ basis consists of 28 spinorbitals,
\[(1s)^2 (2s)^2 (2p)^6(3s)^2(3p)^6(3d)^{10}~,\]
implying that the dimension of the two-particle space is 378. In the spin and angular momentum coupled version of the program we only have to store 14 blocks, the largest of which has a dimension of 10. The scaling of the heaviest matrix computations in the program goes as $O(n^3)$. If we compare the computational cost a such a computation on an uncoupled 2DM with the cost on 2DM in the fully coupled program, we see a increase in efficiency with a factor of 20821. The speedup becomes even more spectacular for the larger cc-pVTZ basis, which has 60 spinorbitals,
\[(1s)^2 (2s)^2 (2p)^6(3s)^2(3p)^6(3d)^{10} (4s)^2 (4p)^6 (4d)^{10} (4f)^{14}~,\]
and for which the dimension of two-particle space is 1770. This reduces to storing 22 blocks, the largest of which has a dimension of 21. The reduction factor in the computational cost of an elementary matrix computation on a 2DM is 125948. Finally, for the cc-pVQZ basis there are 110 spinorbitals
\[(1s)^2 (2s)^2 (2p)^6(3s)^2(3p)^6(3d)^{10} (4s)^2 (4p)^6 (4d)^{10} (4f)^{14} (5s)^2 (5p)^6(5d)^{10}(5f)^{14}(5g)^{18},\]
with a two-particle space dimension of 5995, which reduces to 30 blocks with a maximal dimension of 46. The reduction factor in the computational cost of an elementary matrix computation on a 2DM is 456052. This example shows that symmetry can reduce the complexity of the problem enormously and should be exploited to the full.
\begin{table}
\caption{\label{table_2DM_reduction_LSp} Dimensions and degeneracies of the 2DM blocks for first-row atoms (from $\mathrm{Li}$ to $\mathrm{Ne}$) in a cc-pVDZ, cc-pVTZ and cc-pVQZ basis set.}
\centering
\begin{tabular}{|c|c|c|c|c|c|c|c|c|c|}
\hline
&\multicolumn{3}{c}{dimension}\vline&deg&&\multicolumn{3}{c}{dimension}\vline&deg\\
\hline
$(L^\pi S)$ & DZ & TZ& QZ & &$(L^\pi S)$ & DZ & TZ& QZ &\\ 
\hline
$(0^+0)$&10&20&35&1&$(0^+1)$&4&10&20&3\\
$(1^+0)$&1&4&10&3&$(1^+1)$&4&10&20&9\\
$(1^-0)$&10&20&40&3&$(1^-1)$&8&20&40&9\\
$(2^+0)$&7&21&46&5&$(2^+1)$&4&15&36&15\\
$(2^-0)$&2&8&20&5&$(2^-1)$&2&8&20&15\\
$(3^+0)$&0&4&15&7&$(3^+1)$&1&7&21&21\\
$(3^-0)$&2&12&34&7&$(3^-1)$&2&12&34&21\\
$(4^+0)$&1&7&26&9&$(4^+1)$&0&4&20&27\\
$(4^-0)$&0&2&12&9&$(4^-1)$&0&2&12&27\\
$(5^+0)$&0&0&4&11&$(5^+1)$&0&1&7&33\\
$(5^-0)$&0&2&12&11&$(5^-1)$&0&2&12&33\\
$(6^+0)$&0&1&7&13&$(6^+1)$&0&0&4&39\\
$(6^-0)$&0&0&2&13&$(6^-1)$&0&0&2&39\\
$(7^+0)$&0&0&0&15&$(7^+1)$&0&0&1&45\\
$(7^-0)$&0&0&2&15&$(7^-1)$&0&0&2&45\\
$(8^-0)$&0&0&1&17&$(8^-1)$&0&0&0&51\\
\hline
\end{tabular}
\end{table}
\subsection{Two-index constraints for atomic systems}
The 1DM for atomic systems using the spin and angular momentum averaged ensemble reads:
\begin{align}
\nonumber\rho_{(am_a\sigma_a);(cm_c\sigma_c)} =& \sum_iw_i~\frac{1}{[\mathcal{S}]^2[\mathcal{L}]^2}\sum_{\mathcal{M_SM_L}}
\bra{\Psi^{N\Pi}_{\mathcal{SM_SLM_L},i}}
a^\dagger_{am_a\sigma_a}a_{cm_c\sigma_c}
\ket{\Psi^{N\Pi}_{\mathcal{SM_SLM_L},i}}\\
=& \delta_{l_al_c}\delta_{m_am_c}\delta_{\sigma_a\sigma_c}\rho^{(l_a)}_{n_an_c}~,
\end{align}
and is diagonal in $\sigma,m$ and $l$. A block labeled by single-particle angular momentum $l$ has a degeneracy of $2(2l+1)$ and a dimension equal to the number of spatial orbitals $(nl)$ that have angular momentum $l$, {\it e.g.} for the cc-pVDZ basis the dimension of the $s$ block equals three, since the basis has three types of $s$ orbitals, $1s$, $2s$ and $3s$. The 1DM can be derived from the 2DM, in a similar fashion as in Eq.~(\ref{1DM_of_2DM_sc}):
\begin{equation}
\rho^{(l)}_{n_an_c} = \frac{1}{2[l]^2}\frac{1}{N-1}\sum_\pi\sum_{LS}[L]^2[S]^2\sum_b\sqrt{(ab)(cb)}~\Gamma^{L^\pi S}_{(n_al)b;(n_cl)b}~.
\end{equation}
The matrix constraints in spin and angular momentum coupled form are derived through procedures that are completely analogous as those discussed for spin symmetry in the previous Section. Since two types of recoupling have to be carried through (spin and angular momentum), all expressions become a bit more elaborate. For the sake of completeness we nevertheless provide them in the remainder of this Section.
\paragraph{The $\mathcal{Q}$ condition:}
the spin and angular momentum coupled $\mathcal{Q}$ condition has the same block reduction as the 2DM (as shown in Table \ref{table_2DM_reduction_LSp}). It is defined through a spin and angular momentum averaged ensemble:
\begin{equation}
\mathcal{Q}(\Gamma)^{L^\pi S}_{ab;cd} = \sum_iw_i~\frac{1}{[\mathcal{S}]^2[\mathcal{L}]^2}\sum_{\mathcal{M_SM_L}}
\bra{\Psi^{N\Pi}_{\mathcal{SM_SLM_L},i}}
B^{L^\pi S}_{ab}{B^\dagger}^{L^\pi S}_{cd}
\ket{\Psi^{N\Pi}_{\mathcal{SM_SLM_L},i}}~,
\end{equation}
with $B^\dagger$ given by (\ref{tp_op_scac}). Expressed as a function of the 2DM this becomes:
\begin{align}
\nonumber\mathcal{Q}(\Gamma)^{L^\pi S}_{ab;cd} =& \frac{1}{\sqrt{(ab)(cd)}}\left(\delta_{ac}\delta_{bd}+\pi(-1)^{L+S}\delta_{ad}\delta_{bc}\right)\frac{2\mathrm{Tr}~
\Gamma}{N(N-1)} + \Gamma^{L^\pi S}_{ab;cd}\\
\nonumber&\qquad-\frac{1}{\sqrt{(ab)(cd)}}\left(\delta_{bd}\delta_{l_al_c}\rho^{(l_a)}_{n_an_c} + \delta_{ac}\delta_{l_bl_d}\rho^{(l_b)}_{n_bn_d}\right.\\
&\qquad\qquad\qquad\left.+ \pi(-1)^{L+S}\left[ \delta_{ad}\delta_{l_bl_c}\rho^{(l_b)}_{n_bn_c} + \delta_{bc}\delta_{l_al_d}\rho^{(l_a)}_{n_an_d}\right]\right)~.
\end{align}
\paragraph{The $\mathcal{G}$ condition:}
for the $\mathcal{G}$ condition it is necessary to define a particle-hole operator using a slightly modified annihilation operator to construct a good spherical tensor operator:
\begin{equation}
{A^\dagger}^{L^\pi S}_{ab} = \left[a^\dagger_a \otimes \tilde{a}_b \right]^{LS}~,
\end{equation}
with $\tilde{a}$ defined as:
\begin{equation}
\tilde{a}_{nlm\sigma} = (-1)^{l + m}(-1)^{\frac{1}{2}+\sigma}a_{nl;-m;-\sigma}~.
\label{LS_anni_op}
\end{equation}
With this operator the spin and angular momentum coupled $\mathcal{G}$ condition becomes:
\begin{equation}
\mathcal{G}(\Gamma)^{L^\pi S}_{ab;cd} = \sum_iw_i~\frac{1}{[\mathcal{S}]^2[\mathcal{L}]^2}\sum_{\mathcal{M_SM_L}}
\bra{\Psi^{N\Pi}_{\mathcal{SM_SLM_L},i}}
{A^\dagger}^{L^\pi S}_{ab}{A}^{L^\pi S}_{cd}
\ket{\Psi^{N\Pi}_{\mathcal{SM_SLM_L},i}}~.
\end{equation}
After some angular momentum recoupling the $\mathcal{G}$ condition can be expressed as a function of the 2DM as:
\begin{equation}
\mathcal{G}(\Gamma)^{L^\pi S}_{ab;cd} = \delta_{bd}\delta_{l_al_c}\rho^{(l_a)}_{n_an_c} -\sqrt{(ad)(cb)}\sum_{\pi'}\sum_{L'S'}[L']^2[S']^2 
\left\{
\begin{matrix}
\frac{1}{2}&\frac{1}{2}&S\\
\frac{1}{2}&\frac{1}{2}&S'
\end{matrix}
\right\}
\left\{
\begin{matrix}
l_a&l_b&L\\
l_c&l_d&L'
\end{matrix}
\right\}
\Gamma^{{L'}^{\pi'} S'}_{ad;cb}~.
\end{equation}
\begin{table}
\caption{\label{table_ph_reduction_LSp}Dimensions and degeneracies of the $\mathcal{G}$-matrix blocks for first-row atoms (from $\mathrm{Li}$ to $\mathrm{Ne}$) in a cc-pVDZ, cc-pVTZ and cc-pVQZ basis set.}
\centering
\begin{tabular}{|c|c|c|c|c|c|c|c|c|c|}
\hline
&\multicolumn{3}{c}{dimension}\vline&deg&&\multicolumn{3}{c}{dimension}\vline&deg\\
\hline
$(L^\pi S)$ & DZ & TZ& QZ & &$(L^\pi S)$ & DZ & TZ& QZ &\\ 
\hline
$(0^+0)$&14&30&55&1&$(0^+1)$&14&30&55&3\\
$(1^+0)$&5&14&30&3&$(1^+1)$&5&14&30&9\\
$(1^-0)$&16&40&80&3&$(1^-1)$&16&40&80&9\\
$(2^+0)$&11&36&82&5&$(2^+1)$&11&36&82&15\\
$(2^-0)$&4&16&40&5&$(2^-1)$&4&16&40&15\\
$(3^+0)$&1&11&36&7&$(3^+1)$&1&11&36&21\\
$(3^-0)$&4&24&68&7&$(3^-1)$&4&24&68&21\\
$(4^+0)$&1&11&46&9&$(4^+1)$&1&11&46&27\\
$(4^-0)$&0&4&24&9&$(4^-1)$&0&4&24&27\\
$(5^+0)$&0&1&11&11&$(5^+1)$&0&1&11&33\\
$(5^-0)$&0&4&24&11&$(5^-1)$&0&4&24&33\\
$(6^+0)$&0&1&11&13&$(6^+1)$&0&1&11&39\\
$(6^-0)$&0&0&4&13&$(6^-1)$&0&0&4&39\\
$(7^+0)$&0&0&1&15&$(7^+1)$&0&0&1&45\\
$(7^-0)$&0&0&4&15&$(7^-1)$&0&0&4&45\\
$(8^-0)$&0&0&1&17&$(8^-1)$&0&0&1&51\\
\hline
\end{tabular}
\end{table}
In Table \ref{table_ph_reduction_LSp} the dimensions and degeneracies of the $\mathcal{G}$-matrix blocks for first-row atoms in a cc-pVDZ, cc-pVTZ and cc-pVQZ basis set are shown. The reduction factor in the computational cost is even more spectacular for the $\mathcal{G}$ matrix, {\it i.e.} respectively 28595, 144794, and 489488.
\subsection{The three-index constraints for atomic systems}
We choose to couple the first two particles to good intermediate angular momentum $L_{ab}$ and intermediate spin $S_{ab}$, with intermediate parity $\pi_{ab} = (-1)^{l_a+l_b}$. The operator creating three particles coupled to good total spin and angular momentum in this way is:
\begin{align}
{B^\dagger}^{L^\pi S}_{ab(L_{ab}^{\pi_{ab}}S_{ab})c} =& \frac{1}{\sqrt{(ab)}}\sum_{\sigma_a\sigma_b}\sum_{M^S_{ab}\sigma_c}\braket{\frac{1}{2}\sigma_a \frac{1}{2}\sigma_b}{S_{ab}M^S_{ab}}\braket{S_{ab}M^S_{ab}\frac{1}{2}\sigma_c}{SM_S}\\
\nonumber&\sum_{m_am_b}\sum_{M^L_{ab}m_c}\braket{l_am_al_bm_b}{L_{ab}M^L_{ab}}\braket{L_{ab}M^L_{ab}l_cm_c}{LM_L}a^\dagger_{am_a\sigma_a}a^\dagger_{bm_b\sigma_b}a^\dagger_{cm_c\sigma_c}~.
\end{align}
Only those triplets $abc$ that can couple to total angular momentum $L$ and that have the right parity $\pi$ are allowed in a block $(L^\pi S)$, {\it i.e.} the triplet $l_al_bl_c$ has to satisfy the relations:
\begin{equation}
|L_{ab} - l_c| \leq L \leq L_{ab} + l_c~,\qquad |l_a - l_b| \leq L_{ab} \leq l_a + l_b~,\qquad\text{and}\qquad (-1)^{l_a+l_b+l_c} = \pi~.
\end{equation}
In each $L^\pi S$ block of three-particle space there are many ways to form a complete basis set. The choice of a complete basis is more difficult than for three-particle spin coupling because when $l\geq 1$ it is possible to put more than three particles in the same shell $nl$. Not all intermediate angular momenta $L_{ab}$ can be reached however, because of the antisymmetry of the total three-particle state. The symmetry of the state under the exchange of the first two spatial indices is determined by the intermediate angular momentum, spin and parity:
\begin{equation}
\ket{abc;(L_{ab}^{\pi_{ab}}S_{ab})LM_LSM_S} = \pi_{ab}(-1)^{L_{ab}+S_{ab}}\ket{bac;(L_{ab}^{\pi_{ab}}S_{ab})LM_LSM_S}~.
\end{equation}
The part of three-particle space which is antisymmetrical under exchange of the first two indices ({\it i.e.} $\pi_{ab}(-1)^{S_{ab}+L_{ab}} = -1$) is completely spanned by the the states $a < b < c$.
If we want to extract information about states with a different ordering, {\it e.g.} $c < b < a$, we have to recouple, using two times formula (\ref{recoupling_6j}):
\begin{align}
\nonumber\ket{abc;(L_{ab}^{\pi_{ab}}S_{ab})LM_LSM_S} = \sqrt{\frac{(cb)}{(ab)}}&[S_{ab}][L_{ab}]\sum_{S_{cb}L_{cb}}[S_{cb}][L_{cb}]
\left\{
\begin{matrix}
S&\frac{1}{2}&S_{cb}\\
\frac{1}{2}&\frac{1}{2}&S_{ab}
\end{matrix}
\right\}
\\
&\left\{
\begin{matrix}
L&l_a&L_{cb}\\
l_b&l_c&L_{ab}
\end{matrix}
\right\}
\ket{cba;(L_{cb}^{\pi_{cb}}S_{cb})LM_LSM_S}~.
\end{align}
The part of three-particle space which is symmetrical under the exchange of the first two indices is spanned in part by states of the type $a<b<c$, but other types of states can occur. A second type is where the first two spatial indices are equal, $a=b\neq c$. The two equal indices can couple to all intermediate angular momenta allowed by the triangle rule, provided
\begin{equation}
(-1)^{L_{ab}+S_{ab}} = 1~.
\end{equation}
A new type of state, which didn't occur when only spin was taken into account, consists of the three particles in the same shell, $a=b=c$.
There are many states of the form:
\begin{equation}
\ket{aaa;(L_{aa}^{+}S_{aa})LM_LSM_S}~,
\end{equation}
which have finite norm but are not linearly independent. A suitable set has to be chosen which spans the full space. This can be done by constructing the overlap matrix:
\begin{align}
\nonumber S_{L'_{aa}S'_{aa};L_{aa}S_{aa}}=& \langle aaa;({L'}_{aa}^+ {S'}_{aa})LM_L SM_S|aaa;(L_{aa}^+ S_{aa})LM_LSM_S\rangle\\
=&\delta_{{S'}_{aa}S_{aa}}\delta_{{L'}_{aa}L_{aa}}+ 2[L_{aa}][L_{aa}'][S_{aa}][S_{aa}']
\left\{
\begin{matrix}
\frac{1}{2}&\frac{1}{2}&S_{aa}\\
\frac{1}{2}&S&S_{aa}'
\end{matrix}
\right\}
\left\{
\begin{matrix}
l&l&L_{aa}\\
l&L&L_{aa}'
\end{matrix}
\right\}~,
\end{align}
and taking the eigenvectors with nonzero eigenvalues as basis states.

To simplify the equations, we introduce the notation $X_{ab}$ for the collection of intermediate quantum numbers $L_{ab}^{\pi_{ab}}S_{ab}$. The phase under the exchange of the first two indices is written as:
\begin{equation}
(-1)^{X_{ab}} = \pi_{ab}(-1)^{L_{ab}+S_{ab}}~,
\end{equation}
and by $[X_{ab}]$ the product $[S_{ab}][L_{ab}]$ is implied.
\paragraph{The $\mathcal{T}_1$ condition:}
the spin and angular momentum coupled $\mathcal{T}_1$ condition is defined through a spin and angular momentum averaged ensemble:
\begin{align}
\nonumber\mathcal{T}_1&(\Gamma)^{L^\pi S(X_{ab};X_{de})}_{abc;dez} = 
\sum_i w_i \frac{1}{[\mathcal{L}]^2[\mathcal{S}]^2}\sum_\mathcal{M_SM_L}\left( \bra{\Psi^{N\Pi}_{\mathcal{SM_SLM_L},i}}{B^\dagger}^{L^\pi S}_{ab(X_{ab})c}~{B}^{L^\pi S}_{de(X_{de})z}\ket{\Psi^N_{\mathcal{SM_SLM_L},i}}\right.\\
&\left.\qquad\qquad\qquad\qquad\qquad\qquad +\bra{\Psi^{N\Pi}_{\mathcal{LM_LSM_S},i}}{B}^{L^\pi S}_{de(X_{de})z}~{B^\dagger}^{L^\pi S}_{ab(X_{ab})c}\ket{\Psi^{N\Pi}_{\mathcal{LM_LSM_S},i}}\right)~,
\label{T1_scac_ensemble}
\end{align}
with
\begin{align}
{B^\dagger}^{L^\pi S}_{ab(X_{ab})c} &=\frac{1}{\sqrt{(ab)}}\left[\left[a^\dagger_{a}\otimes a^\dagger_b\right]^{L_{ab}S_{ab}}\otimes a^\dagger_c\right]^{LS}~.
\end{align}
This can be expressed as a function of the 2DM using anticommutation relations and angular momentum recoupling in much the same way as the spin-coupled $\mathcal{T}_1$ condition, leading to the expression:
\begin{align}
\nonumber\mathcal{T}_1&(\Gamma)^{L^\pi S(X_{ab};X_{de})}_{abc;dez} = \left[\frac{2\mathrm{Tr}~\Gamma}{N(N -1)}\right]\frac{\delta_{cz}\delta_{X_{ab}X_{de}}}{\sqrt{(ab)(de)}}\left(\delta_{ad}\delta_{be} + (-1)^{X_{ab}}\delta_{ae}\delta_{bd}\right) + \delta_{cz}\delta_{X_{ab}X_{de}}\Gamma^{X_{ab}}_{ab;de}\\
\nonumber& - \frac{\delta_{X_{ab}X_{de}}}{\sqrt{(ab)(de)}}\left[\delta_{cz}\left(\delta_{be}\delta_{l_al_d}\rho^{(l_a)}_{n_an_d}+(-1)^{X_{ab}}\delta_{ae}\delta_{l_bl_d}\rho^{(l_b)}_{n_bn_d} + (-1)^{X_{de}}\delta_{bd}\delta_{l_al_e}\rho^{(l_a)}_{n_an_e}\right.\right.\\
\nonumber& \left.\left. \qquad\qquad\qquad\qquad\qquad\qquad\qquad + \delta_{ad}\delta_{l_bl_e}\rho^{(l_b)}_{n_bn_e}\right) + \left(\delta_{ad}\delta_{be} + (-1)^{X_{ab}}\delta_{bd}\delta_{ae}\right)\delta_{l_cl_z}\rho^{(l_c)}_{n_cn_z}\right]\\
\nonumber&-\frac{[X_{ab}][X_{de}]}{\sqrt{(ab)(de)}} 
\left\{\begin{matrix}
S&\frac{1}{2}&S_{ab}\\
\frac{1}{2}&\frac{1}{2}&S_{de}
\end{matrix}\right\}
\left[
\left\{\begin{matrix}
L&l_c&L_{ab}\\
l_b&l_a&L_{de}
\end{matrix}\right\}
\left(
\delta_{az}\delta_{cd}\delta_{l_bl_e}\rho^{(l_b)}_{n_bn_e} + (-1)^{X_{de}}\delta_{az}\delta_{ec}\delta_{l_bl_d}\rho^{(l_b)}_{n_bn_d}\right.\right.\\
\nonumber&\left.\left.\qquad+\delta_{az}\delta_{be}\delta_{l_cl_d}\rho^{(l_c)}_{n_cn_d} + (-1)^{X_{de}}\delta_{az}\delta_{bd}\delta_{l_cl_e}\rho^{(l_c)}_{n_cn_e} +\left(\delta_{be}\delta_{cd} +(-1)^{X_{de}}\delta_{bd}\delta_{ce}\right)\delta_{l_al_z}\rho^{(l_a)}_{n_an_z}\right)\right.\\
\nonumber&\left.+ (-1)^{X_{ab}}
\left\{\begin{matrix}
L&l_c&L_{ab}\\
l_a&l_b&L_{de}
\end{matrix}\right\}
\left(
\delta_{bz}\delta_{cd}\delta_{l_al_e}\rho^{(l_a)}_{n_an_e} +(-1)^{X_{de}}\delta_{bz}\delta_{ec}\delta_{l_al_d}\rho^{(l_a)}_{n_an_d}+\delta_{bz}\delta_{ae}\delta_{l_cl_d}\rho^{(l_c)}_{n_cn_d}\right. \right.\\
\nonumber&\left.\left.\qquad\qquad\qquad\qquad + (-1)^{X_{de}}\delta_{bz}\delta_{ad}\delta_{l_cl_e}\rho^{(l_c)}_{n_cn_e} +\left(\delta_{ae}\delta_{cd}+(-1)^{X_{de}}\delta_{ce}\delta_{ad}\right)\delta_{l_bl_z}\rho^{(l_b)}_{n_bn_z}\right)\right]\\ 
\nonumber& + [X_{ab}][X_{de}]
\left\{\begin{matrix}
S&\frac{1}{2}&S_{ab}\\
\frac{1}{2}&\frac{1}{2}&S_{de}
\end{matrix}\right\}
\left[
\delta_{bz}(-1)^{X_{ab}+X_{de}}\sqrt{\frac{(ac)}{(ab)}}
\left\{\begin{matrix}
L&l_c&L_{ab}\\
l_a&l_b&L_{de}
\end{matrix}\right\}
\Gamma^{S_{de}}_{ac;de} \right.\\
&\nonumber\left.\qquad\qquad\qquad\qquad+ \delta_{az}\sqrt{\frac{(bz)}{(ab)}}
\left\{\begin{matrix}
L&l_c&L_{ab}\\
l_b&l_a&L_{de}
\end{matrix}\right\}
\Gamma^{X_{de}}_{cb;de}
+\delta_{cd}\sqrt{\frac{(ez)}{(de)}}
\left\{\begin{matrix}
L&l_z&L_{de}\\
l_e&l_d&L_{ab}
\end{matrix}\right\}
\Gamma^{X_{ab}}_{ab;ze}\right.\\
\nonumber&\left.\qquad\qquad\qquad\qquad\qquad\qquad+\delta_{ce}(-1)^{X_{ab}+X_{de}}\sqrt{\frac{(dz)}{(de)}}
\left\{\begin{matrix}
L&l_z&L_{de}\\
l_d&l_e&L_{ab}
\end{matrix}\right\}
\Gamma^{X_{ab}}_{ab;dz}\right]\\
\nonumber&+\frac{[X_{ab}][X_{de}]}{\sqrt{(ab)(de)}}\sum_{X'} [X']^2
\left\{\begin{matrix}
S&\frac{1}{2}&S'\\
\frac{1}{2}&\frac{1}{2}&S_{ab}
\end{matrix}\right\}
\left\{\begin{matrix}
S&\frac{1}{2}&S'\\
\frac{1}{2}&\frac{1}{2}&S_{de}
\end{matrix}\right\}\\
\nonumber&\qquad\qquad\left[\delta_{ad}\sqrt{(bc)(ez)}
\left\{\begin{matrix}
L&l_a&L'\\
l_b&l_c&L_{ab}
\end{matrix}\right\}
\left\{\begin{matrix}
L&l_d&L'\\
l_e&l_z&L_{de}
\end{matrix}\right\}
\Gamma^{X'}_{bc;ez}\right.\\
\nonumber&\left.\qquad\qquad\qquad + \delta_{bd}(-1)^{X_{ab}}\sqrt{(ac)(ez)}
\left\{\begin{matrix}
L&l_b&L'\\
l_a&l_c&L_{ab}
\end{matrix}\right\}
\left\{\begin{matrix}
L&l_d&L'\\
l_e&l_z&L_{de}
\end{matrix}\right\}
\Gamma^{X'}_{ac;ez}\right.\\
\nonumber&\left.\qquad\qquad\qquad+\delta_{ae}(-1)^{X_{de}}\sqrt{(bc)(dz)}
\left\{\begin{matrix}
L&l_a&L'\\
l_b&l_c&L_{ab}
\end{matrix}\right\}
\left\{\begin{matrix}
L&l_e&L'\\
l_d&l_z&L_{de}
\end{matrix}\right\}
\Gamma^{X'}_{bc;dz}\right.\\
&\left.\qquad\qquad\qquad+\delta_{be}(-1)^{X_{ab}+X_{de}}\sqrt{(ac)(dz)}
\left\{\begin{matrix}
L&l_b&L'\\
l_a&l_c&L_{ab}
\end{matrix}\right\}
\left\{\begin{matrix}
L&l_e&L'\\
l_d&l_z&L_{de}
\end{matrix}\right\}
\Gamma^{X'}_{ac;dz}\right]~.
\end{align}
In Table \ref{T1_table_red} the dimensions and degeneracies of the $\mathcal{T}_1$ matrix blocks are displayed. Without including symmetries the use of three-index conditions would be computationally infeasible. For the small cc-pVDZ basis the size of the uncoupled matrix is 3276, for the TZ already 34220 and for the QZ basis the dimension becomes 215820. The sizes of the blocks in the matrix with all symmetries exploited, however, remain manageable even for the largest basisset. The reduction factor in the computational cost of an elementary $\mathcal{T}_1$ matrix operation is 166466 for the DZ, 933490 for the TZ, and 3092934 for the QZ basis.
\begin{table}
\caption{\label{T1_table_red}Dimensions and degeneracies of the $\mathcal{T}_1$-matrix blocks for first-row atoms (from $\mathrm{Li}$ to $\mathrm{Ne}$) in a cc-pVDZ, cc-pVTZ and cc-pVQZ basis set.}
\centering
\begin{tabular}{|c|c|c|c|c|c|c|c|c|c|}
\hline
&\multicolumn{3}{c}{dimension}\vline&deg&&\multicolumn{3}{c}{dimension}\vline&deg\\
\hline
$(L^\pi S)$ & DZ & TZ& QZ & &$(L^\pi S)$ & DZ & TZ& QZ &\\ 
\hline
$(0^+\frac{1}{2})$	&	27	&	110	&	338	&2&	$(0^+\frac{3}{2})$	&		5	&	32	&	118&4	\\
$(0^-\frac{1}{2})$	&	4	&	27	&	110	&	2 &$(0^-\frac{3}{2})$	&		6	&	26	&	84	&4\\
$(1^+\frac{1}{2})$	&	24	&	130	&	477	&6&	$(1^+\frac{3}{2})$	&		17	&	80	&	274&12	\\
$(1^-\frac{1}{2})$	&	44	&	206	&	692	&6&	$(1^-\frac{3}{2})$	&		16	&	85	&	306&12	\\
$(2^+\frac{1}{2})$	&	38	&	226	&	857	&10&	$(2^+\frac{3}{2})$	&		11	&	90	&	374&20	\\
$(2^-\frac{1}{2})$	&	24	&	164	&	668	&10&	$(2^-\frac{3}{2})$	&		12	&	85	&	346&20	\\
$(3^+\frac{1}{2})$	&	12	&	136	&	672	&14&	$(3^+\frac{3}{2})$	&		8	&	76	&	358&28	\\
$(3^-\frac{1}{2})$	&	20	&	182	&	828	&14&	$(3^-\frac{3}{2})$	&		8	&	79	&	382&28	\\
$(4^+\frac{1}{2})$	&	8	&	116	&	671	&18&	$(4^+\frac{3}{2})$	&		1	&	46	&	298&36	\\
$(4^-\frac{1}{2})$	&	4	&	85	&	550	&18&	$(4^-\frac{3}{2})$	&		2	&	44	&	282&36	\\
$(5^+\frac{1}{2})$	&	1	&	42	&	376	&22&	$(5^+\frac{3}{2})$	&		0	&	22	&	195&44	\\
$(5^-\frac{1}{2})$	&	2	&	60	&	462	&22&	$(5^-\frac{3}{2})$	&		0	&	22	&	208&44	\\
$(6^+\frac{1}{2})$	&	0	&	24	&	275	&26&	$(6^+\frac{3}{2})$	&		0	&	8	&	118&52	\\
$(6^-\frac{1}{2})$	&	0	&	15	&	218	&26&	$(6^-\frac{3}{2})$	&		0	&	8	&	112&52	\\
$(7^+\frac{1}{2})$	&	0	&	4	&	109	&30&	$(7^+\frac{3}{2})$	&		0	&	2	&	57	&60\\
$(7^-\frac{1}{2})$	&	0	&	8	&	144	&30&	$(7^-\frac{3}{2})$	&		0	&	1	&	60	&60\\
$(8^+\frac{1}{2})$	&	0	&	2	&	65	&34&	$(8^+\frac{3}{2})$	&		0	&	0	&	22	&68\\
$(8^-\frac{1}{2})$	&	0	&	1	&	46	&34&	$(8^-\frac{3}{2})$	&		0	&	0	&	22	&68\\
$(9^+\frac{1}{2})$	&	0	&	0	&	15	&38&	$(9^+\frac{3}{2})$	&		0	&	0	&	8	&76\\
$(9^-\frac{1}{2})$	&	0	&	0	&	24	&38&	$(9^-\frac{3}{2})$	&		0	&	0	&	8	&76\\
$(10^+\frac{1}{2})$	&	0	&	0	&	8	&42&	$(10^+\frac{3}{2})$	&		0	&	0	&	1	&84\\
$(10^-\frac{1}{2})$	&	0	&	0	&	4	&42&	$(10^-\frac{3}{2})$	&		0	&	0	&	2	&84\\
$(11^+\frac{1}{2})$	&	0	&	0	&	1	&46&	$(11^+\frac{3}{2})$	&		0	&	0	&	0	&92\\
$(11^-\frac{1}{2})$	&	0	&	0	&	2	&46&	$(11^-\frac{3}{2})$	&		0	&	0	&	0	&92\\
\hline
\end{tabular}
\end{table}
\paragraph{The $\mathcal{T}_2$ condition:}
the spin and angular momentum coupled form for the $\mathcal{T}_2$ map is defined as:
\begin{align}
\nonumber\mathcal{T}_2&(\Gamma)^{L^\pi S(X_{ab};X_{de})}_{abc;dez} = 
\sum_i w_i \frac{1}{[\mathcal{L}]^2[\mathcal{S}]^2}\sum_\mathcal{M_SM_L}\left( \bra{\Psi^{N\Pi}_{\mathcal{SM_SLM_L},i}}{A^\dagger}^{L^\pi S}_{ab(X_{ab})c}~{A}^{L^\pi S}_{de(X_{de})z}\ket{\Psi^N_{\mathcal{SM_SLM_L},i}}\right.\\
&\left.\qquad\qquad\qquad\qquad\qquad\qquad +\bra{\Psi^{N\Pi}_{\mathcal{LM_LSM_S},i}}{A}^{L^\pi S}_{de(X_{de})z}~{A^\dagger}^{L^\pi S}_{ab(X_{ab})c}\ket{\Psi^{N\Pi}_{\mathcal{LM_LSM_S},i}}\right)~.
\label{T2_scac_ensemble}
\end{align}
The coupled two-particle-one-hole creation operator is given by:
\begin{align}
{A^\dagger}^{L^\pi S}_{ab(X_{ab})c} &=\frac{1}{\sqrt{(ab)}}\left[\left[a^\dagger_{a}\otimes a^\dagger_b\right]^{L_{ab}S_{ab}}\otimes \tilde{a}_c\right]^{LS}~,
\end{align}
where the spherical tensor operator $\tilde{a}_c$ (see Eq.~(\ref{LS_anni_op})) appears. As there is no symmetry involving the third index, the ordering of indices is easier than for $\mathcal{T}_1$, {\it i.e.} the ordering is $a\leq b, c$ when $(-1)^{X_{ab}} = 1$ and $a<b,c$ when $(-1)^{X_{ab}} = -1$. Through anticommutation and angular momentum recoupling one obtains an expression of $\mathcal{T}_2$ as a function of the 2DM:
\begin{align}
\nonumber\mathcal{T}_2&(\Gamma)^{L^\pi S(X_{ab};X_{de})}_{abc;dez} = \frac{\delta_{X_{ab}X_{de}}}{\sqrt{(ab)(de)}}\left(\delta_{ad}\delta_{be}+(-1)^{X_{ab}}\delta_{ae}\delta_{bd}\right)\delta_{l_cl_z}\rho^{(l_c)}_{n_cn_z} + \delta_{cz}\delta_{X_{ab}X_{de}}\Gamma^{X_{ab}}_{ab;de}\\
\nonumber&-\frac{[X_{ab}][X_{de}]}{\sqrt{(ab)(de)}}\sum_{X'} [X']^2
\left\{
\begin{matrix}
S&\frac{1}{2}&S_{de}\\
\frac{1}{2}&S'&\frac{1}{2}\\
S_{ab}&\frac{1}{2}&\frac{1}{2}
\end{matrix}
\right\}
\left(\delta_{ad}\sqrt{(ce)(bz)}
\left\{
\begin{matrix}
L&l_z&L_{de}\\
l_c&L'&l_e\\
L_{ab}&l_b&l_a
\end{matrix}
\right\}
\Gamma^{X'}_{ce;zb}\right.\\
\nonumber&\left.\qquad\qquad\qquad\qquad+ (-1)^{X_{ab}}\delta_{bd}\sqrt{(ce)(za)}
\left\{
\begin{matrix}
L&l_z&L_{de}\\
l_c&L'&l_e\\
L_{ab}&l_a&l_b
\end{matrix}
\right\}
\Gamma^{X'}_{ce;za}\right.\\
\nonumber&\left.\qquad\qquad\qquad\qquad+ (-1)^{X_{de}}\delta_{ae}\sqrt{(cd)(zb)}
\left\{
\begin{matrix}
L&l_z&L_{de}\\
l_c&L'&l_d\\
L_{ab}&l_b&l_a
\end{matrix}
\right\}
\Gamma^{X'}_{cd;zb}\right.\\*
&\left. \qquad\qquad\qquad\qquad\qquad+ (-1)^{X_{ab}+X_{de}}\delta_{be}\sqrt{(cd)(za)}
\left\{
\begin{matrix}
L&l_z&L_{de}\\
l_c&L'&l_d\\
L_{ab}&l_a&l_b
\end{matrix}
\right\}
\Gamma^{X'}_{cd;za}\right)~.
\end{align}
In Table~\ref{T2_table_red} the relevant matrix dimensions for $\mathcal{T}_2$ are shown. For comparison, the uncoupled $\mathcal{T}_2$ matrix has a dimension in the three basissets, DZ, TZ and QZ, of respectively 10548, 106200 and 659450. The reduction factor in the computational cost of an elementary $\mathcal{T}_2$ matrix operation is 201137 for the DZ, 1015334 for the TZ, and 3235624 for the QZ basis.
\begin{table}
\caption{\label{T2_table_red}Dimensions and degeneracies of the $\mathcal{T}_2$-matrix blocks for first-row atoms (from $\mathrm{Li}$ to $\mathrm{Ne}$) in a cc-pVDZ, cc-pVTZ and cc-pVQZ basis set.}
\centering
\begin{tabular}{|c|c|c|c|c|c|c|c|c|c|}
\hline
&\multicolumn{3}{c}{dimension}\vline&deg&&\multicolumn{3}{c}{dimension}\vline&deg\\
\hline
$(L^\pi S)$ & DZ & TZ& QZ & &$(L^\pi S)$ & DZ & TZ& QZ &\\ 
\hline
$(0^+\frac{1}{2})$	&		85	&		336	&		1023	&2&		$(0^+\frac{3}{2})$	&		32	&		142	&		456&4	\\	
$(0^-\frac{1}{2})$	&		14	&		85	&		336	& 2 &		$(0^-\frac{3}{2})$	&		10	&		53	&		194&4	\\	
$(1^+\frac{1}{2})$	&		72	&		390	&		1431	&6&		$(1^+\frac{3}{2})$	&		41	&		210	&		751&12	\\	
$(1^-\frac{1}{2})$	&		132	&		618	&		2076	&6&		$(1^-\frac{3}{2})$	&		60	&		291	&		998&12	\\	
$(2^+\frac{1}{2})$	&		113	&		676	&		2567	&10&		$(2^+\frac{3}{2})$	&		49	&		316	&		1231&20	\\	
$(2^-\frac{1}{2})$	&		70	&		488	&		1998	&10&		$(2^-\frac{3}{2})$	&		36	&		249	&		1014&20	\\	
$(3^+\frac{1}{2})$	&		37	&		410	&		2020	&14&		$(3^+\frac{3}{2})$	&		20	&		212	&		1030&28	\\	
$(3^-\frac{1}{2})$	&		62	&		550	&		2490	&14&		$(3^-\frac{3}{2})$	&		28	&		261	&		1210&28	\\	
$(4^+\frac{1}{2})$	&		24	&		348	&		2013	&18&		$(4^+\frac{3}{2})$	&		9	&		162	&		969&36	\\	
$(4^-\frac{1}{2})$	&		12	&		255	&		1650	&18&		$(4^-\frac{3}{2})$	&		6	&		129	&		832&36	\\	
$(5^+\frac{1}{2})$	&		2	&		124	&		1124	&22&		$(5^+\frac{3}{2})$	&		1	&		64	&		571&44	\\	
$(5^-\frac{1}{2})$	&		6	&		179	&		1384	&22&		$(5^-\frac{3}{2})$	&		2	&		82	&		670&44	\\	
$(6^+\frac{1}{2})$	&		1	&		74	&		829	&26&		$(6^+\frac{3}{2})$	&		0	&		32	&		393&52	\\	
$(6^-\frac{1}{2})$	&		0	&		46	&		656	&26&		$(6^-\frac{3}{2})$	&		0	&		23	&		330&52	\\	
$(7^+\frac{1}{2})$	&		0	&		12	&		327	&30&		$(7^+\frac{3}{2})$	&		0	&		6	&		166&60	\\	
$(7^-\frac{1}{2})$	&		0	&		24	&		432	&30&		$(7^-\frac{3}{2})$	&		0	&		9	&		204&60	\\	
$(8^+\frac{1}{2})$	&		0	&		6	&		194	&34&		$(8^+\frac{3}{2})$	&		0	&		2	&		87&68	\\	
$(8^-\frac{1}{2})$	&		0	&		2	&		136	&34&		$(8^-\frac{3}{2})$	&		0	&		1	&		68&68	\\	
$(9^+\frac{1}{2})$	&		0	&		0	&		46	&38&		$(9^+\frac{3}{2})$	&		0	&		0	&		23&76	\\	
$(9^-\frac{1}{2})$	&		0	&		1	&		74	&38&		$(9^-\frac{3}{2})$	&		0	&		0	&		32&76	\\	
$(10^+\frac{1}{2})$	&		0	&		0	&		24	&42&		$(10^+\frac{3}{2})$	&		0	&		0	&		9&84	\\	
$(10^-\frac{1}{2})$	&		0	&		0	&		12	&42&		$(10^-\frac{3}{2})$	&		0	&		0	&		6&84	\\	
$(11^+\frac{1}{2})$	&		0	&		0	&		2	&46&		$(11^+\frac{3}{2})$	&		0	&		0	&		1&92	\\	
$(11^-\frac{1}{2})$	&		0	&		0	&		6	&46&		$(11^-\frac{3}{2})$	&		0	&		0	&		2&92	\\	
$(12^+\frac{1}{2})$	&		0	&		0	&		1	&50&		$(12^+\frac{3}{2})$	&		0	&		0	&		0&100	\\	
\hline
\end{tabular}
\end{table}
\paragraph{The $\mathcal{T}_2'$ condition} The spin and angular momentum coupled $\mathcal{T}_2'$ map is identical to the regular $\mathcal{T}_2$ map when the total spin $S$ is $\frac{3}{2}$ or the total angular momentum $L$ is larger than the maximal single-particle angular momentum in the basis set $l_{\text{max}}$.  When $S=\frac{1}{2}$ and $L \leq l_\text{max}$ the $\mathcal{T}_2'$ matrix has the form:
\begin{equation}
\mathcal{T}_2'(\Gamma)^{l^\pi \frac{1}{2}(X_{ab},X_{de})}_{abc,n_m;dez,n_n} = 
\left(
\begin{matrix}
\mathcal{T}_2(\Gamma)^{l^\pi \frac{1}{2}(X_{ab};X_{de})}_{abc;dez} & \omega^{(X_{ab})}_{abc;(n_nl)}\\
{\omega^\dagger}^{(X_{de})}_{(n_ml);dez}&\rho^{(l)}_{n_mn_n}
\end{matrix}
\right)~,
\end{equation}
in which $\omega$ can be written as a function of the 2DM:
\begin{align}
\omega^{(X_{ab})}_{abc;(n_nl)} =& \frac{[L_{ab}]}{[l]}(-1)^{1+X_{ab}}\sqrt{(nc)}\Gamma^{X_{ab}}_{ab;(n_nl)c}~.
\end{align}
\section{\label{hub_sym}Symmetry in the one-dimensional Hubbard model}
\begin{figure}
\centering
\includegraphics[scale=0.5]{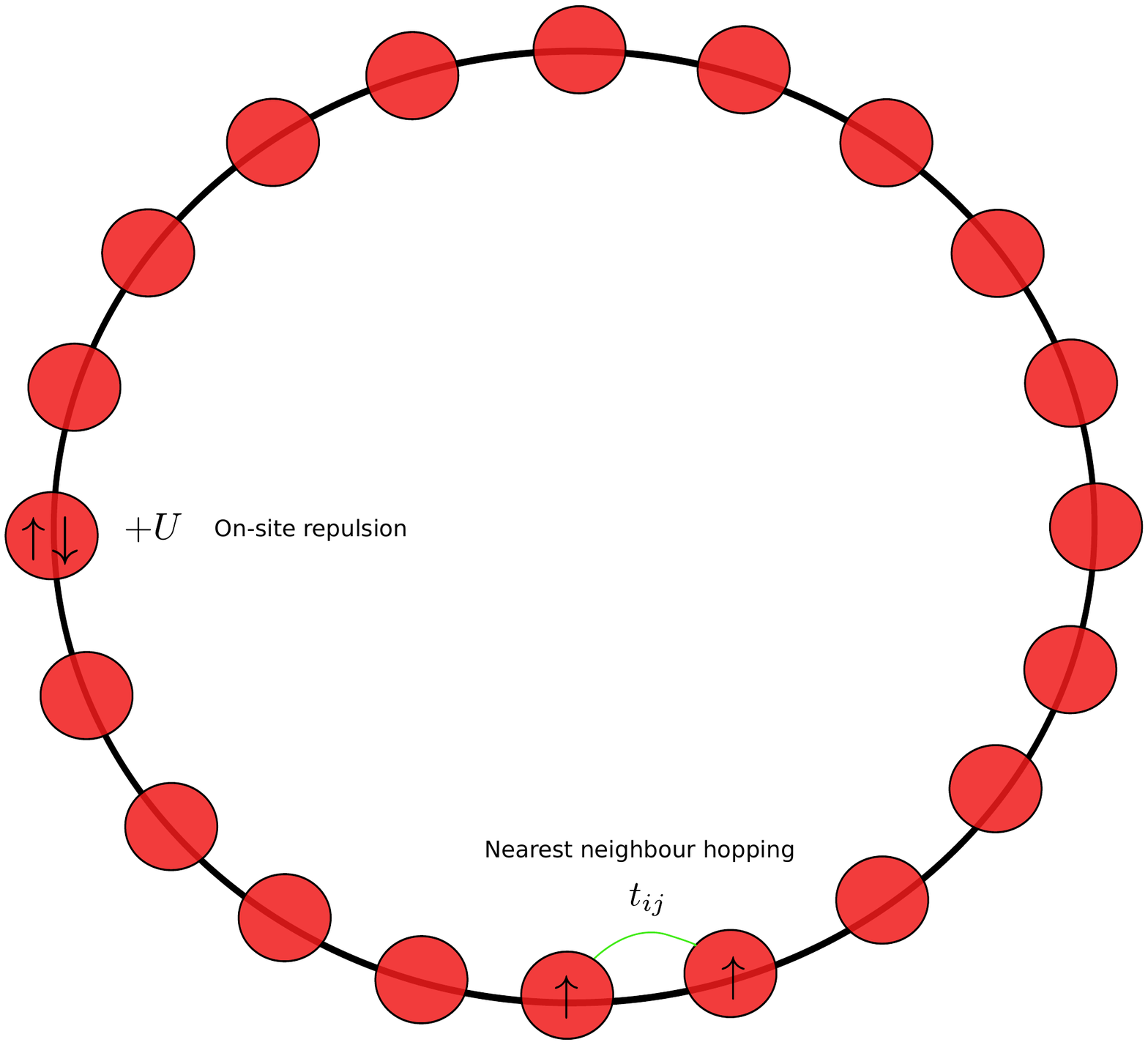}
\caption{\label{hubbard_fig} Illustration of the two terms present in the Hubbard Hamiltonian, electrons can jump to nearest-neighbour sites with amplitude $t_{ij}$. When two electrons are on the same site, there is an energy penalty of $U$.}
\end{figure}
The one-dimensional Hubbard model \cite{hubbard} is the simplest model possessing non-trivial correlations present in a solid, described by the Hamiltonian:
\begin{equation}
\hat{H} = -t\sum_{i\sigma}\left(a^\dagger_{i;\sigma}a_{i+1;\sigma} + a^\dagger_{i+1;\sigma}a_{i;\sigma}\right) + U\sum_{i}a^\dagger_{i\uparrow}a_{i\uparrow}a^\dagger_{i\downarrow}a_{i\downarrow}~.
\label{hubbard}
\end{equation}
It pertains to a one-dimensional lattice, with sites labeled by $i = 1,\ldots,L$. Periodic boundary conditions (PBC) are assumed throughout the thesis.

The complexity of this seemingly simple schematic Hamiltonian lies in the competition between the first term, called the hopping term, which delocalizes the electrons, and the second on-site repulsion term which is diagonal in the site basis. In Figure \ref{hubbard_fig} a graphic representation of the two terms in the Hubbard Hamiltonian is shown. 
As there is no preferred direction in spin space in Eq.~(\ref{hubbard}), spin symmetry can be exploited and all the results from Section~\ref{spinsym} are taken over. However, far more symmetries are present in the model, which allow for a huge reduction of the matrix dimensions involved.
\subsection{Translational invariance}
\begin{figure}
\centering
\includegraphics[scale=0.5]{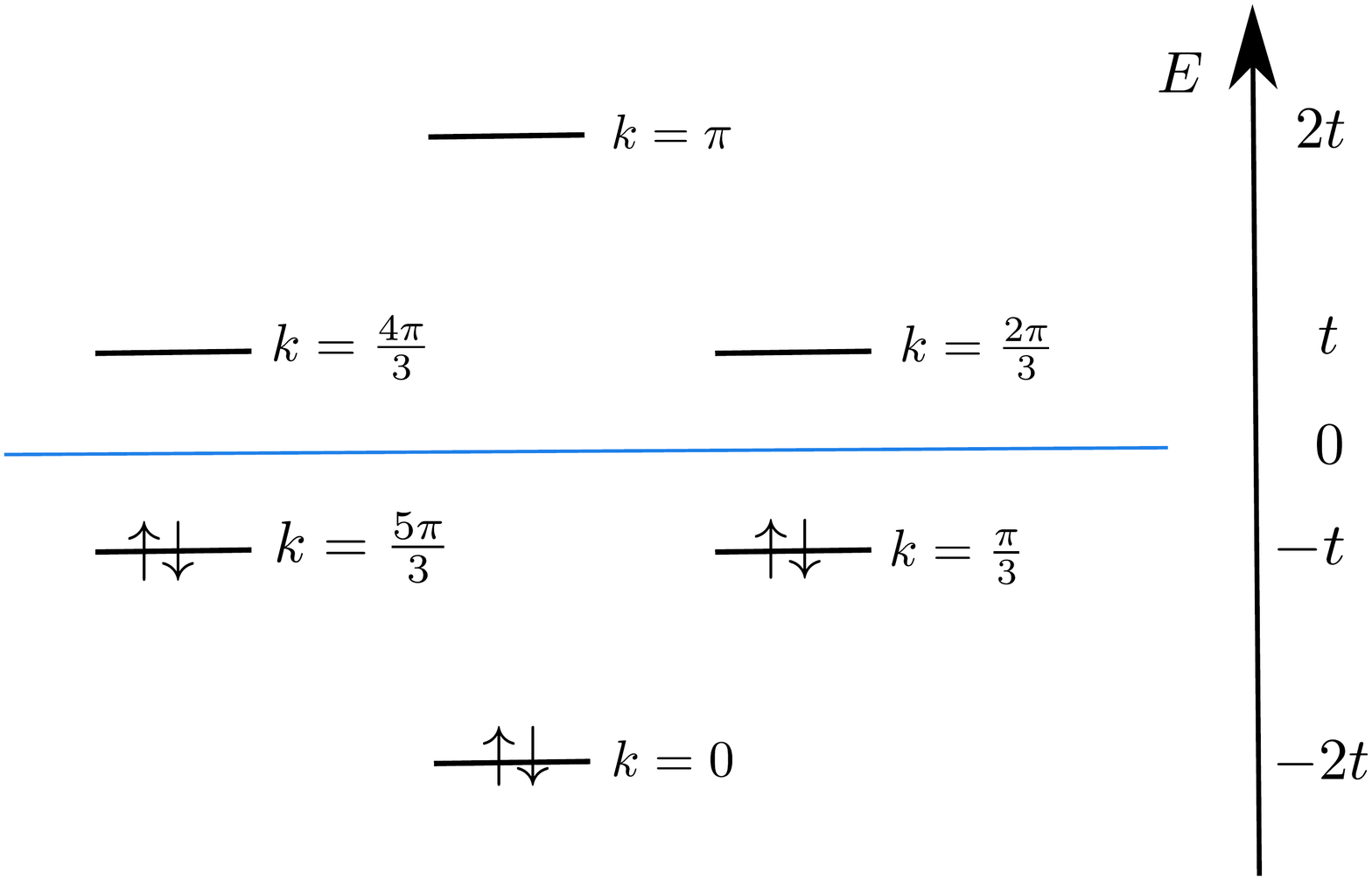}
\caption{\label{fig_hubbard_ni} A representation of the structure of the ground state when $U=0$ for a half-filled lattice of length $L=6$.}
\end{figure}
When periodic boundary conditions are assumed (as in Figure \ref{hubbard_fig}) the Hamiltonian is invariant under translations along the lattice. This is an abelian symmetry for which it is easy to block diagonalize the 2DM and its matrix maps, as the correct basis transformation in single-particle space automatically block diagonalizes all matrices on higher-order particle space. Translational invariance is exploited by Fourier transforming the site basis to quasi-momentum space:
\begin{equation}
\ket{k\sigma} = \sqrt{\frac{1}{L}}\sum_j e^{ikj}\ket{j\sigma}~,
\label{fourier}
\end{equation}
where $L$ is the lattice size, and $k$ takes on the values $\frac{2\pi n}{L}$ for $n = 0,\ldots,L-1$.
The kinetic or hopping part of the Hamiltonian becomes diagonal in this basis:
\begin{equation}
\hat{H}_{\text{hop}} = -2t \sum_{k\sigma} \cos{k} ~ a^\dagger_{k\sigma}a_{k\sigma}~,
\label{sp_hubbard}
\end{equation}
from which it follows that in the non-interacting ($U=0$) ground state, the electrons occupy the states with momenta lower than the Fermi level, as illustrated in Figure~\ref{fig_hubbard_ni} for a half-filled lattice of length $L=6$.
\subsubsection{The 2DM for translationally invariant systems}
The eigenstates of the Hubbard Hamiltonian have a good total quasi-momentum $\mathcal{K}$. The 2DM for these states, expressed in the quasi-momentum single-particle basis:
\begin{equation}
\Gamma^{S}_{k_ak_b;k_ck_d} = \sum_i w_i \frac{1}{[\mathcal{S}]^2} \sum_\mathcal{M}\bra{\Psi^{N\mathcal{K}}_{\mathcal{SM},i}}{B^\dagger}^{S}_{k_ak_b}~B^{S}_{k_ck_d}\ket{\Psi^{N\mathcal{K}}_{\mathcal{SM},i}}~,
\label{2DM_mom}
\end{equation}
where
\begin{equation}
{B^\dagger}^{S}_{k_ak_b} = \frac{1}{\sqrt{(k_ak_b)}}\left[a^\dagger_{k_a}\otimes a^\dagger_{k_b}\right]^{S}~,
\label{2DM_op_transin}
\end{equation}
is automatically block diagonal, because the only non-zero matrix elements in Eq.~(\ref{2DM_mom}) are those which conserve momentum: $(k_a + k_b)\%2\pi =(k_c + k_d)\%2\pi$. This means we have $L$ blocks $\Gamma^{SK}$ for every $S$, with two-particle states that satisfy $K = (k_a+k_b)\%2\pi$, and a block dimension that scales linearly with lattice size $L$.
\subsubsection{Two-index constraints for translationally invariant systems}
The spin-symmetric matrix constraints simplify considerably by including translational invariance, because the 1DM is automatically diagonal in the quasi-momentum basis:
\begin{equation}
\rho_{k} = \sum_i w_i \frac{1}{[\mathcal{S}]^2} \sum_\mathcal{M}\bra{\Psi^{N\mathcal{K}}_{\mathcal{SM},i}}a^\dagger_{k\sigma}a_{k\sigma}\ket{\Psi^{N\mathcal{K}}_{\mathcal{SM},i}}~.
\label{1DM_ti}
\end{equation}
The translationally invariant 1DM can be derived from the 2DM as:
\begin{equation}
\rho_k = \frac{1}{N-1}\sum_S\frac{[S]^2}{2}\sum_{k'}(kk')\Gamma^{SK}_{kk';kk'}~.
\label{1DM_sc_ti_written_up}
\end{equation}
\paragraph{The $\mathcal{Q}$ map:}
the translationally invariant $\mathcal{Q}$ map is defined as:
\begin{equation}
\mathcal{Q}(\Gamma)^{SK}_{k_ak_b;k_ck_d} = \sum_i w_i \frac{1}{[\mathcal{S}]^2} \sum_\mathcal{M}\bra{\Psi^{N\mathcal{K}}_{\mathcal{SM},i}}{B}^{SK}_{k_ak_b}~{B^\dagger}^{SK}_{k_ck_d}\ket{\Psi^{N\mathcal{K}}_{\mathcal{SM},i}}~,
\end{equation}
with $B^\dagger$ given by Eq.~(\ref{2DM_op_transin}). The expression of the $\mathcal{Q}$ map as a function of the 2DM becomes quite simple:
\begin{equation}
\mathcal{Q}(\Gamma)^{SK}_{k_ak_b;k_ck_d} = \Gamma^{SK}_{k_ak_b;k_ck_d} + \frac{\left(\delta_{k_ak_c}\delta_{k_bk_d}+(-1)^S\delta_{k_ak_d}\delta_{k_bk_c}\right)}{\sqrt{(k_ak_b)(k_ck_d)}}\left[\frac{2\mathrm{Tr}~\Gamma}{N(N-1)}-\rho_{k_a}-\rho_{k_b}\right]~.
\label{ti_Q_map}
\end{equation}
The $\mathcal{Q}$ map is also block diagonal in $K=(k_a+k_b)\%2\pi$, and has the same dimensional reduction as the 2DM.
\paragraph{The $\mathcal{G}$ condition:}
the $\mathcal{G}$ condition is slightly more complicated, because the correct annihilation or hole operator is given by:
\begin{equation}
\tilde{a}_{k\sigma} = (-1)^{\frac{1}{2}+\sigma} a_{\bar{k}-\sigma}~,\qquad\text{with}\qquad \bar{k} = -k\%2\pi~.
\label{good_sto_mom}
\end{equation}
Using this operator the translationally invariant $\mathcal{G}$ map becomes:
\begin{equation}
\mathcal{G}(\Gamma)^{SK}_{k_ak_b;k_ck_d} = \sum_i w_i \frac{1}{[\mathcal{S}]^2}\sum_\mathcal{M} \bra{\Psi^{\mathcal{K}N}_{\mathcal{SM},i}}{A^\dagger}^{SK}_{k_ak_b}~{A}^{SK}_{k_ck_d}\ket{\Psi^{N\mathcal{K}}_{\mathcal{SM},i}}~,
\end{equation}
where $K = (k_a+k_b)\%2\pi = (k_c+k_d)\%2\pi$ and with the particle-hole operator defined by:
\begin{equation}
{A^\dagger}^{SK}_{k_ak_b} = \left[a^\dagger_{k_a}\otimes \tilde{a}_{k_b}\right]^{SK}~.
\end{equation}
The $\mathcal{G}$ map can be expressed as a function of the 2DM:
\begin{equation}
\mathcal{G}(\Gamma)^{SK}_{k_ak_b;k_ck_d} = \delta_{k_bk_d}\delta_{k_ak_c}\rho_{k_a} - \sqrt{(k_a\bar{k}_d)(k_b\bar{k}_c)}\sum_{S'}
\left\{
\begin{matrix}
\frac{1}{2}&\frac{1}{2}&S\\
\frac{1}{2}&\frac{1}{2}&S'
\end{matrix}
\right\}
\Gamma^{S'K'}_{k_a\bar{k}_d;k_c\bar{k}_b}~,
\end{equation}
from which one can see that the blocks in the $\mathcal{G}$ matrix with $K = (k_a + k_b)\%2\pi= (k_c + k_d)\%2\pi$, correspond to the blocks with $K' = (k_a +\bar{k}_d)\%2\pi= (k_c +\bar{k}_b)\%2\pi$ in the 2DM, as can expected for a particle-hole transformed quantity.
\subsubsection{Three-index constraints for translationally invariant systems}
Three-particle space is built up of the same type of states as in the spin-coupled case (See Table~\ref{dp_dim}), but now divided into $L$ separate Sections with the same three-particle momentum $K = (k_a +k_b+k_c)\%2\pi$. From now on Latin indices $a$ are used to represent momenta $k_a$, in order to ease notation.
\paragraph{The $\mathcal{T}_1$ map:}
in translationally invariant form, the single-particle part of the $\mathcal{T}_1$ map simplifies considerably:
\begin{align}
\nonumber\mathcal{T}_1(\Gamma)^{SK(S_{ab};S_{de})}_{abc;dez} &=   \frac{\delta_{cz}\delta_{S_{ab}S_{de}}}{\sqrt{(ab)(de)}}\left(\delta_{ad}\delta_{be} + (-1)^{S_{ab}}\delta_{ae}\delta_{bd}\right) \left[\left(\frac{2\mathrm{Tr}~\Gamma}{N(N -1)}\right) -\rho_{a}-\rho_{b}-\rho_{c}\right]\\
\nonumber&-\frac{[S_{ab}][S_{de}]}{\sqrt{(ab)(de)}} 
\left\{
\begin{matrix}S&\frac{1}{2}&S_{ab}\\
\frac{1}{2}&\frac{1}{2}&S_{de}
\end{matrix}
\right\}\left(\delta_{az}\delta_{cd}\delta_{be}+(-1)^{S_{de}}\delta_{az}\delta_{ce}\delta_{bd} + (-1)^{S_{ab}}\delta_{bz}\delta_{ae}\delta_{cd}\right.\\
\nonumber&\left.\qquad\qquad\qquad+ (-1)^{S_{ab}+S_{de}}\delta_{bz}\delta_{ce}\delta_{ad}\right)\left[\rho_{a}+\rho_{b}+\rho_{c}\right]+ \delta_{cz}\delta_{S_{ab}S_{de}}\Gamma^{S_{ab}K_{ab}}_{ab;de}\\
\nonumber& + [S_{ab}][S_{de}]
\left\{
\begin{matrix}S&\frac{1}{2}&S_{ab}\\
\frac{1}{2}&\frac{1}{2}&S_{de}
\end{matrix}
\right\}
\left[\delta_{az}\sqrt{\frac{(bz)}{(ab)}}\Gamma^{S_{de}K_{cb}}_{cb;de} + \delta_{bz}(-1)^{S_{ab}+S_{de}}\sqrt{\frac{(ac)}{(ab)}}\Gamma^{S_{de}K_{ac}}_{ac;de}\right.\\
\nonumber&\left.\qquad\qquad\qquad\qquad\qquad+ \delta_{ce}(-1)^{S_{ab}+S_{de}}\sqrt{\frac{(dz)}{(de)}}\Gamma^{S_{ab}K_{ab}}_{ab;dz}+\delta_{cd}\sqrt{\frac{(ez)}{(de)}}\Gamma^{S_{ab}}_{ab;ze}\right]\\
\nonumber&+\frac{[S_{ab}][S_{de}]}{\sqrt{(ab)(de)}}\sum_{S'} [S']^2\left\{\begin{matrix}S&\frac{1}{2}&S'\\\frac{1}{2}&\frac{1}{2}&S_{ab}\end{matrix}\right\}\left\{\begin{matrix}S&\frac{1}{2}&S'\\\frac{1}{2}&\frac{1}{2}&S_{de}\end{matrix}\right\}\left[\delta_{ad}\sqrt{(bc)(ez)}\Gamma^{S'K_{bc}}_{bc;ez}\right.\\
\nonumber&\left.\qquad\qquad + \delta_{bd}(-1)^{S_{ab}}\sqrt{(ac)(ez)}\Gamma^{S'K_{ac}}_{ac;ez} + \delta_{ae}(-1)^{S_{de}}\sqrt{(bc)(dz)}\Gamma^{S'K_{bc}}_{bc;dz}\right.\\
&\left.\qquad\qquad\qquad\qquad\qquad +\delta_{be}(-1)^{S_{ab}+S_{de}}\sqrt{(ac)(dz)}\Gamma^{S'K_{ac}}_{ac;dz}\right]~.
\label{T1_sc_ti_2DM}
\end{align}
In the above equation the intermediate momenta appearing in the 2DM part of the equation are defined as:
\begin{equation}
K_{ab} = (a+b)\%2\pi~.
\end{equation}
It can be seen that if three-particle momentum conservation is fulfilled, {\it i.e.} $(a+b+c)\%2\pi = (d+e+z)\%2\pi$, then the terms containing the 2DM conserve two-particle momentum. The three-particle matrix reduces to a block-diagonal matrix with $L$ quasi-momentum blocks for every $S$, each of which scales quadratically with the lattice size $L$.
\paragraph{The $\mathcal{T}_2$ map}
For the $\mathcal{T}_2$ map we need to define a translationally invariant two-particle-one-hole operator:
\begin{equation}
{B^\dagger}^{SK}_{ab(S_{ab})c} = \frac{1}{\sqrt{(ab)}}\left[\left[a^\dagger_a \otimes a^\dagger_b \right]^{S_{ab}K_{ab}}\otimes \tilde{a}_c\right]^{SK}~,
\label{pph_sc_op_ti}
\end{equation}
where $\tilde{a}$ is used, as defined in Eq.~(\ref{good_sto_mom}), and $K = (a+b+c)\%2\pi$. The definition of $\mathcal{T}_2$ is then exactly the same as in Eq.~(\ref{T2_sc_ensemble}), when performing the usual angular momentum recoupling and anticommutation relations, and can be rewritten as a function of the 2DM:
\begin{align}
\nonumber\mathcal{T}_2(\Gamma)^{SK(S_{ab};S_{de})}_{abc;dez} &= \frac{\delta_{S_{ab}S_{de}}}{\sqrt{(ab)(de)}}\delta_{cz}\left(\delta_{ad}\delta_{be}+(-1)^{S_{ab}}\delta_{ae}\delta_{bd}\right)\rho_{c} + \delta_{cz}\delta_{S_{ab}S_{de}}\Gamma^{S_{ab}K_{ab}}_{ab;de}\\*
\nonumber&-\frac{[S_{ab}][S_{de}]}{\sqrt{(ab)(de)}}\sum_{S'} [S']^2
\left\{
\begin{matrix}
S&\frac{1}{2}&S_{de}\\
\frac{1}{2}&S'&\frac{1}{2}\\
S_{ab}&\frac{1}{2}&\frac{1}{2}
\end{matrix}
\right\}
\left(\delta_{ad}\sqrt{(\bar{c}e)(b\bar{z})}\Gamma^{S'K_{\bar{c}e}}_{\bar{c}e;\bar{z}b}\right.\\*
\nonumber&\left.+ (-1)^{S_{ab}}\delta_{bd}\sqrt{(\bar{c}e)(\bar{z}a)}\Gamma^{S'K_{\bar{c}e}}_{\bar{c}e;\bar{z}a} + (-1)^{S_{de}}\delta_{ae}\sqrt{(\bar{c}d)(\bar{z}b)}\Gamma^{S'K_{\bar{c}d}}_{\bar{c}d;\bar{z}b}\right.\\*
&\left. \qquad\qquad\qquad\qquad\qquad\qquad+ (-1)^{S_{ab}+S_{de}}\delta_{be}\sqrt{(\bar{c}d)(\bar{z}a)}\Gamma^{S'K_{\bar{c}d}}_{\bar{c}d;\bar{z}a}\right)~.
\label{T2_sc_ti}
\end{align}
Again one sees that if the two-particle-one-hole momentum is conserved, {\it i.e.} $(a+b+c)\%2\pi=(d+e+z)\%2\pi$, then the two-particle momenta appearing in the rows and columns of the 2DM are the same, {\it e.g.} if $a = d$ then $(\bar{c}+e)\%2\pi=(\bar{z}+b)\%2\pi$. The $\mathcal{T}_2$ map also becomes block diagonal, with each block scaling quadratically with lattice size.
\paragraph{The $\mathcal{T}_2'$ map}
The $\mathcal{T}_2'$ condition can be similarly treated. For $S=\frac{3}{2}$ it again reduces to the regular $\mathcal{T}_2$ map. For $S = \frac{1}{2}$ the dimension of each $K$ block increases by one:
\begin{equation}
\mathcal{T}_2'(\Gamma)^{\frac{1}{2}k(S_{ab};S_{de})}_{abc;dez} = 
\left(
\begin{matrix}
\mathcal{T}_2(\Gamma)^{\frac{1}{2}k(S_{ab};S_{de})}_{abc;dez} & \omega^{(S_{ab})}_{abc;k}\\
{\omega^\dagger}^{(S_{de})}_{k:dez}&\rho_k\\
\end{matrix}
\right)~,
\label{T2P_sc_ti}
\end{equation}
with 
\begin{equation}
\omega^{(S_{ab})}_{abc;k} = \frac{[S_{ab}]}{\sqrt{2}}(-1)^{1+S_{ab}}\sqrt{(k\bar{c})}\Gamma^{S_{ab}K_{ab}}_{ab;k\bar{c}}~.
\end{equation}
\subsection{Parity}
The Hubbard model with periodic boundary conditions (PBC) has more symmetries than spin and translational invariance, one of them being parity. As we saw in Section~\ref{sym_atom}, which deals with atomic systems, parity follows from the symmetry under the inversion of space, {\it i.e.} $\mathbf{r}\rightarrow -\mathbf{r}$. One can readily appreciate from Figure~\ref{hubbard_fig} that a similar discrete symmetry is present here, {\it i.e.} the Hamiltonian is invariant for the inversion of the site index $i\rightarrow-i\%L$. From the Fourier transform in Eq.~(\ref{fourier}) it can be seen that the effect of this operator on a momentum state is to transform $k$ into $\bar{k} = -k\%2\pi$. However, this operation does not commute with translation, which means that the eigenstates of the Hamiltonian cannot have good momentum \emph{and} parity at the same time. From now on we assume that the number of lattice sites $L$ is even.

As a consequence, if the ground state has momentum $\ket{\mathcal{K}}\neq 0$ or $\pi$ then it is doubly degenerate, forming a doublet with $\ket{\overline{\mathcal{K}}}$. In what follows we will use this degeneracy to exploit both translational invariance and parity to reduce the dimensions of the matrices involved in the program. The following considerations are valid for every $\mathcal{K}$. We start with the simplest case, the 1DM.
\subsubsection{Translationally invariant 1DM with parity}
Single-particle space is built out of $L$ different momentum states having up or down spin, $\ket{k\sigma}$, for $0 \leq k < 2\pi$ and $\sigma=\pm\frac{1}{2}$. If we transform to a basis with good parity, momentum is no longer a good quantum number:
\begin{equation}
\ket{\tilde{k}^\pi\sigma} = {}^\rho N_k \left(\ket{k\sigma} + \pi\ket{\bar{k}\sigma}\right)~,\qquad\text{with}\qquad 0\leq\tilde{k}\leq\pi~.
\end{equation}
Two states, $k=0$ and $k=\pi$ \footnote{Note that we use $\pi$ for both the parity quantum number, $\pi=\pm1$ as for the transcendent number; in what follows it is always clear from the context what interpretation should be given to $\pi$.}
are mapped on themselves, and only positive parity states can be formed with these momenta. They have norms ${}^\rho N_k = \frac{1}{2}$. For the other states, $0<k<\pi$, both positive and negative parity combinations can be constructed, with norm ${}^\rho N_k=\frac{1}{\sqrt{2}}$. To take advantage of both symmetries at the same time, we define the 1DM using an ensemble of the $\ket{\mathcal{K}}$ and $\ket{\mathcal{\bar{K}}}$ states:
\begin{align}
\rho_{\tilde{k}^\pi\tilde{k}'^{\pi'}} =& \sum_iw_i\frac{1}{2}\frac{1}{[\mathcal{S}]^2}\sum_{\mathcal{M}}\sum_{p\in\pm}\bra{\Psi^{N(p\mathcal{K})}_{\mathcal{SM},i}}a^\dagger_{\tilde{k}^\pi\sigma}a_{\tilde{k}'^{\pi'}\sigma}\ket{\Psi^{N(p\mathcal{K})}_{\mathcal{SM},i}}\\
=&\frac{{}^{\rho}N_k^2}{2}\sum_{p\in\pm}\left[\delta_{kk'}\left(\rho^{p\mathcal{K}}_{k}+\pi\pi'\rho^{p\mathcal{K}}_{\bar{k}}\right) + \delta_{k\bar{k}'}\left(\pi'\rho^{p\mathcal{K}}_{k}+\pi\rho^{p\mathcal{K}}_{\bar{k}}\right)\right]~,
\label{1DM_ti_par}
\end{align}
in which $\rho^{p\mathcal{K}}_k$ is a regular translationally invariant 1DM as defined in Eq.~(\ref{1DM_ti}). Because of parity symmetry we have,
\begin{equation}
\rho^{\mathcal{K}}_k = \rho^{\bar{\mathcal{K}}}_{\bar{k}}~,
\end{equation}
from which it follows that Eq.~(\ref{1DM_ti_par}) is diagonal in $\pi$. If we now define:
\begin{equation}
\rho_k = \frac{1}{2}\left(\rho_k^{\mathcal{K}}+\rho_k^{\bar{\mathcal{K}}}\right)~,
\end{equation}
the translationally invariant 1DM with good parity reduces to:
\begin{equation}
\rho_{\tilde{k}^\pi} = \frac{1}{2}\left(\rho_k + \rho_{\bar{k}}\right)~,
\end{equation}
which can be seen to be independent of parity. The 1DM is still diagonal in $\tilde{k}$, but less elements have to be stored, since for $0<\tilde{k}<\pi$ there is a degeneracy in parity.
\subsubsection{Translationally invariant 2DM with parity}
In contrast to parity in atomic systems, the parity of a two-particle state for translationally invariant systems is not the product of the parities of the single-particle states building up the two-particle state. Instead, the parity is inherently a two-particle property, and the single-particle states building up the two-particle states have no good parity:
\begin{equation}
\ket{ab;S\tilde{K}^\pi} = {}^{\Gamma}N_{ab}^K\left(\ket{ab;SK} + \pi\ket{\bar{a}\bar{b};S\bar{K}}\right)~,
\label{tp_ti_par}
\end{equation}
with
\begin{equation}
0\leq \tilde{K}\leq\pi\qquad\text{and}\qquad0\leq a,b < 2\pi~.
\end{equation}
In general the 2DM is defined using a $p\mathcal{K}$ ensemble:
\begin{align}
\Gamma^{S\tilde{K}^\pi}_{ab;cd} =& \sum_iw_i\frac{1}{2}\frac{1}{[\mathcal{S}]^2}\sum_{\mathcal{M}}\sum_{p\in\pm}\bra{\Psi^{N(p\mathcal{K})}_{\mathcal{SM},i}}{B^\dagger}^{S\tilde{K}^\pi}_{ab}B^{S\tilde{K}^\pi}_{cd}\ket{\Psi^{N(p\mathcal{K})}_{\mathcal{SM},i}}~,
\end{align}
where
\begin{equation}
{B^\dagger}^{S\tilde{K}^\pi}_{ab} = {}^{\Gamma}N^{K}_{ab}\left({B^\dagger}^{SK}_{ab}+\pi{B^\dagger}^{S\bar{K}}_{\bar{a}\bar{b}}\right)~,
\end{equation}
and with ${B^\dagger}^{SK}$ defined as in Eq.~(\ref{2DM_op_transin}). Because of this $p\mathcal{K}$-ensemble definition and the fact that parity symmetry implies that,
\begin{equation}
^{\mathcal{K}}\Gamma^{SK}_{ab;cd} = ~^{-\mathcal{K}}\Gamma^{S\bar{K}}_{\bar{a}\bar{b};\bar{c}\bar{d}}~,
\end{equation}
one sees that the 2DM becomes diagonal in two-particle parity.
As was the case for the 1DM, the $\tilde{K} = 0$ and $\pi$ are mapped on themselves, but because the single-particle momenta $a,b$ change, both positive and negative parity combinations can now be formed. Let us take a look at the different possibilities:
\paragraph{$\mathbf{{\tilde{K}} = 0}$:}
for $\tilde{K} = 0$ the single-particle indices in Eq.~(\ref{tp_ti_par}) have to satisfy:
\begin{equation}
(a+b)\%2\pi = 0\qquad\text{or}\qquad a = \bar{b}~.
\end{equation}
This means that $\tilde{K} = 0$ states can be written as:
\begin{equation}
\ket{a\bar{a};S0^\pi} = {}^{\Gamma}N_{a\bar{a}}^0\left(\ket{a\bar{a};S0} + \pi\ket{\bar{a}a;S0}\right)~,
\label{tp_ti_par_0}
\end{equation}
in which the second term is equal to the first but with exchanged single-particle indices. From previous discussions we know that the symmetry of the two-particle state under the exchange of the single-particle indices depends on the two-particle spin, {\it i.e.}
\begin{equation}
\ket{a\bar{a};S0} = (-1)^S\ket{\bar{a}a;S0}~.
\end{equation}
One can see from Eq.~(\ref{tp_ti_par_0}) that for $S= 0$ only the positive parity states, and for $S=1$ only negative parity states remain. The norm is given ${}^{\Gamma}N^0_{a\bar{a}} = \frac{1}{2}$. In the $\tilde{K}=0$ case, the definition of the parity-symmetric 2DM as a function of the regular translationally invariant 2DM then reduces to:
\begin{equation}
\Gamma^{S0^\pi}_{ab;cd} = \delta_{\pi(-1)^S}\Gamma^{S0}_{ab;cd}~.
\label{2DM_ti_par_0}
\end{equation}
\paragraph{$\mathbf{0<{\tilde{K}} <\boldsymbol\pi}$:}
for $0<\tilde{K} <\pi$ the first and second term in Eq.~(\ref{tp_ti_par}) consist of different single-particle indices $a\neq \bar{b}$, implying that both positive and negative parity combinations can be constructed, with norm ${}^{\Gamma}N_{ab}^K = \frac{1}{\sqrt{2}}$. As shown for the 1DM, the $p\mathcal{K}$ ensemble makes the 2DM diagonal in, and independent of, parity. Hence every block is twofold degenerate. Since $K\neq\bar{K}$ there are only two terms remaining in the definition of the parity-symmetric 2DM:
\begin{equation}
\Gamma^{S\tilde{K}^\pi}_{ab;cd} = \frac{1}{2}\left(\Gamma^{SK}_{ab;cd}+\Gamma^{S\bar{K}}_{\bar{a}\bar{b};\bar{c}\bar{d}}\right)~.
\label{2DM_ti_par_K}
\end{equation}
\paragraph{$\mathbf{{\tilde{K}} = \boldsymbol\pi}$:}
Finally, for this block $K$ again equals $\bar{K}$. In this case there is always one state that is mapped on itself, and for which only a positive parity combination can be formed, {\it i.e.} $\ket{0\pi;S\pi^+}$, with norm ${}^{\Gamma}N_{0\pi}^\pi = \frac{1}{2}$. For all the other states in this block both positive and negative parity combinations can be formed, with norm ${}^{\Gamma}N_{ab}^\pi=\frac{1}{\sqrt{2}}$. Because of the $p\mathcal{K}$ ensemble, the 2DM falls apart in a positive and negative parity block, and since $K=\bar{K}$, four terms remain in the definition of the 2DM:
\begin{equation}
\Gamma^{S\pi^\pi}_{ab;cd} = {}^{\Gamma}N_{ab}^\pi~{}^{\Gamma}N_{cd}^\pi\left(\Gamma^{S\pi}_{ab;cd} + \Gamma^{S\pi}_{\bar{a}\bar{b};\bar{c}\bar{d}} + \pi\left[\Gamma^{S\pi}_{ab;\bar{c}\bar{d}}+\Gamma^{S\pi}_{\bar{a}\bar{b};cd}\right]\right)~.
\label{2DM_ti_par_pi}
\end{equation}
We observe from Eq.~(\ref{2DM_ti_par_pi}) that the original $K=\pi$ block reduces to a positive and negative parity block, for both the $S=0$ and $S=1$ part. Also note that there is \emph{no} degeneracy between the positive and negative parity block!
\subsubsection{Parity symmetric form of the two-index constraints}
Before deriving the parity-symmetric form of the two-index constraints, the derivation of the parity-symmetric form of the 1DM out of the 2DM is considered. Expressed as a function of the regular translationally invariant 2DM, the parity-symmetric 1DM is given by:
\begin{align}
\nonumber\rho_{\tilde{a}} =& \frac{1}{2}\left(\rho_a + \rho_{\bar{a}}\right)\\
=& \frac{1}{N-1}\sum_S\frac{[S]^2}{2}\sum_{b}(ab)\frac{1}{2}\left[\Gamma^{SK}_{ab;ab}+\Gamma^{S\bar{K}}_{\bar{a}\bar{b};\bar{a}\bar{b}}\right]~.
\end{align}
From the previous Section it can be deduced that, for all values of $\tilde{K}$:
\begin{equation}
\Gamma^{SK}_{ab;cd} + \Gamma^{S\bar{K}}_{\bar{a}\bar{b};\bar{c}\bar{d}} = \frac{1}{2~{}^{\Gamma}N_{ab}^K~{}^{\Gamma}N_{cd}^K}\left[\Gamma^{S\tilde{K}^{(+)}}_{ab;cd} + \Gamma^{S\tilde{K}^{(-)}}_{ab;cd}\right]~,
\end{equation}
from which it follows that:
\begin{equation}
\rho_{\tilde{a}}= \frac{1}{N-1}\sum_S\frac{[S]^2}{2}\sum_{b}\frac{(ab)}{4~{}^{\Gamma}{N_{ab}^K}^2}\sum_\pi\Gamma^{S\tilde{K}^\pi}_{ab;ab}~.
\end{equation}
Bear in mind that the single-particle momentum $a$ in the above equations is in the range $0\leq a\leq\pi$, but the momentum over which is summed, $b$, runs over the entire range $0\leq b<2\pi$.
\paragraph{The $\mathcal{Q}$ condition:}
the $\mathcal{Q}$ matrix is analogous to the 2DM $\Gamma$, and the expression of the parity-symmetric $\mathcal{Q}$ as a function of the regular translationally invariant $\mathcal{Q}$ matrix is exactly the same as for the 2DM in Eqs.~(\ref{2DM_ti_par_0}), (\ref{2DM_ti_par_K}) and (\ref{2DM_ti_par_pi}). The expression of the parity-symmetric $\mathcal{Q}$ map as a function of the parity-symmetric 2DM is then easy to derive using Eq.~(\ref{ti_Q_map}), leading to the following expression:
\begin{align}
\mathcal{Q}(\Gamma)^{S\tilde{K}^\pi}_{ab;cd} =& \Gamma^{SK^\pi}_{ab;cd} + \frac{\left(\delta_{ac}\delta_{bd}+(-1)^S\delta_{ad}\delta_{bc}\right)}{\sqrt{(ab)(cd)}}\left[\frac{2\mathrm{Tr}~\Gamma}{N(N-1)}-\rho_{\tilde{a}}-\rho_{\tilde{b}}\right]~.
\label{Q_2DM_sc_ti}
\end{align}
This expression is valid for all values of $\tilde{K}$. The only difficulty is that the indices on the left-hand side can have the values $0\leq a<2\pi$, while for the 1DM only terms $\leq \pi$ are stored. This is solved by using the \emph{tilde} above the index, by which is implied that:
\begin{equation}
\tilde{a} = 
\left\{
\begin{matrix}
a\qquad\text{if}\qquad a \leq \pi\\
\bar{a}\qquad\text{if}\qquad a > \pi
\end{matrix}
\right.
~.
\end{equation}
\paragraph{The $\mathcal{G}$ condition:}
the parity-symmetric form of a particle-hole state is defined as:
\begin{equation}
\ket{ab;S\tilde{K}^\pi} =~^\mathcal{G}N^K_{ab} \left(\left[a^\dagger_a \otimes \tilde{a}_b\right]^{SK} + \left[a^\dagger_{\bar{a}} \otimes \tilde{a}_{\bar{b}}\right]^{S\bar{K}}\right)\ket{0}~,
\end{equation}
in which the hole operator $\tilde{a}_{k\sigma}$ is defined as in Eq.~(\ref{good_sto_mom}). Using this parity-symmetric particle-hole operator the $\mathcal{G}$ map is defined in a $p\mathcal{K}$ ensemble, which once again renders the matrix diagonal in particle-hole parity. The particle-hole states can be divided into two classes, on the one hand $\tilde{K}=0\text{ or }\pi$, and on the other hand those states which are mapped on a different momentum. For simplicity we first consider this last class.
\paragraph{$\mathbf{0 < \tilde{K} < \boldsymbol\pi}$:} in this case one can construct both positive and negative parity combinations, with norm $^\mathcal{G}N_{ab}^K = \frac{1}{\sqrt{2}}$. The resulting $\mathcal{G}$ matrix contains only two terms, because of momentum conservion, and is independent of particle-hole parity, so every block is twofold degenerate:
\begin{equation}
\mathcal{G}(\Gamma)^{S\tilde{K}^\pi}_{ab;cd} = \frac{1}{2}\left[\mathcal{G}(\Gamma)^{SK}_{ab;cd}+\mathcal{G}(\Gamma)^{S\bar{K}}_{\bar{a}\bar{b};\bar{c}\bar{d}}\right]~.
\end{equation}
This implies the following expression of the $\mathcal{G}$ map as a function of the parity-symmetric 2DM:
\begin{align}
\mathcal{G}(\Gamma)^{S\tilde{K}^\pi}_{ab;cd} =& \delta_{ac}\delta_{bd}\rho_{\tilde{a}}-\frac{\sqrt{(a\bar{d})(c\bar{b})}}{4~{}^{\Gamma}N^{K'}_{a\bar{d}}~{}^{\Gamma}N^{K'}_{c\bar{b}}}\sum_{S'}[S']^2
\left\{
\begin{matrix}
\frac{1}{2}&\frac{1}{2}&S\\
\frac{1}{2}&\frac{1}{2}&S'
\end{matrix}
\right\}
\sum_{\pi'}
\Gamma^{S'\tilde{K}'^{\pi'}}_{a\bar{d};c\bar{b}}~.
\end{align}
\paragraph{$\mathbf{\tilde{K} = 0}$ and $\mathbf{\tilde{K} = \boldsymbol\pi}$:} both the $\tilde{K} = 0$ and $\tilde{K} = \pi$ blocks are mapped on themselves. For $\tilde{K} = 0$ the action of the parity operator is again to exchange the single-particle momenta, but in contrast with the two-particle case, there is no symmetry between the particle and the hole index. As a consequence positive and negative parity combinations for both $\tilde{K} = 0$ and $\tilde{K} = \pi$ can be constructed, with norms $^\mathcal{G}N_{ab}^K=\frac{1}{\sqrt{2}}$. There are a few exceptions however: for $\tilde{K} = 0$, the states with $a = b = 0$ and $a = b = \pi$, and for $\tilde{K} = \pi$ the states with $a = 0$, $b = \pi$ and $a=\pi,b= 0$, are mapped on themselves and only occur in the positive parity block, with norm $^{\mathcal{G}}N_{ab}^{K} = \frac{1}{2}$. The general form of the parity-symmetric $\mathcal{G}$ map when ${K} ={\bar{K}}$, as a function of the regular translationally invariant $\mathcal{G}$ is:
\begin{align}
\mathcal{G}(\Gamma)^{S\tilde{K}^\pi}_{ab;cd} =& ~^{\mathcal{G}}N^K_{ab}~^{\mathcal{G}}N^K_{cd}\left[
\mathcal{G}(\Gamma)^{SK}_{ab;cd}
+\mathcal{G}(\Gamma)^{SK}_{\bar{a}\bar{b};\bar{c}\bar{d}}
+\pi\left(\mathcal{G}(\Gamma)^{SK}_{ab;\bar{c}\bar{d}}
+\mathcal{G}(\Gamma)^{SK}_{\bar{a}\bar{b};cd}
\right)
\right]~.
\end{align}
In this case the expression of $\mathcal{G}$ as a function of the 2DM is a bit more complicated:
\begin{align}
\nonumber\mathcal{G}(\Gamma)^{S\tilde{K}^\pi}_{ab;cd} =& \delta_{ac}\delta_{bd}\rho_{\tilde{a}}
-~^\mathcal{G}N^K_{ab} ~^\mathcal{G}N^K_{cd}
\sum_{S'}[S']^2
\left\{
\begin{matrix}
\frac{1}{2}&\frac{1}{2}&S\\
\frac{1}{2}&\frac{1}{2}&S'
\end{matrix}
\right\}
\left[
\frac{\sqrt{(a\bar{d})(c\bar{b})}}{4~{}^{\Gamma}N^{K'}_{a\bar{d}}{}^{\Gamma}N^{K'}_{c\bar{b}}}
\sum_{\pi'}
\Gamma^{S'\tilde{K}'^{\pi'}}_{a\bar{d};c\bar{b}}\right.\\
&\left. \qquad\qquad\qquad\qquad\qquad\qquad+ \pi\frac{\sqrt{(a{d})(c{b})}}{4~{}^{\Gamma}N^{K''}_{\bar{a}\bar{d}}~{}^{\Gamma}N^{K''}_{c{b}}}
\sum_{\pi'}
\Gamma^{S'\tilde{K}''^{\pi''}}_{\bar{a}\bar{d};c{b}}\right]~.
\end{align}
\subsubsection{Parity symmetric form of the three-index constraints}
The construction of a complete parity-symmetric three-particle basis involves a lot of bookkeeping. The main problem is that the parity transformation acting on the states as defined in Table~{\ref{dp_dim}} does not always transform them into states that are either identical or orthogonal to the original state, as was the case in the two-particle and particle-hole space. We define a parity-symmetric three-particle state as:
\begin{equation}
\ket{abc;(S_{ab})S\tilde{K}^\pi}= ~^{\mathcal{T}_1}N^{SK}_{ab(S_{ab})c}\left(\ket{abc;(S_{ab})SK}+\pi\ket{\bar{a}\bar{b}\bar{c};(S_{ab})S\bar{K}}\right)~.
\label{par_T1}
\end{equation}
A description of the reduction of the three-particle basis in a parity-symmetric form is given by considering three cases separately.
\paragraph{$\mathbf{\tilde{K} = 0}$:}
the spin-symmetric basis was built up of states as shown in Table~\ref{dp_dim}. The same ordering holds in every $K$ block when translational invariance is included, but only states for which $(a+b+c)\%2\pi=K$ are allowed in the block.

If $S=\frac{1}{2}$, both $S_{ab} = 0$ and $S_{ab} = 1$ are possible intermediary spins. For $S_{ab} = 0$ there are two possible configurations: $a=b\neq c$ and $a<b<c$. For $S_{ab} = 1$ only the latter configuration remains. For the first type of configuration, {\it i.e.} $\ket{aab;(0)\frac{1}{2}0}$, the parity operation leads to a one-to-one mapping of states with $a<\pi$ on states with $a>\pi$. These states are linearly independent, which means both positive and negative parity states can be formed, with norm $^{\mathcal{T}_1}N^{\frac{1}{2}0}_{aa(0)b} = \frac{1}{\sqrt{2}}$. The state with $a = \pi$, {\it i.e.} the state $\ket{\pi\pi0;(0)\frac{1}{2}0}$, is mapped on itself, and only a positive parity state can be formed, with norm $^{\mathcal{T}_1}N^{\frac{1}{2}0}_{\pi\pi(0)0}=\frac{1}{2}$. 

For $a<b<c$, both positive and negative parity states can be formed when $a>0$. However, when $a = 0$, the state is mapped by a parity transformation on a non-orthogonal state:
\begin{equation}
\ket{0\bar{a}\bar{b};(S_{ab})\frac{1}{2}0} = [S_{ab}](-1)^{S_{ab}}\sum_{S_{ac}}[S_{ac}](-1)^{S_{ac}}
\left\{
\begin{matrix}
\frac{1}{2}&\frac{1}{2}&S_{ac}\\
\frac{1}{2}&\frac{1}{2}&S_{ab}
\end{matrix}
\right\}
\ket{0{a}{b};(S_{ac})\frac{1}{2}0}~.
\end{equation}
Constructing good parity combinations out of these states, as in Eq.~(\ref{par_T1}), mixes up the intermediary spin quantum number $S_{ab}$, which can no longer be used to label our basis. This is best avoided, and it is more convenient to use the configuration:
\begin{equation}
\ket{ab0;(S_{ab})S0^\pi} = \frac{1}{2}\left(\ket{ab0;(S_{ab})S0} + \pi\ket{ba0;(S_{ab})\frac{1}{2}0}\right)~,
\end{equation}
when one of the momenta is zero. These states have good parity while simultaneously keeping $S_{ab}$ as a good quantum number. When $S_{ab} = 0$ only the positive parity combination is possible, when $S_{ab} = 1$ only the negative combination remains.

For $S = \frac{3}{2}$, the states are completely antisymmetrical in the momentum indices $abc$, and the construction of a parity-symmetric basis becomes much simpler. The configurations $a<b<c$ span a complete basis. When $a > 0$, parity maps these states on linearly independend states, implying that both positive and negative parity states can be formed. When $a=0$, the action of the parity operator is:
\begin{equation}
\hat{P}\ket{0ab;(1)\frac{3}{2}0} = \ket{0ba;(1)\frac{3}{2}0} = -\ket{0ab;(1)\frac{3}{2}0}~,
\end{equation}
so only negative parity combinations can be formed.
\paragraph{$\mathbf{0<\tilde{K}<\boldsymbol\pi}$:}
when the momentum of a three-particle state is between 0 and $\pi$, it is mapped by the parity operation on an orthogonal state, and both positive and negative parity states can be constucted using Eq.~(\ref{par_T1}).
\paragraph{$\mathbf{\tilde{K}=\boldsymbol\pi}$:}
three-particle states with $\tilde{K}=\pi$ are mapped on states which are not necessarily orthogonal, so analogous problems arise as for $\tilde{K}=0$.

For $S=\frac{1}{2}$ and $S_{ab} = 0$, we have the configuration $a=b\neq c$. There is one state of this type which is mapped on itself and for which only a positive parity combination can be formed, {\it i.e.} $\ket{00\pi;(0)\frac{1}{2}\pi}$. For the other states of this type both positive and negative combinations can be formed. When $a<b<c$ there is no problem if $b\neq \pi$. If $b=\pi$, however, the parity operator transforms the state into a non-orthogonal state:
\begin{equation}
\ket{\bar{a}\pi\bar{b};(S_{ab})\frac{1}{2}\pi} = [S_{ab}](-1)^{S_{ab}}\sum_{S_{ac}}[S_{ac}](-1)^{S_{ac}}
\left\{
\begin{matrix}
\frac{1}{2}&\frac{1}{2}&S_{ac}\\
\frac{1}{2}&\frac{1}{2}&S_{ab}
\end{matrix}
\right\}
\ket{{a}\pi{b};(S_{ac})\frac{1}{2}\pi}~,
\end{equation}
which means that a parity-symmetric state will again mix up states with different intermediate spin. To avoid this we use a different configuration in this block when one of the indices is equal to $\pi$:
\begin{equation}
\ket{ab\pi;(S_{ab})\frac{1}{2}\pi^\pi} = \frac{1}{2}\left(\ket{ab\pi;(S_{ab})\frac{1}{2}\pi} + \pi\ket{ba\pi;(S_{ab})\frac{1}{2}\pi}\right)~.
\end{equation}
This state has good parity and intermediate spin and only two states remain, {\it i.e.} the positive parity state for $S_{ab} = 0$ and the negative parity state for $S_{ab} = 1$.

When $S=\frac{3}{2}$ the complete antisymmetry in the momenta again makes the construction of a parity-symmetric basis easier. Almost all states are mapped on a orthogonal state, so positive and negative parity combinations can be formed. Only when one of the indices is equal to $\pi$, the positive parity state vanishes because of the antisymmetry.
\paragraph{The $\mathcal{T}_1$ map:}
the parity-symmetric translationally invariant $\mathcal{T}_1$ condition is defined using a $p\mathcal{K}$ ensemble:
\begin{align}
\nonumber\mathcal{T}_1(\Gamma)^{S(S_{ab};S_{de})\tilde{K}^\pi}_{abc;dez} =& \sum_iw_i\frac{1}{2}\frac{1}{[\mathcal{S}]^2}\sum_{\mathcal{M}}\sum_{p\in\pm}\left(\bra{\Psi^{N(p\mathcal{K})}_{\mathcal{SM},i}}{B^\dagger}^{S\tilde{K}^\pi}_{ab(S_{ab})c}B^{S\tilde{K}^\pi}_{de(S_{de})z}\ket{\Psi^{N(p\mathcal{K})}_{\mathcal{SM},i}}\right.\\
&\left.\qquad\qquad\qquad\qquad\qquad+\bra{\Psi^{N(p\mathcal{K})}_{\mathcal{SM},i}}{B}^{S\tilde{K}^\pi}_{de(S_{de})z} {B^\dagger}^{S\tilde{K}^\pi}_{ab(S_{ab})c}\ket{\Psi^{N(p\mathcal{K})}_{\mathcal{SM},i}}\right)~,
\end{align}
where
\begin{equation}
{B^\dagger}^{S\tilde{K}^\pi}_{ab(S_{ab})c} = ~^{\mathcal{T}_1}N^{SK}_{ab(S_{ab})c}\left({B^\dagger}^{SK}_{ab(S_{ab})c}+\pi{B^\dagger}^{S\bar{K}}_{\bar{a}\bar{b}(S_{ab})\bar{c}}\right)~,
\end{equation}
and with ${B^\dagger}$ defined as in Eq.~(\ref{dp_sc_op}). Because of the $p\mathcal{K}$ ensemble the $\mathcal{T}_1$ condition is diagonal in both parity and momentum $\tilde{K}$. For $0<\tilde{K}<\pi$ the parity-symmetric $\mathcal{T}_1$ is expressed as a function of the regular translationally invariant $\mathcal{T}_1$ as:
\begin{equation}
\mathcal{T}_1(\Gamma)^{S(S_{ab};S_{de})\tilde{K}}_{abc;dez} = ~^{\mathcal{T}_1}N^{SK}_{ab(S_{ab})c}~^{\mathcal{T}_1}N^{SK}_{de(S_{de})z}\left(\mathcal{T}_1(\Gamma)^{S(S_{ab};S_{de})K}_{abc;dez}+\mathcal{T}_1(\Gamma)^{S(S_{ab};S_{de})\bar{K}}_{\bar{a}\bar{b}\bar{c};\bar{d}\bar{e}\bar{z}}\right)~,
\end{equation}
which is independent of parity, and as such doubly degenerate. The blocks with $\tilde{K} = 0$ or $\pi$ are split up in a positive and negative parity block. Since $K=\bar{K}$ four terms remain in the definition of the parity-symmetric $\mathcal{T}_1$:
\begin{align}
\nonumber\mathcal{T}_1(\Gamma)^{S(S_{ab};S_{de})\tilde{K}}_{abc;dez} =& ~^{\mathcal{T}_1}N^{SK}_{ab(S_{ab})c}~^{\mathcal{T}_1}N^{SK}_{de(S_{de})z}\left(\mathcal{T}_1(\Gamma)^{S(S_{ab};S_{de})K}_{abc;dez}+\mathcal{T}_1(\Gamma)^{S(S_{ab};S_{de}){K}}_{\bar{a}\bar{b}\bar{c};\bar{d}\bar{e}\bar{z}}\right.\\
&\left.\qquad\qquad\qquad\qquad+\pi\left[\mathcal{T}_1(\Gamma)^{S(S_{ab};S_{de})K}_{abc;\bar{d}\bar{e}\bar{z}}+\mathcal{T}_1(\Gamma)^{S(S_{ab};S_{de}){K}}_{\bar{a}\bar{b}\bar{c};{d}{e}{z}}\right]
\right)~.
\end{align}
When expressed as a function of the 2DM, the part of the $\mathcal{T}_1$ condition that is a function of the 1DM (which we call $\mathcal{T}_1^{\text{sp}}$) is the same for all $\tilde{K}$ values:
\begin{align}
&\nonumber\frac{\mathcal{T}^{\text{sp}}_1(\Gamma)^{S\tilde{K}(S_{ab};S_{de})}_{abc;dez}}{~^{\mathcal{T}_1}N^{SK}_{ab(S_{ab})c}~^{\mathcal{T}_1}N^{SK}_{de(S_{de})z}} =   \frac{\delta_{cz}\delta_{S_{ab}S_{de}}}{\sqrt{(ab)(de)}}\left(\delta_{ad}\delta_{be} + (-1)^{S_{ab}}\delta_{ae}\delta_{bd}\right) \left[\left(\frac{2\mathrm{Tr}~\Gamma}{N(N -1)}\right) -\rho_{\tilde{a}}-\rho_{\tilde{b}}-\rho_{\tilde{c}}\right]\\
\nonumber&\qquad\qquad\qquad-\frac{[S_{ab}][S_{de}]}{\sqrt{(ab)(de)}} 
\left\{
\begin{matrix}S&\frac{1}{2}&S_{ab}\\
\frac{1}{2}&\frac{1}{2}&S_{de}
\end{matrix}
\right\}\left(\delta_{az}\delta_{cd}\delta_{be}+(-1)^{S_{de}}\delta_{az}\delta_{ce}\delta_{bd} + (-1)^{S_{ab}}\delta_{bz}\delta_{ae}\delta_{cd}\right.\\
&\left.\qquad\qquad\qquad\qquad\qquad+ (-1)^{S_{ab}+S_{de}}\delta_{bz}\delta_{ce}\delta_{ad}\right)\left[\rho_{\tilde{a}}+\rho_{\tilde{b}}+\rho_{\tilde{c}}\right]~.
\end{align}
The remaining part, which is a function of the 2DM, is different when $0<\tilde{K}<\pi$ or $\tilde{K} = 0,\pi$. The parity-symmetric expression of this part of the $\mathcal{T}_1$ map can be obtained by replacing all 2DM terms appearing in Eq.~(\ref{T1_sc_ti_2DM}) by their parity symmetric counterparts. For $0<\tilde{K}<\pi$ the correct replacement is:
\begin{equation}
\delta_{cz}\Gamma^{S_{ab}K_{ab}}_{ab;de} \rightarrow \delta_{cz}\frac{{~^{\mathcal{T}_1}N^{SK}_{ab(S_{ab})c}~^{\mathcal{T}_1}N^{SK}_{de(S_{de})z}}}{2~^\Gamma N^{K_{ab}}_{ab}~^\Gamma N^{K_{ab}}_{de}}\sum_{\pi'}\Gamma^{S_{ab}\tilde{K}_{ab}^{\pi'}}_{ab;de}~.
\end{equation}
For $\tilde{K} = 0$ or $\pi$, the replacement becomes:
\begin{align}
\delta_{cz}\Gamma^{S_{ab}K_{ab}}_{ab;de} \rightarrow& {~^{\mathcal{T}_1}N^{SK}_{ab(S_{ab})c}~^{\mathcal{T}_1}N^{SK}_{de(S_{de})z}}\left(\frac{\delta_{cz}\sum_{\pi'}\Gamma^{S_{ab}\tilde{K}_{ab}^{\pi'}}_{ab;de}}{2~^\Gamma N^{K_{ab}}_{ab}~^\Gamma N^{K_{ab}}_{de}} + \pi\frac{\delta_{c\bar{z}}\sum_{\pi'}\Gamma^{S_{ab}\tilde{K}_{ab}^{\pi'}}_{ab;\bar{d}\bar{e}}}{2~^\Gamma N^{K_{ab}}_{ab}~^\Gamma N^{K_{ab}}_{\bar{d}\bar{e}}} \right)~.
\end{align}
\paragraph{The $\mathcal{T}_2$ map:}
the construction of a parity-symmetric two-particle-one-hole matrix is easier, as there is no symmetry between the particle and hole indices. A two-particle-one-hole state with good parity can be constructed through:
\begin{equation}
\ket{abc;(S_{ab})S\tilde{K}^\pi} = ~^{\mathcal{T}_2}N^{SK}_{ab(S_{ab})c}\left(\ket{abc;(S_{ab})S{K}} +\pi\ket{\bar{a}\bar{b}\bar{c};(S_{ab})S{K}}\right)~.
\end{equation}
When $0<\tilde{K}<\pi$ all states are mapped on orthogonal ones by the parity operator, so positive and negative parity states can be formed. For $\tilde{K}=0$ or $\pi$, some states are mapped on themselves, {\it i.e.} for $\tilde{K} = 0$:
\[\ket{0\pi\pi;(S_{ab})S0}~,\qquad\ket{000;(S_{ab})S0}~,\qquad\text{and}\qquad\ket{\pi\pi0;(S_{ab})S0}~,\]
and for $\tilde{K}=\pi$:
\[\ket{00\pi;(S_{ab})S\pi}~,\qquad\ket{0\pi0;(S_{ab})S\pi}~,\qquad\text{and}\qquad\ket{\pi\pi\pi;(S_{ab})S\pi}~.\]
For these states only positive parity combinations can be formed, with norm $^{\mathcal{T}_2}N^{SK}_{ab(S_{ab})c} =\frac{1}{2}$.
There is another type of state for which not all parity combinations can be formed, because of the symmetry between the first two indices. For $\tilde{K} = 0$ this is:
\begin{equation}
\ket{\bar{a}\bar{b}0;(S_{ab})S0} = \ket{ba0;(S_{ab})S0}= (-1)^{S_{ab}}\ket{ab0;(S_{ab})S0}~,
\end{equation}
and for $\tilde{K}=\pi$:
\begin{equation}
\ket{\bar{a}\bar{b}\pi;(S_{ab})S\pi} = \ket{ba\pi;(S_{ab})S\pi}= (-1)^{S_{ab}}\ket{ab\pi;(S_{ab})S\pi}~.
\end{equation}
For these states, if $S_{ab} = 0$ only positive parity combinations can be formed, and if $S_{ab} = 1$ only negative combinations are allowed. The $\mathcal{T}_2$ condition is defined using the $p\mathcal{K}$ ensemble:
\begin{align}
\nonumber\mathcal{T}_2(\Gamma)^{S(S_{ab};S_{de})\tilde{K}^\pi}_{abc;dez} =& \sum_iw_i\frac{1}{2}\frac{1}{[\mathcal{S}]^2}\sum_{\mathcal{M}}\sum_{p\in\pm}\left(\bra{\Psi^{N(p\mathcal{K})}_{\mathcal{SM},i}}{B^\dagger}^{S\tilde{K}^\pi}_{ab(S_{ab})c}B^{S\tilde{K}^\pi}_{de(S_{de})z}\ket{\Psi^{N(p\mathcal{K})}_{\mathcal{SM},i}}\right.\\
&\left.\qquad\qquad\qquad\qquad\qquad+\bra{\Psi^{N(p\mathcal{K})}_{\mathcal{SM},i}}{B}^{S\tilde{K}^\pi}_{de(S_{de})z} {B^\dagger}^{S\tilde{K}^\pi}_{ab(S_{ab})c}\ket{\Psi^{N(p\mathcal{K})}_{\mathcal{SM},i}}\right)~,
\end{align}
where
\begin{equation}
{B^\dagger}^{S\tilde{K}^\pi}_{ab(S_{ab})c} = ~^{\mathcal{T}_2}N^{SK}_{ab(S_{ab})c}\left({B^\dagger}^{SK}_{ab(S_{ab})c}+\pi{B^\dagger}^{S\bar{K}}_{\bar{a}\bar{b}(S_{ab})\bar{c}}\right)~,
\end{equation}
and with ${B^\dagger}$ defined as in Eq.~(\ref{pph_sc_op_ti}). There are two different expressions of the $\mathcal{T}_2$ map, depending on the momentum of the block $\tilde{K}$. As before, when $0<\tilde{K}<\pi$ the blocks become degenerate in parity, and the $\mathcal{T}_2$ map is defined as:
\begin{equation}
\mathcal{T}_2(\Gamma)^{S(S_{ab};S_{de})\tilde{K}^\pi}_{abc;dez} =~^{\mathcal{T}_2}N^K_{ab(S_{ab})c}~^{\mathcal{T}_2}N^K_{de(S_{de})c}\left[\mathcal{T}_2(\Gamma)^{S(S_{ab};S_{de})K}_{abc;dez}+\mathcal{T}_2(\Gamma)^{S(S_{ab};S_{de})\bar{K}}_{\bar{a}\bar{b}\bar{c};\bar{d}\bar{e}\bar{z}}\right]~.
\end{equation}
The $\tilde{K} = 0$ and $\tilde{K} =\pi$ blocks, however, split up in a positive and negative parity part, which are not degenerate:
\begin{align}
\nonumber\mathcal{T}_2(\Gamma)^{S(S_{ab};S_{de})\tilde{K}^\pi}_{abc;dez} =&~^{\mathcal{T}_2}N^K_{ab(S_{ab})c}~^{\mathcal{T}_2}N^K_{de(S_{de})c}\left[\mathcal{T}_2(\Gamma)^{S(S_{ab};S_{de})K}_{abc;dez}+\mathcal{T}_2(\Gamma)^{S(S_{ab};S_{de}){K}}_{\bar{a}\bar{b}\bar{c};\bar{d}\bar{e}\bar{z}}\right.\\
&\left.\qquad\qquad\qquad\qquad+\pi\left(
\mathcal{T}_2(\Gamma)^{S(S_{ab};S_{de})K}_{\bar{a}\bar{b}\bar{c};dez}+\mathcal{T}_2(\Gamma)^{S(S_{ab};S_{de}){K}}_{{a}{b}{c};\bar{d}\bar{e}\bar{z}}
\right)\right]~.
\end{align}
Expressed as a function of the 2DM, the parity-symmetric $\mathcal{T}_2$ map stays largely unchanged from Eq.~(\ref{T2_sc_ti}). The only thing one has to do, is replace the translationally invariant 1DM by its parity-symmetric counterpart:
\begin{equation}
\rho_c\rightarrow\rho_{\tilde{c}}~,
\end{equation}
replace the first 2DM term as:
\begin{equation}
\delta_{cz}\Gamma^{S_{ab}K_{ab}}_{ab;de} \rightarrow \delta_{cz}\frac{{~^{\mathcal{T}_2}N^{SK}_{ab(S_{ab})c}~^{\mathcal{T}_2}N^{SK}_{de(S_{de})z}}}{2~^\Gamma N^{K_{ab}}_{ab}~^\Gamma N^{K_{ab}}_{de}}\sum_{\pi'}\Gamma^{S_{ab}\tilde{K}_{ab}^{\pi'}}_{ab;de}~,
\end{equation}
if $0<\tilde{K}<\pi$ or:
\begin{align}
\delta_{cz}\Gamma^{S_{ab}K_{ab}}_{ab;de} \rightarrow& {~^{\mathcal{T}_2}N^{SK}_{ab(S_{ab})c}~^{\mathcal{T}_2}N^{SK}_{de(S_{de})z}}\left(\frac{\delta_{cz}\sum_{\pi'}\Gamma^{S_{ab}\tilde{K}_{ab}^{\pi'}}_{ab;de}}{2~^\Gamma N^{K_{ab}}_{ab}~^\Gamma N^{K_{ab}}_{de}} + \pi\frac{\delta_{c\bar{z}}\sum_{\pi'}\Gamma^{S_{ab}\tilde{K}_{ab}^{\pi'}}_{ab;\bar{d}\bar{e}}}{2~^\Gamma N^{K_{ab}}_{ab}~^\Gamma N^{K_{ab}}_{\bar{d}\bar{e}}} \right)~,
\end{align}
if $\tilde{K} = 0$ or $\pi$, and replace the four remaining 2DM terms by:
\begin{equation}
\delta_{ad}\sqrt{\frac{(\bar{c}e)(\bar{z}b)}{(ab)(de)}}\Gamma^{SK_{\bar{c}e}}_{\bar{c}e;\bar{z}b} \rightarrow \delta_{ad}\sqrt{\frac{(\bar{c}e)(\bar{z}b)}{(ab)(de)}}\frac{~^{\mathcal{T}_2}N^{SK}_{ab(S_{ab})c}~^{\mathcal{T}_2}N^{SK}_{de(S_{de})z}}{2~^\Gamma N^{K_{\bar{c}e}}_{\bar{c}e}~^\Gamma N^{K_{\bar{c}e}}_{\bar{z}b}}\sum_{\pi'}\Gamma^{S'K_{\bar{c}e}^{\pi'}}_{\bar{c}e;\bar{z}b}~,
\end{equation}
if $0<\tilde{K}<\pi$ and by:
\begin{align}
\nonumber\delta_{ad}\sqrt{\frac{(\bar{c}e)(\bar{z}b)}{(ab)(de)}}\Gamma^{SK_{\bar{c}e}}_{\bar{c}e;\bar{z}b} \rightarrow& \delta_{ad}\sqrt{\frac{(\bar{c}e)(\bar{z}b)}{(ab)(de)}}\frac{~^{\mathcal{T}_2}N^{SK}_{ab(S_{ab})c}~^{\mathcal{T}_2}N^{SK}_{de(S_{de})z}}{2~^\Gamma N^{K_{\bar{c}e}}_{\bar{c}e}~^\Gamma N^{K_{\bar{c}e}}_{\bar{z}b}}\sum_{\pi'}\Gamma^{S'K_{\bar{c}e}^{\pi'}}_{\bar{c}e;\bar{z}b}\\
&+\pi \delta_{\bar{a}d}\sqrt{\frac{({c}e)(\bar{z}\bar{b})}{(ab)(de)}} \frac{~^{\mathcal{T}_2}N^{SK}_{ab(S_{ab})c}~^{\mathcal{T}_2}N^{SK}_{de(S_{de})z}}{2~^\Gamma N^{K_{{c}e}}_{{c}e}~^\Gamma N^{K_{{c}e}}_{\bar{z}\bar{b}}}\sum_{\pi'}\Gamma^{S'K_{{c}e}^{\pi'}}_{{c}e;\bar{z}\bar{b}}~,
\end{align}
if $\tilde{K} = 0$ or $\pi$.
\paragraph{The $\mathcal{T}'_2$ map}
For the parity-symmetric $\mathcal{T}_2'$ map there is only a small change compared to Eq.~(\ref{T2P_sc_ti}). As derived in the previous paragraph for the $\mathcal{T}_2$ constraint, the terms with $0<\tilde{K}<\pi$ become twofold degenerate in parity and only half of them needs to be stored. The terms with $\tilde{K} = 0$ or $\pi$ fall apart in a positive and negative parity block. The change in the parity-symmetric $\mathcal{T}_2'$ compared to the regular translationally invariant $\mathcal{T}_2'$ is that for $\tilde{K} = 0$ and $\pi$, only the positive parity part is different from the $\mathcal{T}_2$ map. The blocks that are different than the standard $\mathcal{T}_2$ have the form:
\begin{equation}
\mathcal{T}_2'(\Gamma)^{\frac{1}{2}(S_{ab};S_{de})\tilde{k}^\pi}_{abc;dez} = 
\left(
\begin{matrix}
\mathcal{T}_2(\Gamma)^{\frac{1}{2}(S_{ab};S_{de})\tilde{k}^\pi}_{abc;dez} & \omega^{(S_{ab})^\pi}_{abc;\tilde{k}}\\
{\omega^\dagger}^{(S_{de})^\pi}_{\tilde{k};dez}&\rho_{\tilde{k}}\\
\end{matrix}
\right)~.
\end{equation}
The expression of the parity-symmetric $\omega$ as a function of the 2DM is:
\begin{equation}
\omega^{(S_{ab})^\pi}_{abc;\tilde{k}} = \frac{[S_{ab}]}{\sqrt{2}}(-1)^{1+S_{ab}}\sqrt{(\tilde{k}\bar{c})}~^{\mathcal{T}_2}N^{\frac{1}{2}\tilde{k}}_{ab(S_{ab})c}\frac{\sum_{\pi'}\Gamma^{S_{ab}K_{ab}^{\pi'}}_{ab;\tilde{k}\bar{c}}}{{2~^\Gamma N^{K_{ab}}_{ab}~^\Gamma N^{K_{ab}}_{\tilde{k}\bar{c}}}}~,
\end{equation}
for $0<\tilde{k}<\pi$ and:
\begin{align}
\nonumber\omega^{(S_{ab})^\pi}_{abc;\tilde{k}} =& \frac{[S_{ab}]}{\sqrt{2}}(-1)^{1+S_{ab}}~^{\mathcal{T}_2}N^{\frac{1}{2}\tilde{k}}_{ab(S_{ab})c}
\left(
\frac{\sqrt{\tilde{k}\bar{c}}\sum_{\pi'}\Gamma^{S_{ab}K_{ab}^{\pi'}}_{ab;\tilde{k}\bar{c}}}{{2~^\Gamma N^{K_{ab}}_{ab}~^\Gamma N^{K_{ab}}_{\tilde{k}\bar{c}}}}\right.\\
&\left.\qquad\qquad\qquad\qquad\qquad\qquad\qquad\qquad +\pi\frac{\sqrt{(\tilde{k}{c})}\sum_{\pi'}\Gamma^{S_{ab}K_{\bar{a}\bar{b}}^{\pi'}}_{\bar{a}\bar{b};\tilde{k}c}}{{2~^\Gamma N^{K_{\bar{a}\bar{b}}}_{\bar{a}\bar{b}}~^\Gamma N^{K_{\bar{a}\bar{b}}}_{\tilde{k}{c}}}}
\right)~,
\end{align}
for $\tilde{k}=\pi$ or 0.
   
The listing of all the above expressions is quite tedious but must be done for completeness and future work in this area. The extra bookkeeping connected with including additional symmetries pays off, however, and leads to considerable speedup factors.

\chapter{\label{applications}Applications}
In this Chapter the $N$-representability constraints and semidefinite programming algorithms developed in the previous Chapters are applied to some physical systems. In the first Section, the spin and angular momentum symmetry decomposition discussed in Section \ref{sym_atom} is used to study the isoelectronic series of Beryllium, Neon and Silicon, see also \cite{atomic}. The next Section deals with the dissociation of diatomic molecules, and the failure of the standard $N$-representability conditions to describe the dissociation limit \cite{helen_1}. An analysis is made as to why these constraints fail, and new conditions are derived to fix this pathological behaviour \cite{qsep,helen_2}. In the last Section the one-dimensional Hubbard model is studied, exploiting all the symmetries discussed in Section \ref{hub_sym}. The performance of the different constraints discussed in Chapter \ref{n_rep} is studied for various fillings of the lattice, and it is found that higher-order constraints are needed to correctly describe the strong-correlation limit.

\section{\label{isoelectronic_series}The isoelectronic series of Beryllium, Neon and Silicon}
The isoelectronic series of atomic systems consitute a well-known benchmark problem in electronic structure theory, and a good test for the performance of any many-body method. The atomic Hamiltonian is given by Eq.~(\ref{ham_atom}). For a fixed number of electrons $N$, one can consider the central charge $Z$ as a variable, and this range of systems is called the $N$-electron atomic isoelectronic series, though usually it is named after the neutral atom species, {\it e.g.} the Beryllium series for $N=4$. If the interaction between the electrons is neglected, we recover the hygdrogenic Hamiltonian. An interesting thing about the Hydrogenic Hamiltonian is that there is an accidental symmetry, which leads to an additional degeneracy in the single-particle spectrum, corresponding to the main quantum number $n$. When the interelectronic interaction is switched on, however, this symmetry is broken and the degeneracy disappears. Using the electronic series, one can look at the transition between these two regimes. Starting from the neutral atom, $Z=N$, the central charge is increased, which is equivalent to scaling down the electron interactions. In the limit $Z\rightarrow\infty$ the hydrogenic Hamiltonian is recovered. For Beryllium, the incipient degeneracy of the $2s$ and the $2p$ levels causes the ground state to be increasingly multireference as $Z$ increases (as is shown in Fig.~\ref{beryl}), which is hard to describe for many methods, including Hartree-Fock and DFT.
\begin{figure}
\centering
\includegraphics[scale=0.5]{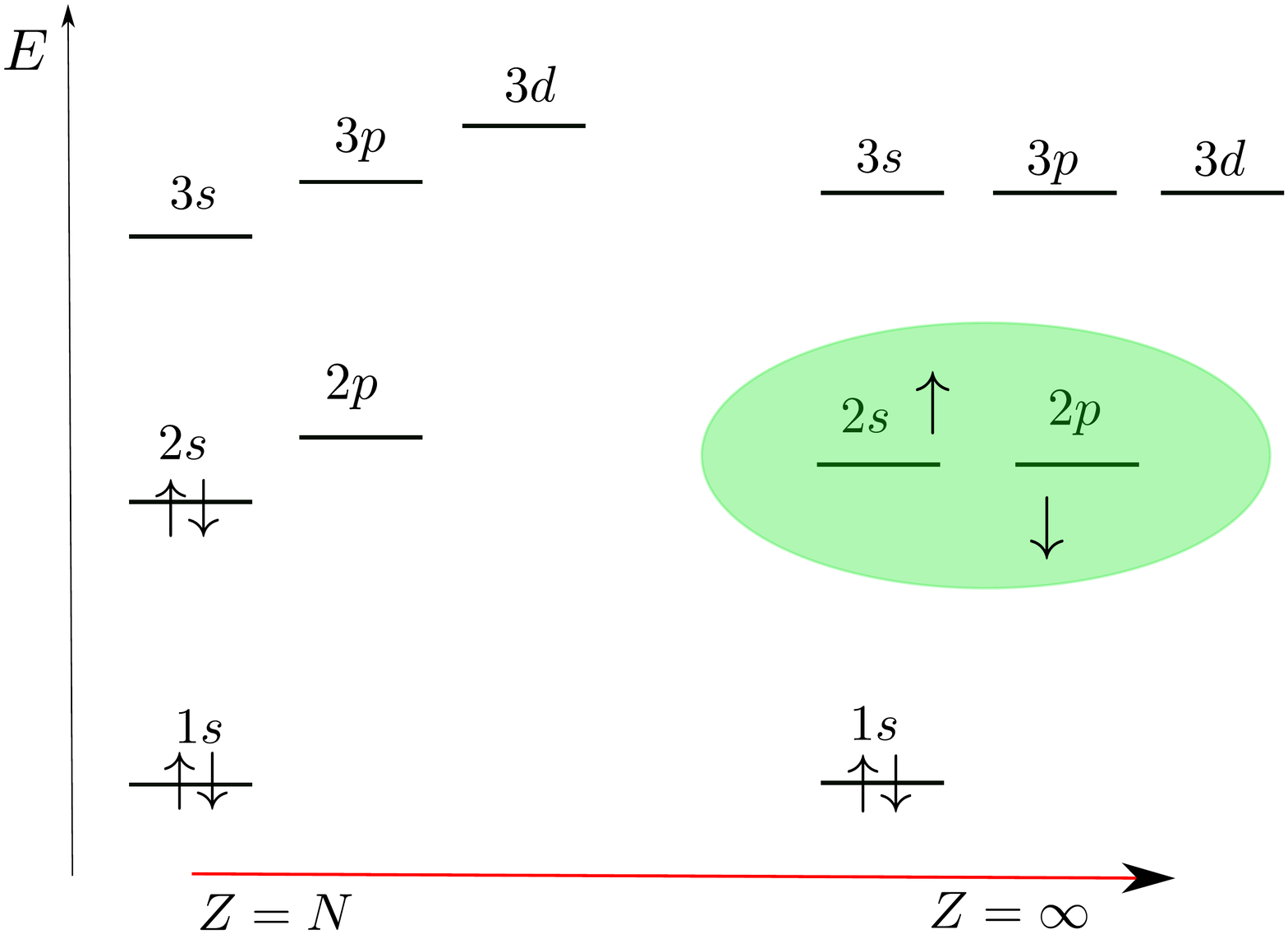}
\caption{\label{beryl} Schematic representation of the influence of the increase of $Z$ for the ground state of Beryllium. When $Z=N$, the ground state is dominated by a single reference $(1s)^2(2s)^2$ singlet. The admixture of the $(1s)^2(2p)^2$ configuration is small, but increases and becomes competitive as $Z\rightarrow\infty$.}
\end{figure}
Neon has ten electrons, and the ground state therefore has a closed shell $(1s)^2(2s)^2(2p)^6$ configuration, which means we don't expect any near-degeneracy problems. For Silicon which has 14 electrons in a $(1s)^2(2s)^2(2p)^6(3s)^2(3p)^2$ configuration, there is not only a near-degeneracy problem, but the additional complication that the ground state is not a spin and angular momentum singlet.

In this study the focus lies on three issues: the performance of the standard $\mathcal{IQG}$ $N$-representability conditions in multireference situations, the quality of the variationally obtained 2DM, which we inspect by looking at different physical properties obtainable from the 2DM. The third issue is the basis set dependence of the results, which we are able to check because we are not limited to small basis sets through the use of spin and angular momentum symmetry. The specific SDP algorithm used in this study is the simple dual-only potential reduction program explained in Section \ref{pr_sdp}.
\subsection{Spin and angular momentum constraints}
Because we know the ground-state spin and angular momentum of the systems studied, we can improve results by restricting our search to these symmetry sectors, as discussed in Chapter~\ref{symmetry}. This is done by imposing linear equality constraints, as explained in Section \ref{pr_sdp_le}, the explicit form of which is derived in this Section.
\subsubsection{\label{spinconstraints_singlet}Imposing the spin constraints for $\mathcal{S} = 0$}
The spin-coupled form of the $\hat{\mathcal{S}}_z$ operator can be written as:
\begin{equation}
\label{S_Z}
\hat{\mathcal{S}}_z = \frac{1}{\sqrt{2}}\sum_a\left[a^\dagger_{a}\otimes\tilde{a}_{a}\right]^1_0~,
\end{equation}
which is a particle-hole operator as defined in Eq.~(\ref{ph_op_sc}). This operator lives in particle-hole space, and we can force the vector
\begin{equation}
\{\hat{\mathcal{S}}_z\}^S_{ab} = \frac{1}{\sqrt{2}}\delta_{S1}\delta_{ab}~,
\end{equation}
to be an eigenvector of $\mathcal{G}(\Gamma)$ with eigenvalue zero. This can be seen from the fact that
\begin{equation}
\sum_c\mathcal{G}(\Gamma)^1_{ab;cc} = \sum_i w_i\bra{\Psi_i^N}\left[a^\dagger_a \otimes a^\dagger_b\right]^S \hat{\mathcal{S}}_z\ket{\Psi_i^N} = 0~,
\end{equation}
when the $\ket{\Psi_i^N}$ are spin singlets.
In doing this we automatically impose the same constraints on $\mathcal{G}(\Gamma)$ for $\mathcal{S}_x$ and $\mathcal{S}_y$ due to the threefold degeneracy of the $S=1$ block of the 2DM. It can be seen that in this case the expectation value of the total spin is zero,
\begin{equation}
\bra{\Psi^N}\hat{\mathcal{S}}^2\ket{\Psi^N} = \bra{\Psi^N}\hat{\mathcal{S}}_x^2+\hat{\mathcal{S}}_y^2+\hat{\mathcal{S}}_z^2\ket{\Psi^N} =0~.
\end{equation}
So the condition to be imposed on the density matrix becomes:
\begin{equation}
\label{50}
\sum_S [S]^2\left[\frac{1}{2}\frac{1}{N-1} - (-1)^S\left\{\begin{matrix}\frac{1}{2} & \frac{1}{2} & 1 \\ \frac{1}{2} & \frac{1}{2} & S\end{matrix}\right\}\right]\sum_b\Gamma^S_{ab;cb} = 0~.
\end{equation}
We will call this necessary condition for $\Gamma$ to be derivable from a singlet wave-function ensemble the projection of a two-particle density matrix on a spin-singlet state. It can be seen that there are as many constraint matrices as the dimension of 1DM space. Because of the zero eigenvalues in the $\mathcal{G}$ matrix the projected density matrix is on the edge of the feasible region during the whole of the minimization process, and as a result, the cost function is infinity. This is easily circumvented by taking the pseudo-inverse of the $\mathcal{G}$ matrix, which excludes the $\mathcal{S}_z$-state from the inversion process.

The linear constraints for imposing the spin-singlet condition are given by:
\begin{equation}
\forall ~~ k \leq l \qquad:\qquad\mathrm{Tr}~\Gamma\left(^{[kl]}E\right)= 0~, 
\end{equation}
in which the constraint matrices $ ^{[kl]}E$ have the following form:
\begin{eqnarray}
\nonumber^{[kl]}E^{S}_{ab;cd} &=& \left({}^{[kl]}f^{S}_{ab;cd}\right) + (-1)^S \left({}^{[kl]}f^{S}_{ba;cd}\right) + (-1)^S \left({}^{[kl]}f^{S}_{ab;dc}\right)+ \left({}^{[kl]}f^{S}_{ba;dc}\right) \\
&&+ \left({}^{[kl]}f^{S}_{cd;ab}\right) + (-1)^S \left({}^{[kl]}f^{S}_{dc;ab}\right)  + (-1)^S \left({}^{[kl]}f^{S}_{cd;ba}\right) + \left({}^{[kl]}f^{S}_{dc;ba}\right)~,
\end{eqnarray}
with
\begin{equation}
^{[kl]}f^{S}_{ab;cd} = \left[\frac{1}{2}\frac{1}{N-1} - (-1)^S\left\{\begin{matrix}\frac{1}{2} & \frac{1}{2} & 1 \\ \frac{1}{2} & \frac{1}{2} & S\end{matrix}\right\}\right]\delta_{ak}\delta_{cl}\delta_{bd}~.
\end{equation}
\subsubsection{Imposing the spin constraints for $\mathcal{S}\neq 0$}
For higher-spin multiplets we use the spin-averaged ensemble (see Section \ref{spin_averaged_ensemble}), in which the 2DM has the same simple structure as for the singlet case.
The expectation value of the $\hat{\mathcal{S}}^2$ spin operator is forced to be exact, using the linear constraint
\begin{equation}
\mathrm{Tr}~\Gamma \{\hat{\mathcal{S}}^2\} = \mathcal{S}(\mathcal{S}+1)~,
\end{equation}
where  $\{\hat{\mathcal{S}}^2\}$ is the two-particle matrix representation of the $\hat{\mathcal{S}}^2$ operator,
\begin{equation}
\{\hat{\mathcal{S}}^2\}^{S}_{ab;cd} = \left[\frac{3}{2}\left(\frac{2-N}{N-1}\right) + S(S+1)\right]
\left(\delta_{ac}\delta_{bd} + (-1)^S\delta_{ad}\delta_{bc}\right)~.
\end{equation}
There is only one linear constraint for nonzero spin, in contrast to the numerous constraints for the projection on a singlet state.
It can therefore be expected that the spin constraints ({\it i.e.} the constraints on the 2DM ensuring that it is derivable from a wave function with good total spin) are less accurate than those for the singlet case. It is, in fact, known how to cure this situation \cite{maz_spin} by considering not the
spin-averaged ensemble but rather the 2DM derived from the highest-weight member ($\mathcal{M}=\mathcal{S}$) of the multiplet.
Similar to the spin-singlet projection, one can then impose the condition that, since the spin-raising ladder operator $\hat{\mathcal{S}}_+$ destroys the wave function, the $\mathcal{G}(\Gamma)$ matrix must have a zero eigenvalue (with an eigenvector in particle-hole space corresponding to the $\hat{\mathcal{S}}_+$ operator).
In such a highest-weight scheme, the spin restrictions for the $\mathcal{S}\neq 0$ case are put on the same footing as for the singlet case; in fact, the highest-weight and the spin-averaged ensemble scheme are equivalent for the singlet case.
However, the highest-weight scheme for $\mathcal{S}\neq 0$ requires one to keep track of more matrices and is computationally more demanding by about a factor of 10. We therefore used the ensemble scheme even for the nonsinglets ({\it i.e.} the Si atom), though we checked some cases with the highest-weight method for the spin.
\subsubsection{Spin and angular momentum projection}
Orbital angular momentum is an additional conserved quantum number in atomic systems. The same principles as in the previous Section can be used, but the implementation is a bit more complicated. It can be shown that in a spin and angular momentum coupled basis the $z$-projections of $\mathcal{S}$ and $\mathcal{L}$ become
\begin{eqnarray}
\mathcal{S}_z &=& \frac{1}{\sqrt{2}}\sum_{nl}[l]\left[a_{nl}^\dagger\otimes\tilde{a}_{nl}\right]^{\left(0^+1\right)}~,\\
\mathcal{L}_z &=& \sqrt{\frac{2}{3}}\sum_{nl}[l]\hat{l}\left[a_{nl}^\dagger\otimes\tilde{a}_{nl}\right]^{\left(1^+0\right)}~.
\end{eqnarray}
Following the same argument as before, it can be imposed that the density matrix is derivable from an eigenstate with zero eigenvalue of respectively the $\mathcal{S}$ and $\mathcal{L}$ operators when 
\begin{eqnarray}
\sum_{c}[l_c]\mathcal{G}\left(\Gamma\right)^{\left(0^+1\right)}_{ab;cc} &=& 0~,\\
\sum_{c}[l_c]\hat{l}_c\mathcal{G}\left(\Gamma\right)^{\left(1^+0\right)}_{ab;cc} &=& 0~.
\end{eqnarray}
This can be translated into linear constraints on the 2DM, given by:
\begin{equation}
\forall ~~ k \leq l \qquad:\qquad\mathrm{Tr}~\Gamma \left({}^{[kl]}_{\mathcal{J}}E\right)= 0~,
\end{equation}
where the constraint matrices $ ^{[kl]}E $ have the following form:
\begin{align}
\nonumber^{[kl]}_{~~\mathcal{J}}E^{(L^\pi S)}_{ab;cd} =& \left({}^{[kl]}_{~~\mathcal{J}}f^{(L^\pi S)}_{ab;cd}\right) + \pi(-1)^{L+S} \left({}^{[kl]}_{~~\mathcal{J}}f^{L^\pi S}_{ba;cd}\right) + \left({}^{[kl]}_{~~\mathcal{J}}f^{(L^\pi S)}_{ba;dc}\right) + \pi(-1)^{S + L} \left({}^{[kl]}_{~~\mathcal{J}}f^{(L^\pi S)}_{ab;dc}\right)\\
& + \left({}^{[kl]}_{~~\mathcal{J}}f^{L^\pi S}_{cd;ab}\right) + \pi(-1)^{S+L} \left({}^{[kl]}_{~~\mathcal{J}}f^{L^\pi S}_{dc;ab}\right) + \left({}^{[kl]}_{~~\mathcal{J}}f^{L^\pi S}_{dc;ba}\right) + \pi(-1)^{L+S} \left({}^{[kl]}_{~~\mathcal{J}}f^{L^\pi S}_{cd;ba}\right)~,
\end{align}
where $\mathcal{J}$ can mean either $\mathcal{S}$ or $\mathcal{L}$. The parity $\pi$ has been introduced in Section~\ref{sym_atom}. The precise expression for the $f$-coefficients reads:
\begin{align}
 ^{[kl]}_{~~\mathcal{S}}f^{(L^\pi S)}_{ab;cd} =& \delta_{l_kl_l}\left[\frac{1}{2}\frac{1}{N-1} - \frac{(-1)^S}{[l_k]}\left\{\begin{matrix}\frac{1}{2} & \frac{1}{2} & 1 \\ \frac{1}{2} & \frac{1}{2} & S\end{matrix}\right\}\right] \delta_{ak}\delta_{cl}\delta_{bd}~,\\
 ^{[kl]}_{~~\mathcal{L}}f^{(L^\pi S)}_{ab;cd} =& \delta_{l_kl_l}\left[\frac{\hat{l}^2_a}{N-1} - \frac{1}{2}\left(\hat{l}_a^2 + \hat{l}_b^2 - \hat{L}^2\right)\right] \delta_{ak}\delta_{cl}\delta_{bd}~.
\end{align}
The projection on nonzero spin and angular momentum is a less strict condition, using the appropriate ensemble averaged over the third-component of $\mathcal{L}$ (see Section~\ref{sym_atom}). The expectation values of $\mathcal{L}$ and $\mathcal{S}$ are projected on the desired values:
\begin{eqnarray}
\mathrm{Tr}~\Gamma \{\hat{\mathcal{S}}^2\} &=& \mathcal{S}\left(\mathcal{S} + 1\right)~,\\
\mathrm{Tr}~\Gamma \{\hat{\mathcal{L}}^2\} &=&  \mathcal{L}\left(\mathcal{L} + 1\right)~,
\end{eqnarray}
where the $\{\hat{\mathcal{S}}^2\}$ and $\{\hat{\mathcal{L}}^2\}$ are the two-particle matrix representations of the $\hat{\mathcal{S}^2}$ and $\hat{\mathcal{L}^2}$ operators respectively:
\begin{align}
\{\hat{\mathcal{S}}^2\}^{(L^\pi S)}_{ab;cd} =& \left(\delta_{ac}\delta_{bd} + \pi(-1)^{L+S}\delta_{ad}\delta_{bc}\right)\left[\frac{3}{2}\left(\frac{2-N}{N-1}\right) + \hat{S}^2\right] ~,\\
\{\hat{\mathcal{L}}^2\}^{(L^\pi S)}_{ab;cd} =& \left(\delta_{ac}\delta_{bd} + \pi(-1)^{L+S}\delta_{ad}\delta_{bc}\right)\left[\frac{2-N}{N-1}\left(\hat{l}_a^2 + \hat{l}_b^2\right) + \hat{L}^2\right] ~.
\end{align}
\subsection{Results and discussion}
Using the dual-only potential reduction algorithm, the isoelectronic series of Be, Ne and Si were calculated from the neutral atom up to a central charge $Z = 28$. Beryllium and Neon are both elements with a singlet ground state. In the Silicon ground state the total spin and angular momentum are both one, which allows us to assess the quality of the spin and angular momentum constraints for ${\mathcal{S},\mathcal{L}} \neq 0$. In order to study the basis set dependence, the properties of the ground state of the Be and Ne series were calculated in a cc-pVDZ, a cc-pVTZ and a cc-pVQZ basis set \cite{dunning}. The Si series was  only calculated in a cc-pVDZ and a cc-pVTZ basis set \cite{woon}. Spherical harmonic (and not Cartesian) basis functions are used throughout this Section. With the density matrices obtained from the v2DM, several properties were studied. These are compared to estimates for non-relativistic energies based on experimental data \cite{davidson,chakravorty}, and to the results of coupled-cluster (CCSD) calculations, and in some cases, to full-configuration-interaction (full-CI) calculations. 

The basis functions used were those of the neutral atom, but with a rescaling $r\rightarrow rZ/N$ for the positive ions with $Z > N$.
The CCSD and full-CI results were obtained using the MOLPRO program \cite{molpro}.

\subsubsection{Ground-state energy}
The ground-state energies, calculated with various methods and in the different basis sets, are shown in Tables \ref{gse_Be_DZ}, \ref{gse_Be_TZ} and \ref{gse_Be_QZ} for Beryllium, Tables \ref{gse_Ne_DZ}, \ref{gse_Ne_TZ} and \ref{gse_Ne_QZ} for Neon and Tables \ref{gse_Si_DZ} and \ref{gse_Si_TZ} for the Silicon isoelectronic series. Even in the best case (Be in cc-pVQZ), the calculated energies are at least 20 mHartree removed from the experimental estimate in \cite{davidson,chakravorty}. This is due to the difficulty of describing the interelectronic cusp in the exact wave function using finite single-particle basis sets.
\begin{table}
\centering
\caption{\label{gse_Be_DZ}Ground-state energies of the Be series in the cc-pVDZ basis set using different methods.}
\begin{tabular}{|c|ccccc|}
\hline
Z&v2DM&HF&CCSD&full-CI&expt.\\
\hline
4&-14.617473&-14.572338&-14.617369&-14.61741&-14.66736\\
5&-24.275712&-24.216056&-24.27566&-24.275684&-24.34892\\
6&-36.387458&-36.316267&-36.387421&-36.387439&-36.53493\\
7&-50.940925&-50.860695&-50.940896&-50.940909&-51.22284\\
8&-67.931909&-67.844323&-67.931884&-67.931896&-68.41171\\
9&-87.358767&-87.265015&-87.358746&-87.358755&-88.10113\\
10&-109.22078&-109.12175&-109.22076&-109.22077&-110.29089\\
11&-133.51761&-133.414&-133.51759&-133.5176&-134.98088\\
12&-160.24908&-160.14145&-160.24906&-160.24907&-162.17102\\
13&-189.41511&-189.30392&-189.41509&-189.4151&-191.86127\\
14&-221.01564&-220.90129&-221.01563&-221.01564&-224.0516\\
15&-255.05067&-254.93347&-255.05066&-255.05066&-258.742\\
16&-291.52018&-291.4004&-291.52017&-291.52017&-295.93244\\
17&-330.42417&-330.30206&-330.42416&-330.42416&-335.62293\\
18&-371.76264&-371.63841&-371.76263&-371.76263&-377.81344\\
19&-415.5356&-415.40942&-415.53559&-415.53559&-422.50398\\
20&-461.74304&-461.61508&-461.74303&-461.74304&-469.69455\\
21&-510.38498&-510.25537&-510.38498&-510.38498&-519.38513\\
22&-561.46143&-561.3303&-561.46142&-561.46142&-571.57572\\
23&-614.97237&-614.83983&-614.97237&-614.97237&-626.26633\\
24&-670.91783&-670.78398&-670.91783&-670.91783&-683.45695\\
25&-729.29781&-729.16274&-729.2978&-729.2978&-743.14758\\
26&-790.1123&-789.97609&-790.11229&-790.1123&-805.33822\\
27&-853.36132&-853.22404&-853.36131&-853.36131&-870.02886\\
28&-919.04486&-918.90659&-919.04486&-919.04486&-937.21951\\
\hline
\end{tabular}
\end{table}
\begin{table}
\centering
\caption{\label{gse_Be_TZ}The ground-state energies of the Be series in the cc-pVTZ basis set using different methods.}
\begin{tabular}{|c|ccccc|}
\hline
Z&v2DM&HF&CCSD&full-CI&expt.\\
\hline
4&-14.625431&-14.572873&-14.623559&-14.62381&-14.66736\\
5&-24.300695&-24.234557&-24.299207&-24.29943&-24.34892\\
6&-36.473162&-36.394215&-36.471944&-36.47214&-36.53493\\
7&-51.137349&-51.045734&-51.136311&-51.136486&-51.22284\\
8&-68.290965&-68.186797&-68.290046&-68.290206&-68.41171\\
9&-87.932793&-87.816285&-87.931958&-87.932107&-88.10113\\
10&-110.06209&-109.93353&-110.06132&-110.06146&-110.29089\\
11&-134.67837&-134.53811&-134.67764&-134.67778&-134.98088\\
12&-161.7813&-161.62969&-161.78061&-161.78074&-162.17102\\
13&-191.37066&-191.20808&-191.37&-191.37011&-191.86127\\
14&-223.44627&-223.27309&-223.44563&-223.44574&-224.0516\\
15&-258.00801&-257.82461&-258.0074&-258.00751&-258.742\\
16&-295.05582&-294.86255&-295.05522&-295.05533&-295.93244\\
17&-334.58961&-334.38682&-334.58904&-334.58913&-335.62293\\
18&-376.60935&-376.39737&-376.60879&-376.60888&-377.81344\\
19&-421.115&-420.89414&-421.11445&-421.11454&-422.50398\\
20&-468.10654&-467.87712&-468.106&-468.10609&-469.69455\\
21&-517.58394&-517.34625&-517.58342&-517.5835&-519.38513\\
22&-569.5472&-569.30152&-569.54669&-569.54677&-571.57572\\
23&-623.99631&-623.7429&-623.99581&-623.99589&-626.26633\\
24&-680.93126&-680.67038&-680.93077&-680.93084&-683.45695\\
25&-740.35204&-740.08394&-740.35156&-740.35163&-743.14758\\
26&-802.25866&-801.98356&-802.25818&-802.25825&-805.33822\\
27&-866.65111&-866.36925&-866.65064&-866.65071&-870.02886\\
28&-933.5294&-933.24098&-933.52894&-933.529&-937.21951\\
\hline
\end{tabular}
\end{table}
\begin{table}
\centering
\caption{\label{gse_Be_QZ}Ground-state energies of the Be series in the cc-pVQZ basis set using different methods.}
\begin{tabular}{|c|ccccc|}
\hline
Z&v2DM&HF&CCSD&full-CI&expt.\\
\hline
4&-14.642807&-14.572968&-14.639589&-14.640124&-14.66736\\
5&-24.321254&-24.236385&-24.317643&-24.31822&-24.34892\\
6&-36.500934&-36.40257&-36.497178&-36.497761&-36.53493\\
7&-51.177145&-51.065945&-51.173335&-51.173918&-51.22284\\
8&-68.347448&-68.22364&-68.343621&-68.344203&-68.41171\\
9&-88.010503&-87.87405&-88.00667&-88.007254&-88.10113\\
10&-110.16555&-110.01625&-110.16171&-110.1623&-110.29089\\
11&-134.81213&-134.64967&-134.8083&-134.80889&-134.98088\\
12&-161.95&-161.77398&-161.94616&-161.94677&-162.17102\\
13&-191.57899&-191.38895&-191.57514&-191.57575&-191.86127\\
14&-223.699&-223.49443&-223.69514&-223.69577&-224.0516\\
15&-258.30998&-258.09031&-258.3061&-258.30675&-258.742\\
16&-295.4119&-295.17653&-295.40801&-295.40867&-295.93244\\
17&-335.00476&-334.75303&-335.00085&-335.00153&-335.62293\\
18&-377.08857&-376.81977&-377.08463&-377.08533&-377.81344\\
19&-421.66334&-421.37673&-421.65938&-421.66011&-422.50398\\
20&-468.72912&-468.42388&-468.72513&-468.72588&-469.69455\\
21&-518.28593&-517.96121&-518.28191&-518.28269&-519.38513\\
22&-570.33383&-569.9887&-570.32977&-570.33058&-571.57572\\
23&-624.87286&-624.50634&-624.86877&-624.86961&-626.26633\\
24&-681.90309&-681.51414&-681.89895&-681.89983&-683.45695\\
25&-741.42459&-741.01206&-741.4204&-741.42132&-743.14758\\
26&-803.43742&-803.00013&-803.43318&-803.43414&-805.33822\\
27&-867.94167&-867.47832&-867.93738&-867.93838&-870.02886\\
28&-934.93744&-934.44663&-934.93309&-934.93413&-937.21951\\
\hline
\end{tabular}
\end{table}
\begin{table}
\centering
\caption{\label{gse_Ne_DZ}The ground-state energies of the Ne series in the cc-pVDZ basis set using different methods.}
\begin{tabular}{|c|ccccc|}
\hline
Z&v2DM&HF&CCSD&full-CI&expt.\\
\hline
10&-128.70843&-128.48878&-128.67964&-128.68088&-128.9376\\
11&-161.80049&-161.59591&-161.77283&-161.77411&-162.0659\\
12&-198.88784&-198.70208&-198.86199&-198.86309&-199.2204\\
13&-239.97194&-239.80393&-239.94802&-239.94883&-240.3914\\
14&-285.04223&-284.88894&-285.02004&-285.02061&-285.5738\\
15&-334.08381&-333.94195&-334.06299&-334.06338&-334.7642\\
16&-387.08194&-386.94882&-387.06219&-387.06246&-387.9608\\
17&-444.02427&-443.89781&-444.00531&-444.00551&-445.1622\\
18&-504.90101&-504.77972&-504.88268&-504.88282&-506.3673\\
19&-569.70474&-569.58754&-569.68689&-569.68701&-571.5754\\
20&-638.42981&-638.31595&-638.41239&-638.41248&-640.7891\\
21&-711.07205&-710.96091&-711.05497&-711.05505&-713.9988\\
22&-787.6282&-787.51937&-787.61143&-787.6115&-791.2132\\
23&-868.09589&-867.98895&-868.07932&-868.07938&-872.4291\\
24&-952.47304&-952.36781&-952.45674&-952.45679&-957.6463\\
25&-1040.7584&-1040.6545&-1040.7422&-1040.7422&-1046.8646\\
26&-1132.9505&-1132.8479&-1132.9345&-1132.9345&-1140.0838\\
27&-1229.0485&-1228.9471&-1229.0327&-1229.0328&-1237.3039\\
28&-1329.0518&-1328.9514&-1329.0361&-1329.0361&-1338.5247\\
\hline
\end{tabular}
\end{table}
\begin{table}
\centering
\caption{\label{gse_Ne_TZ}Ground-state energies of the Ne series in the cc-pVTZ basis set using different methods.}
\begin{tabular}{|c|cccc|}
\hline
Z&v2DM&HF&CCSD&expt.\\
\hline
10&-128.86088&-128.53186&-128.81081&-128.9376\\
11&-161.97703&-161.65496&-161.92829&-162.0659\\
12&-199.11372&-198.79861&-199.06598&-199.2204\\
13&-240.26728&-239.9582&-240.22028&-240.3914\\
14&-285.43166&-285.12786&-285.38525&-285.5738\\
15&-334.6021&-334.30313&-334.55624&-334.7642\\
16&-387.77553&-387.48107&-387.73017&-387.9608\\
17&-444.95002&-444.6598&-444.9051&-445.1622\\
18&-506.12426&-505.83808&-506.07977&-506.3673\\
19&-571.29734&-571.01502&-571.25329&-571.5754\\
20&-640.46864&-640.18995&-640.42496&-640.7891\\
21&-713.63756&-713.36231&-713.59424&-713.9988\\
22&-790.80365&-790.53161&-790.76063&-791.2132\\
23&-871.96637&-871.69743&-871.9237&-872.4291\\
24&-957.12544&-956.85936&-957.08305&-957.6463\\
25&-1046.2804&-1046.0171&-1046.2383&-1046.8646\\
26&-1139.431&-1139.1702&-1139.3892&-1140.0838\\
27&-1236.5769&-1236.3184&-1236.5353&-1237.3039\\
28&-1337.7178&-1337.4616&-1337.6764&-1338.5247\\
\hline
\end{tabular}
\end{table}
\begin{table}
\centering
\caption{\label{gse_Ne_QZ}The ground-state energies of the Ne series in the cc-pVQZ basis set using different methods.}
\begin{tabular}{|c|cccc|}
\hline
Z&v2DM&HF&CCSD&expt.\\
\hline
10&-128.92686&-128.54347&-128.87106&-128.9376\\
11&-162.05038&-161.67155&-161.99595&-162.0659\\
12&-199.19913&-198.82303&-199.14502&-199.2204\\
13&-240.36392&-239.98965&-240.30977&-240.3914\\
14&-285.53886&-285.16605&-285.48453&-285.5738\\
15&-334.72067&-334.3492&-334.66615&-334.7642\\
16&-387.90757&-387.53731&-387.85281&-387.9608\\
17&-445.09826&-444.72921&-445.04331&-445.1622\\
18&-506.29194&-505.92402&-506.23681&-506.3673\\
19&-571.48788&-571.12105&-571.43261&-571.5754\\
20&-640.68561&-640.31973&-640.63013&-640.7891\\
21&-713.88444&-713.51958&-713.8289&-713.9988\\
22&-791.08414&-790.72018&-791.02848&-791.2132\\
23&-872.28418&-871.92117&-872.22853&-872.4291\\
24&-957.48456&-957.12224&-957.42872&-957.6463\\
25&-1046.6846&-1046.3231&-1046.6288&-1046.8646\\
26&-1139.8845&-1139.5236&-1139.8285&-1140.0838\\
27&-1237.0837&-1236.7235&-1237.0277&-1237.3039\\
28&-1338.2821&-1337.9226&-1338.2261&-1338.5247\\
\hline
\end{tabular}
\end{table}
\begin{table}
\centering
\caption{\label{gse_Si_DZ}Ground-state energies of the Si series in the cc-pVDZ basis sets using different methods. The results under v2DM were calculated using the ensemble avaraged spin projection, those under v2DM$^*$ were calculated using the maximal weight method.}
\begin{tabular}{|c|ccccc|}
\hline
&v2DM&v2DM$^*$&HF&CCSD&expt.\\
\hline
14&-288.93962&-288.92921&-288.84644&-288.91895&-289.359\\
15&-340.36765&&-340.27338&-340.34709&-340.872\\
16&-396.10801&&-396.01679&-396.08749&-396.869\\
17&-456.09635&&-456.00926&-456.0759&-457.337\\
18&-520.29362&-520.27860&-520.21067&-520.27348&-522.269\\
19&-588.68067&&-588.60149&-588.66097&-591.66\\
20&-661.24791&&-661.17202&-661.22871&-665.507\\
21&-737.99017&&-737.91714&-737.97148&-743.808\\
22&-818.90446&-818.88808&-818.83388&-818.88621&-826.559\\
23&-903.98889&&-903.92042&-903.97105&-913.762\\
24&-993.24222&&-993.1756&-993.22475&--1005.413\\
25&-1086.6636&&-1086.5986&-1086.6465&-1101.513\\
26&-1184.2525&-1184.2381&-1184.1889&-1184.2357&-1202.061\\
27&-1286.0084&&-1285.9461&-1285.9919&-1307.057\\
28&-1391.9311&&-1391.8699&-1391.9147&-1416.5\\
\hline
\end{tabular}
\end{table}

\begin{table}
\centering
\caption{\label{gse_Si_TZ} The ground-state energies of the Si series in the cc-pVTZ basis sets using different methods. The results under v2DM were calculated using the ensemble avaraged spin projection.}
\begin{tabular}{|c|cccc|}
\hline
&v2DM&HF&CCSD&expt.\\
\hline
14&-289.02515&-288.85215&-288.9835&-289.359\\
15&-340.50472&-340.33467&-340.46205&-340.872\\
16&-396.43974&-396.27384&-396.39711&-396.869\\
17&-456.80372&-456.64236&-456.7617&-457.337\\
18&-521.58&-521.42294&-521.53858&-522.269\\
19&-590.75683&-590.60373&-590.7159&-591.66\\
20&-664.32396&-664.17433&-664.28328&-665.507\\
21&-742.2714&-742.12467&-742.23072&-743.808\\
22&-824.58949&-824.44532&-824.54882&-826.559\\
23&-911.27007&-911.12778&-911.22912&-913.762\\
24&-1002.3054&-1002.1648&-1002.2643&-1005.413\\
25&-1097.6895&-1097.5501&-1097.6482&-1101.513\\
26&-1197.4173&-1197.2789&-1197.3759&-1202.061\\
27&-1301.4847&-1301.347&-1301.4431&-1307.057\\
28&-1409.8882&-1409.7512&-1409.8466&-1416.5\\
\hline
\end{tabular}
\end{table}
\begin{figure}
\centering
\includegraphics[scale=0.5]{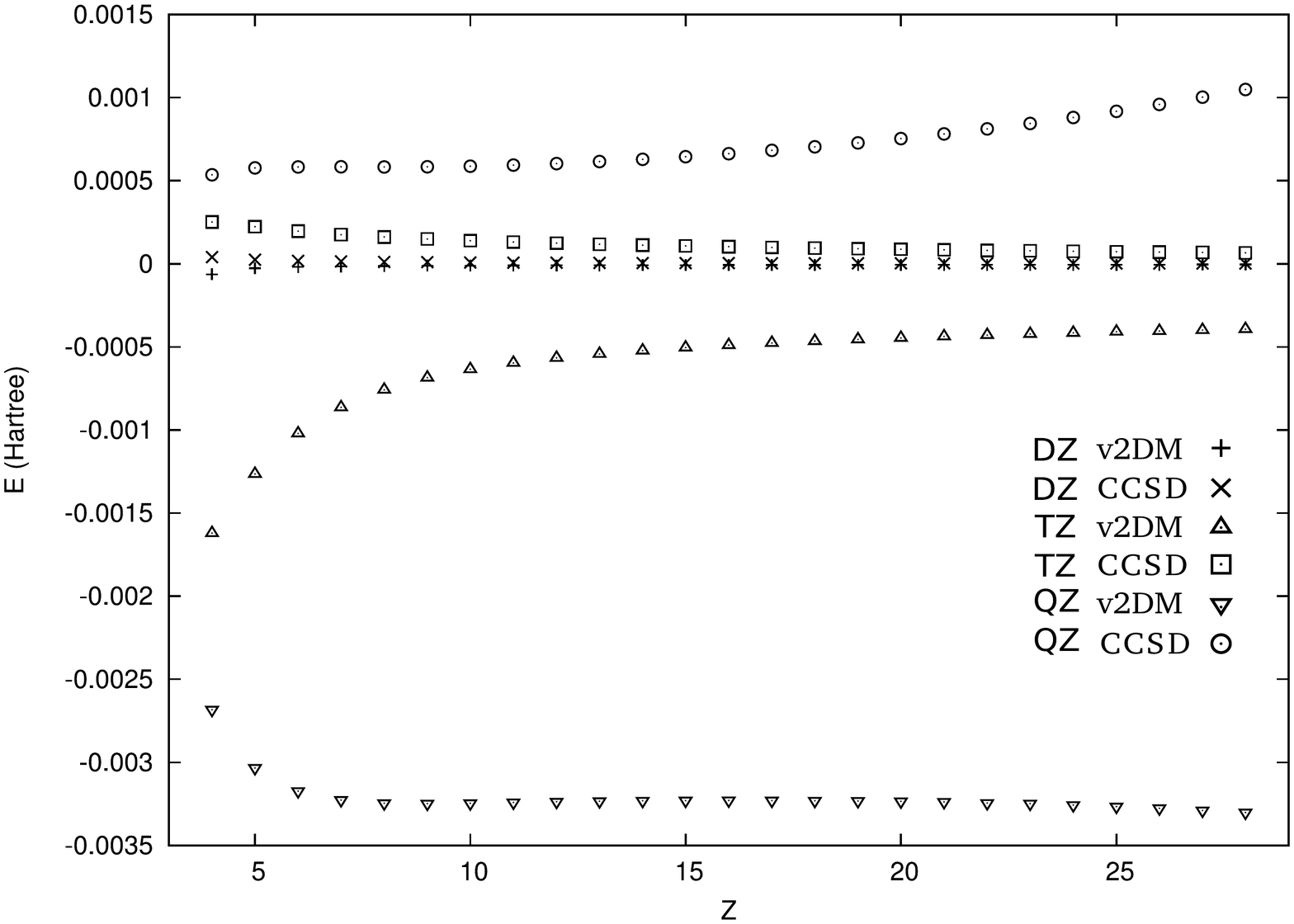}
\caption{\label{beryl_diff_fci} Difference between approximate (CCSD or v2DM) and full-CI energies for the Be series in all three basis sets.}
\end{figure}
More relevant is the difference between the v2DM (and CCSD) energies as compared to full-CI in the same basis set. This is shown in Figure~\ref{beryl_diff_fci} for the case of the Be series. Note that the CCSD energy is always above, the v2DM energy below, the full-CI energy. For v2DM, this simply reflects the nature of the variational problem. For the smallest cc-pVDZ basis set, v2DM and CCSD have about the same level of accuracy. The difference with full-CI grows as the basis set size increases for both CCSD and v2DM, but this effect is worse for the v2DM. 

As far as the $Z$-dependence is concerned, the trend differs markedly for the cc-pV(D,T)Z and for the cc-pVQZ basis set. As $Z$ increases there is a growing accuracy for the smaller basis sets in both CCSD and v2DM, whereas for cc-pVQZ the accuracy decreases for CCSD and becomes constant for v2DM. The reason for this difference is not clear, though it is probably connected to the incipient degeneracy of the $2s$ and $2p$ states and the quality of its description in the various basis sets, as is more fully explained in the next Section. It should be noted that the v2DM results are overall very accurate, even in the worst case ($Z=28$, cc-pVQZ) differing less than 3 mHartree from full-CI.     

For the Ne series, full-CI calculations were only possible in the cc-pVDZ basis. From the results collected in Tables~\ref{gse_Ne_DZ}, \ref{gse_Ne_TZ} and \ref{gse_Ne_QZ} it is seen that the v2DM accuracy is significantly less than for Be, the largest deviation to full-CI (28 mHartree) appearing for the neutral atom. This is actually comparable to v2DM results for molecules under $\mathcal{IQG}$ conditions, so it is likely that because of the small number of electrons the Be results are an exceptionally favorable case. This is also borne out by the Si results in Tables~\ref{gse_Si_DZ} and \ref{gse_Si_TZ}, showing a maximal deviation between CCSD and v2DM energies of 21 mHartree for the neutral atom.

\subsubsection{Correlation energy}
Correlation energies were calculated by taking the difference of v2DM or CCSD energies with the Hartree-Fock results in the same basis set. 
The results labeled ``experimental'' are the estimates in \cite{chakravorty}.
\paragraph{Beryllium series:}
in Fig.~\ref{beryl_corr_ener} the v2DM correlation energy is shown as a function of central charge $Z$ for the 
different basis sets. Note that on the plot the difference between the CCSD and full-CI correlation energies would not be visible. 
The experimental curve is linear in $Z$, as a direct consequence of the near-degeneracy of the ground state \cite{chakravorty}.
One can calculate a perturbative series expansion of the exact and Hartree-Fock energy in powers of $\frac{1}{Z}$; the corresponding series for the correlation 
energy starts with a constant if the hydrogenic ground state is nondegenerate, or with a linear term in $Z$ in case of degeneracy.   
The v2DM correlation energy does not follow this trend: it goes linear in the beginning, but becomes concave in the cc-pVDZ and cc-pVTZ basis, or convex in the cc-pVQZ basis. This failure, however, is not related to the v2DM method as the trend is the same in full-CI. It simply reflects the fact that the incipient degeneracy is not well described in these basis sets. This can also be seen by calculating the $Z=1$ hydrogen spectrum (corresponding to the 
$Z\rightarrow \infty$ situation, when the electron-electron interaction can be neglected) in the basis sets: the $2s$ and $2p$ energies are not degenerate, but differ by 5.8 mHartree (cc-pVDZ), 2.0 mHartree (cc-pVTZ) and -2.3 mHartree (cc-pVQZ). Note that for cc-pVQZ the $2p$ energy actually drops below the $2s$ energy, explaining the different (convex/concave) behavior of the 
curves. To make sure we also performed calculations in the cc-pVDZ basis after rescaling  ($r\rightarrow\alpha r$) it in such a way that the hydrogenic 
$2s$-$2p$ degeneracy is exact. In this basis the v2DM correlation energy (also shown in Fig.~\ref{beryl_corr_ener}) indeed has the correct linear behavior. It is clear from the above discussion that v2DM is indeed capable of providing accurate correlation energies in the presence of near-degeneracies, when other many-body techniques (like density functional theory or MP2) can fail.
\begin{figure}
\centering
\includegraphics[scale=0.6]{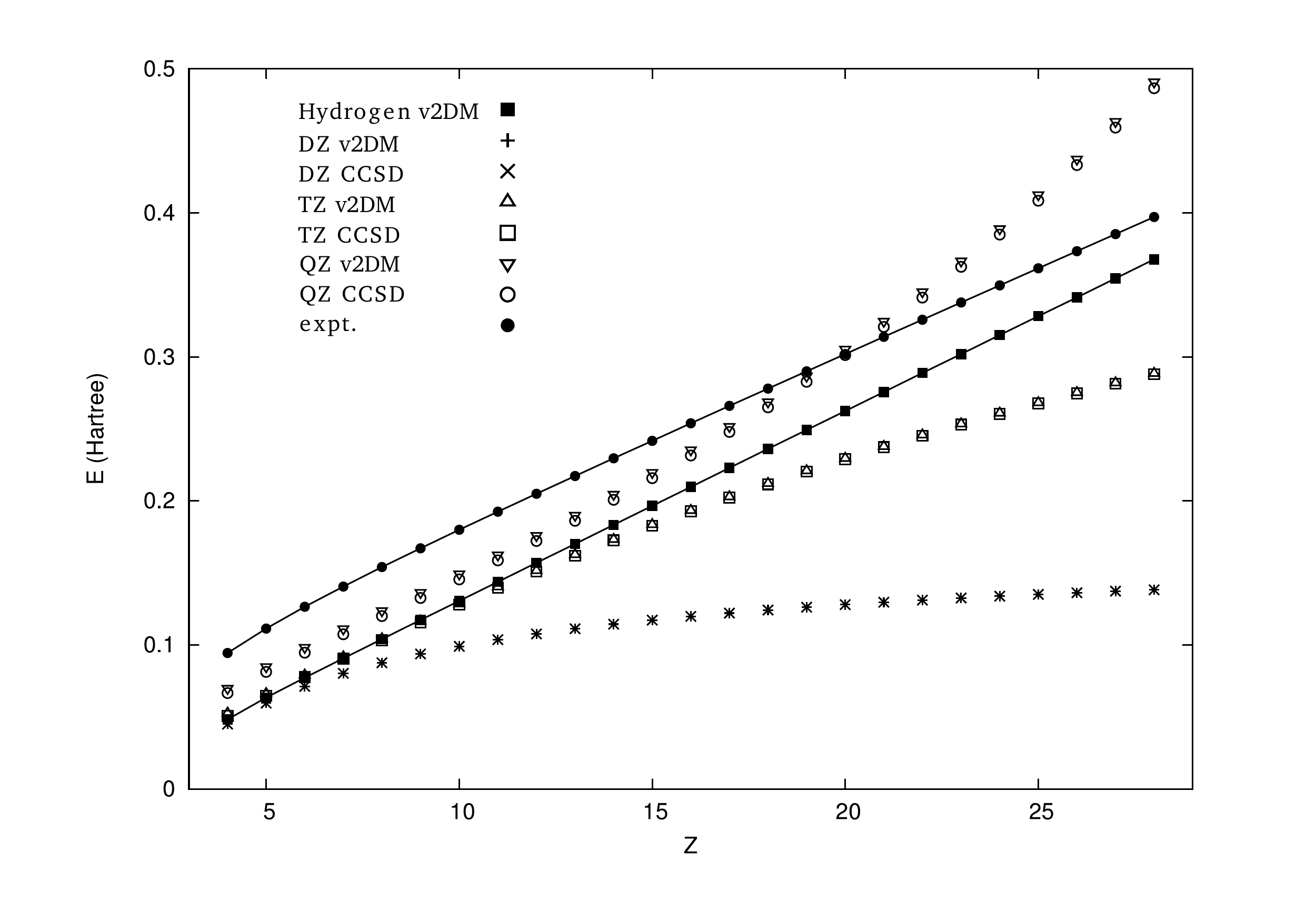}
\caption{\label{beryl_corr_ener}v2DM correlation energy for the Be series in all three basis sets, and in a rescaled basis set that exhibits hydrogen-like behaviour (degeneracy between the $2s$ and $2p$ level). For comparison, the CCSD and experimental values are also shown. Note that for the cc-pVDZ basis, CCSD and v2DM results coincide.}
\end{figure}
\paragraph{Neon series:}
\begin{figure}
\centering
\includegraphics[scale=0.6]{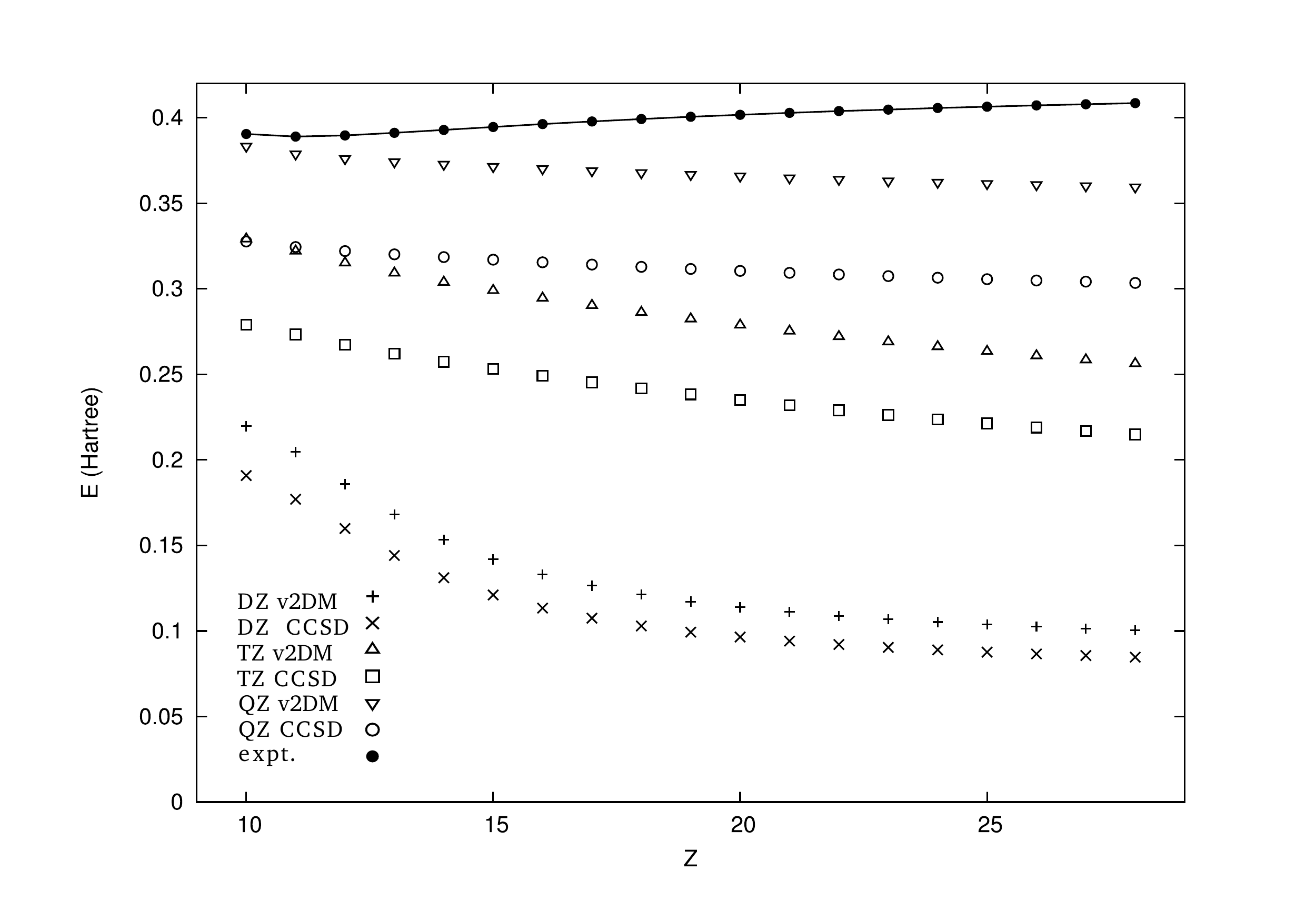}
\caption{\label{neon_corr_ener}v2DM correlation energy for the Ne series in all three basis sets. For comparison, the CCSD and experimental values are also shown.}
\end{figure}
in Fig. \ref{neon_corr_ener} the correlation energy is shown for all three basis sets as a  function of $Z$. Because Ne is a closed shell atom, there is no 
near-degeneracy for large $Z$ values and the exact correlation should be asymptotically constant in $Z$, as is indeed visible in the experimental curve. 
Due to basis set effects, this constant behavior is imperfectly realized, but the v2DM follows the same trends as CCSD for all basis sets.  
Note that the approximation to a constant behavior at large $Z$ is best for the largest basis set. The decrease in correlation energy for increasing $Z$, 
in contrast to the slight rise in the experimental correlation energy, can be attributed   
to the fact that the basis sets were optimized for the neutral atom. While the rescaling procedure fixes the nuclear cusp, the resulting basis set is 
obviously far from optimal for highly charged ions.
\paragraph{Silicon series:}
\begin{figure}
\centering
\includegraphics[scale=0.6]{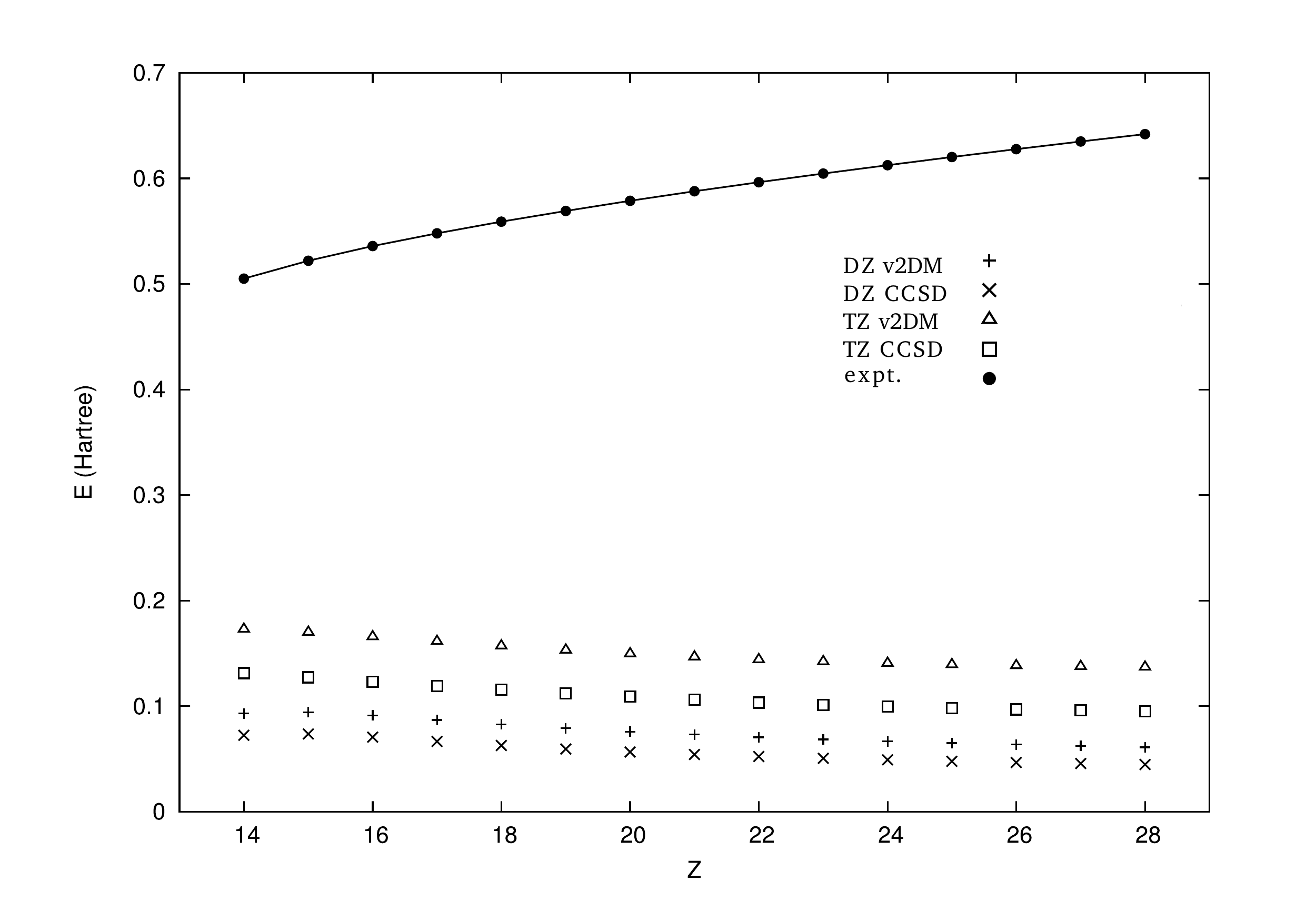}
\caption{\label{Silicon_corr_ener}v2DM correlation energy for the Si series in the cc-pVDZ and cc-pVTZ basis set. For comparison, the CCSD and experimental values are also shown.}
\end{figure}
for silicon, only the cc-pVDZ and cc-pVTZ basis have been used [Fig. \ref{Silicon_corr_ener}]. As was the case for Be, the theoretical linear rise with $Z$ is thwarted by 
imperfections in the basis sets. However the v2DM correlation energy closely tracks the CCSD one. The Si ground state is a spin triplet. The results in Tables~\ref{gse_Si_DZ} and \ref{gse_Si_TZ} have been obtained using the spin-averaged ensemble, as explained in Section~\ref{spin_averaged_ensemble}. In order to assess the quality of the spin constraints, we have also performed calculations using the highest-weight method, for $Z$=14, 18, 22, and 26 with the cc-pVDZ basis set, the resulting energies are also reported in Table~\ref{gse_Si_DZ}. The energy differences between the approaches are sizeable, with differences as large as 20 mHartree, reflecting the weaker nature of the spin constraints imposed in the spin-averaged scheme. However, the discrepancy between the two approaches is stable for increasing $Z$. 
\subsubsection{Ionization energies}
\begin{figure}
\centering
\includegraphics[scale=0.5]{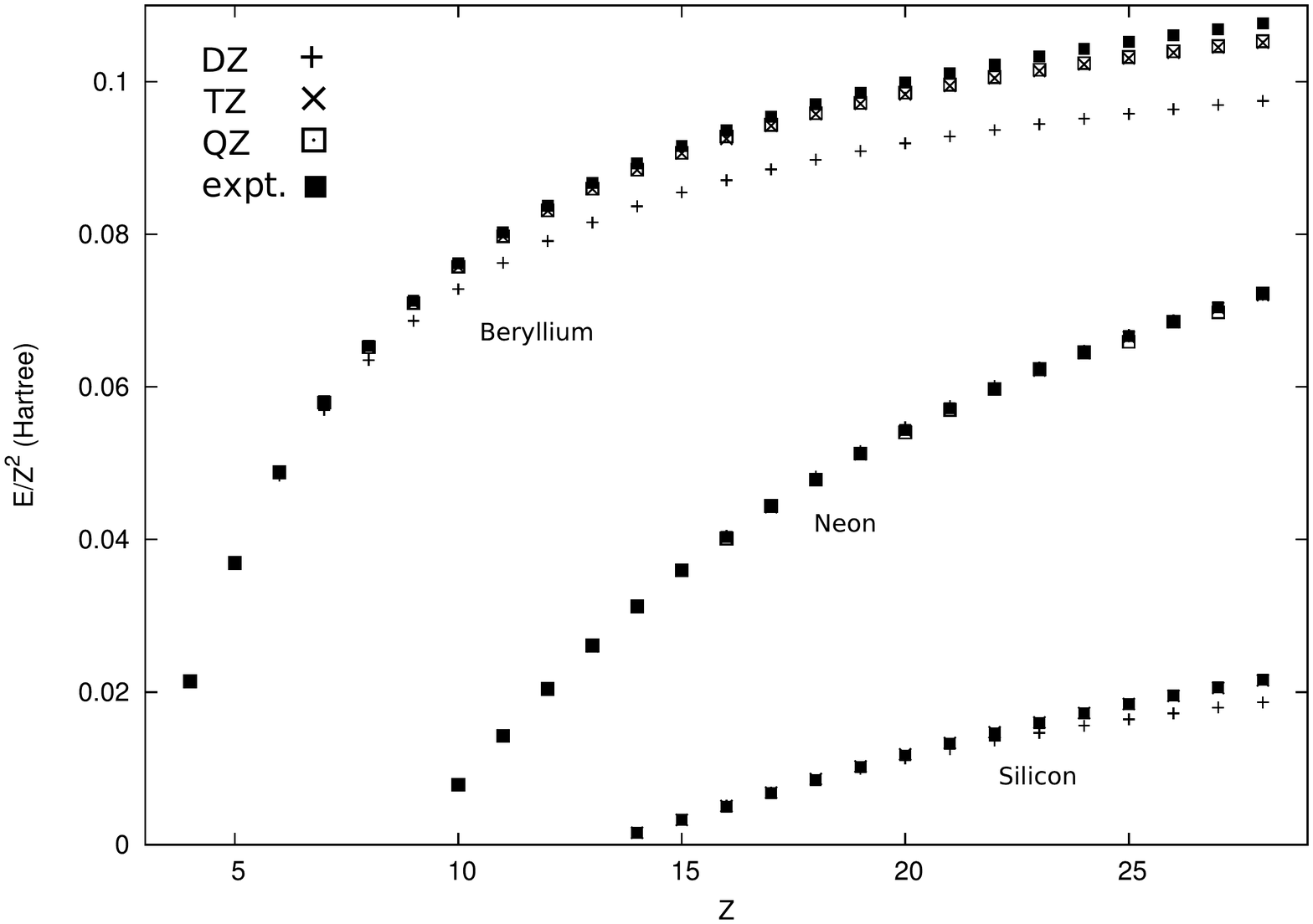}
\caption{\label{ion_E}Ionization energy scaled with $\frac{1}{Z^2}$, for the Be, Ne and Si series in the different basis sets compared with experimental results.}
\end{figure}
It is important to gauge the quality of the 2DM by analyzing other properties than just the energy, {\it e.g.} the  ionization energies of the different atomic ions, which can be easily calculated using the extended Koopmans' theorem (EKT) \cite{smith,day,morrell}. The EKT provides a single-particle picture of the ground state, with single-particle energies and spectroscopic factors. The ionization energies are shown in Figure~\ref{ion_E}; the agreement between calculated and  
experimental values is very good, pointing to the realistic nature of the variationally obtained 2DM. 
The good agreement with experiment reflects the fact that the error in the description of the interelectronic cusp largely cancels since the ionization energy is 
an energy difference. For Be and Ne it is clear that the basis set limit is nearly reached at the cc-pVTZ -- cc-pVQZ level. Even for Si the experimental ionization energy 
is closely reproduced.
\subsubsection{Correlated Hartree-Fock-like single-particle energies}
\begin{figure}
\centering
\includegraphics[scale=0.5]{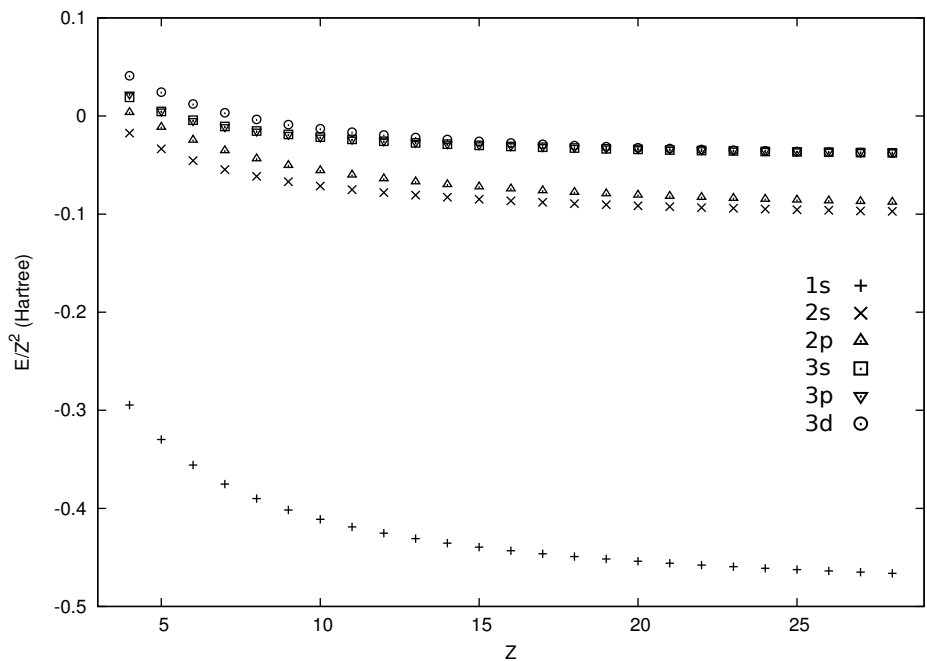}
\caption{\label{sp_levels}Single-particle levels obtained in a correlated Hartree-Fock-like scheme (see text) for the Be series in a cc-pVDZ basis set.}
\end{figure}
Another single-particle picture is given by the correlated Hartree-Fock-like single-particle orbitals and energies. These are constructed by diagonalizing the single-particle Hamiltonian:
\begin{equation}
h_{\alpha\gamma} = \left(T + U\right)_{\alpha\gamma} + \sum_{\beta\delta}V_{\alpha\beta;\gamma\delta}\rho_{\beta\delta}~,
\label{hartree_fock_like}
\end{equation}
where the 1DM $\rho$ appearing in Eq.~(\ref{hartree_fock_like}) is constructed from the variationally determined 2DM. As an example of this method, the single-particle energies for the isoelectronic series of Be in a cc-pVDZ basis are shown in Figure~\ref{sp_levels}. Notice that when $Z$ increases, the energy levels approach those of the hydrogen atom. Similar behavior is present for the other basis sets and for the Ne and Si isoelectronic series. 
\subsubsection{Natural occupations}
\begin{figure}
\centering
\includegraphics[scale=0.5]{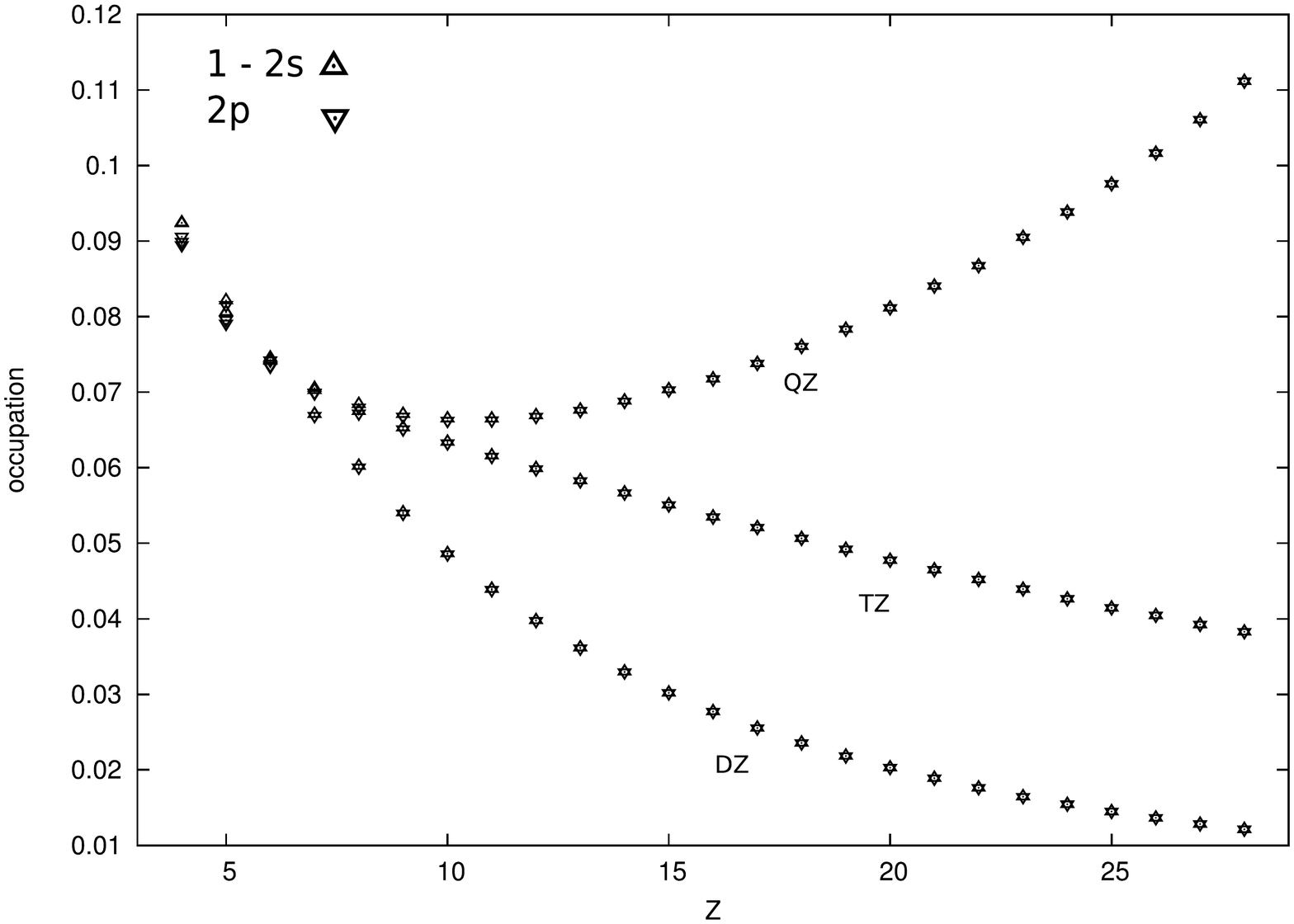}
\caption{\label{nocc_2s2p}Natural occupation of the $2p$-orbital and one minus the occupation of $2s$-orbital for the Be series, in all three basis sets.}
\end{figure}
The eigenvalues of the 1DM ({\it i.e.} the natural orbital occupation numbers) provide additional insight into the quality of the electron correlation effects included in v2DM. The occupation numbers from v2DM are always very close to those from full-CI, differing by at most 0.005. Of particular interest are the occupations of the quasi-degenerate $2s$ and $2p$ orbitals in Be. These are shown in Fig.~\ref{nocc_2s2p}. The sum of the $2s$ and $2p$ occupations is nearly 1 and increasingly so for large $Z$. This implies that only the $2s$ and $2p$ are partially occupied in the large-$Z$ limit. The shapes of the curves reflect the aforementioned imperfections in the basis sets, with the $2s$ below the $2p$ for cc-pVDZ and cc-pVTZ, and above the $2p$ for cc-pVQZ.

\section{\label{diatomic}Dissociation of diatomic molecules}
A diatomic molecule, consisting of $N$ electrons spread out over two atoms, $A$ with charge $Z_A$ on $\mathbf{R}_A$ and $B$ with charge $Z_B$ on $\mathbf{R}_B$, is described by the Hamiltonian:
\begin{equation}
\hat{H} = -\sum_{i = 1}^N\left(\frac{1}{2}\nabla_i^2 + \frac{Z_A}{|\mathbf{R}_A - \mathbf{r}_i|} + \frac{Z_B}{|\mathbf{R}_B - \mathbf{r}_i|} \right) + \sum_{i<j}^N\frac{1}{|\mathbf{r_i} - \mathbf{r_j}|}~.
\end{equation}
One speaks of dissociation, or bond stretching, when the distance between the two bonded atoms $|\mathbf{R}_A-\mathbf{R}_B|$ is increased from its equilibrium value, and the dissociation limit corresponds to the situation where there is no interaction left between the atoms. Dissociation of molecules is very important in chemistry, because a good understanding of it is essential to understand why chemical reactions (which involve the breaking of bonds and the formation of new bonds) take place. Many electronic structure methods have great difficulty describing the dissociation limit, however, because when a molecule is pulled out of equilibrium, multireference effects occur.
\subsection{v2DM dissociates molecules into fractionally charged atoms}
\begin{figure}
\centering
\includegraphics[scale=0.7]{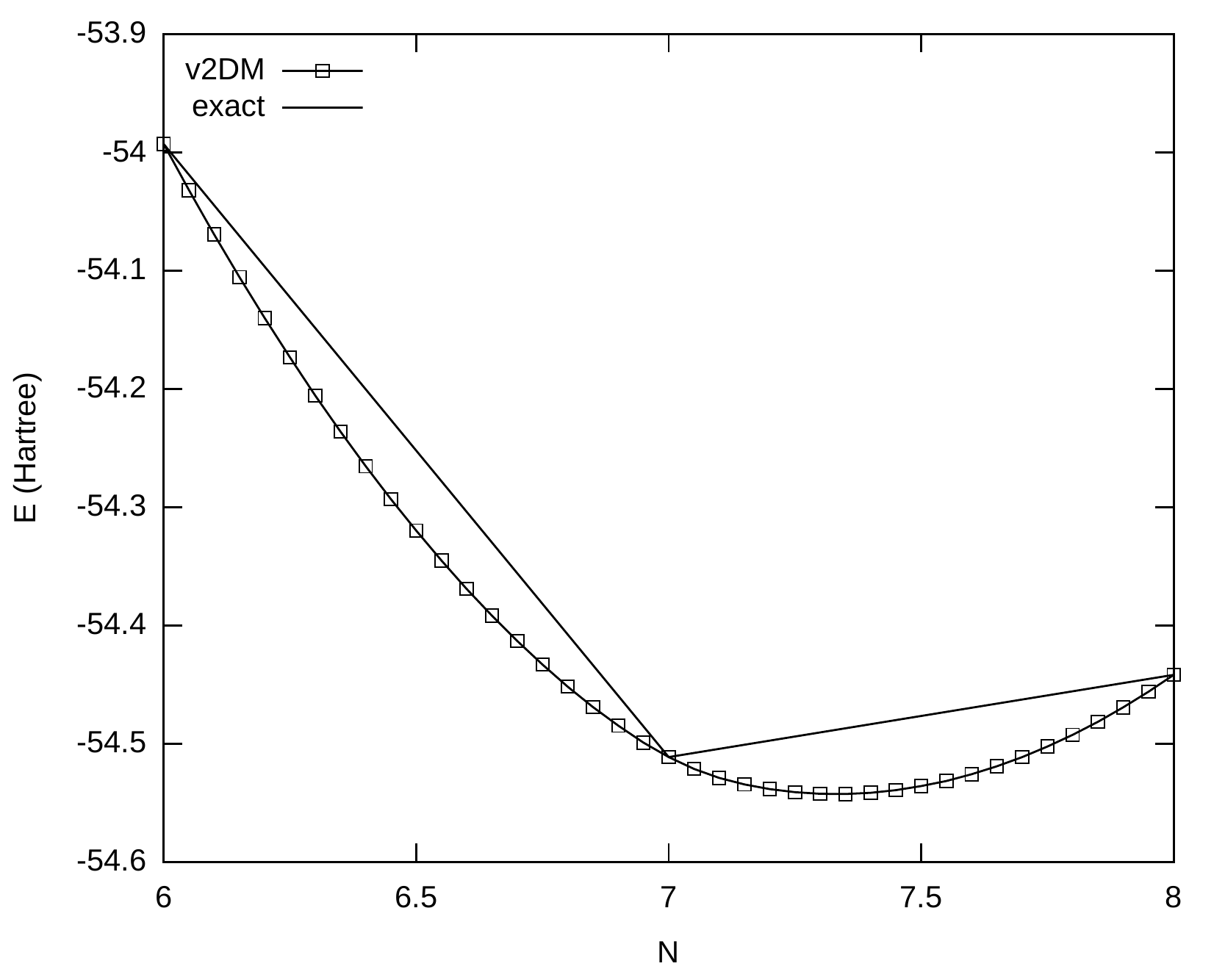}
\caption{\label{N_frac}Energy versus electron number for the Nitrogen atom. One can see that the v2DM method results in a smooth convex curve, while the exact result is piecewise linear.}
\end{figure}
To test how well the $\mathcal{IQG}$ conditions describe molecular dissociation the 14-electron $\mathrm{NO^+}$ molecule was studied. From the ionization energies one can predict that the dissociation products should be $\mathrm{N}$ and $\mathrm{O^+}$. Rather surprisingly the v2DM result was to divide the charge between the $\mathrm{N}$ and the $\mathrm{O}$, assigning 6.53 electrons to $\mathrm{N}$ and 7.47 electrons to $\mathrm{O}$. We also found that the energy was far too low. The same failure occurs in DFT, where the problem results from the wrong behaviour of the energy as a function of the number of electrons \cite{cohen,frac_N_1,frac_N_2,mori,perdew,ruzs}. It turns out that in v2DM the same error lies at the basis of the faulty dissociation limit. As can be seen in Fig.~\ref{N_frac} for $\mathrm{N}$ and Fig.~\ref{O_frac} for $\mathrm{O}$, the $E$ vs. $N$ curve is convex, while the exact curve is piecewise linear between the integer occupations. This can be understood from the discussion in Section~\ref{subsystem}, where it was shown that the correct 2DM description of a system with fractional electron number is through an ensemble of integer-$N$ 2DM's. The results in Figs.~\ref{N_frac} and \ref{O_frac}, however, were produced in a naive approach by considering $N$ as a continuous parameter in the integer $N$ v2DM procedure as discussed in Chapter~\ref{SDP}. This naive approach is relevant, however, because it is exactly what happens when considering the individual atoms as subsystems of an integer $N$ molecule, using approximate $N$-representability conditions, as explained in Section~\ref{applicability_of_ssc}.
\begin{figure}
\centering
\includegraphics[scale=0.7]{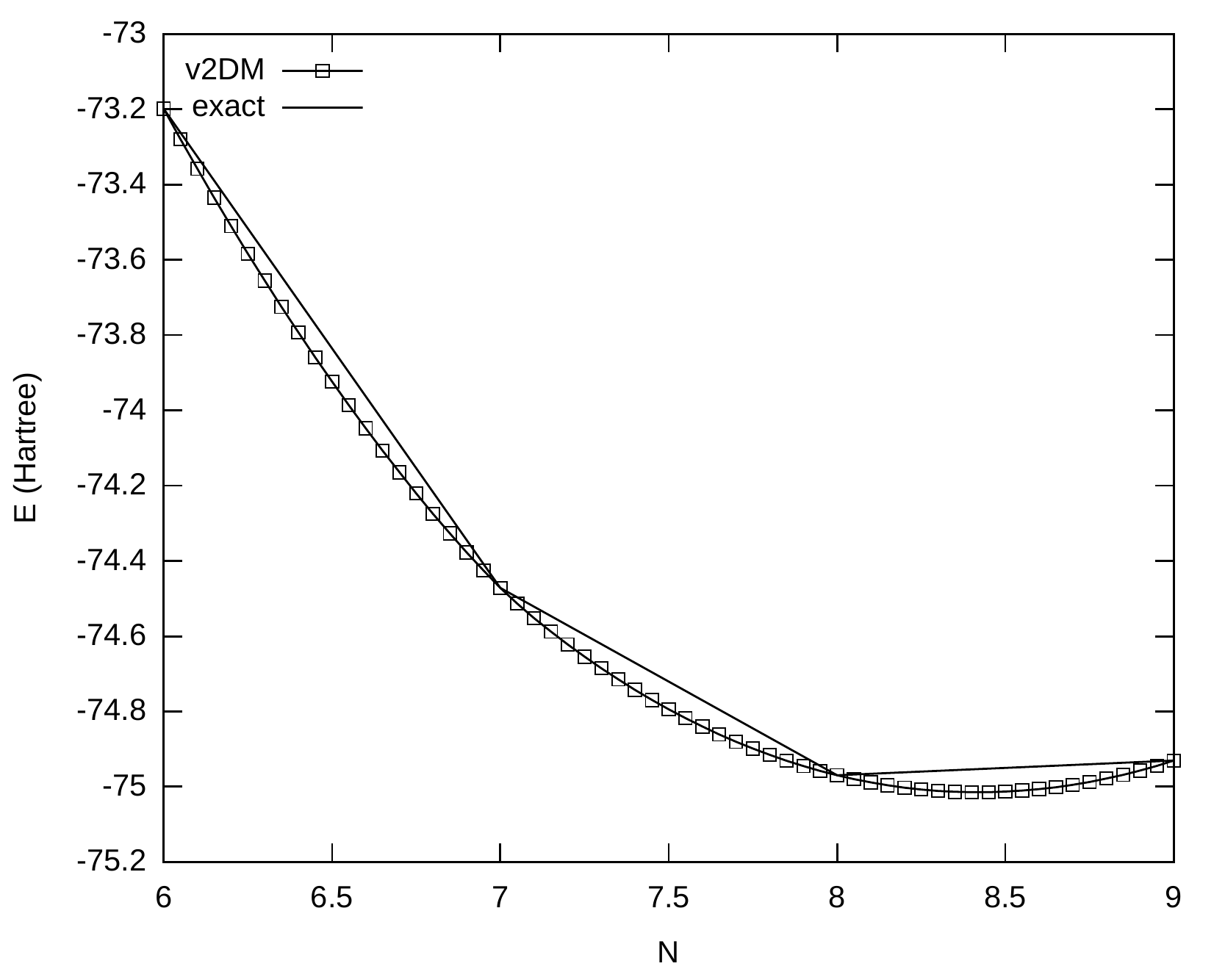}
\caption{\label{O_frac}Energy versus electron number for the Oxygen atom (See caption Figure~\ref{N_frac}).}
\end{figure}
\begin{figure}
\centering
\includegraphics[scale=0.7]{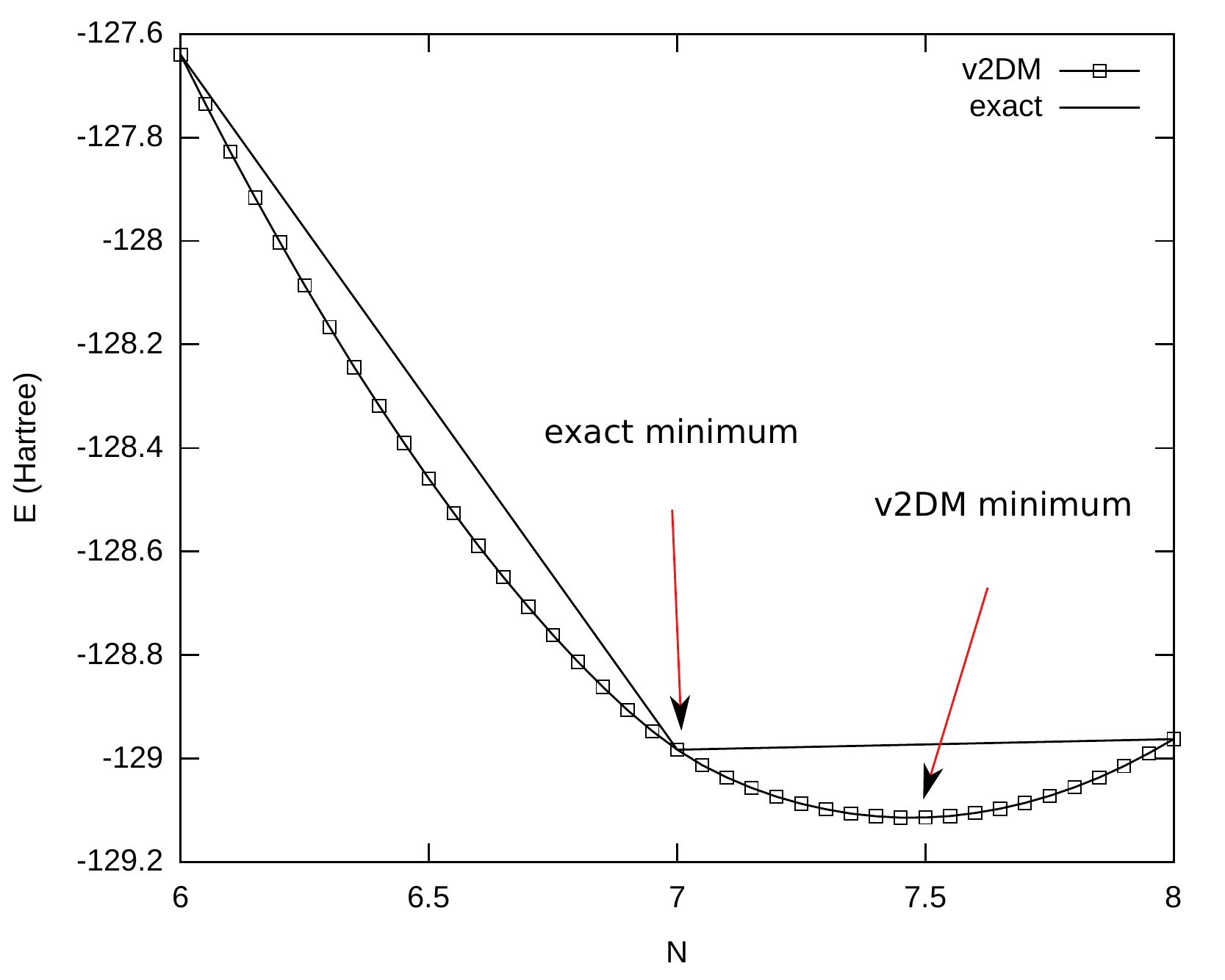}
\caption{\label{sum_frac}Energy versus electron number for the sum of the Nitrogen and Oxygen energies. $N$ is the number of electrons on the Oxygen atom. The v2DM method clearly has a minimum at fractional $N$.}
\end{figure}
From Fig.~\ref{sum_frac} one can see how this convex curve results in the fractionally occupied dissociation products. In this figure, the sum of the energies of $\mathrm{N}$ and $\mathrm{O}$ are plotted, for fractional occupations which sum to 14 electrons. One can see that the lowest energy is obtained from the fractional occupation of Oxygen of 7.47, which is exactly the result we get from the molecular calculation in the dissociation limit! It is worth noting that the same behaviour is present when adding the $\mathcal{T}$ conditions, and as such adding them does not fix the errors in the dissociation limit.
\subsection{Subsystem constraints cure the dissociation limit}
In this Section we show that the result of Section~\ref{subsystem} can be used to fix the dissociation problem in v2DM. From Theorem~\ref{theo_sub_constr} we know that if a 2DM describing a molecule is $N$-representable, then the 2DM's of the the atomic parts of the molecules are fractional $N$-representable, and must obey the inequalities:
\begin{equation}
\mathrm{Tr}~\rho^\mathcal{V} t^\mathcal{V} + \mathrm{Tr}~\Gamma^\mathcal{V} V^\mathcal{V} \geq E^{\bar{N}}_0\left(\hat{H}^\mathcal{V}\right)~,
\label{subsystem_eq_at}
\end{equation}
for every two-particle Hamiltonian $\hat{H}^\mathcal{V}$ defined on the atomic subspace $\mathcal{V}$. From Figs.~\ref{N_frac} and \ref{O_frac} we can see that these constraints are grossly violated when we choose the atomic Hamiltonian for the subsystem Hamiltonian $\hat{H}^{\mathcal{V}}$, where the exact fractional-$N$ energy $E_0^{\bar{N}}$ is the piecewise linear curve. If the inequalities (\ref{subsystem_eq_at}) for the atomic subsystems in the molecule are enforced, the $E$ vs. $N$ curve automatically becomes piecewise linear, and the dissociation of the molecule results in integer occupied fragments.
\subsubsection{Implementation}
In this Section we discuss how to implement the subsystem constraints in the v2DM approach. The procedure for applying the atom-$A$ subsystem constraint in a diatomic $AB$ can be summarized as follows: 
\begin{enumerate}
\item Solve the atomic v2DM problem for a central charge $Z_A$ at various electron numbers, using the single-particle orbitals centered on $A$.
In practice, only electron numbers near atomic neutrality are important. This generates the atomic energy $E^{\bar{N}}_A$ as a function of fractional 
electron number $\bar{N}$ (i.e. a piecewise linear curve as in Figs. \ref{N_frac} and \ref{O_frac}). 
\item For each internuclear distance $R_{AB}$, calculate the transformation matrix between the (nonorthogonal) atomic basis of the $A$-centered orbitals $\ket{i_A}$, and the orthonormal molecular basis $\ket{\alpha}$ that is used in the SDP program,
\begin{equation}
X^A_{\alpha ,i} = \langle i_A|\alpha\rangle~.
\end{equation}  
The $X^A$ matrix is easily constructed with standard quantities in molecular modelling packages, 
\begin{equation}
X^A_{i,\alpha} = \sum_{j_D} C_{\alpha ,j_D} S_{i_A;j_D}~,
\end{equation}
with $C$ the expansion coefficients of the $|\alpha\rangle$ molecular basis in terms of all the nonorthogonal orbitals centered on the various atoms,
\begin{equation}
|\alpha\rangle = \sum_{j_D} C_{\alpha ; j_D}|j_D\rangle~,
\end{equation}
and with $S_{i_A ;j_D} = \langle i_A|j_D \rangle$ the overlap matrix for the atom-centered basis functions.
The orthogonal projector on the subspace spanned by the $A$-centered orbitals, when expressed in terms of the nonorthogonal basis set, reads \cite{bultinck}
\begin{equation}
\tilde{P}^A = \sum_{ij}(S^{-1}_A)_{ij}|i_A\rangle\langle j_A|~,
\end{equation}
where $S_A$ is the block of the overlap matrix corresponding to the $A$-centered orbitals, and $S_A^{-1}$ is the inverse of this block.
\item Perform the molecular v2DM calculation with the extra linear inequality: 
\begin{equation}
\label{26}
\mbox{Tr}~\left(t^A\rho^N\right)+\mbox{Tr}~\left(V^A\Gamma^N\right) \geq E_A^{\bar{N}=\mbox{Tr}\left(1^A\rho^N\right)}~.
\end{equation}   
Here: 
\begin{eqnarray}
(t^A)_{\alpha\gamma} &=& \sum_{ij}\langle i_A |\tilde{t}^A| j_A\rangle W^A_{i\alpha}W^A_{j\gamma}~,\\
(1^A)_{\alpha\gamma} &=& \sum_{ij} (S_A)_{ij}  W^A_{i\alpha}W^A_{j\gamma}~,\\
(V^A)_{\alpha\beta;\gamma\delta} &=& \sum_{ijkl} \langle i_A j_A|\tilde{V}^A|k_Al_A\rangle~,
W^A_{i\alpha}W^A_{j\beta}W^A_{k\gamma}W^A_{l\delta}~,
\end{eqnarray}
and $t_A$ is the kinetic energy plus attraction to nucleus $A$. The coefficients $W$ are given by:
\begin{equation}
\label{W}
W^A_{i\alpha} = \sum_j X^A_{\alpha j}\left(S_A^{-1}\right)_{ji}~.
\end{equation}
The inequality (\ref{26}) is nothing but the application of Eq.~(\ref{subsystem_eq_at}) in the subspace defined by the single-particle orbitals centered on $A$ and using the Hamiltonian of atom $A$. 
Obviously, atom $B$ generates a similar inequality.
\end{enumerate}
\subsubsection{Numerical verification}
In order to show the value of the subsystem constraints, we present the potential energy surface of $\mathrm{BeB}^+$, computed for a separation ranging from 1 to 9 \AA. The main interest here is a proof-of-principle of the fact that the new constraints indeed severely restrict the variational freedom in the v2DM. We therefore opted for the fairly small Dunning-Hay basis \cite{dh}, making a comparison to full-CI calculations still feasible. $\mathrm{BeB}^+$ is a good example since this 8 electron system dissociates into $\mathrm{Be}$ and $\mathrm{B}^+$. Application of the $\mathcal{I}$, $\mathcal{Q}$ and $\mathcal{G}$ conditions for an 8 electron system is expected to yield energies that are significantly too low compared to full-CI. However, at the dissociation limit, the energy of the molecule should be equivalent to that of isolated $\mathrm{Be}$ and $\mathrm{B}^+$. As shown in \cite{atomic}, for both $\mathrm{Be}$ and $\mathrm{B}^+$ the $\mathcal{IQG}$ energy is very nearly equal to the full-CI energy, so application of the subspace constraints should result in much higher $\mathcal{IQG}$ energies.

Table~\ref{BeB_t} contains the energy of the molecule at different internuclear distances, computed at the full-CI level of theory as well as the variationally optimized 2DM energy with and without subspace constraints (see also Fig. \ref{BeB_p}).
\begin{table}
\centering
\caption{\label{BeB_t}Difference (in mHartree) between full-CI energy and variationally optimized 2DM energy without (2DM) and with (2DM+) subspace constraints.}
\begin{tabular}{|ccccccc|}
\hline
$R$&2DM&2DM+&\vline&R&2DM&2DM+\\
\hline
1.25  &  16.55 &  16.55 &\vline&3.50  &  13.30 &  13.30 \\
1.50  &  10.86 &  10.86 &\vline&3.75  &  17.17 &  17.17 \\
1.75  &   9.56 &   9.56 &\vline&4.00  &  19.12 &  19.12 \\
2.00  &   9.94 &   9.94 &\vline&4.50  &  19.90 &  18.73 \\
2.10  &  10.07 &  10.07 &\vline&5.00  &  21.28 &  14.07 \\
2.20  &  10.07 &  10.07 &\vline&5.50  &  23.07 &   9.35 \\
2.30  &  10.30 &  10.30 &\vline&6.00  &  24.48 &   5.27 \\
2.40  &  10.59 &  10.59 &\vline&6.50  &  25.67 &   2.44 \\
2.50  &  10.77 &  10.77 &\vline&7.00  &  26.95 &   1.26 \\
2.60  &  10.91 &  10.91 &\vline&7.50  &  27.90 &   0.66 \\
2.75  &  11.30 &  11.30 &\vline&8.00  &  28.88 &   0.53 \\
3.00  &  12.00 &  12.00 &\vline&8.50  &  29.63 &   0.38 \\
3.25  &  12.56 &  12.56 &\vline&9.00  &  30.39 &   0.37\\
\hline
\end{tabular}
\end{table}

Table \ref{BeB_t} very clearly shows that, at larger separation, the difference between the full-CI and 2DM energies is substantial when using only the $\mathcal{I}$, $\mathcal{Q}$ and $\mathcal{G}$ constraints, growing as large as 30~mHartree. The subspace constraints succeed in reducing this error by approximately two orders of magnitude. As expected, the remaining error is very small because $\mathcal{IQG}$ yield energies for the atomic 4-electron isoelectronic series that are very near to full-CI energies. The present new constraints are clearly very succesful. As Figure \ref{BeB_p} shows, the constraints are active most for separations above 4.5~\AA. The nearer to complete dissociation, the more of the error is recovered by the subspace constraints. As shown by van Aggelen {\it et al.} \cite{helen_2}, not only are the energies improved, also chemical observables and chemical concepts are substantially better for the 2DM obtained when including the subspace constraints. As an example, the Mulliken population \cite{mulliken} on the Be atom at $9~\text{\AA}$ is $+0.38$ when not using the subspace constraints, whereas inclusion of the subspace constraints yields a charge of $0.00$, consistent with reality.
The added constraints result in a much better description of molecular dissociation. Neither atom still suffers from fractional occupancy at the dissociation limit. Addition of each subsystem constraint does not slow down the v2DM, as it adds a fairly simple linear inequality constraint.
\begin{figure}
\centering
\includegraphics[scale=0.5]{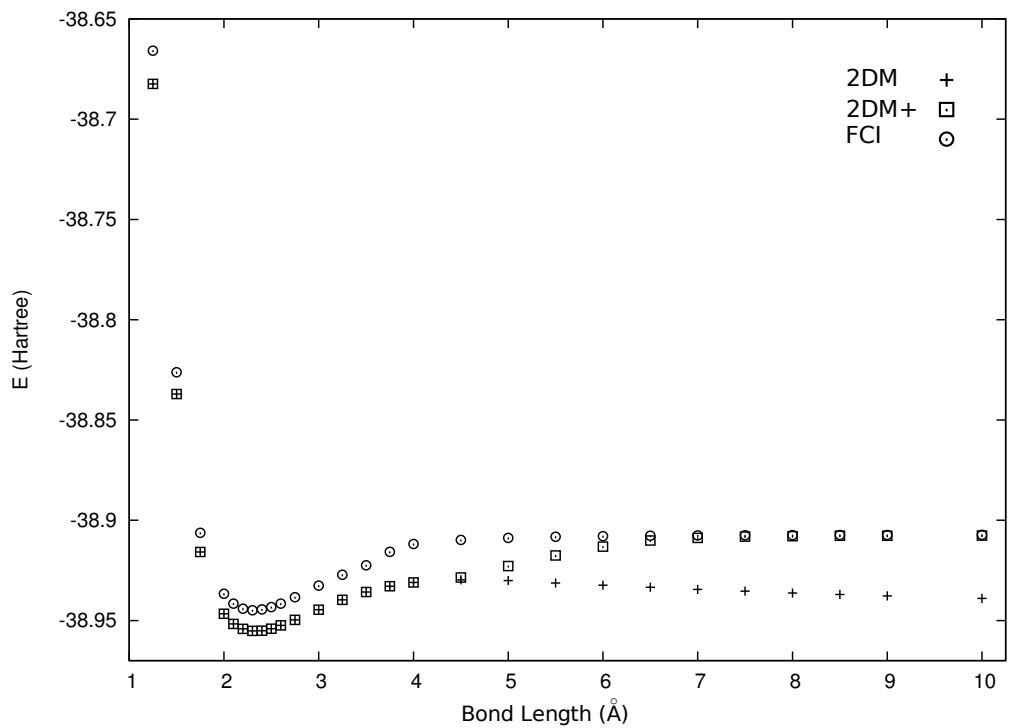}
\caption{\label{BeB_p}Singlet dissociation energy curve of $\text{BeB}^+$ calculated in full-CI, and determined variationally without (2DM) and with (2DM+) subspace constraints.}
\end{figure}
\subsubsection{Extension to polyatomic systems}
For polyatomic systems, the addition of correct subsystem constraints is much less straightforward than in the diatomic case. The subsystems are no longer only atomic systems, but can also be polyatomic systems, for which the constraints become dependent upon the geometry of the subsystems, {\it i.e.} the interatomic distance. The number of possible subspaces, grows rapidly with the number of atoms present in the molecule. It is unlikely that all of these constraints are needed and only a few will probably be violated when not included. A detailed study \cite{helen_3} was made by van Aggelen {\it et al.} of the performance of subsystem constraints in the triatomic molecular anion $\mathrm{F}_3^-$. It was found that for geometries with either clearly dissociated or short bonds, the correct dissociation can be obtained using only constraints on the spatially separated units in the system. {\it E.g.} when the three $\mathrm{F}$ atoms are so far apart that they don't interact anymore, subsystem constraints for the atomic parts are the only ones useful, or for the case where a diatomic molecule is clearly separated from a single atom, the correct constraints are the geometry-dependent diatomic one, and an atomic one. However, in situations which are not so clear-cut, additional constraints may become active.
\subsubsection{Generalization of the atomic subspace constraints}
The current formulation of the subsystem contraints on diatomic systems has the drawback that they only become active at large distances, and do not improve the energy closer to equilibrium. As a consequence there is a non-parallellity of the v2DM potential energy surface compared to the exact one, which is clearly visible in Fig.~\ref{BeB_p}. Ideas to improve on this exist but haven't been tried yet. One approach would be to modify the atomic Hamiltonian by including effects from its environment. This would result in an iterative optimization of the 2DM, where the 2DM from the previous molecular calculation is used to create the subspace Hamiltonian for the constraint in the next iteration. The environment would appear only in the single-particle part of the subspace Hamiltonian corresponding to atom A, by adding to the Hamiltonian of atom A the nuclear attraction to atom B as well as the Hartree-Fock potential generate by the current iterative approximation for $\Gamma$.
\section{\label{1dhub_sec}The one-dimensional Hubbard model}
The one-dimensional Hubbard model, as introduced in Section~\ref{hub_sym}, is described by the Hamiltonian:
\begin{equation}
\hat{H} = -t\sum_{i\sigma}\left(a^\dagger_{i;\sigma}a_{i+1;\sigma} + a^\dagger_{i+1;\sigma}a_{i;\sigma}\right) + U\sum_{i}a^\dagger_{i\uparrow}a_{i\uparrow}a^\dagger_{i\downarrow}a_{i\downarrow}~.
\label{hubbard_ham}
\end{equation}
This Hamiltonian is the simplest model describing the non-trivial correlations in a solid state lattice as a competition between the delocalizing hopping term and the local on-site interaction. In this Section we present and discuss the results of v2DM calculations, taking advantage of all the symmetries as described in Section~\ref{hub_sym}, on a 50-site lattice with the $\mathcal{IQG}$ conditions, and on a 20-site lattice with the $\mathcal{IQG}\mathcal{T}_1\mathcal{T}_2$ conditions. The Hubbard model has been studied before using the v2DM method, see {\it e.g.} \cite{maz_hub,shenvi,nakata_last}, but up to now only the half-filled lattice was considered. In this study we consider different filling factors, and extract various properties like the ground-state energy, two-particle correlation functions and the momentum distribution in order to assess the quality of the variationally obtained 2DM. 
\subsection{Results}
The v2DM results discussed in this Section were all obtained using the primal-dual predictor corrector semidefinite programming algorithm described in Section~\ref{primal_dual_algo}. Although the one-dimensional Hubbard model can be solved exactly using the Bethe ansatz \cite{bethe,liebwu,essler,1D_hub}, it is hard to extract information about the solution for finite systems. For the calculations on a 20-site lattice, we compare the data with the quasi-exact results obtained through a variational Matrix Product State (MPS) algorithm \cite{schollwock,verstraete,chan}, written by co-worker Sebastian Wouters \cite{sebastian}. For the 50-site lattice, however, this is no longer computationally feasible. At half filling a simplification in the Bethe-ansatz equations occurs, which allows to calculate the ground-state energy of finite systems by solving a set of non-linear equations (Lieb-Wu)\cite{liebwu_2}. At other fillings no data is available for comparison.
\subsubsection{Ground-state energy}
\begin{figure}
\centering
$
\begin{array}{c}
\includegraphics[scale=0.7]{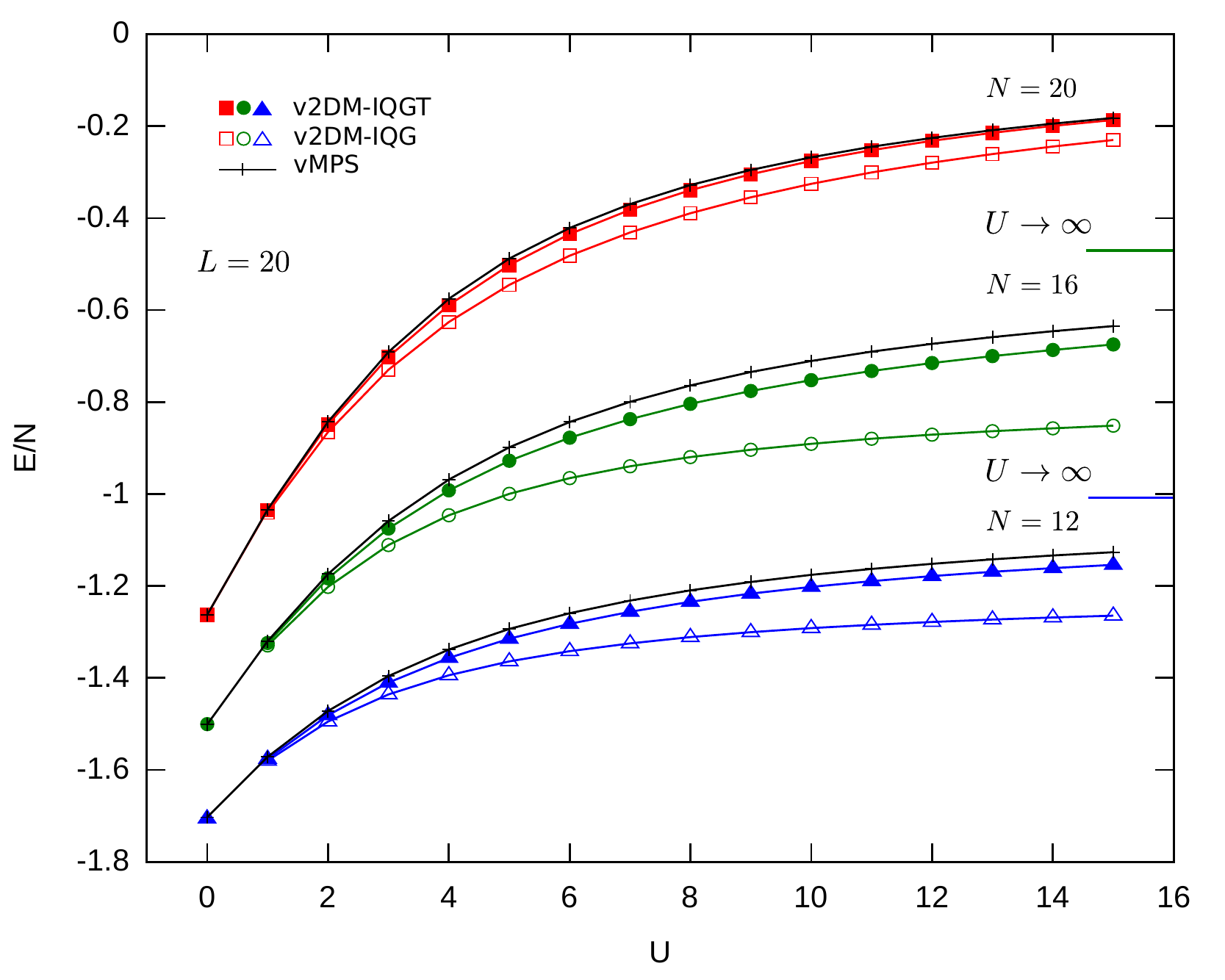}\\
\includegraphics[scale=0.7]{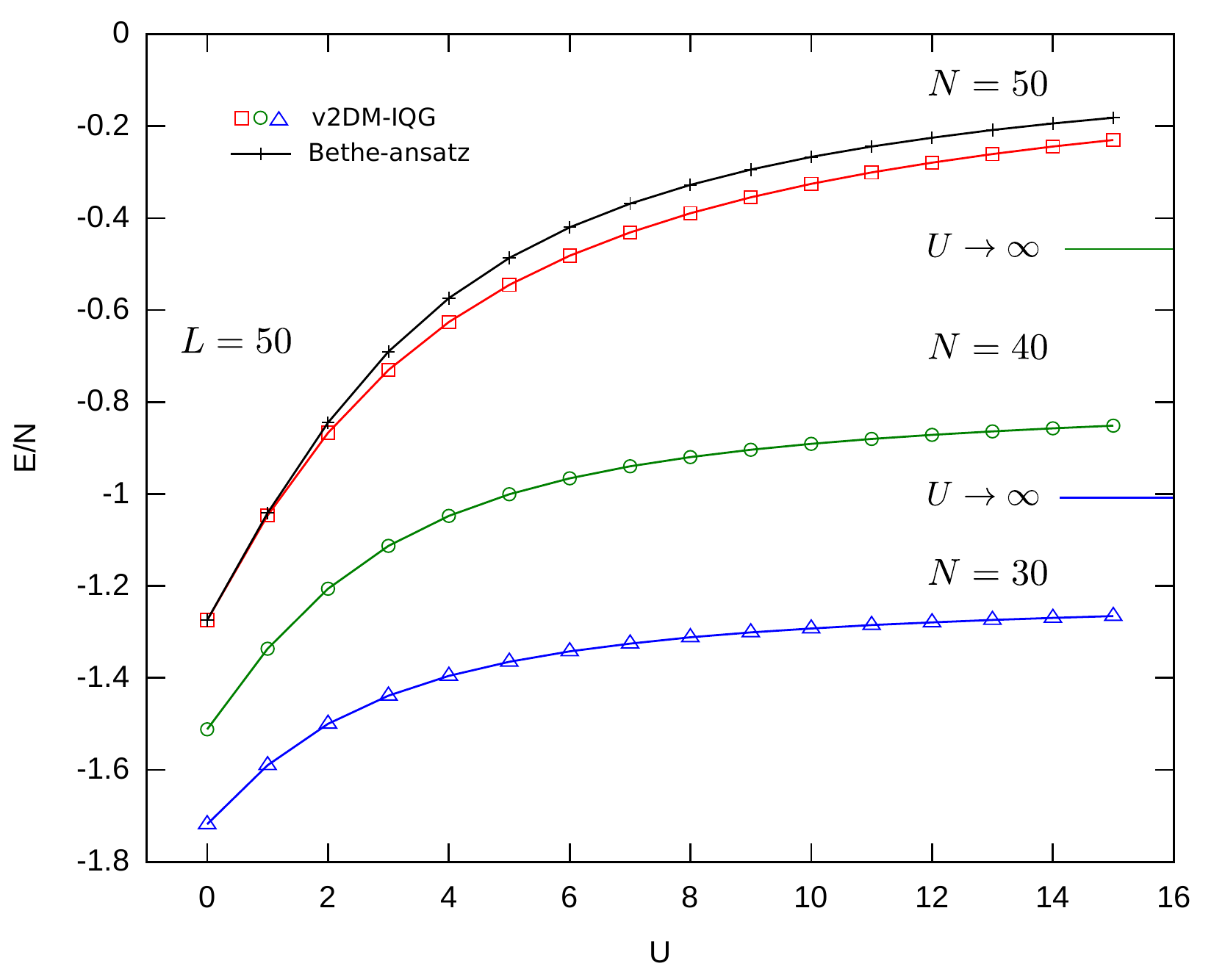}
\end{array}
$
\caption{\label{energy_1dhub}Ground-state energy per particle as a function of on-site repulsion $U$ of the Hubbard model for a 20-site (top) and 50-site (bottom) lattice at half-, $\frac{4}{10}$  and $\frac{3}{10}$ filling. For the 20-site lattice a comparison is made between v2DM using the $\mathcal{IQG}$ and $\mathcal{IQGT}$ conditions, and a quasi-exact result using MPS. For the 50-site lattice only $\mathcal{IQG}$ conditions are feasible, and these have been compared to the exact (Bethe-ansatz) result for half filling.}
\end{figure}
In Fig.~\ref{energy_1dhub} the ground-state energy per particle of the one-dimensional Hubbard model is plotted as a function of the on-site repulsion $U$ (the hopping parameter $t$ will always be taken equal to unity). In the top figure the v2DM results for the 20-site lattice are shown for three different fillings, 12 particles ($\frac{3}{10}$), 16 particles ($\frac{4}{10}$) and half filling. These were calculated using both the $\mathcal{IQG}$ and the $\mathcal{IQGT}$ conditions, and are compared to the quasi-exact variational MPS results. In the bottom figure the v2DM results for the 50-site lattice are shown for the same fillings ({\it i.e.} 30 particles ($\frac{3}{10}$), 40 particles ($\frac{4}{10}$) and half filling). For the 50-site lattice it was only possible to perform the calculations using the $\mathcal{IQG}$ conditions, and compare to the exact solution obtained by solving the Lieb-Wu equations for the half-filled lattice \cite{liebwu_2}. 

One interesting thing to notice is that the $\mathcal{IQG}$ energy per particle for the 20-site lattice and the 50-site lattice, at the same filling, are very similar. This is due to the periodic boundary conditions which make the results converge quite rapidly for increasing lattice size $L$, implying that one can already extract relevant results for the thermodynamic limit by studying relatively small lattices. This fast convergence can be clearly seen in Fig.~\ref{enconv}, where we plotted the energy per particle of a Hubbard model with $U=1$ at half filling, as a function of the lattice size $L$. 
\begin{figure}
\centering
\includegraphics[scale=0.7]{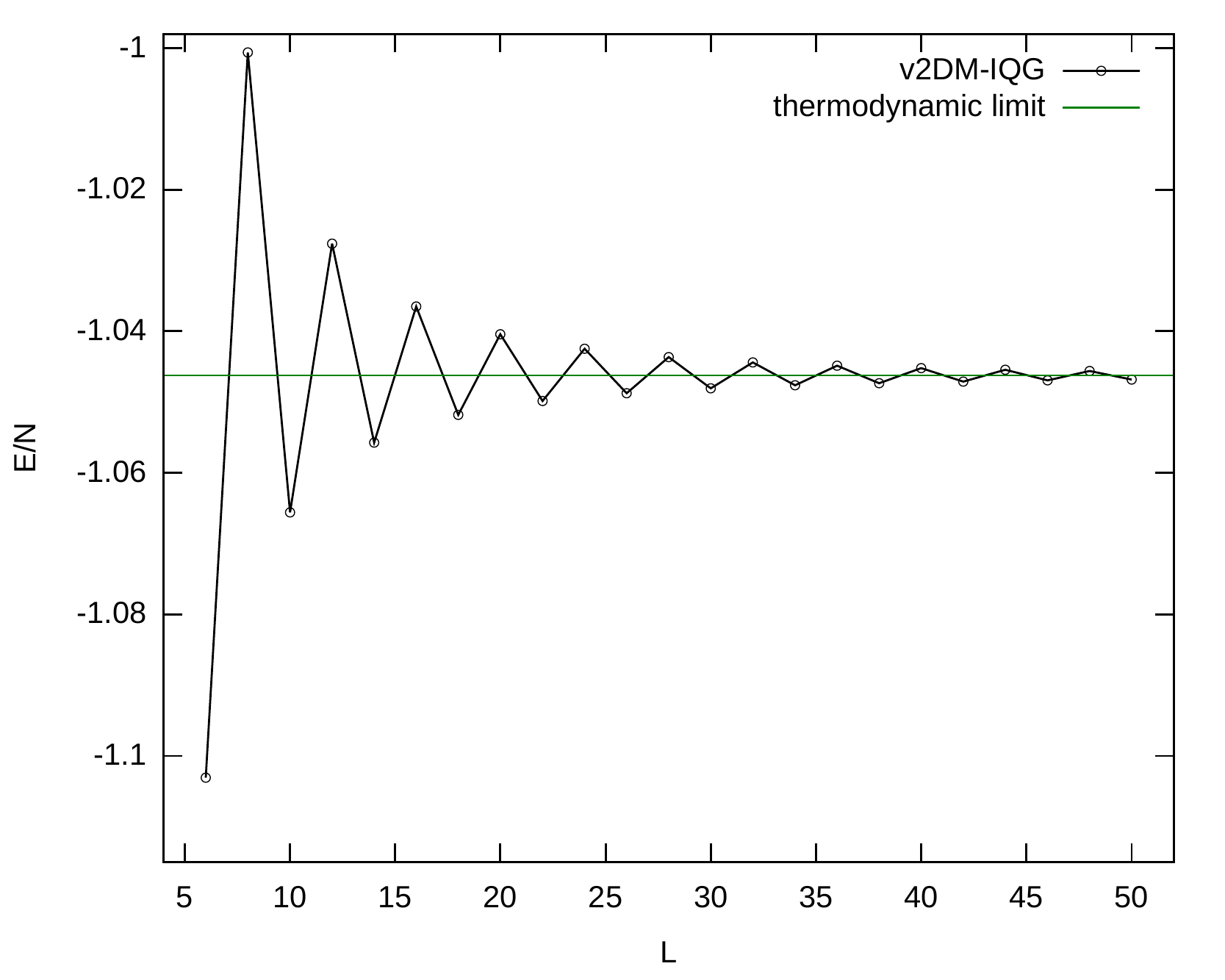}
\caption{\label{enconv}Ground-state energy per particle as a function of the lattice size $L$, indicating how fast the finite size results converge to the thermodynamic limit for half filling.}
\end{figure}

Another thing to remark in Fig.~\ref{energy_1dhub} is that for the 20-site lattice, the difference between $\mathcal{IQG}$ and $\mathcal{IQGT}$ is rather small for the half-filled lattice, but larger for the other fillings, and that the difference gets larger when $U$ increases. For the 50-site lattice we see that the $\mathcal{IQG}$ result agrees nicely with the solution of the Lieb-Wu equations. For the other fillings no reference data are available. There are, however, two limits of the model that are exactly solvable. The first limit is the case of no interaction, {\it i.e.} $U=0$, for which the solution has already been given in Section~\ref{hub_sym} (see Fig.~\ref{hubbard_fig}). The Hamiltonian reduces to a single-particle operator, which means this limit is already described correctly by including the $\mathcal{I}$ and $\mathcal{Q}$ conditions alone. The other exactly solvable limit is when $U\rightarrow+\infty$. In this limit the physics of the model decouples into two independent parts, one describing the spin of the system, and the other the movement of the particles (this is called spin-charge separation \cite{ogata}). This decoupling shows up in the Bethe-ansatz wave function: the charge degrees of freedom are described by a Slater determinant of spinless fermions, whereas the spin degrees of freedom become equivalent to a spin-$\frac{1}{2}$ Heisenberg model. The single-particle energy spectrum changes slightly compared to Eq.~(\ref{sp_hubbard}) because the boundary conditions for spinless fermions are periodic/antiperiodic if $N$ is even/odd \cite{ogata,krivnov}:
\begin{equation}
\epsilon_k = -2t\cos{k}\qquad\text{where}\qquad
\left\{
\begin{matrix}
k = \frac{2\pi n}{L}&\qquad\text{if}\qquad N\%2=0\\
k = \frac{(2n+1)\pi}{L}&\qquad\text{if}\qquad N\%2=1
\end{matrix}
\right.~.
\label{1dhub_sclim}
\end{equation}
When the lattice is half-filled all the single-particle states are occupied, and the total energy sums up to zero, which is correctly described by the $\mathcal{IQG}$ results in Fig.~\ref{energy_1dhub}. Away from half filling, however, the energy has a finite limit which can be calculated using Eq.~(\ref{1dhub_sclim}). From the figure we can see that the $\mathcal{IQG}$ conditions do not suffice to correctly describe the large-$U$ limit. Only when the $\mathcal{T}$ conditions are added, the results converge to the right limit. Calculations at very large values of $U$ have been performed that confirm this statement, and these results are shown in Table~\ref{hub_largeU}.
\begin{table}
\centering
\caption{\label{hub_largeU}Energy per site of the v2DM calculations away from half filling at large values of $U$, compared to the MPS and the Bethe-ansatz results where available.}
\begin{tabular}{|c|c|ccccc|}
\hline
$L$&$N$&$U$&$\mathcal{IQG}$&$\mathcal{IQGT}$&vMPS&exact\\
\hline
\multirow{6}{*}{20}&
\multirow{3}{*}{12}&50&-1.2259&-1.0804&-1.0488&*\\
&&100&-1.2177&-1.0646&-1.03116&*\\
&&$\infty$&*&* &* &-1.0008\\
\cline{2-7}
&\multirow{3}{*}{16}&50&-0.7972&-0.5458&-0.5205&*\\
&&100&-0.7860&-0.5179&-0.49513&*\\
&&$\infty$&*&* &* &-0.4639\\
\hline
\end{tabular}
\begin{tabular}{|c|c|ccc|}
\hline
$L$&$N$&$U$&$\mathcal{IQG}$&exact\\
\hline
\multirow{6}{*}{50}&
\multirow{3}{*}{30}&50&-1.2272&*\\
&&100&-1.2191&*\\
&&$\infty$&*&-1.0008\\
\cline{2-5}
&\multirow{3}{*}{40}&50&-0.7974&*\\
&&100&-0.7862&*\\
&&$\infty$&*&-0.4671\\
\hline
\end{tabular}
\end{table}
\subsubsection{Momentum distribution}
\begin{figure}
\centering
$
\begin{array}{c}
\includegraphics[scale=0.7]{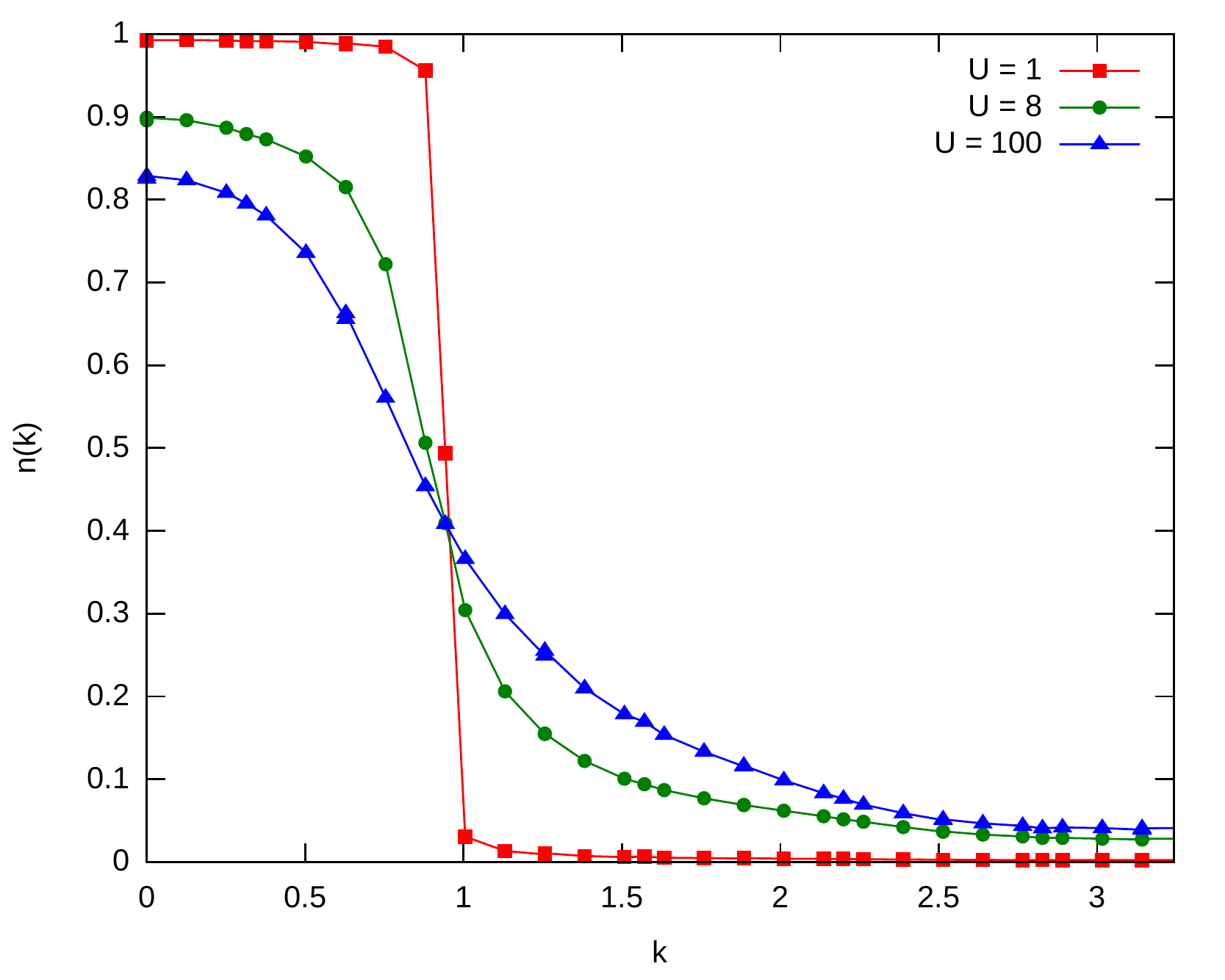}\\
\includegraphics[scale=0.7]{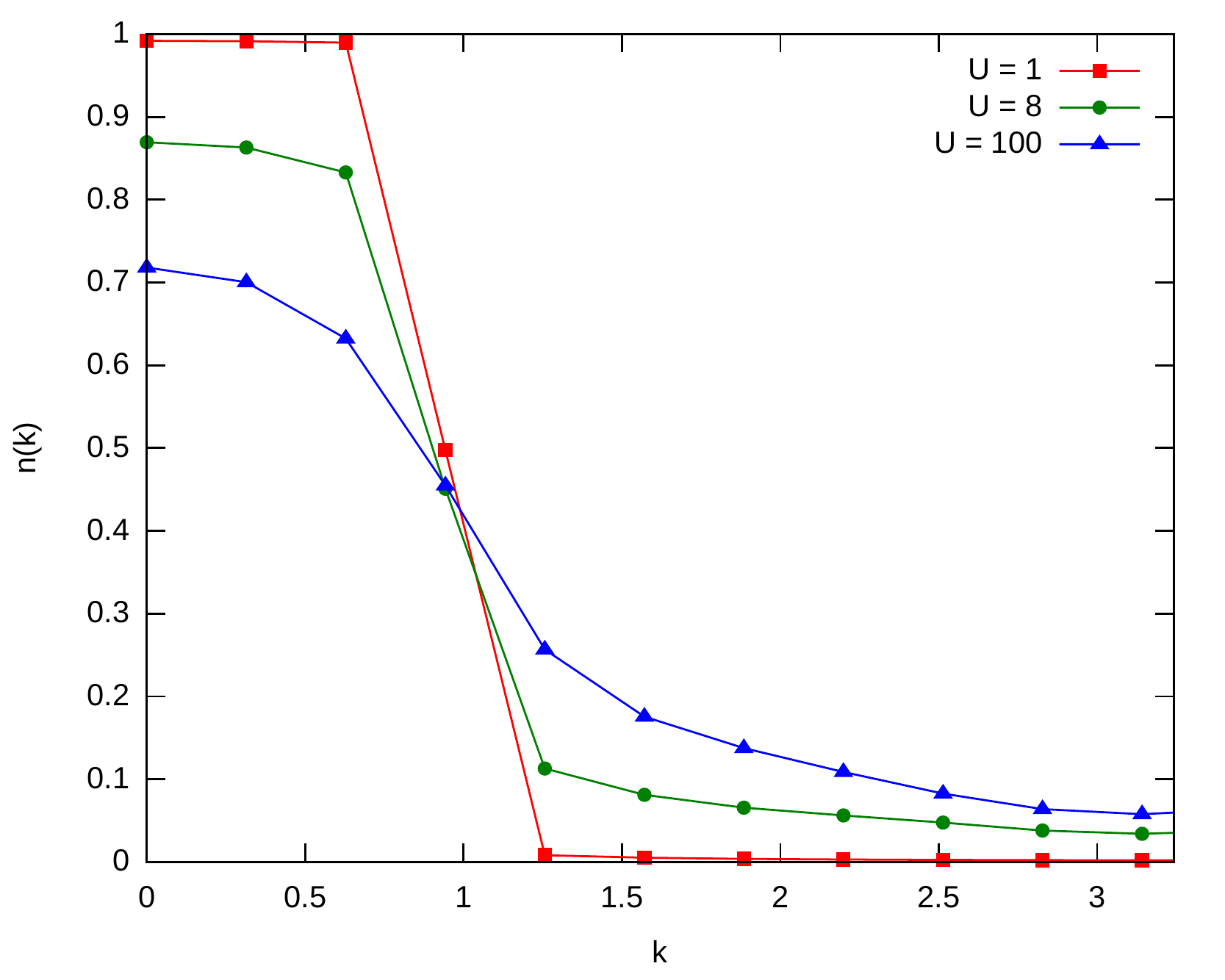}
\end{array}
$
\caption{\label{momentum_0.3}Momentum distribution in the ground state for a lattice filling of $\frac{3}{10}$ and on-site repulsion $U= 1,8$ and 100 calculated with the $\mathcal{IQG}$ conditions (top) and the $\mathcal{IQGT}$ conditions (bottom).}
\end{figure}
\begin{figure}
\centering
$
\begin{array}{c}
\includegraphics[scale=0.7]{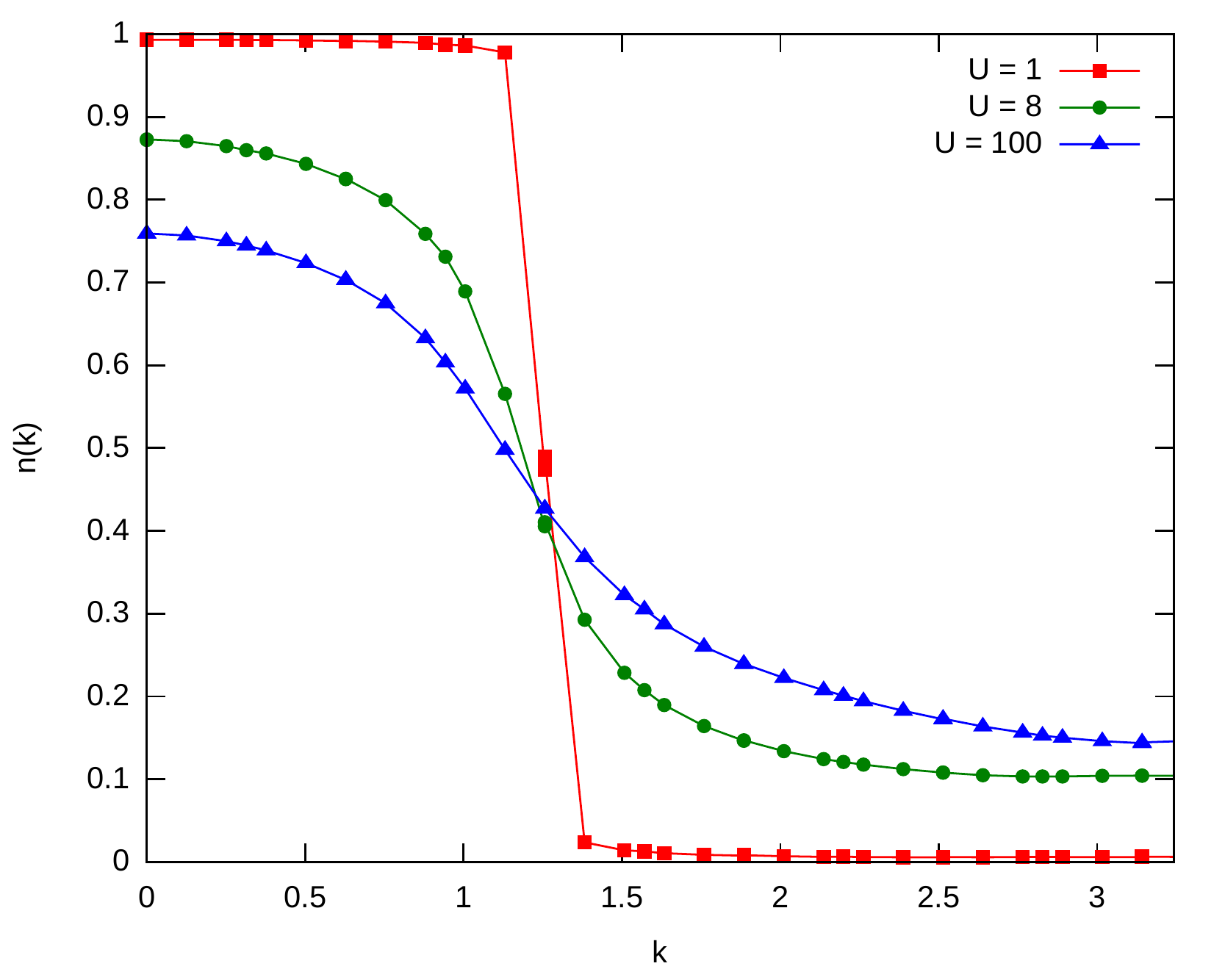}\\
\includegraphics[scale=0.7]{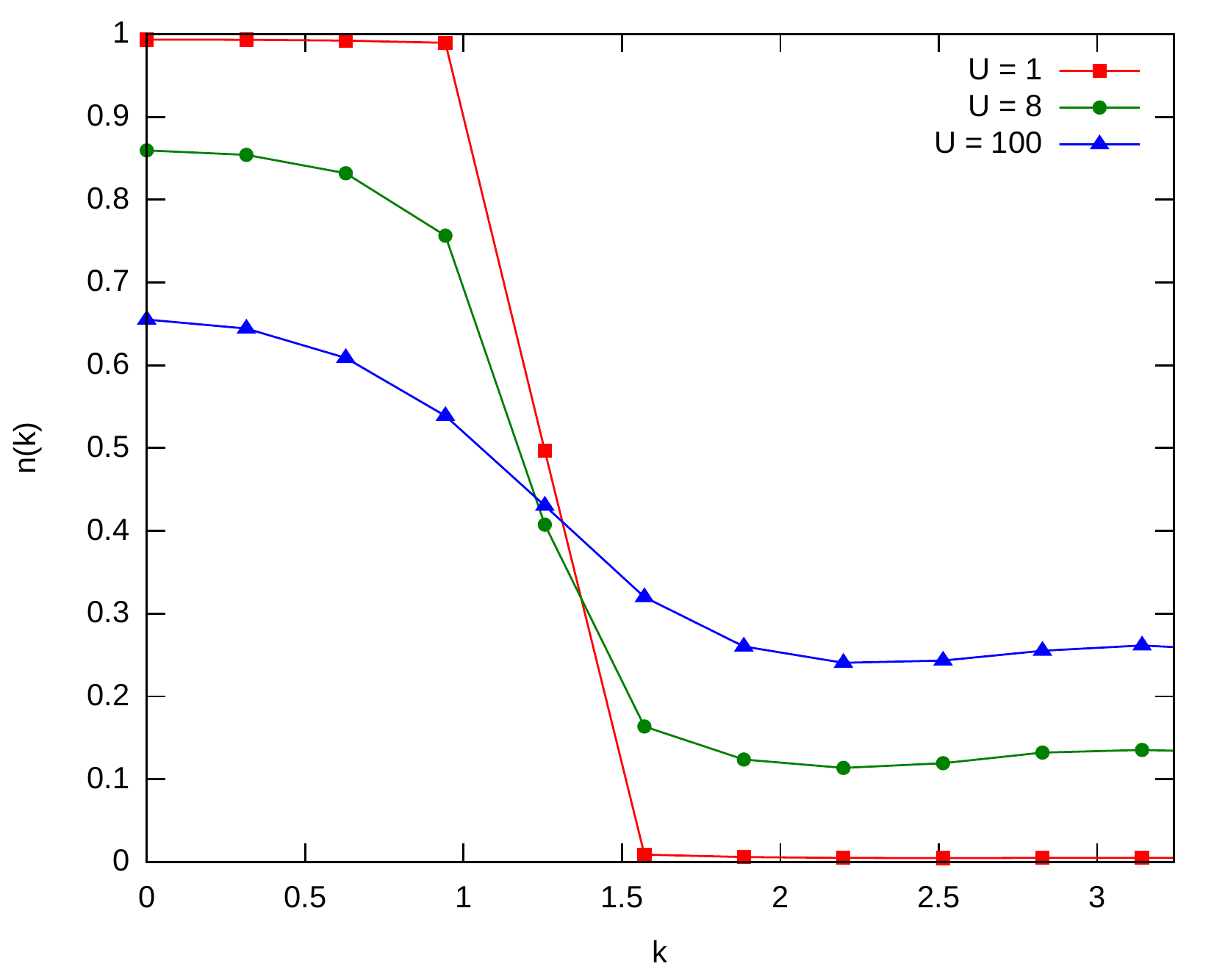}
\end{array}
$
\caption{\label{momentum_0.4}Momentum distribution in the ground state for a lattice filling of $\frac{4}{10}$ and on-site repulsion $U= 1, 8$ and 100 calculated with the $\mathcal{IQG}$ conditions (top) and the $\mathcal{IQGT}$ conditions (bottom).}
\end{figure}
The momentum distribution function is easy to extract from the 2DM, as it is just the 1DM expressed in momentum space (See Eq.~(\ref{1DM_sc_ti_written_up})).
The results of the v2DM calculations for $\frac{3}{10}$ and $\frac{4}{10}$ fillings using both $\mathcal{IQG}$ and $\mathcal{IQGT}$ calculations are presented in Figs.~\ref{momentum_0.3} and \ref{momentum_0.4}, for three different values of the on-site repulsion $U$. In case of $\mathcal{IQG}$, the values of the momentum distribution for the same filling factor but different lattice sizes $L=20$ and $L=50$ lie very nearly on the same curve, and are therefore plotted together. Before the results are discussed we note that the one-dimensional Hubbard model is \emph{not} a Fermi liquid, but a Luttinger liquid \cite{tomonaga,luttinger,mattis}. The momentum distribution therefore has no discontinuity at the Fermi level $k_F$ in the thermodynamic limit. For finite systems, however, one has only discrete momentum values so one has to perform finite size scaling to see whether Luttinger or Fermi liquid behaviour appears. Since we only performed calculations at two systems sizes, finite size scaling is not possible. Nevertheless, conclusions about the quality of the variationally obtained 2DM can be extracted by comparing our results to numerical calculations in the weak-correlation limit using Quantum Monte Carlo (QMC) \cite{sorella}, and in the strong-correlation limit by solving the simplified Bethe-ansatz \cite{ogata}. The agreement between the $\mathcal{IQG}$ and $\mathcal{IQGT}$ results is quite good for the $U=1$ curve, both for the $\frac{3}{10}$ and $\frac{4}{10}$ filled lattice. These momentum distributions have the form one would expect for small values of $U$, {\it i.e.} a Fermi-sea type distribution with highly occupied low momentum states and with a steep drop around the Fermi momentum $k_F$. As $U$ increases the distribution becomes more spread out, and we note that the agreement between $\mathcal{IQG}$ and $\mathcal{IQGT}$ deteriorates. The $\mathcal{IQGT}$ results tend to move away more from the Fermi-sea behaviour than the $\mathcal{IQG}$ results, which is in agreement with the earlier observations for the energy.

For large values of $U$ one would expect a discontinuity around $2k_F$, because of the spinless fermion description of the charge part of system. This jump is absent, in agreement with exact results \cite{ogata}, which shows that the behaviour in the strong-correlation limit is more subtle. For $\frac{4}{10}$ filling we notice that the momentum distribution has a non-monotonous behaviour: the occupation first drops and then rises for higher momenta in the $\mathcal{IQGT}$ description of the large-$U$ systems. This behaviour is absent in the $\mathcal{IQG}$ description. This rise is physical and theoretically understood from the Bethe-ansatz solution at infinite $U$ \cite{ogata}; its absence in the 2DM determined with $\mathcal{IQG}$ is another indication of the fact that the two-index $\mathcal{IQG}$ conditions are not sufficient to capture the physics in the strong-correlation limit.
\subsubsection{Correlation functions}
Two-particle correlation functions are important quantities in the analysis of lattice systems, because they usually display the physics present in the system (for instance the appearance of magnetism). In this Section we show that in our approach, these correlation functions are easily extracted from the 2DM, and compare our results to those in \cite{sorella,ogata}.
\paragraph{Charge correlation}
\begin{figure}
\centering
$
\begin{array}{c}
\includegraphics[scale=0.7]{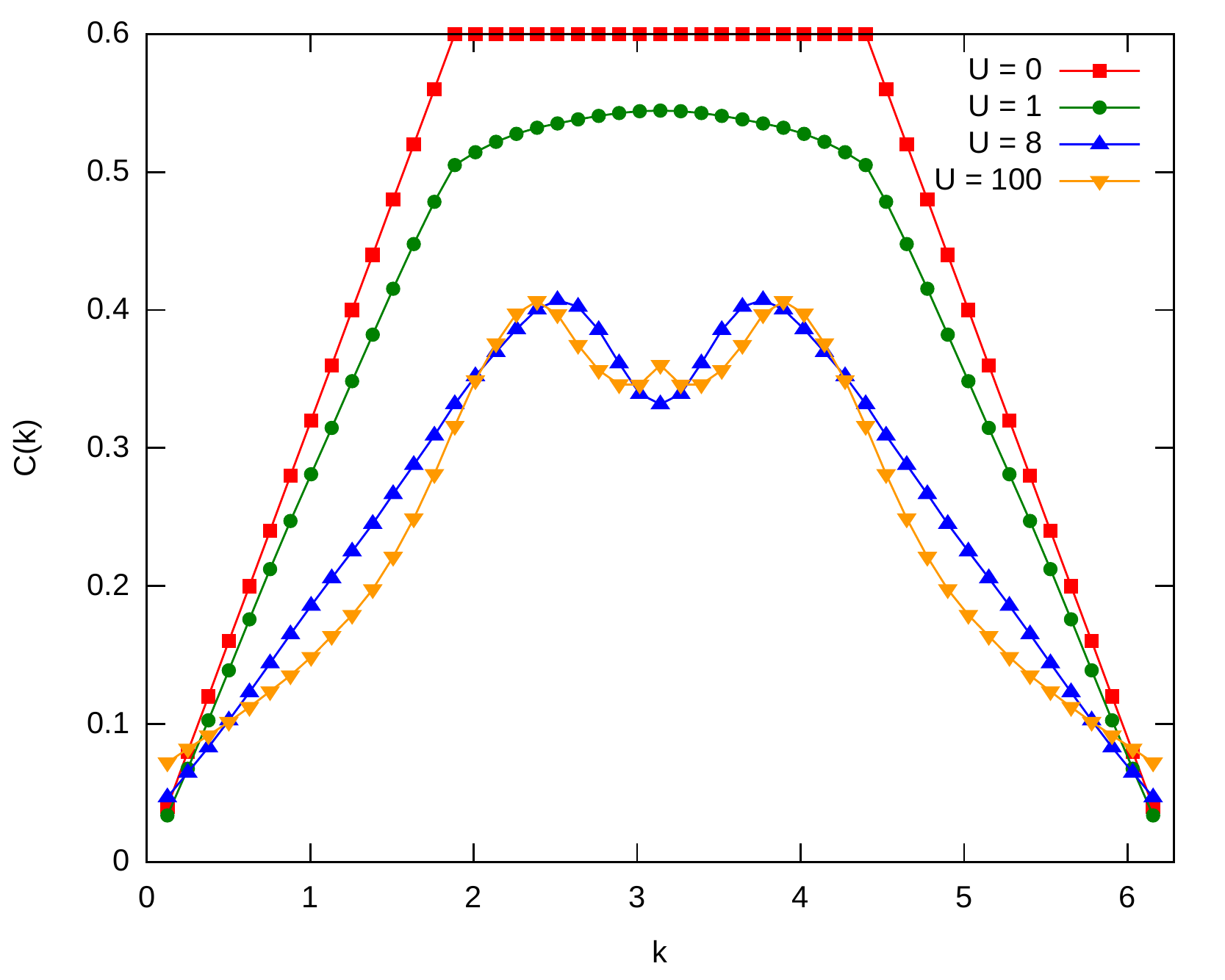}\\
\includegraphics[scale=0.7]{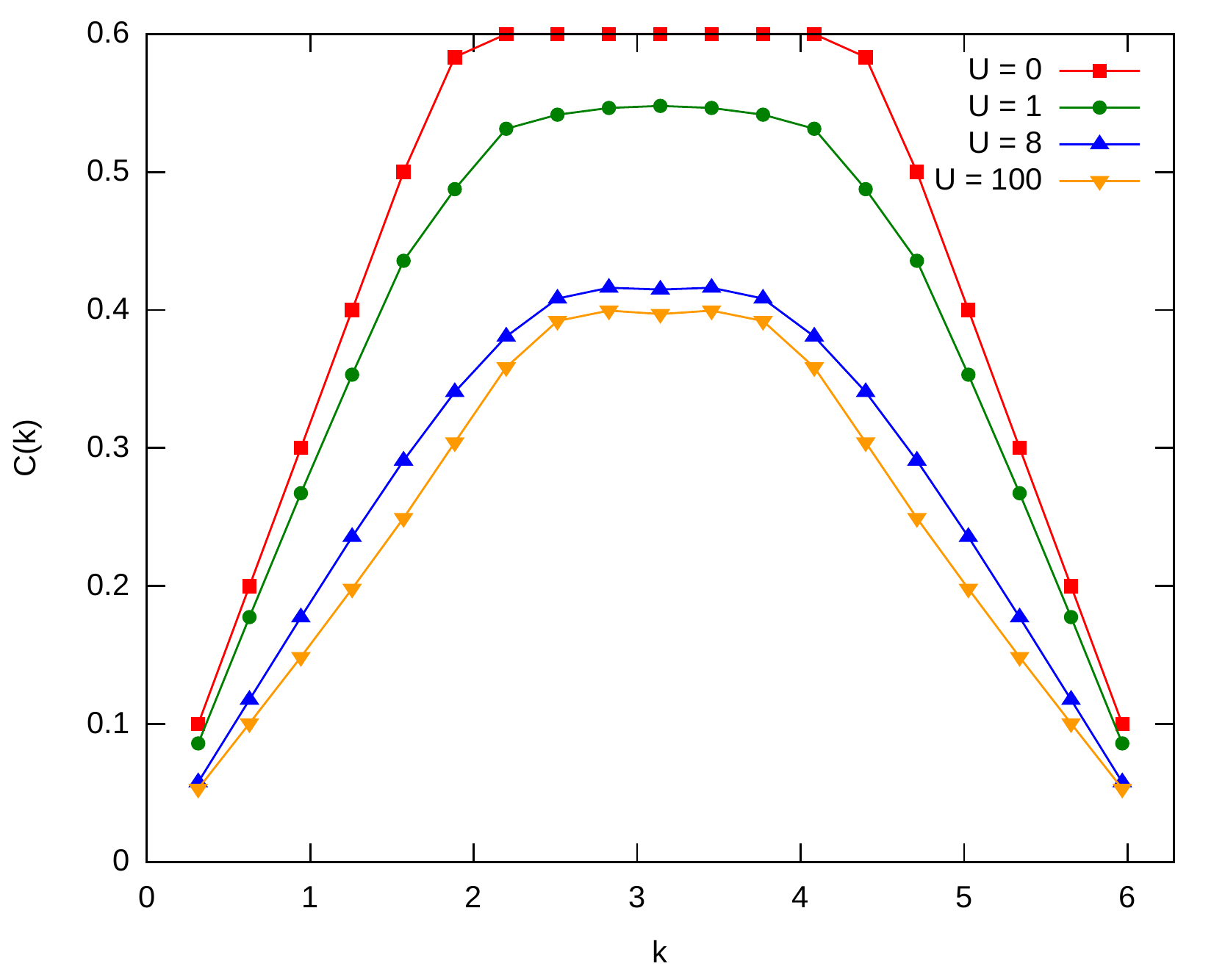}
\end{array}
$
\caption{\label{chargecorr_0.3}Two-particle charge correlation function $C(k)$, as a function of momentum, for a $\frac{3}{10}$ filled lattice and various values of on-site repulsion $U$, using $\mathcal{IQG}$ (top) and $\mathcal{IQGT}$ (bottom) conditions.}
\end{figure}
\begin{figure}
\centering
$
\begin{array}{c}
\includegraphics[scale=0.7]{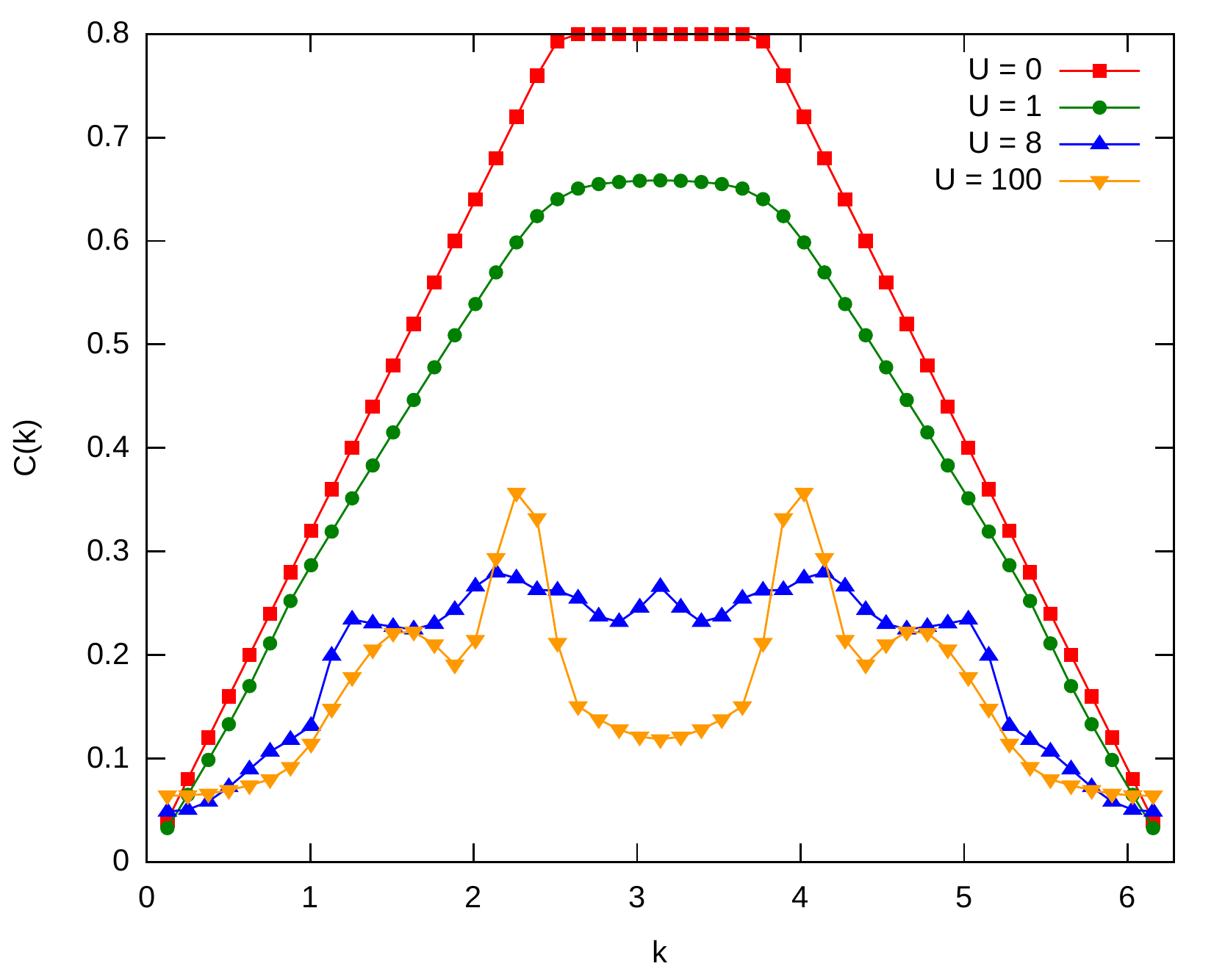}\\
\includegraphics[scale=0.7]{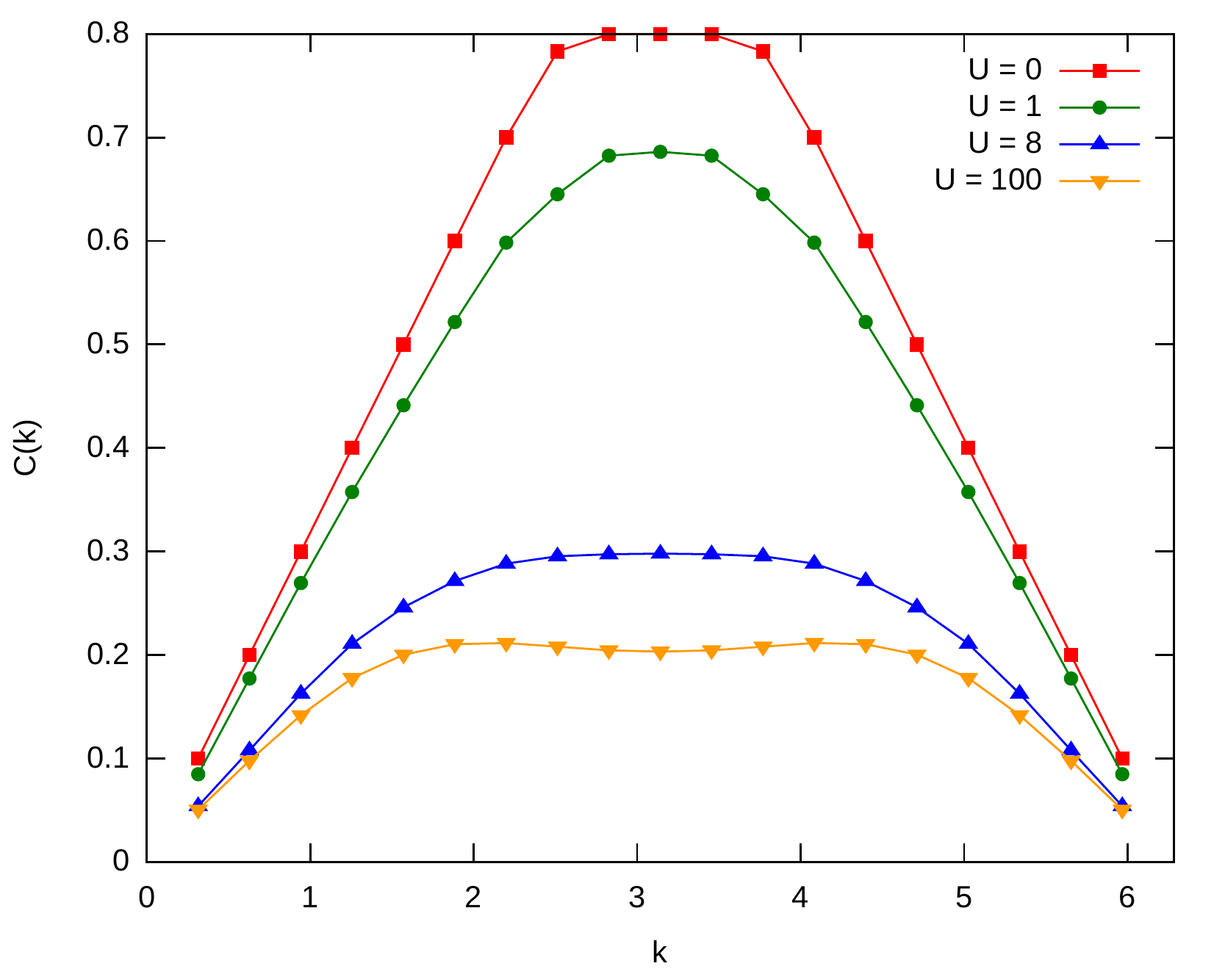}
\end{array}
$
\caption{\label{chargecorr_0.4}Two-particle charge correlation function $C(k)$, as a function of momentum, for a $\frac{4}{10}$ filled lattice and various values of on-site repulsion $U$, using $\mathcal{IQG}$ (top) and $\mathcal{IQGT}$ (bottom) conditions.}
\end{figure}
\begin{figure}
\centering
$
\begin{array}{c}
\includegraphics[scale=0.7]{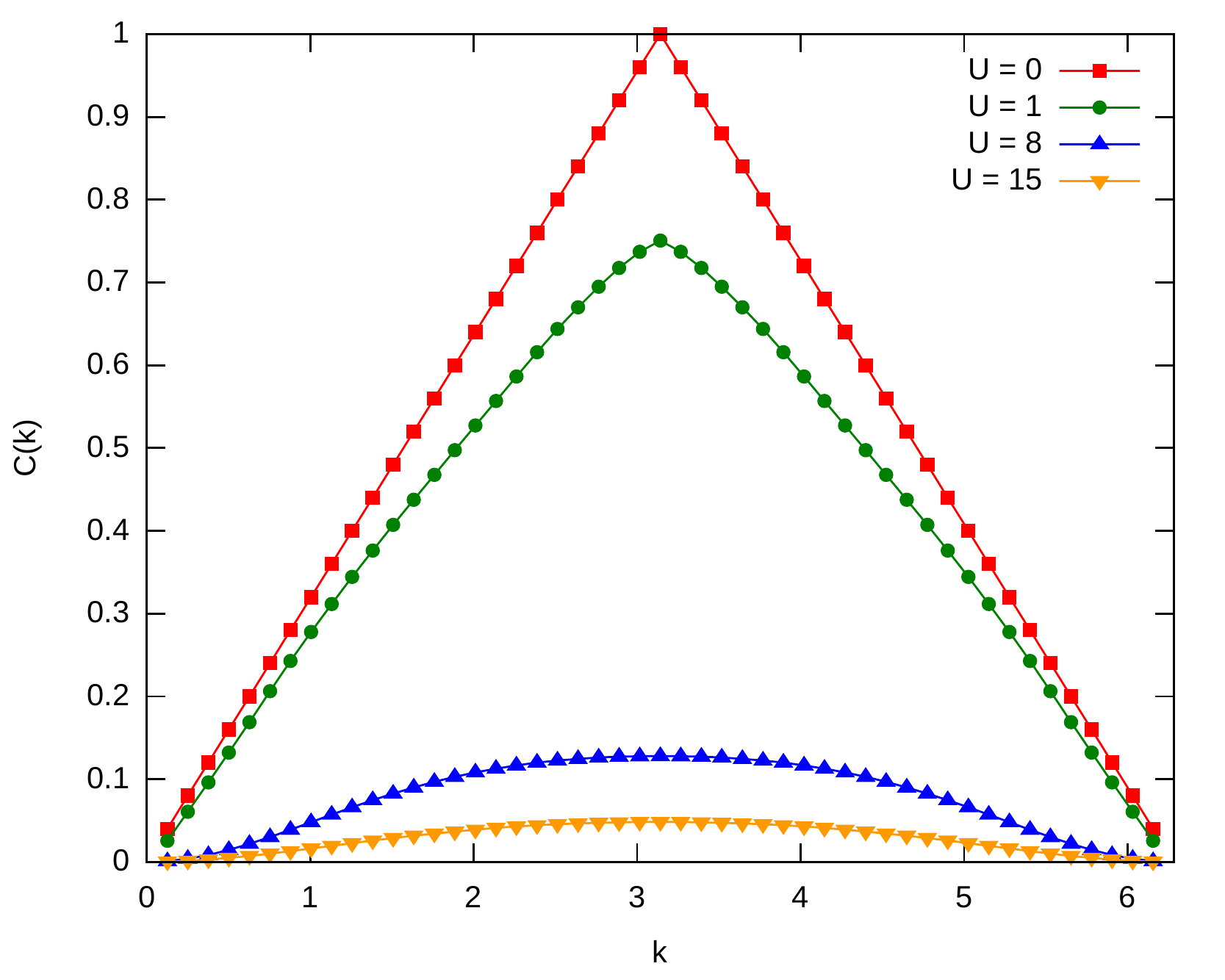}\\
\includegraphics[scale=0.7]{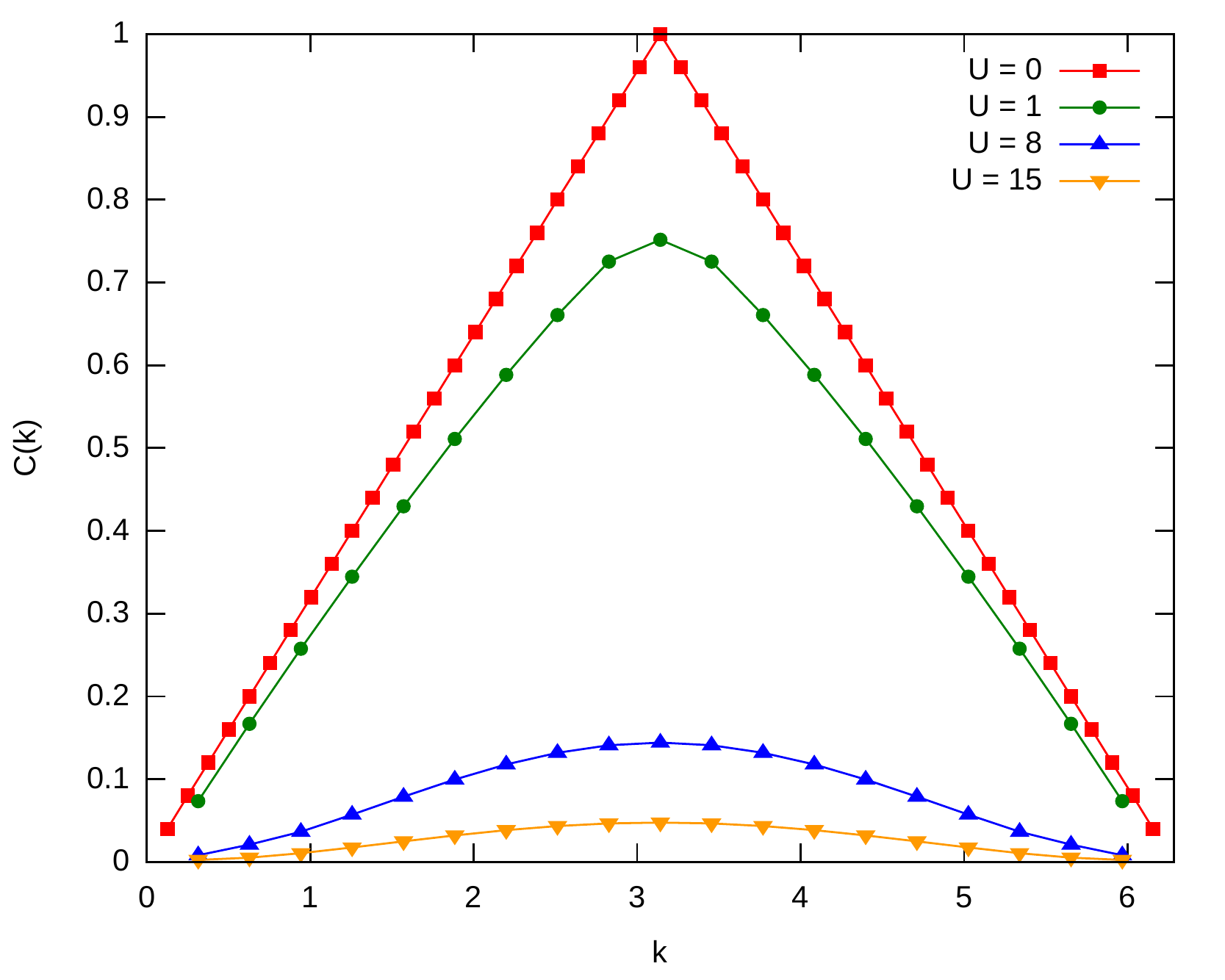}
\end{array}
$
\caption{\label{chargecorr_0.5}Two-particle charge correlation function $C(k)$, as a function of momentum, for a half-filled lattice and various values of on-site repulsion $U$, using $\mathcal{IQG}$ (top) and $\mathcal{IQGT}$ (bottom) conditions.}
\end{figure}
The two-particle charge correlation function is defined as:
\begin{equation}
C(r) = \langle \hat{n}_j\hat{n}_{j+r}\rangle = \sum_{\sigma\sigma'}\langle a^\dagger_{j\sigma}a_{j\sigma}a^\dagger_{j+r;\sigma'}a_{j+r;\sigma'} \rangle~,
\label{tpcorr}
\end{equation}
in which the notation $\langle . \rangle$ denotes the expectation value. The function is independent of the specific choice of the index $j$ because of the periodic boundary conditions. The expression in Eq.~(\ref{tpcorr}) can be written in terms of the $\mathcal{G}(\Gamma)$ matrix:
\begin{equation}
C(r) = \sum_{\sigma\sigma'} \mathcal{G}(\Gamma)_{j\sigma j\sigma;(j+r)\sigma'(j+r)\sigma'}~,
\end{equation}
and in fact only the singlet part of the $\mathcal{G}$ matrix appears:
\begin{equation}
C(r) = 2~\mathcal{G}(\Gamma)^{0}_{j j;(j+r)(j+r)}~.
\end{equation}
In translationally invariant systems one usually takes the Fourier transform of the correlation function,
\begin{equation}
C(k) = \sum_r e^{ikr} C(r) = 2\sum_{k_ak_b}\sum_{k_ck_d}\mathcal{G}(\Gamma)^{0k}_{k_ak_b;k_ck_d}~.
\end{equation}
In Figs.~\ref{chargecorr_0.3}, \ref{chargecorr_0.4} and \ref{chargecorr_0.5}, $C(k)$ has been plotted for $\frac{3}{10}$, $\frac{4}{10}$  and half filling respectively, using both $\mathcal{IQG}$ and $\mathcal{IQGT}$ conditions. Comparing the $\mathcal{IQG}$ with the $\mathcal{IQGT}$ results the same trends can be noticed as for the energy and the momentum distributions. For half filling (Fig.~\ref{chargecorr_0.5}) the $\mathcal{IQG}$ and $\mathcal{IQGT}$ results are in nice agreement. Moving away from half-filling (Figs.~\ref{chargecorr_0.3} and \ref{chargecorr_0.4}) there is only agreement for small values of $U$. For larger values of $U$ strange oscillations appear in the $\mathcal{IQG}$ results. So in this limit not only the energy, but the entire physical content of the $\mathcal{IQG}$-2DM cannot be trusted. This is once again an indication that the $\mathcal{IQG}$ conditions fail to describe the strong-correlation limit away from half-filling. The $\mathcal{IQGT}$ results compare well, both in shape and magnitude, with the results from Quantum Monte Carlo \cite{sorella}, and the Bethe-ansatz results in the strong-correlation limit \cite{ogata}.
\paragraph{Spin correlation}
\begin{figure}
\centering
$
\begin{array}{c}
\includegraphics[scale=0.7]{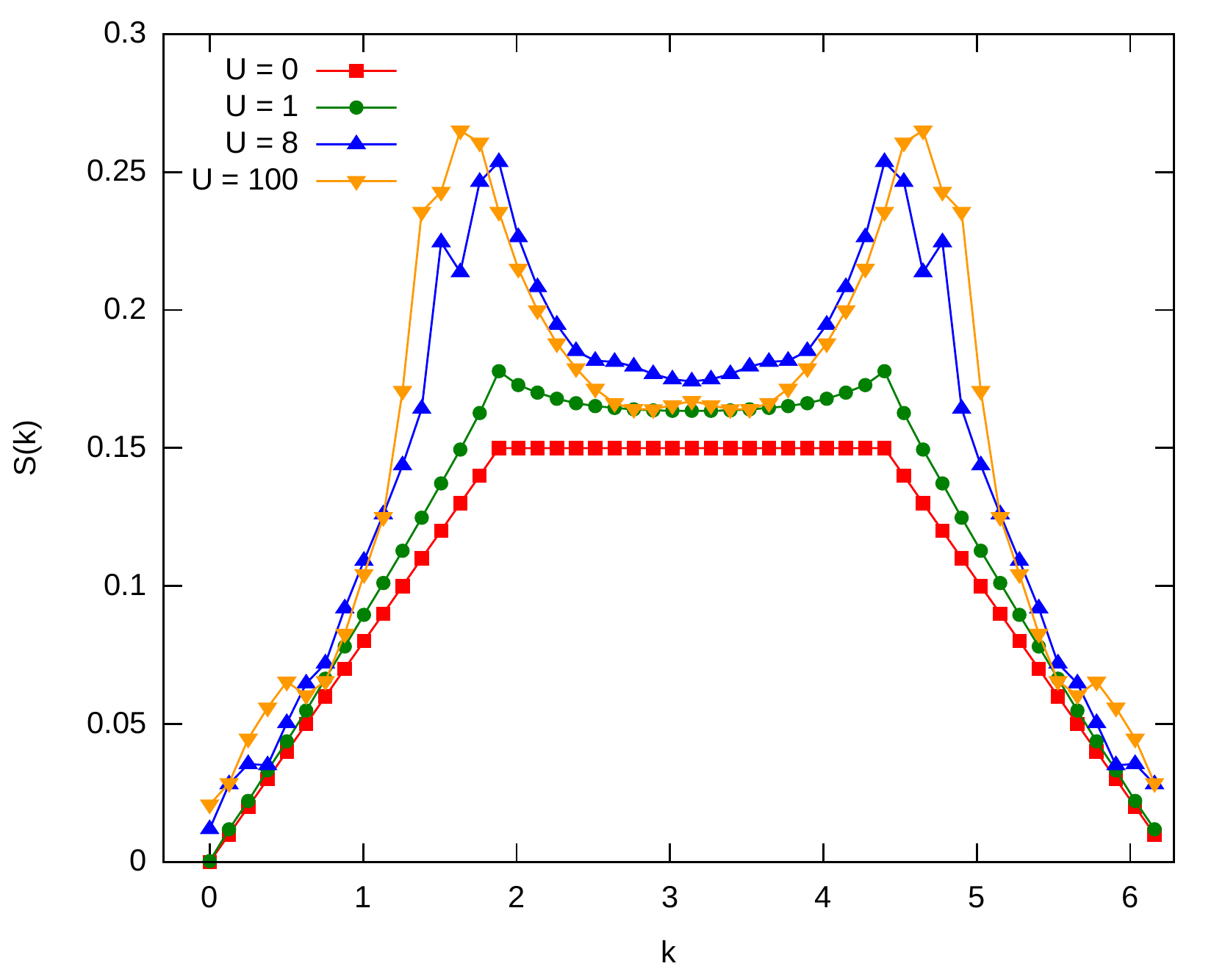}\\
\includegraphics[scale=0.7]{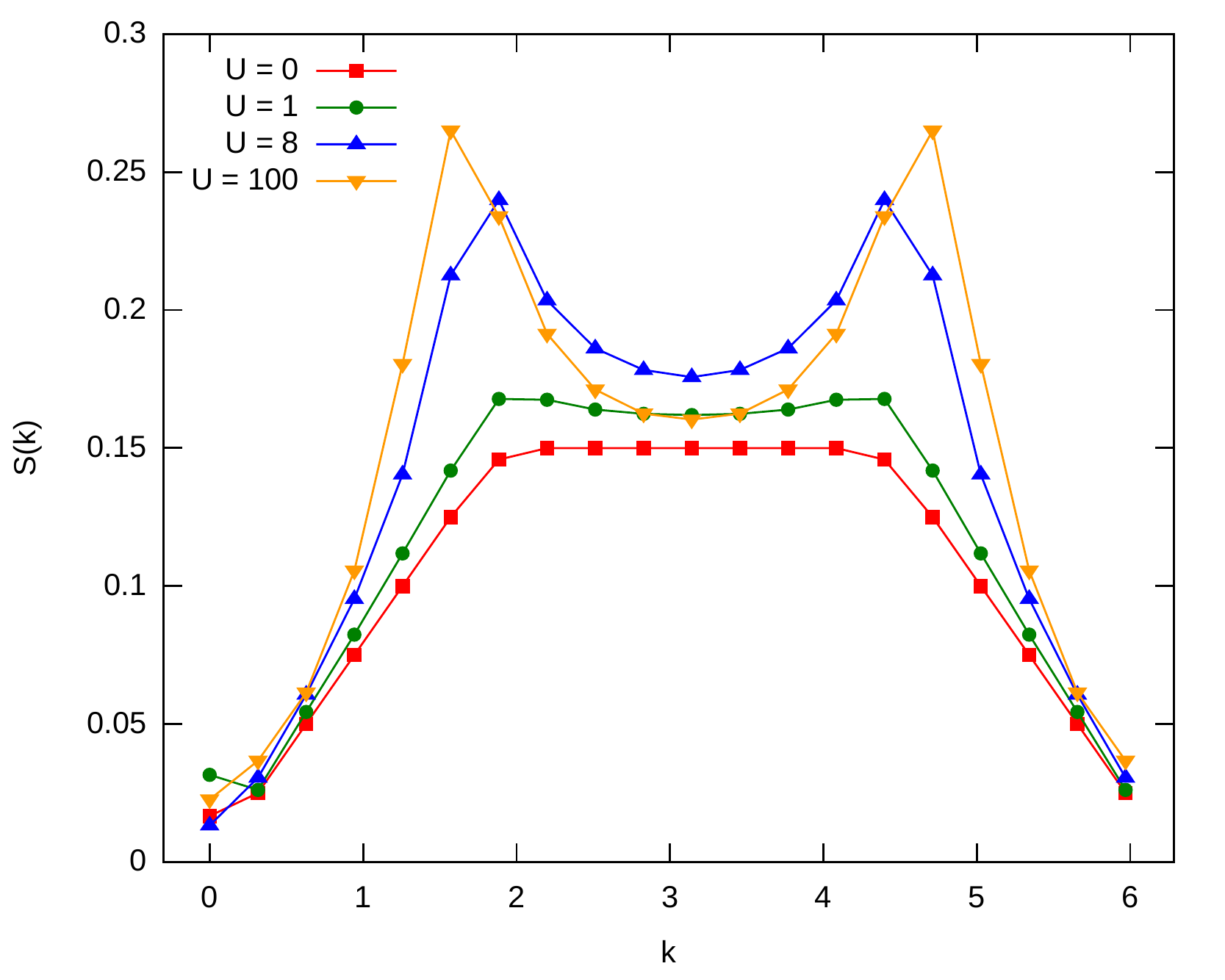}
\end{array}
$
\caption{\label{spincorr_0.3}Two-particle spin correlation function $S(k)$, as a function of momentum, for a $\frac{3}{10}$ filled lattice and various values of on-site repulsion $U$, using $\mathcal{IQG}$ (top) and $\mathcal{IQGT}$ (bottom) conditions.}
\end{figure}
\begin{figure}
\centering
$
\begin{array}{c}
\includegraphics[scale=0.7]{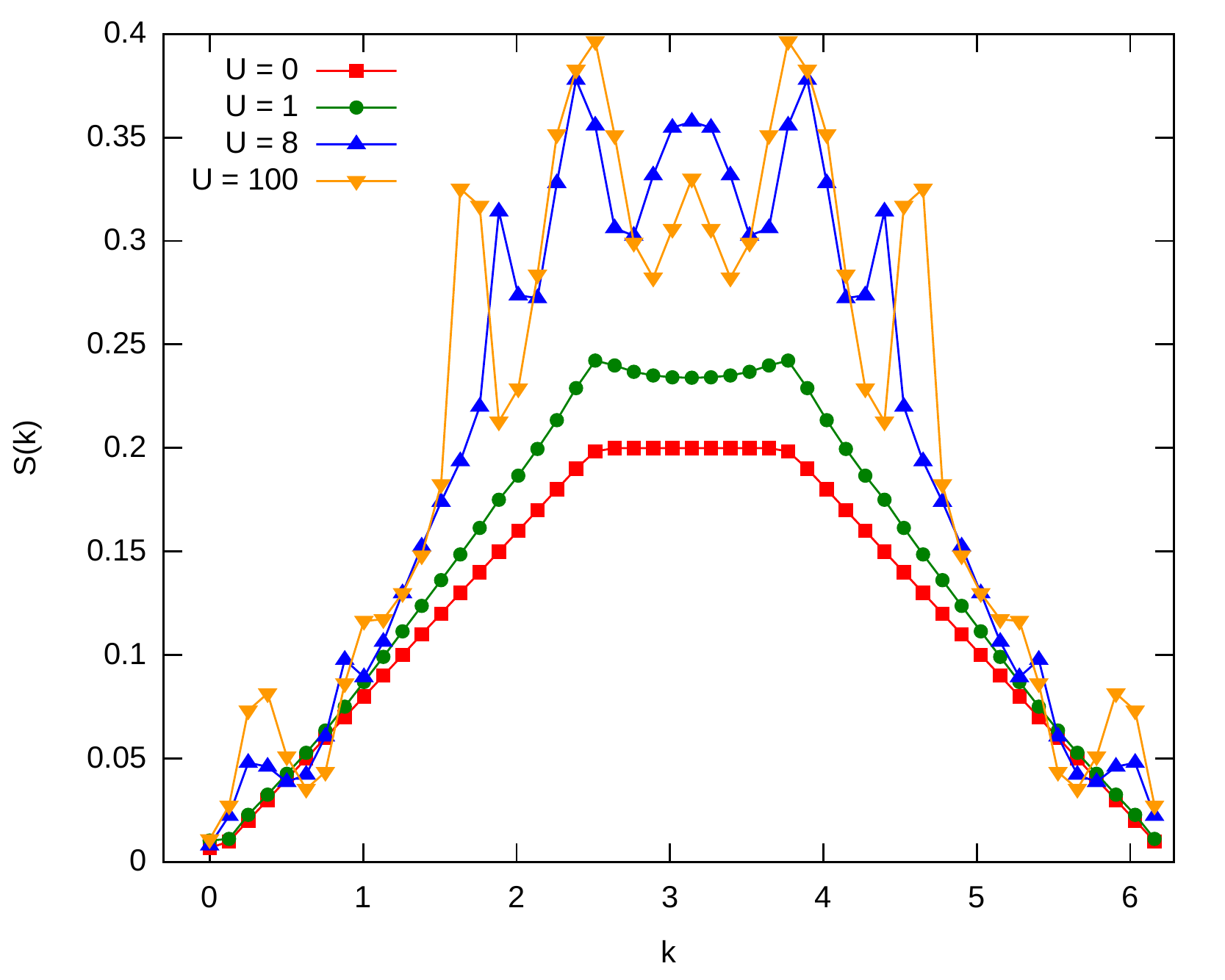}\\
\includegraphics[scale=0.7]{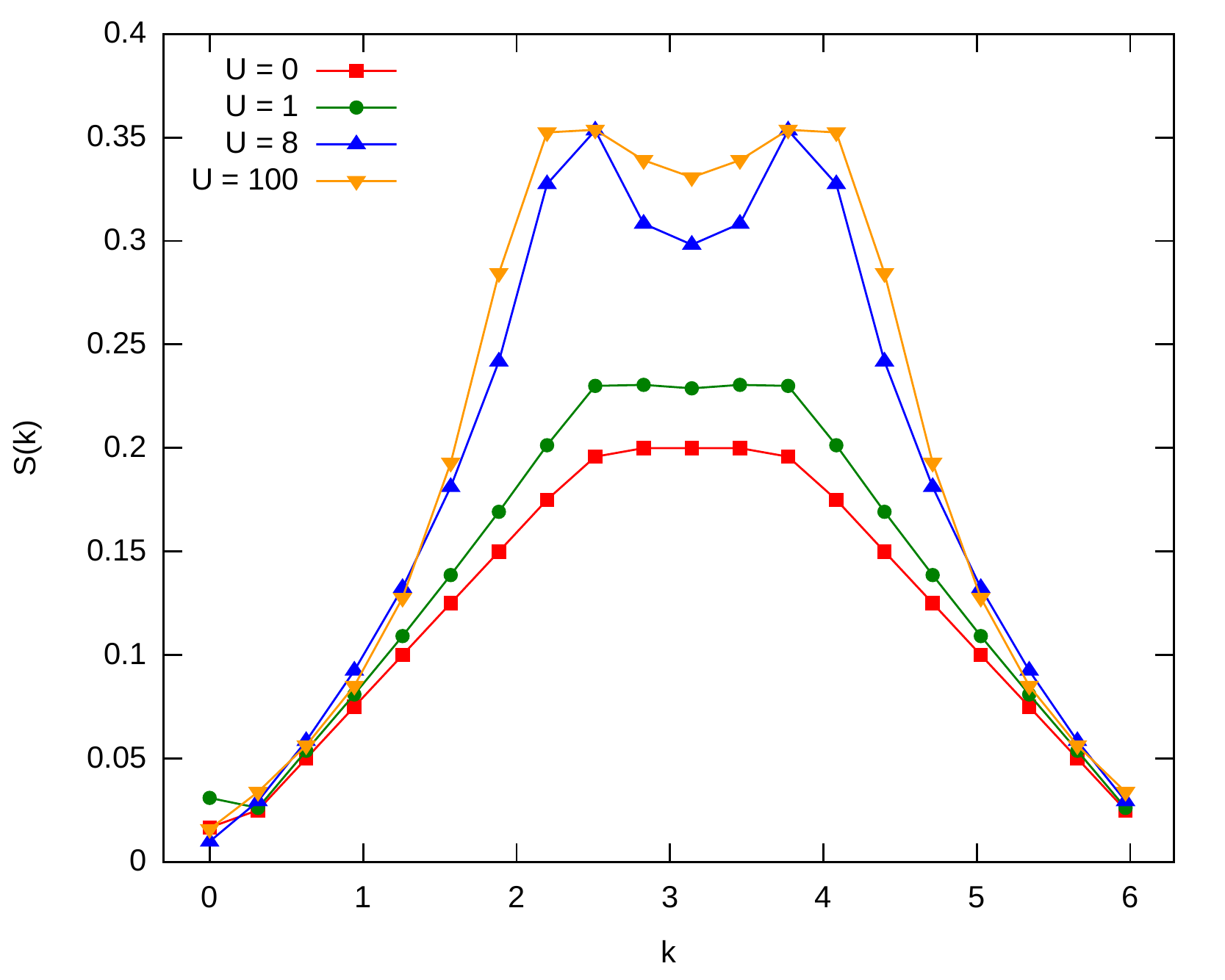}
\end{array}
$
\caption{\label{spincorr_0.4}Two-particle spin correlation function $S(k)$, as a function of momentum, for a $\frac{4}{10}$ filled lattice and various values of on-site repulsion $U$, using $\mathcal{IQG}$ (top) and $\mathcal{IQGT}$ (bottom) conditions.}
\end{figure}
\begin{figure}
\centering
$
\begin{array}{c}
\includegraphics[scale=0.7]{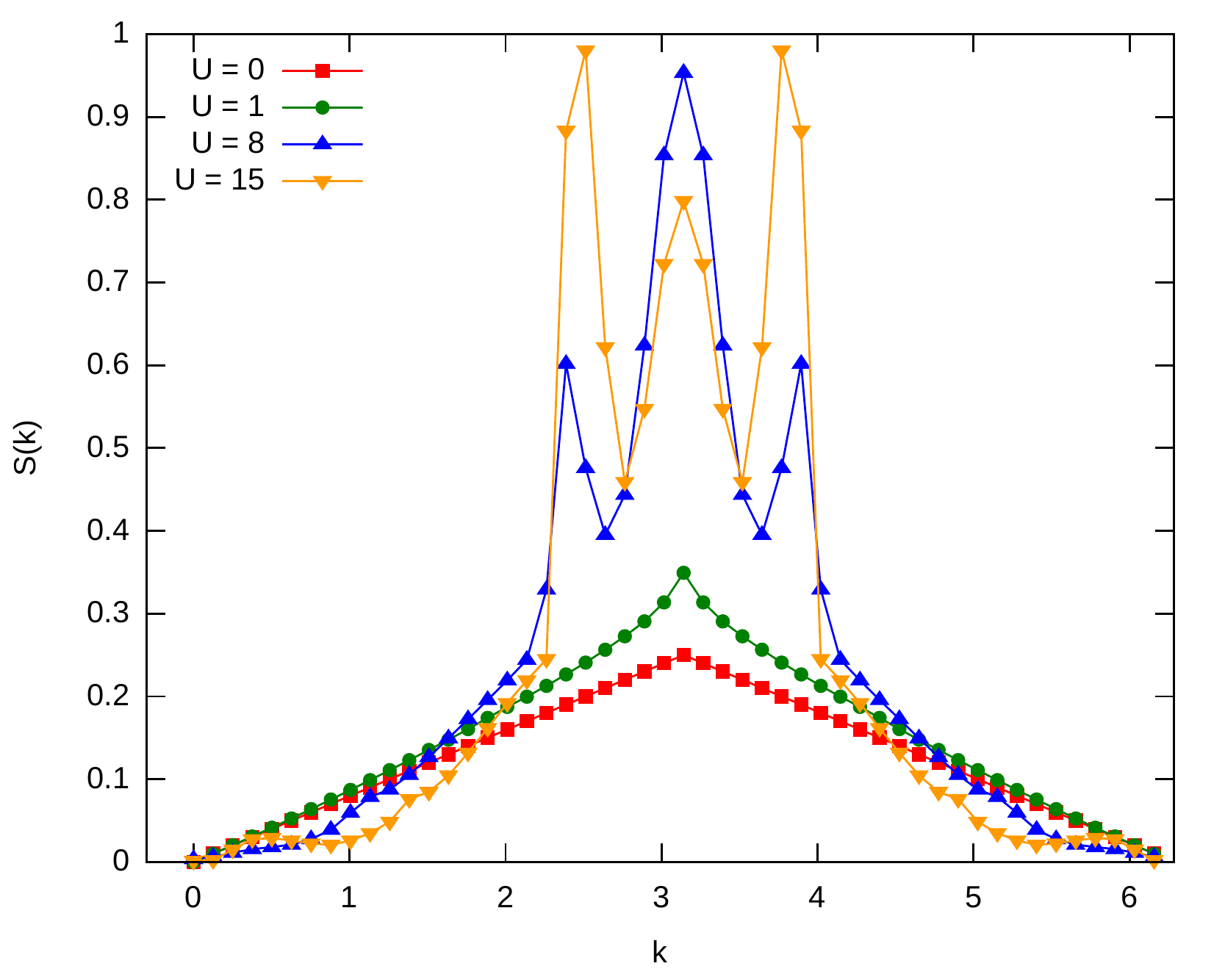}\\
\includegraphics[scale=0.7]{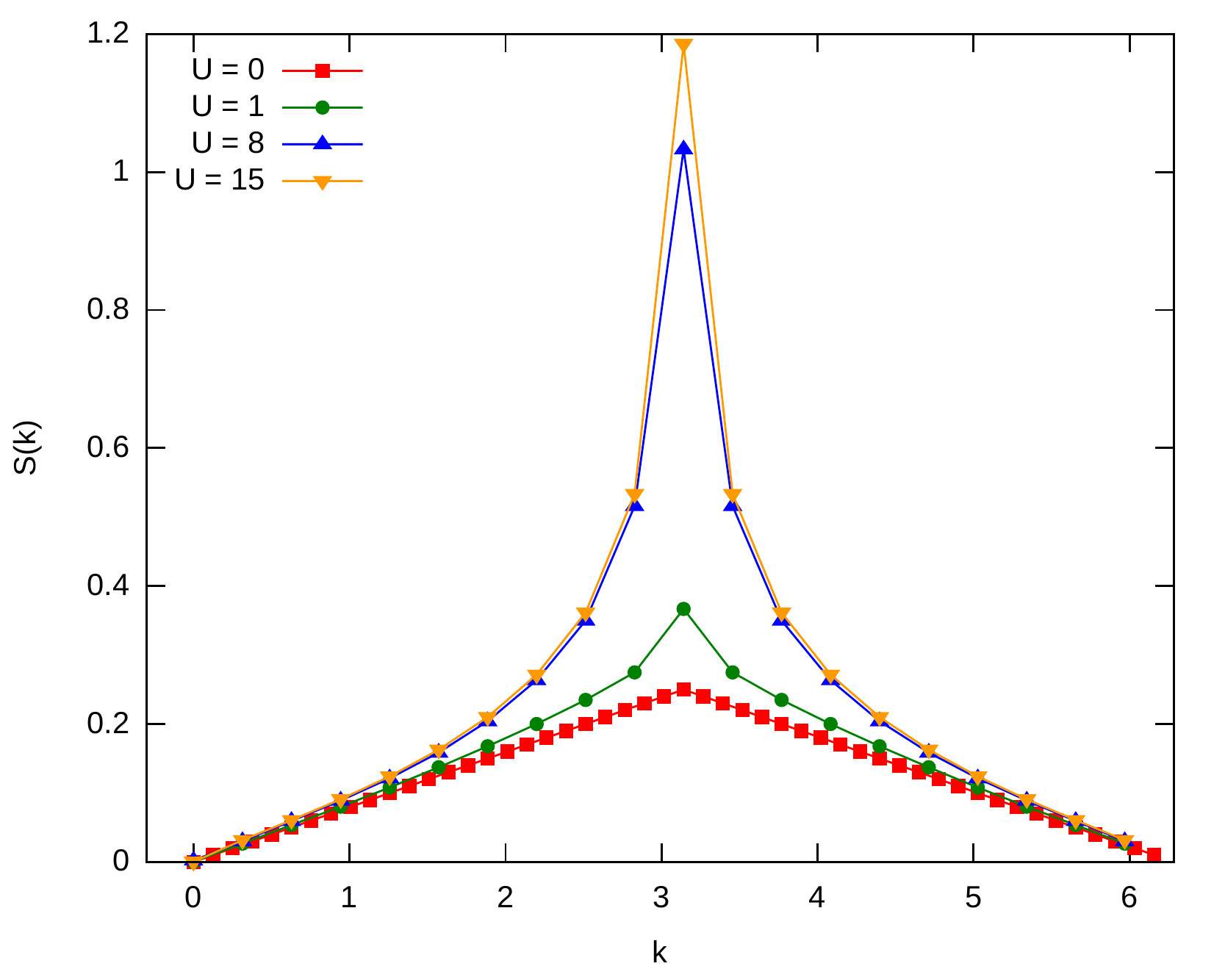}
\end{array}
$
\caption{\label{spincorr_0.5}Two-particle spin correlation function $S(k)$, as a function of momentum, for a half-filled lattice and various values of on-site repulsion $U$, using $\mathcal{IQG}$ (top) and $\mathcal{IQGT}$ (bottom) conditions.}
\end{figure}
The two-particle spin-correlation function defined as:
\begin{equation}
S(r) = \langle S_z^j S_z^{j+r}\rangle = \sum_{\sigma\sigma'} \sigma\sigma' \langle a^\dagger_{j\sigma}a_{j\sigma}a^\dagger_{j+r;\sigma'}a_{j+r;\sigma'} \rangle~,
\end{equation}
can be expressed as a function of the $\mathcal{G}$ matrix:
\begin{equation}
S(r) = \sum_{\sigma\sigma'}\sigma\sigma'\mathcal{G}(\Gamma)_{j\sigma j\sigma;(j+r)\sigma'(j+r)\sigma'}~.
\end{equation}
Written in terms of the spin-coupled $\mathcal{G}$ matrix, only the triplet $S=1$ part contributes:
\begin{equation}
S(r) = \frac{1}{2}\mathcal{G}(\Gamma)^1_{jj;(j+r)(j+r)}~.
\end{equation}
Fourier transforming $S(r)$ yields the momentum dependent spin-correlation function:
\begin{equation}
S(k) = \sum_r e^{ikr}S(r) = \frac{1}{2}\sum_{k_ak_b}\sum_{k_ck_d}\mathcal{G}^{1k}_{k_ak_b;k_ck_d}~.
\end{equation}
We have plotted this object in Figs.~\ref{spincorr_0.3}, \ref{spincorr_0.4} and \ref{spincorr_0.5} for $\frac{3}{10}$, $\frac{4}{10}$ and half filling respectively, using both $\mathcal{IQG}$ as $\mathcal{IQGT}$ conditions. Unsurprisingly, a good agreement is observed between the $\mathcal{IQG}$ and $\mathcal{IQGT}$ results for small values of $U$. At larger values of $U$, away from half filling, results are very poor with the $\mathcal{IQG}$ conditions, especially in the $\frac{4}{10}$-filled case, where the correlation function is wildly oscillating. More surprising, is that the spin-correlation function for the half-filled lattice in the large-$U$ limit is also incorrect in the $\mathcal{IQG}$ approximation. In the strong-correlation limit, for half-filling, the spin part of the Hubbard model is identical to the Heisenberg model \cite{ogata}, for which the spin-correlation function has a singularity at two times the Fermi momentum $2k_F=\pi$, which is exactly what we see in the $\mathcal{IQGT}$ results. Below half-filling the singularity in the large-$U$ limit splits and shifts to smaller values of $k$, as obserbed in the $\mathcal{IQGT}$ figures \ref{spincorr_0.3} and \ref{spincorr_0.4}. This is in agreement with the results in \cite{ogata} and \cite{sorella}. In conclusion we can say that the 2DM obtained with the $\mathcal{IQGT}$ conditions, correctly describes the physics that governs the spin-correlation function, whereas the $\mathcal{IQG}$ conditions do not. It is also important to note, that even though the $\mathcal{IQG}$ results for the energy are good for the half-filled lattice, the 2DM is flawed, because the spin-correlation function is not correctly described.
\subsection{Failure in the strong-correlation limit}
The standard two-index conditions, $\mathcal{IQG}$, have a problem describing the strong-correlation limit, $U\rightarrow\infty$, of the Hubbard model. The three-index conditions, however, describe the physics of the model correctly. In this Section we analyse why this is the case, and derive some non-standard constraints that try to fix this behaviour at the two-index level.

The strong-correlation limit of the Hubbard model has a very clear physical structure \cite{ogata,krivnov}. The spin part of the system, described by the Heisenberg model, is decoupled from the charge part of the system, which is described by spinless particles hopping around in a lattice. The statistics of these particles is, however, much more complex than the original spin-$\frac{1}{2}$ fermions, because the blocking of double occupation only occurs in site space. There is no reason why the momentum states couldn't have an occupation larger than one. The statistics of these particles, called hard-core fermions, has been studied in some detail and is called orthostatistics \cite{palev,mishra,kishore}.

When analysing the structure of the $\mathcal{IQG}$-2DM, it is found that there is no double occupation of lattice sites, {\it i.e.} on has correctly
\begin{equation}
\Gamma^0_{ii;jk}\rightarrow 0~.
\end{equation}
Therefore the error in the energy results from a huge overestimation of the hopping energy of the spinless particles, due to the faulty $\mathcal{IQG}$ description of the hopping of particles in a singly-occupied lattice. Based on this analysis we next derive two constraints that amend the pathological behaviour of $\mathcal{IQG}$ in the strong-correlation limit.
\subsubsection{The non-linear hopping constraint}
The first constraint imposes a bound on the hopping energy through knowledge of the exact result $T^\infty_{\text{min}}$ in the $U\rightarrow\infty$ limit  and the doubly-occupied fraction of particles on the lattice.

We can divide the full $N$-particle Hilbert space into orthogonal subspaces characterized by a fixed number $n$ of \emph{doubly occupied} sites, {\it i.e.} the eigenspaces of:
\begin{equation}
\hat{n} = \sum_ia^\dagger_{i\uparrow}a^\dagger_{i\downarrow}a_{i\downarrow}a_{i\uparrow}~.
\end{equation}
An arbitrary wave function can now be decomposed as:
\begin{equation}
\ket{\Psi}= \sum_n\ket{\Psi^{(n)}}~,
\end{equation}
in which the $\ket{\Psi^{(n)}}$ are not necessarily normalized. The expectation value of the number of doubly occupied sites is given by:
\begin{align}
\langle n\rangle = \sum_i\Gamma_{i\uparrow i\downarrow;i\uparrow i\downarrow} =& \sum_n n\braket{\Psi^{(n)}}{\Psi^{(n)}}\\
\label{beta_gp}\geq& \sum_{n\neq0}\braket{\Psi^{(n)}}{\Psi^{(n)}}\\
=& 1 - \braket{\Psi^{(0)}}{\Psi^{(0)}}~.
\end{align}
From this we can derive a lower bound for the occupation of singly-occupied space, {\it i.e.}
\begin{equation}
\braket{\Psi^{(0)}}{\Psi^{(0)}} \geq 1 - \overline{\overline{\Gamma}}_P~,
\end{equation}
in which we have introduced the \emph{pair trace} symbol:
\begin{equation}
\overline{\overline{\Gamma}}_P = \sum_i\Gamma_{i\uparrow i\downarrow;i\uparrow i\downarrow}~.
\end{equation}
Now take an arbitrary wave function and decompose it into its singly-occupied component and the remainder:
\begin{equation}
\ket{\Psi} = \alpha\ket{\Psi^{\text{sos}}} + \beta\ket{\Psi^\perp}~,
\end{equation}
in which
\begin{equation}
\alpha\ket{\Psi^{\text{sos}}} = \ket{\Psi^{(0)}}~.
\end{equation}
The expectation value of the hopping term $\hat{T}$ in an arbitrary wave function can then be expressed as:
\begin{equation}
\mathrm{Tr}~\Gamma T^{(2)} = |\alpha|^2\bra{\Psi^{\text{sos}}}\hat{T}\ket{\Psi^{\text{sos}}} + |\beta|^2\bra{\Psi^\perp}\hat{T}\ket{\Psi^\perp} + (\alpha^*\beta + \beta^*\alpha)\bra{\Psi^{\text{sos}}}\hat{T}\ket{\Psi^\perp}~.
\label{hopp_exp}
\end{equation}
Both the lowest eigenvalue $T^0_\text{min}$ of the hopping operator on the full space (see Section~\ref{hub_sym} and in Fig.~\ref{hubbard_fig}), as well as the lowest eigenvalue $T^\infty_{\text{min}}$ on singly-occupied space (see Eq.~(\ref{1dhub_sclim})) are known. Being the lowest eigenvalues the following inequalities hold:
\begin{equation}
T^0_{\text{min}} \leq \bra{\Psi^{\perp}}\hat{T}\ket{\Psi^{\perp}}~, \qquad\text{and}\qquad T^\infty_{\text{min}} \leq\bra{\Psi^{\text{sos}}}\hat{T}\ket{\Psi^{\text{sos}}}~. 
\end{equation}
Eq.~(\ref{hopp_exp}) can now be rewritten as an inequality:
\begin{align}
\nonumber\mathrm{Tr}~\Gamma T \geq& |\alpha|^2 T^\infty_{\text{min}} + |\beta|^2 T^0_{\text{min}} - 2|\alpha||\beta||T^0_{\text{min}}|\\
=& |\alpha|^2 T^\infty_{\text{min}} + |\beta|^2 T^0_{\text{min}} + 2|\alpha||\beta|T^0_{\text{min}} = f(\beta)~.
\label{ineq_beta}
\end{align}
The second line in Eq.~(\ref{ineq_beta}) follows from the fact that the lowest eigenvalue $T^0_\text{min}$ is always negative. Eq.~(\ref{beta_gp}) implies that $|\beta|^2 \leq \overline{\overline{\Gamma}}_P$, so the inequality (\ref{ineq_beta}) still holds upon replacing $|\beta|^2$ by $\overline{\overline{\Gamma}}_P$, as long as $f$ is a decreasing function. It follows that (\ref{beta_gp}) can be formulated as a non-linear inequality constraint which can be expressed solely in terms of the 2DM:
\begin{equation}
\mathrm{Tr}~\Gamma T \geq (1 - \overline{\overline{\Gamma}}_P)T^\infty_\text{min} + T^0_\text{min}\left(\overline{\overline{\Gamma}}_P + 2\sqrt{\overline{\overline{\Gamma}}_P \left(1 - \overline{\overline{\Gamma}}_P\right)}\right) = f\left(\overline{\overline{\Gamma}}_P\right)~.
\end{equation}
\begin{figure}
\centering
\includegraphics[scale=0.7]{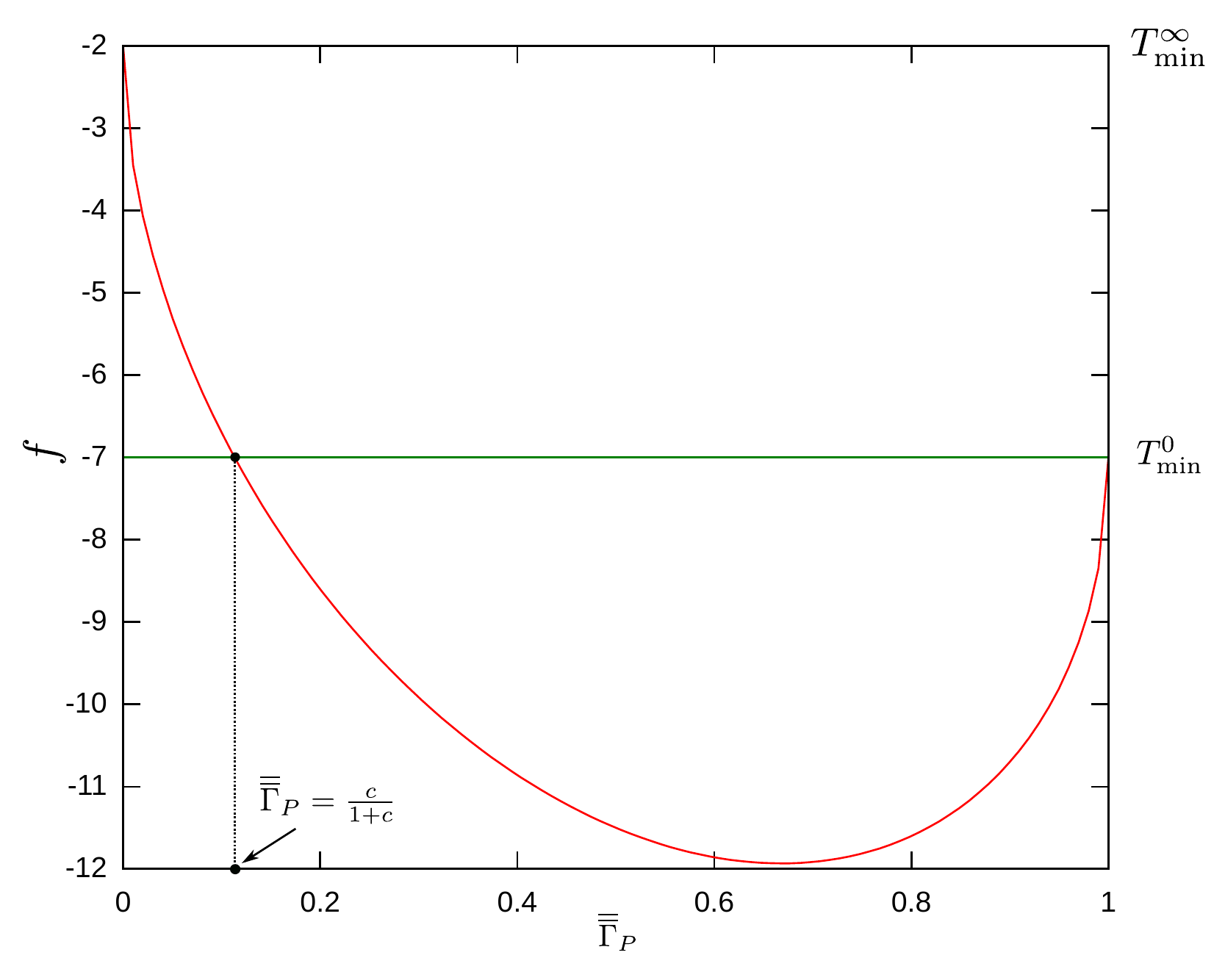}
\caption{\label{nlcon} The inequality curve $f\left(\overline{\overline{\Gamma}}_P\right)$ for a 6-site lattice with 5 particles.}
\end{figure}
This constraint is necessary as long as $f\left(\overline{\overline{\Gamma}}_P\right) \geq T^0_\text{min}$, {\it i.e.} when 
\begin{equation}
\overline{\overline{\Gamma}}_P \leq \frac{c}{1+c}~, \qquad\text{with}\qquad c=\left(\frac{T^\infty_\text{min} - T^0_\text{min}}{2T^0_\text{min}}\right)^2~.
\end{equation}
In Fig.~\ref{nlcon} an example is given for the inequality curve $f\left(\overline{\overline{\Gamma}}_P\right)$ for a 6-site lattice with 5 particles, in which case $T^0_\text{min} =-7$ and $T^\infty_\text{min} = -2$. Note that when $\overline{\overline{\Gamma}}_P\rightarrow 0$  the correct strong-interaction limit is restored. For larger values the curve quickly descends to the point where it crosses $T^0_\text{min}$, after which the constraint becomes redundant.

We have implemented this constraint in the dual-only potential reduction program discussed in Section~\ref{pr_sdp}, as it is the most straightforward algorithm to which a non-linear constraint can be added. This potential in Eq.~(\ref{dual_only_pot}) (here expressed as a function of the 2DM $\Gamma$) simply acquires an extra term:
\begin{equation}
\phi(\Gamma) = \mathrm{Tr}~\Gamma H^{(2)} - t\ln\det Z(\Gamma) - t \ln{\left[\mathrm{Tr}~\Gamma T - f^*(\Gamma)\right]}~,
\label{nlpot}
\end{equation}
where
\begin{equation}
f^*(\Gamma) = \left\{
\begin{matrix}
(1 - \overline{\overline{\Gamma}}_P)T^\infty_\text{min} + T^0_\text{min}\left(\overline{\overline{\Gamma}}_P + 2\sqrt{\overline{\overline{\Gamma}}_P \left(1 - \overline{\overline{\Gamma}}_P\right)}\right)&\text{for}&\overline{\overline{\Gamma}}_P \leq \frac{c}{1+c}\\
\frac{2\mathrm{Tr}~\Gamma}{N(N-1)}T^0_\text{min}&\text{for}&\overline{\overline{\Gamma}}_P > \frac{c}{1+c}
\end{matrix}\right.~.
\end{equation}
The algorithm's workhorse is Newton's method, for which we need the gradient and Hessian of the potential in Eq.~(\ref{nlpot}). The gradient can be written in matrix form as in Eq.~(\ref{gradient}), with the addition of an extra term coming from the non-linear constraint:
\begin{equation}
\nabla\phi = \hat{P}_{\mathrm{Tr}}\left(H^{(2)} - t\sum_j \mathcal{L}_j^\dagger\left[\mathcal{L}_j\left(\Gamma\right)^{-1}\right] -\frac{t}{\left[\mathrm{Tr}~\Gamma T - f^*(\Gamma)\right]}\left[T-g(\Gamma)\mathbb{1}_P\right]\right)~,
\label{gradient_nl}
\end{equation}
where the matrix $\mathbb{1}_P$ defined in such a way that:
\begin{equation}
\mathrm{Tr}~\Gamma \mathbb{1}_P = \overline{\overline{\Gamma}}_P~.
\end{equation}
The non-linear gradient function $g$ reads:
\begin{equation}
g(\Gamma) = \left\{
\begin{matrix}
T^0_\text{min} -T^{\infty}_{\text{min}} + \frac{T^0_\text{min}}{\sqrt{\overline{\overline{\Gamma}}_P\left(1 - \overline{\overline{\Gamma}}_P\right)}}\left(1 - 2\overline{\overline{\Gamma}}_P\right)& \text{for}&  \overline{\overline{\Gamma}}_P \leq \frac{c}{1+c}\\
0& \text{for}&  \overline{\overline{\Gamma}}_P > \frac{c}{1+c}
\end{matrix}
\right.
~.
\end{equation}
As in Eq.~(\ref{hessian}), the Hessian can be seen as a map from traceless two-particle matrix space on itself. Compared to Eq.~(\ref{hessian}) there are two extra terms coming from the non-linear constraint. The action of the Hessian map on a traceless two-particle matrix $\Delta$ is given by:
\begin{align}
\nonumber\mathcal{H}\Delta =& t \hat{P}_{\mathrm{Tr}}\left[\sum_k \mathcal{L}^\dagger_k\left(\mathcal{L}_k(\Gamma)^{-1}\mathcal{L}_k(\Delta)\mathcal{L}_k(\Gamma)^{-1}\right)
+ \left(\frac{\mathrm{Tr}~T\Delta -g(\Gamma)\overline{\overline{\Delta}}_P}{\left[\mathrm{Tr}~T\Gamma - f^*(\Gamma)\right]^2}\right)\left[T-g(\Gamma)\mathbb{1}_P\right]\right.\\
&\left.\qquad\qquad\qquad\qquad\qquad+\left(\frac{h(\Gamma)\overline{\overline{\Delta}}_P}{\mathrm{Tr}~T\Gamma - f^*(\Gamma)}\right) \mathbb{1}_P
\right]~,
\label{hessian_nl}
\end{align}
with the non-linear Hessian function $h$ defined as:
\begin{equation}
h(\Gamma) = 
\left\{
\begin{matrix}
-\left[\frac{T^0_\text{min}}{2\left(\overline{\overline{\Gamma}}_P\left(1 - \overline{\overline{\Gamma}}_P\right)\right)^{\frac{3}{2}}}\left(1-2\overline{\overline{\Gamma}}_P\right)^2 + \frac{2T^0_\text{min}}{\sqrt{\overline{\overline{\Gamma}}_P\left(1-\overline{\overline{\Gamma}}_P\right)}}\right]& \text{for}&  \overline{\overline{\Gamma}}_P \leq \frac{c}{1+c}\\
0 &\text{for}&  \overline{\overline{\Gamma}}_P > \frac{c}{1+c}
\end{matrix}
\right.~.
\end{equation}
After the Newton step $\Delta$ has been determined, we perform a search in the direction of $\Delta$ to find the optimal stepsize $\alpha$. In Section~\ref{pr_sdp} we showed that this can be done rapidly using a bisection algorithm to determine the roots of the scalar function in Eq.~(\ref{lsfunc}). An extra term is added to this function through the inclusion of the non-linear constraint, leading to the new line-search function:
\begin{equation}
\nabla_\alpha\phi(\alpha) = \mathrm{Tr}~\Delta H^{(2)} - t \sum_j\left(\sum_i\frac{\lambda_i^{\mathcal{L}_j}}{1 + \alpha\lambda^{\mathcal{L}_j}_i}\right) -t\left(\frac{\mathrm{Tr}~T\Delta - \frac{\partial f^*}{\partial \alpha}\left(\Gamma + \alpha\Delta\right)}{{\mathrm{Tr}~\left(\Gamma+\alpha\Delta\right)T - f^*\left(\Gamma + \alpha\Delta\right)}}\right)~,
\label{lsfunc_nl}
\end{equation}
where the derivative of $f^*$ with respect to $\alpha$ is:
\begin{equation}
\frac{\partial f^*}{\partial \alpha}\left(\Gamma + \alpha \Delta\right) = \left(T_\text{min} - T^0_\text{min}\right)\bar{\bar{\Delta}}_P + \frac{T_\text{min}\bar{\bar{\Delta}}_P}{\sqrt{\left(\bar{\bar{\Gamma}}_P + \alpha\bar{\bar{\Delta}}_P\right)\left(1-\bar{\bar{\Gamma}}_P - \alpha\bar{\bar{\Delta}}_P \right)}}\left[1-2\left(\bar{\bar{\Gamma}}_P -\bar{\bar{\Delta}}_P\right)\right]~.
\end{equation}
\begin{table}
\centering
\begin{tabular}{|c|cccc|cccc|}
\hline
&\multicolumn{4}{c}{$L=6$\qquad$N=5$}\vline&\multicolumn{4}{c}{$L=10$\qquad$N=9$}\vline\\
\hline
$U$ & $\mathcal{IQG}$ & $\mathcal{IQGT}$ &non-lin&exact& $\mathcal{IQG}$ & $\mathcal{IQGT}$ & non-lin&exact\\
\hline
50 & -3.55 & -2.29 & -3.06 & -2.20& -4.89 & -2.54 &-4.64 & -2.46 \\
100 & -3.49 & -2.15 & -2.51& -2.08 & -4.77  & -2.27 &-3.66 & -2.22\\
1000 & -3.44 & -2.03 & -2.05& -2.01 & -4.67 & -2.03 & -2.15 &-2.02\\
\hline
\end{tabular}
\caption{\label{energy_nl_obhf}The ground-state energy of a 6-site lattice with 5 particles (left), and a 10-site lattice with 9 particles (right), for $U = 50,\ 100$ and 1000, exact results compared with v2DM results using $\mathcal{IQG}$ and $\mathcal{IQGT}$ results, and $\mathcal{IQG}$ results with the non-linear constraint.}
\end{table}

The results of such a calculation, on a 6-site lattice with 5 particles and on a 10-site lattice with 9 particles for large values of $U$ are shown in Table~\ref{energy_nl_obhf}, the non-linear hopping constraint results are compared to $\mathcal{IQG}$, $\mathcal{IQGT}$ and exact results. For the 6-site lattice the exact results were calculated using diagonalization, for the 10-site lattice this no longer possible and an MPS algorithm was used \cite{schollwock,verstraete,chan,sebastian}.  It can be seen that the non-linear constraint is succesful as the strong-correlation limit is correctly recovered. It is, however, also observed that the constraint only becomes active at very large values of $U > 30$, and as such does not improve the quality of the $\mathcal{IQG}$ energy curves for the $U$ range displayed in Fig.~\ref{energy_1dhub}. Calculations performed at larger lattice sizes indicate that the value of $U$ where the constraint becomes active increases with lattice size, which renders it even less valuable as a practical tool. In conclusion, we can say that the non-linear hopping constraint guarantees that the energy converges to the right value in the large-$U$ limit, which is what we set out to achieve. But it fails to improve the quality of the energy and 2DM at intermediate values of $U$.
\subsubsection{The Gutzwiller projection constraint}
In the previous Section we constructed a new constraint by looking at the symptom {\it i.e.} the energy in the strong-correlation limit is far below what it should be. We imposed an exact lower bound to the energy for a certain double occupation $\overline{\overline{\Gamma}}_P$. In this Section we derive a constraint that fixes the origin of the problem, being the faulty description of spinless particles hopping on a lattice. The creation and annihilation of particles on a singly-occupied lattice can be described by the so-called Gutzwiller projection operators \cite{gutz_1,gutz_2}:
\begin{eqnarray}
\label{gutz_anni}g_{\alpha} &=& a_{\alpha}\left(1 - a^\dagger_{\bar{\alpha}}a_{\bar{\alpha}}\right),\\
\label{gutz_create}g^\dagger_{\alpha} &=& \left(1 - a^\dagger_{\bar{\alpha}}a_{\bar{\alpha}}\right)a^\dagger_{\alpha}~,
\end{eqnarray}
where $\alpha$ and $\bar{\alpha}$ are single-particle indices on the same site with opposite spin. One can see in Eq.~(\ref{gutz_create}) that a particle in the state $\alpha$ can only by created if there is no particle present in the opposite spin state $\bar{\alpha}$; if there is, the Gutzwiller operators annihilate the state. In analogy with the necessary and sufficient conditions on the 1DM derived in Section~\ref{n_rep_1dm}, we now propose that every Hamiltonian that is expressed as a first-order operator in the Gutzwiller creation and annihilation operators will be correctly optimized if the following conditions are fulfilled:
\begin{align}
\label{gutz_rho}\rho^G \succeq 0 \qquad\text{with}\qquad \rho^G_{\alpha\beta} = \bra{\Psi^N}g^\dagger_{\alpha}g_{\beta}\ket{\Psi^N}~,\\
\label{gutz_q}q^G \succeq 0 \qquad\text{with}\qquad q^G_{\alpha\beta} = \bra{\Psi^N}g_{\alpha}g^\dagger_{\beta}\ket{\Psi^N}~.
\end{align}
It is clear from Eqs. (\ref{gutz_anni}) and (\ref{gutz_create}), that these constraints are third-order operators, in the underlying fermion creation/annihilation operators. This explains why three-index constraints are needed to describe the strong-correlation limit, and why the two-index conditions fail. By imposing the Gutzwiller conditions (\ref{gutz_rho}) and (\ref{gutz_q}) we introduce the relevant correlations present when $U$ increases, without having to resort to the full three-index conditions. The problem is that the Gutzwiller conditions cannot be expressed as a function of the 2DM, so if we want to stick to a 2DM framework, we have to manipulate the conditions somehow. Using the same approach as for the $\mathcal{T}$ conditions, we can anticommute the two Gutzwiller conditions to reduce the rank of the operators by one:
\begin{equation}
G\succeq 0\qquad\text{with}\qquad G_{\alpha\beta} = 
\left\{g^\dagger_\alpha,g_\beta\right\} = \delta_{\alpha\beta}(1 - a^\dagger_{\bar{\alpha}}a_{\bar{\alpha}})~.
\end{equation}
This is obviously a trivial constraint, in fact it is already included in the $\mathcal{IQG}$ conditions. If we expand the Gutzwiller conditions we can express them as a function of first-, second- and third-order operators:
\begin{eqnarray}
\label{rho_G_2DM}\rho^{\text{G}}_{\alpha\beta} &=& \rho_{\alpha\beta} - \Gamma_{\alpha\bar{\beta};\beta\bar{\beta}} - \Gamma_{\alpha\bar{\alpha};\beta\bar{\alpha}} + \bra{\Psi^N}a^\dagger_\alpha a^\dagger_{\bar{\alpha}} a_{\bar{\alpha}} a^\dagger_{\bar{\beta}} a_{\bar{\beta}} a_{\beta}\ket{\Psi^N}~,\\
\label{q_G_2DM}q^\text{G}_{\alpha\beta} &=& q_{\alpha\beta} - \mathcal{G}(\Gamma)_{\bar{\beta}\alpha;\bar{\beta}\beta} - \mathcal{G}(\Gamma)_{\bar{\alpha}\alpha;\bar{\alpha}\beta} +  \bra{\Psi^N}a^\dagger_{\bar{\alpha}} a_{\bar{\alpha}} a_{\alpha} a^\dagger_{\beta}a^\dagger_{\bar{\beta}} a_{\bar{\beta}}\ket{\Psi^N}~.
\end{eqnarray}
One can reduce the order of these constraints by \emph{only} anticommutating the third-order part, which leads to the following two matrix functions:
\begin{eqnarray}
\tilde{\rho}^{\text{G}}\left(\Gamma\right)_{\alpha\beta} &=& \rho_{\alpha\beta} - \Gamma_{\alpha\bar{\beta};\beta\bar{\beta}} - \Gamma_{\alpha\bar{\alpha};\beta\bar{\alpha}} + \mathcal{T}_2(\Gamma)_{\alpha\bar{\alpha}\bar{\alpha};\beta\bar{\beta}\bar{\beta}}~,\\
\tilde{q}^\text{G}\left(\Gamma\right)_{\alpha\beta} &=& q_{\alpha\beta} - \mathcal{G}(\Gamma)_{\bar{\beta}\alpha;\bar{\beta}\beta} - \mathcal{G}(\Gamma)_{\bar{\alpha}\alpha;\bar{\alpha}\beta} + \mathcal{T}_2(\Gamma)_{\alpha\bar{\alpha}\bar{\alpha};\beta\bar{\beta}\bar{\beta}}~.
\end{eqnarray}
These have to be positive-semidefinite, since a positive matrix was added to the original Gutzwiller conditions. Expressing $\mathcal{T}_2$ as a function of the 2DM one can rewrite:
\begin{eqnarray}
\tilde{\rho}^{\text{G}}\left(\Gamma\right)_{\alpha\beta} &=& \rho_{\alpha\beta} - \Gamma_{\alpha\bar{\beta};\beta\bar{\beta}} - \Gamma_{\alpha\bar{\alpha};\beta\bar{\alpha}} + \delta_{\alpha\beta}\rho_{\bar{\alpha}\bar{\alpha}}~,\\
\tilde{q}^{\text{G}}\left(\Gamma\right)_{\alpha\beta} &=& \Gamma_{\bar{\beta}\beta;\bar{\beta}\alpha} + \Gamma_{\bar{\alpha}\beta;\bar{\alpha}\alpha} - \rho_{\alpha\beta} + \delta_{\alpha\beta}\left( \frac{2\mathrm{Tr}~\Gamma}{N(N-1)} - \rho_{\bar{\alpha}\bar{\alpha}}\right)~.
\end{eqnarray}
These matrix positivity conditions are easily included in any of the algorithms described in Chapter~\ref{SDP}, and at a very small computational cost as the dimension of the matrices is the same as that of the 1DM. Taking translational invariance into account, these constraints even reduce to $L$ linear constraints, which can straightforwardly be included in the SDP algorithms from Chapter~\ref{SDP}. 
\begin{figure}
\centering
$
\begin{array}{c}
\includegraphics[scale=0.7]{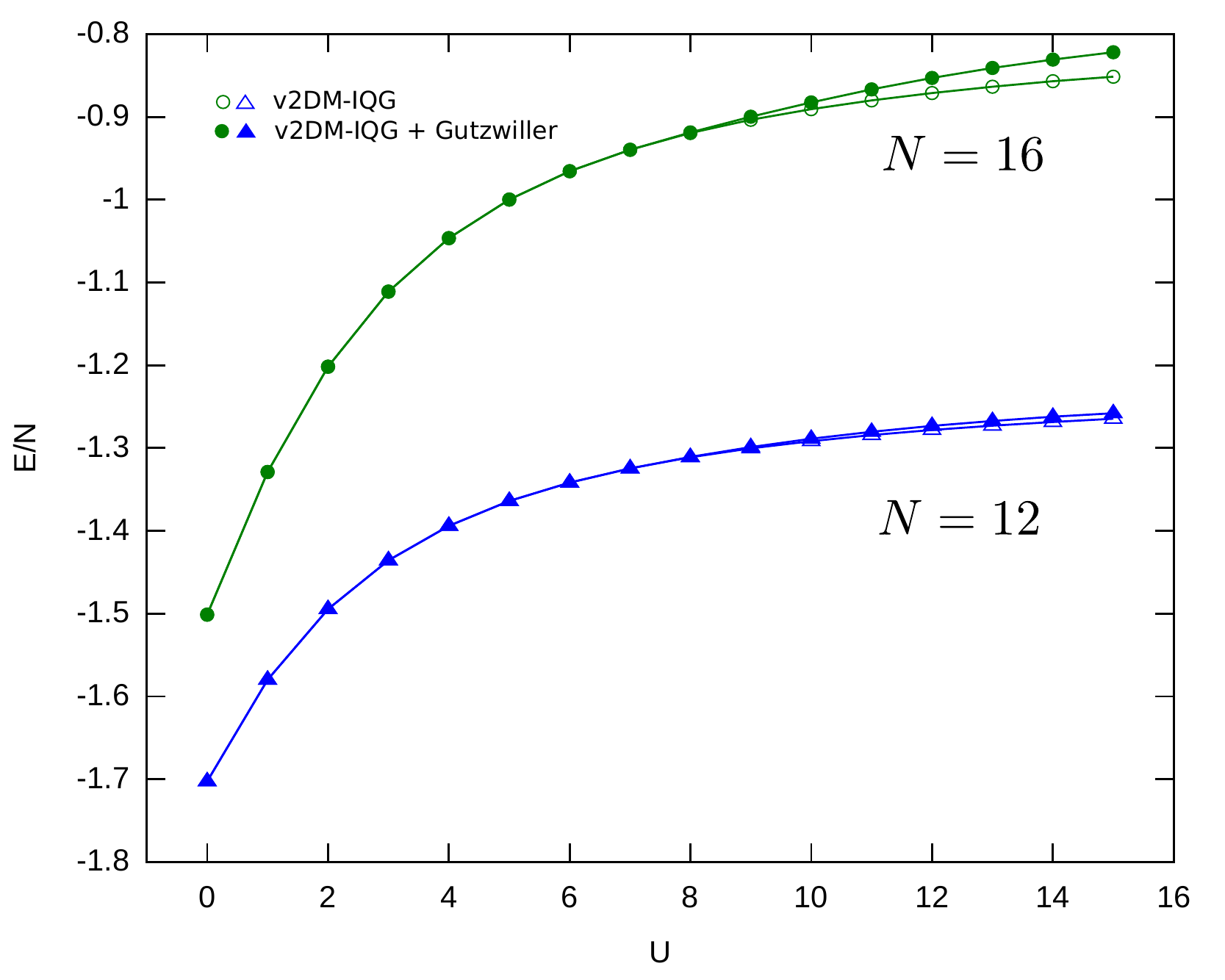}\\
\includegraphics[scale=0.7]{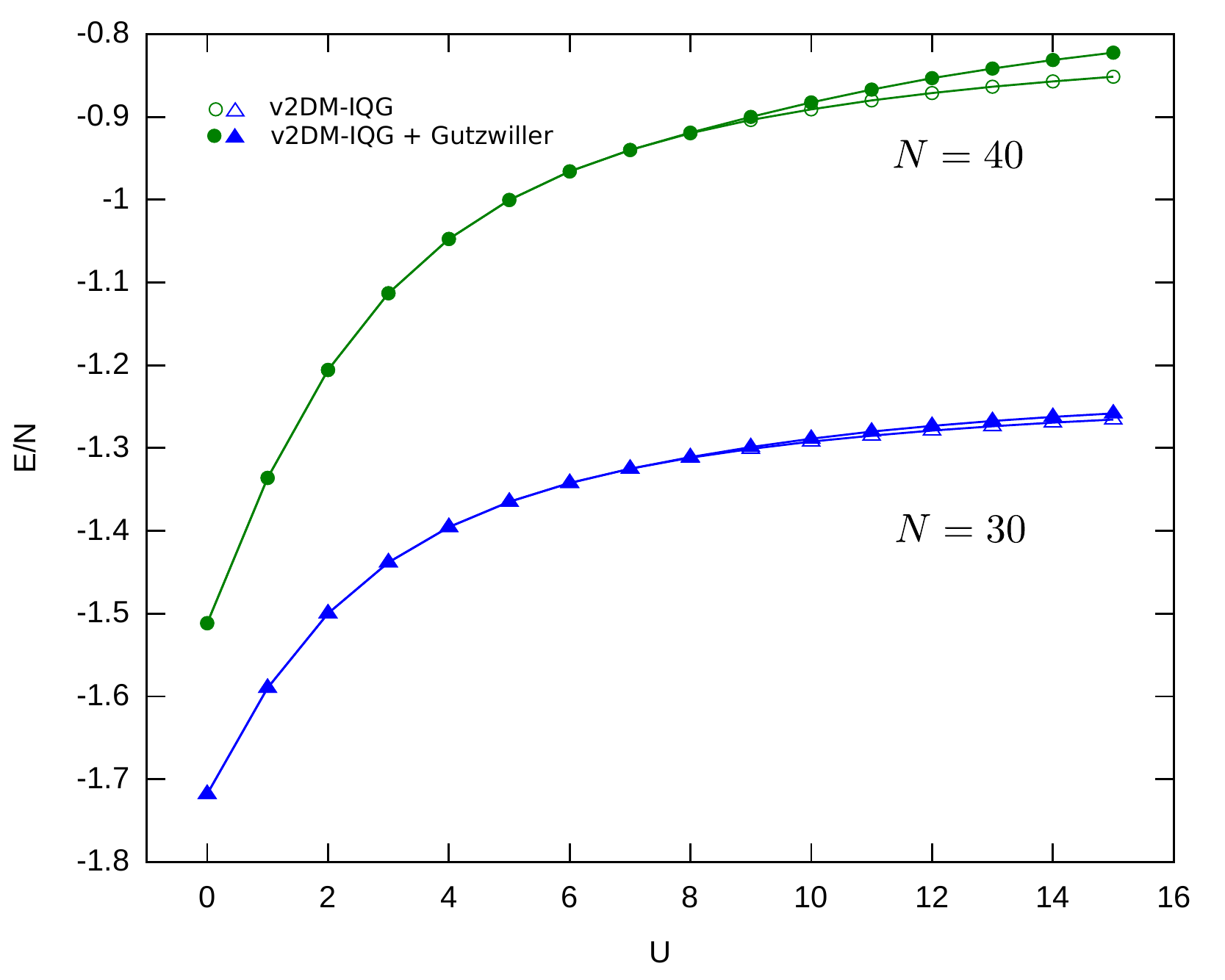}
\end{array}
$
\caption{\label{gutz_energy_1dhub} Ground-state energy per particle as a function of on-site repulsion $U$ of the Hubbard model, for a 20-site (top) and 50-site (bottom) lattice at $\frac{4}{10}$  and $\frac{3}{10}$ filling. A comparison is made between an optimization using only $\mathcal{IQG}$ and one using $\mathcal{IQG}$ and the Gutzwiller conditions.}
\end{figure}
The new constraints were implemented and added to the translationally invariant, parity-symmetric primal-dual program. Calculations were again performed for lattice sizes of 20 and 50, with $\frac{3}{10}$  and $\frac{4}{10}$ filling, the results of which are shown in Fig.~\ref{gutz_energy_1dhub}. It is seen that the constraints clearly include something that is missing in the $\mathcal{IQG}$ conditions, because they become active at reasonably small values of $U > 5$. The amount of improvement when including the Gutzwiller conditions, however, is quite disappointing. Moreover, calculations at higher values of $U$ indicate that the $U\rightarrow\infty$ limit is not restored by adding these conditions. The reason for this is probably the addition of the anticommutator term to the third-order terms in Eqs.~(\ref{rho_G_2DM}) and (\ref{q_G_2DM}). In conclusion we can say that we have identified the origin of the failure of the $\mathcal{IQG}$ conditions, and derived constraints on a subspace of three-particle space that have to be included to fix the limit. If we manipulate these constraints so that they can be expressed as a function of the 2DM, the resulting conditions improve the results, but are not strong enough to fix the strong-correlation limit. The only way out is to introduce an object beyond the 2DM, which is the subject of the next Chapter.

\chapter{\label{v2.5DM}Restoring the strong-correlation limit in the Hubbard model}
\markboth{CHAPTER 6. RESTORING THE STRONG-CORRELATION LIMIT}{CHAPTER 6. RESTORING THE STRONG-CORRELATION LIMIT}
Previously (Section~\ref{1dhub_sec}) we discussed the v2DM study of the one-dimensional Hubbard model. We found that the two-index conditions fail to describe the strong-correlation limit, ({\it i.e.} strong on-site repulsion or $U\rightarrow\infty$), and that three-index conditions $\mathcal{T}_1$, $\mathcal{T}_2$ are needed to incorporate the physics in this limit. We identified the origin of the two-index failure and tried to fix it, cheaply, by applying constraints on a subspace of three-particle space, the so-called Gutzwiller conditions. This was unsuccesful because we approximated these constraints in order to express them as a function of the 2DM. In this Chapter we go beyond the 2DM and identify the minimal object from which both the genuine Gutzwiller conditions (Eqs. (\ref{rho_G_2DM}) and (\ref{q_G_2DM})), \emph{and} the two-index conditions can be simultaneously derived. We establish that the strong-correlation limit is restored, without resorting to the computationally heavy three-index constraints.
\section{An intermediary object: the 2.5DM}
The most compact object from which both the Gutzwiller conditions and the two-index conditions can be derived is the 3DM with one index diagonal, which we call the 2.5DM:
\begin{equation}
W^\lambda_{\alpha\beta;\gamma\delta} = \sum_i w_i \bra{\Psi^N_i}a^\dagger_\alpha a^\dagger_\beta a^\dagger_\lambda a_\lambda a_\delta a_\gamma \ket{\Psi^N_i}~.
\label{2.5DM_unc}
\end{equation}
Just as the 2DM, the 2.5DM can be used as the basic variable in a variational scheme. By permutation of the indices one can derive 6 conditions on this object, called the lifting conditions \cite{mazziotti,hammond,mazz_book}. Instead of using Eq.~(\ref{2.5DM_unc}) we preferred to use the slightly larger object:
\begin{equation}
W^{l[\sigma_l;\sigma_l']}_{a\sigma_ab\sigma_b;c\sigma_cd\sigma_d} = \sum_i w_i \bra{\Psi^N_i}a^\dagger_{a\sigma_a}a^\dagger_{b\sigma_b}a^\dagger_{l\sigma_l}a_{l\sigma_l'}a_{d\sigma_d}a_{c\sigma_c}\ket{\Psi^N_i}~,
\label{2.5DM_spin_unc}
\end{equation}
where the third index is diagonal in the site index only, and not in the spin index. This is a more desirable object since it captures more correlations. In addition we found that the extra flexibility makes a significant difference in the final energy prediction. It is also advantageous compared to the 2.5DM in Eq.~(\ref{2.5DM_unc}) because a spin-coupled version can be constructed, defined as:
\begin{equation}
W^{l}|^{S(S_{ab};S_{cd})}_{ab;cd} = \sum_i w_i \frac{1}{[\mathcal{S}]^2}\sum_\mathcal{M} \bra{\Psi^N_{\mathcal{SM},i}}{B^\dagger}^S_{ab(S_{ab})l}~{B}^S_{cd(S_{cd})l}\ket{\Psi^N_{\mathcal{SM},i}}~,
\label{2.5DM_sc}
\end{equation}
with $B^\dagger$ the three-particle creation operator as defined in Eq.~(\ref{dp_sc_op}):
\begin{equation}
{B^\dagger}^S_{ab(S_{ab})c} = \frac{1}{\sqrt{(ab)}}\left[\left[a^\dagger_a\otimes a^\dagger_b\right]^{S_{ab}}\otimes a^\dagger_c\right]^S~.
\label{2.5DM_op}
\end{equation}
In the same manner as for the spin discussion in Chapter~\ref{symmetry}, a spin-averaged ensemble is used to decompose the 2.5DM into blocks with different three-particle spin $S$. Because of this, the spin-coupled 2.5DM can be optimized faster than the spin-uncoupled 2.5DM in Eq.~(\ref{2.5DM_unc}). The spin-coupled 2.5DM in Eq.~(\ref{2.5DM_sc}) can only be written as a function of the 2.5DM in Eq.~(\ref{2.5DM_spin_unc}), since off-diagonal spin terms are needed to produce the correct coupling, as can be seen from expanding Eq.~(\ref{2.5DM_sc}):
\begin{align}
\nonumber W^{l}|^{S(S_{ab};S_{cd})}_{ab;cd} =& \frac{1}{\sqrt{(ab)(cd)}}\sum_{\sigma_a\sigma_b}\sum_{\sigma_lM_{ab}} \braket{\frac{1}{2}\sigma_a\frac{1}{2}\sigma_b}{S_{ab}M_{ab}}\braket{S_{ab}M_{ab}\frac{1}{2}\sigma_l}{SM}\\
&\qquad\qquad\sum_{\sigma_c\sigma_d}\sum_{M_{cd}\sigma_l'}\braket{\frac{1}{2}\sigma_c\frac{1}{2}\sigma_d}{S_{cd}M_{cd}}\braket{S_{cd}M_{cd}\frac{1}{2}\sigma_l'}{SM}~W^{l[\sigma_l;\sigma_l']}_{a\sigma_ab\sigma_b;c\sigma_cd\sigma_d}~.
\end{align}
In practice we store $L$ blockmatrices, one for each site. Every blockmatrix consist of two blocks, one with spin-$\frac{1}{2}$ and one with spin-$\frac{3}{2}$, with dimensions and degeneracies given in Table~\ref{2.5DM_basis}.
\begin{table}
\centering
\begin{tabular}{|ccccc|}
\hline
$S$&$S_{ab}$&sp ordering&dim&deg\\
\hline
$\frac{1}{2}$ & 0 & $a\leq b$ & $\frac{L(L+1)}{2}$ & 2\\[1pt]
$\frac{1}{2}$ & 1 & $a < b$ & $\frac{L(L-1)}{2}$ & 2\\[1pt]
$\frac{3}{2}$ & 1 & $a < b$ & $\frac{L(L-1)}{2}$ & 4\\
\hline
\end{tabular}
\caption{\label{2.5DM_basis}The dimensions, degeneracies and orbital ordering of the $S=\frac{1}{2}$ and $S=\frac{3}{2}$ blocks appearing in the 2.5DM object defined in Eq.~(\ref{2.5DM_sc}).}
\end{table}
There is another complication that arises when using the 2.5DM, which is that because of the antisymmetry between some single-particle indices in the 2.5DM, relations exist between elements in different blocks. These relations have to be imposed as extra consistency conditions.
\subsection{\label{consistent}The consistency conditions}
To get insight into the problem, we introduce the consistency constraints for the spin-diagonal 2.5DM in Eq.~(\ref{2.5DM_unc}). There, $2L$ blocks are stored, all of dimension ${(2L-1)(L-1)}$, with $\alpha < \beta$ and both $\alpha$ and $\beta$ different from $\lambda$. When $\beta = \delta$ in Eq.~(\ref{2.5DM_unc}) it follows that:
\begin{equation}
W^\lambda_{\alpha\beta;\gamma\beta} = W^\beta_{\alpha\lambda;\gamma\lambda}~.
\end{equation}
When $\beta=\delta$ \emph{and} $\alpha=\gamma$ there are two equalities to take into account:
\begin{equation}
W^\lambda_{\alpha\beta;\alpha\beta} = W^\beta_{\alpha\lambda;\alpha\lambda} = W^\alpha_{\lambda\beta;\lambda\beta}~.
\end{equation}
These equalities have to be included in the SDP program using the method of linear equality constraints explained in Section~\ref{pr_sdp_le}. We have to find an operator that projects on the right subspace (see Eq.~(\ref{P_f})). In the present case the projection operator is very simple: for the case of one equal index pair,
\begin{equation}
W^\lambda_{\alpha\beta;\gamma\beta} \rightarrow \frac{1}{2}\left[W^\lambda_{\alpha\beta;\gamma\beta} + W^\beta_{\alpha\lambda;\gamma\lambda}\right]~,
\end{equation}
and for the case of two equal index pairs,
\begin{equation}
W^\lambda_{\alpha\beta;\alpha\beta} \rightarrow \frac{1}{3}\left[W^\lambda_{\alpha\beta;\alpha\beta} + W^\beta_{\alpha\lambda;\alpha\lambda} + W^\alpha_{\lambda\beta;\lambda\beta}\right]~.
\end{equation}

For the spin-coupled 2.5DM in Eq.~(\ref{2.5DM_sc}) the consistency conditions become a lot more complicated, for two reasons. First, the first and/or second lower index can be equal to the upper index defining the block, which leads to relations between elements in the \emph{same} block. Second, exchanging upper and lower indices is not straightforward, but involves recoupling of the intermediate spin, as was the case for the $\mathcal{T}_1$  conditions (See Eq.~(\ref{recoupling_T1})). This merely leads to a lot of bookkeeping problems. In what follows we list one example for every class of possible symmetries.
\subsubsection{Case of one equality}
When there is one equality between the indices, two different types of projection appear:
\paragraph{$\mathbf{a = l}$~:}
when one lower index is equal to the block index $l$; the $W$ matrix must satisfy the identities:
\begin{equation}
W^l|^{S(S_{ab};S_{cd})}_{lb;cd} = [S_{ab}]\sum_{S_{lb}}[S_{lb}] 
\left\{
\begin{matrix}
S&\frac{1}{2}&S_{lb}\\
\frac{1}{2}&\frac{1}{2}&S_{ab}
\end{matrix}
\right\}
W^l|^{S(S_{lb};_{cd})}_{lb;cd}~,
\end{equation}
and the corresponding projection operator reads:
\begin{equation}
W^l|^{S(S_{ab};S_{cd})}_{lb;cd} \rightarrow \frac{1}{2}\left[W^l|^{S(S_{ab};S_{cd})}_{lb;cd} + [S_{ab}]\sum_{S_{lb}}[S_{lb}]\left\{
\begin{matrix}
S&\frac{1}{2}&S_{lb}\\
\frac{1}{2}&\frac{1}{2}&S_{ab}
\end{matrix}
\right\}
W^l|^{S(S_{lb};_{cd})}_{lb;cd} \right]~.
\end{equation}
\paragraph{$\mathbf{a = c}$~:}
when two lower indices are equal, consistency requires that:
\begin{equation}
W^l|^{S(S_{ab};S_{cd})}_{ab;ad} = [S_{ab}][S_{cd}]\sum_{S_{lb}S_{ld}}[S_{lb}] [S_{ld}]
\left\{
\begin{matrix}
S&\frac{1}{2}&S_{lb}\\
\frac{1}{2}&\frac{1}{2}&S_{ab}
\end{matrix}
\right\}
\left\{
\begin{matrix}
S&\frac{1}{2}&S_{ld}\\
\frac{1}{2}&\frac{1}{2}&S_{cd}
\end{matrix}
\right\}
W^a|^{S(S_{lb};_{ld})}_{lb;ld}~,
\end{equation}
which leads to the projection:
\begin{align}
\nonumber W^l|^{S(S_{ab};S_{cd})}_{ab;ad} \rightarrow& \frac{1}{2}\left[W^l|^{S(S_{ab};S_{cd})}_{lb;cd}+[S_{ab}][S_{cd}]\sum_{S_{lb}S_{ld}}[S_{lb}] [S_{ld}]\right.\\
&\left.\qquad\qquad\qquad\qquad
\left\{
\begin{matrix}
S&\frac{1}{2}&S_{lb}\\
\frac{1}{2}&\frac{1}{2}&S_{ab}
\end{matrix}
\right\}
\left\{
\begin{matrix}
S&\frac{1}{2}&S_{ld}\\
\frac{1}{2}&\frac{1}{2}&S_{cd}
\end{matrix}
\right\}
W^a|^{S(S_{lb};_{ld})}_{lb;ld}
\right]~.
\end{align}
\subsubsection{Case of two equalities}
Three different cases appear, where two equalities exist between the indices:
\paragraph{$\mathbf{a = c}$ and $\mathbf{c = l}$~:}
Now there are \emph{three} symmetry relationships the elements have to satisfy:
\begin{align}
\label{acl_1}W^l|^{S(S_{ab};S_{cd})}_{lb;ld} =& [S_{ab}]\sum_{S_{lb}} [S_{lb}]
\left\{
\begin{matrix}
S&\frac{1}{2}&S_{lb}\\
\frac{1}{2}&\frac{1}{2}&S_{ab}
\end{matrix}
\right\}
W^l|^{S(S_{lb};S_{cd})}_{lb;ld}~,\\
\label{acl_2}W^l|^{S(S_{ab};S_{cd})}_{lb;ld} =& [S_{cd}]\sum_{S_{ld}} [S_{ld}]
\left\{
\begin{matrix}
S&\frac{1}{2}&S_{ld}\\
\frac{1}{2}&\frac{1}{2}&S_{cd}
\end{matrix}
\right\}
W^l|^{S(S_{ab};S_{ld})}_{lb;ld}~,\\
\label{acl_3}W^l|^{S(S_{ab};S_{cd})}_{lb;ld} =& [S_{ab}][S_{cd}]\sum_{S_{lb}S_{ld}}[S_{lb}][S_{ld}]
\left\{
\begin{matrix}
S&\frac{1}{2}&S_{lb}\\
\frac{1}{2}&\frac{1}{2}&S_{ab}
\end{matrix}
\right\}
\left\{
\begin{matrix}
S&\frac{1}{2}&S_{ld}\\
\frac{1}{2}&\frac{1}{2}&S_{cd}
\end{matrix}
\right\}
W^l|^{S(S_{lb};S_{ld})}_{lb;ld}~.
\end{align}
The corresponding projection reads:
\begin{equation}
W^l|^{S(S_{ab};S_{cd})}_{lb;ld} \rightarrow \frac{1}{4}\left[W^l|^{S(S_{ab};S_{cd})}_{lb;ld}+ (\ref{acl_1})_r + (\ref{acl_2})_r + (\ref{acl_3})_r\right]~,
\end{equation}
where the notation $(\ref{acl_1})_r$ is short for the right-hand side of Eq.~(\ref{acl_1}).
\paragraph{$\mathbf{a = c}$ and $\mathbf{b = l}$~:}
two different relationships have to be satisfied:
\begin{align}
\label{acbl_1}W^l|^{S(S_{ab};S_{cd})}_{al;ad} =& [S_{ab}][S_{cd}]\sum_{S_{lb}S_{ld}}[S_{lb}][S_{ld}]
\left\{
\begin{matrix}
S&\frac{1}{2}&S_{lb}\\
\frac{1}{2}&\frac{1}{2}&S_{ab}
\end{matrix}
\right\}
\left\{
\begin{matrix}
S&\frac{1}{2}&S_{ld}\\
\frac{1}{2}&\frac{1}{2}&S_{cd}
\end{matrix}
\right\}
W^a|^{S(S_{lb};S_{ld})}_{ll;ld}~,\\
\label{acbl_2}W^l|^{S(S_{ab};S_{cd})}_{al;ad} =& [S_{ab}](-1)^{S_{ab}}\sum_{S_{al}}[S_{al}](-1)^{S_{al}}
\left\{
\begin{matrix}
S&\frac{1}{2}&S_{al}\\
\frac{1}{2}&\frac{1}{2}&S_{ab}
\end{matrix}
\right\}
W^l|^{S(S_{al};S_{cd})}_{al;ad}~,
\end{align}
which are imposed using the following projection:
\begin{equation}
W^l|^{S(S_{ab};S_{cd})}_{al;ad} \rightarrow \frac{1}{3}\left[W^l|^{S(S_{ab};S_{cd})}_{al;ad}+ (\ref{acbl_1})_r + (\ref{acbl_2})_r\right]~.
\end{equation}
\paragraph{$\mathbf{a = c}$ and $\mathbf{b = d}$~:}
for the last type two equalities have to be satisfied:
\begin{align}
\label{acbd_1}W^l|^{S(S_{ab};S_{cd})}_{ab;ab} =& [S_{ab}][S_{cd}]\sum_{S_{lb}S_{ld}}[S_{lb}][S_{ld}]
\left\{
\begin{matrix}
S&\frac{1}{2}&S_{lb}\\
\frac{1}{2}&\frac{1}{2}&S_{ab}
\end{matrix}
\right\}
\left\{
\begin{matrix}
S&\frac{1}{2}&S_{ld}\\
\frac{1}{2}&\frac{1}{2}&S_{cd}
\end{matrix}
\right\}
W^a|^{S(S_{lb};S_{ld})}_{lb;lb}~,\\
\nonumber W^l|^{S(S_{ab};S_{cd})}_{ab;ab} =& [S_{ab}][S_{cd}](-1)^{S_{ab}+S_{cd}}\sum_{S_{al}S_{cl}}[S_{al}][S_{cl}](-1)^{S_{al}+S_{cl}}\\
\label{acbd_2}&\qquad\qquad\qquad\qquad\qquad\left\{
\begin{matrix}
S&\frac{1}{2}&S_{al}\\
\frac{1}{2}&\frac{1}{2}&S_{ab}
\end{matrix}
\right\}
\left\{
\begin{matrix}
S&\frac{1}{2}&S_{cl}\\
\frac{1}{2}&\frac{1}{2}&S_{cd}
\end{matrix}
\right\}
W^a|^{S(S_{al};S_{cl})}_{al;al}~,
\end{align}
which can be imposed using the projection:
\begin{equation}
W^l|^{S(S_{ab};S_{cd})}_{ab;ab} \rightarrow \frac{1}{3}\left[W^l|^{S(S_{ab};S_{cd})}_{al;ad}+ (\ref{acbd_1})_r + (\ref{acbd_2})_r\right]~.
\end{equation}
\subsubsection{Case of three equalities}
There are two different instances where three equalities hold between the indices.
\paragraph{$\mathbf{a=b}$, $\mathbf{b = c}$ and $\mathbf{d = l}$~:}
For the first type, consistency requires that the following three conditions that are satisfied:
\begin{align}
\label{abcdl_1}W^l|^{S(S_{ab};S_{cd})}_{aa;al} =& [S_{cd}](-1)^{S_{cd}}\sum_{S_{cl}}[S_{cl}](-1)^{S_{cl}}
\left\{
\begin{matrix}
S&\frac{1}{2}&S_{cl}\\
\frac{1}{2}&\frac{1}{2}&S_{cd}
\end{matrix}
\right\}
W^l|^{S(S_{ab};S_{cl})}_{aa;al}~,\\
\label{abcdl_2}W^l|^{S(S_{ab};S_{cd})}_{aa;al} =& [S_{ab}][S_{cd}]\sum_{S_{lb}S_{ld}}[S_{lb}][S_{ld}]
\left\{
\begin{matrix}
S&\frac{1}{2}&S_{lb}\\
\frac{1}{2}&\frac{1}{2}&S_{ab}
\end{matrix}
\right\}
\left\{
\begin{matrix}
S&\frac{1}{2}&S_{ld}\\
\frac{1}{2}&\frac{1}{2}&S_{cd}
\end{matrix}
\right\}
W^a|^{S(S_{lb};S_{ld})}_{la;ll}~,\\
\nonumber W^l|^{S(S_{ab};S_{cd})}_{aa;al} =& [S_{ab}][S_{cd}](-1)^{S_{ab}}\sum_{S_{al}S_{ld}}[S_{al}][S_{ld}](-1)^{S_{al}}\\
\label{abcdl_3}&\qquad\qquad\qquad\qquad\qquad \left\{
\begin{matrix}
S&\frac{1}{2}&S_{al}\\
\frac{1}{2}&\frac{1}{2}&S_{ab}
\end{matrix}
\right\}
\left\{
\begin{matrix}
S&\frac{1}{2}&S_{ld}\\
\frac{1}{2}&\frac{1}{2}&S_{cd}
\end{matrix}
\right\}
W^a|^{S(S_{al};S_{ld})}_{al;ll}~,
\end{align}
leading to the following projection:
\begin{equation}
W^l|^{S(S_{ab};S_{cd})}_{aa;al} \rightarrow \frac{1}{4}\left[W^l|^{S(S_{ab};S_{cd})}_{aa;al}+ (\ref{abcdl_1})_r + (\ref{abcdl_2})_r + (\ref{abcdl_3})_r\right]~.
\end{equation}
\paragraph{$\mathbf{a=c}$, $\mathbf{b = d}$ and $\mathbf{a = l}$~:}
For the final type, there are \emph{four} conditions that have to be fulfilled:
\begin{align}
\label{acbdal_1}W^l|^{S(S_{ab};S_{cd})}_{lb;lb} =& [S_{ab}]\sum_{S_{lb}} [S_{lb}]
\left\{
\begin{matrix}
S&\frac{1}{2}&S_{lb}\\
\frac{1}{2}&\frac{1}{2}&S_{ab}
\end{matrix}
\right\}
W^l|^{S(S_{lb};S_{cd})}_{lb;lb}~,\\
\label{acbdal_2}W^l|^{S(S_{ab};S_{cd})}_{lb;lb} =& [S_{cd}]\sum_{S_{ld}} [S_{ld}]
\left\{
\begin{matrix}
S&\frac{1}{2}&S_{ld}\\
\frac{1}{2}&\frac{1}{2}&S_{cd}
\end{matrix}
\right\}
W^l|^{S(S_{ab};S_{ld})}_{lb;lb}~,\\
\label{acbdal_3}W^l|^{S(S_{ab};S_{cd})}_{lb;lb} =& [S_{ab}][S_{cd}]\sum_{S_{lb}S_{ld}}[S_{lb}][S_{ld}]
\left\{
\begin{matrix}
S&\frac{1}{2}&S_{lb}\\
\frac{1}{2}&\frac{1}{2}&S_{ab}
\end{matrix}
\right\}
\left\{
\begin{matrix}
S&\frac{1}{2}&S_{ld}\\
\frac{1}{2}&\frac{1}{2}&S_{cd}
\end{matrix}
\right\}
W^l|^{S(S_{lb};S_{ld})}_{lb;lb}~,\\
\nonumber W^l|^{S(S_{ab};S_{cd})}_{lb;lb} =& 2[S_{ab}][S_{cd}](-1)^{S_{ab}+S_{cd}}\sum_{S_{al}S_{cl}}[S_{al}][S_{cl}](-1)^{S_{al}+S_{cl}}\\
\label{acbdal_4}&\qquad\qquad\qquad\qquad\qquad \left\{
\begin{matrix}
S&\frac{1}{2}&S_{al}\\
\frac{1}{2}&\frac{1}{2}&S_{ab}
\end{matrix}
\right\}
\left\{
\begin{matrix}
S&\frac{1}{2}&S_{cl}\\
\frac{1}{2}&\frac{1}{2}&S_{cd}
\end{matrix}
\right\}
W^b|^{S(S_{al};S_{cl})}_{ll;ll}~.
\end{align}
This is imposed by the following projection:
\begin{equation}
W^l|^{S(S_{ab};S_{cd})}_{lb;lb} \rightarrow \frac{1}{5}\left[W^l|^{S(S_{ab};S_{cd})}_{lb;lb}+ (\ref{acbdal_1})_r + (\ref{acbdal_2})_r + (\ref{acbdal_3})_r + (\ref{acbdal_4})_r\right]~.
\end{equation}

Symmetry relations between elements in the same site block $l$ causes some of the elements to be linear dependent, implying that a $W$ matrix which satisfies the consistency conditions has zero eigenvalues. This can be circumvented by taking the pseudo-inverse of the $W$ matrix (see also Section~\ref{spinconstraints_singlet}),  which excludes the eigenvectors with zero eigenvalues from the inversion process.
\section{The spin-adapted lifting conditions}
There are six independent matrix positivity conditions that can be expressed as a function of the 2.5DM. They are derived in the standard way, by finding manifestly positive Hamiltonians:
\begin{equation}
\hat{H} = \sum_i \hat{B}^\dagger_i \hat{B}_i~,
\end{equation}
that can be written as a function of the 2.5DM using anticommutation relations and spin recoupling.
\paragraph{The $\mathcal{I}_1$ condition}
The $\mathcal{I}_1$ condition trivially expresses the positivity of each individual site block of the 2.5DM, analogous to the $\mathcal{I}$ condition for the 2DM:
\begin{equation}
\mathcal{I}_1(W)^l\succeq 0\qquad\text{with}\qquad \mathcal{I}_1(W)^l|^{S(S_{ab};S_{cd})}_{ab;cd} = W^l|^{S(S_{ab};S_{cd})}_{ab;cd}~.
\end{equation}
The positivity of this object is clear from its definition in Eq.~(\ref{2.5DM_sc}). Bear in mind that every site block is composed out of two spin blocks, both of which have to be positive semidefinite.
\paragraph{The $\mathcal{Q}_2$ condition:}
the first non-trivial positivity constraint is the $\mathcal{Q}_2$ condition, which is defined as:
\begin{equation}
(\mathcal{Q}_2)^{l}|^{S(S_{ab};S_{cd})}_{ab;cd} = \sum_i w_i \frac{1}{[\mathcal{S}]^2}\sum_\mathcal{M} \bra{\Psi^N_{\mathcal{SM},i}}{B}^S_{ab(S_{ab})l}~{B^\dagger}^S_{cd(S_{cd})l}\ket{\Psi^N_{\mathcal{SM},i}}~,
\label{Q2_sc}
\end{equation}
with $B^\dagger$ given by Eq.~(\ref{2.5DM_op}).
This is the equivalent of the $\mathcal{Q}$ condition for the 2DM, and it is readily seen that every site block has to be positive semidefinite. Using anticommutation relations and some straightforward angular momentum recoupling algebra we can express $\mathcal{Q}_2$ as a function of the 2.5DM:
\begin{align}
\nonumber\mathcal{Q}_2&(W)^l|^{S(S_{ab};S_{cd})}_{ab;cd} =   \left[\frac{2\mathrm{Tr}~W}{N(N -1)(N-2)}\right]\mathbb{1}_{2.5}^l|^{S(S_{ab};S_{cd})}_{ab;cd}-W^l|^{S(S_{ab};S_{cd})}_{ab;cd} + \delta_{S_{ab}S_{cd}}\Gamma^{S_{ab}}_{ab;cd}\\
\nonumber& - \frac{\delta_{S_{ab}S_{cd}}}{\sqrt{(ab)(cd)}}\left[\delta_{bd}\rho_{ac}+(-1)^{S_{ab}}\delta_{ad}\rho_{bc} + (-1)^{S_{cd}}\delta_{bc}\rho_{ad} + \delta_{ac}\rho_{bd}\right.\\
\nonumber&\left.\qquad\qquad\qquad\qquad\qquad\qquad\qquad\qquad\qquad\qquad + \left(\delta_{ac}\delta_{bd} + (-1)^{S_{ab}}\delta_{bc}\delta_{ad}\right)\rho_{ll}\right]\\
\nonumber&-\frac{[S_{ab}][S_{cd}]}{\sqrt{(ab)(cd)}} \left\{\begin{matrix}S&\frac{1}{2}&S_{ab}\\\frac{1}{2}&\frac{1}{2}&S_{cd}\end{matrix}\right\}\left[ (-1)^{S_{ab}+S_{cd}}\delta_{bl}\delta_{dl}\rho_{ac} + (-1)^{S_{cd}}\delta_{al}\delta_{dl}\rho_{bc}+(-1)^{S_{ab}}\delta_{bl}\delta_{cl}\rho_{ad} \right.\\
\nonumber& \left.\qquad\qquad\qquad +\delta_{al}\delta_{cl}\rho_{bd} \left(\delta_{bd}\delta_{lc} +(-1)^{S_{cd}}\delta_{bc}\delta_{ld}\right)\rho_{al} + \left(\delta_{al}\delta_{bd}+(-1)^{S_{ab}}\delta_{bl}\delta_{ad}\right)\rho_{cl} \right.\\
\nonumber&\left.\qquad\qquad\qquad + (-1)^{S_{ab}}\left(\delta_{ad}\delta_{cl}+(-1)^{S_{cd}}\delta_{ld}\delta_{ac}\right)\rho_{bl} + (-1)^{S_{cd}}(\delta_{al}\delta_{bc}+(-1)^{S_{ab}}\delta_{bl}\delta_{ac})\rho_{ld}\right]\\
\nonumber& + [S_{ab}][S_{cd}]\left\{\begin{matrix}S&\frac{1}{2}&S_{ab}\\\frac{1}{2}&\frac{1}{2}&S_{cd}\end{matrix}\right\} \left[\delta_{al}\sqrt{\frac{(bl)}{(ab)}}\Gamma^{S_{cd}}_{lb;cd} + \delta_{bl}(-1)^{S_{ab}+S_{cd}}\sqrt{\frac{(al)}{(ab)}}\Gamma^{S_{cd}}_{al;cd}\right.\\
\nonumber&\left.\qquad\qquad\qquad\qquad\qquad\qquad\qquad+ \delta_{dl}(-1)^{S_{ab}+S_{cd}}\sqrt{\frac{(cl)}{(cd)}}\Gamma^{S_{ab}}_{ab;cl}+\delta_{cl}\sqrt{\frac{(dl)}{(cd)}}\Gamma^{S_{ab}}_{ab;ld}\right]\\
\nonumber&+\frac{[S_{ab}][S_{cd}]}{\sqrt{(ab)(cd)}}\sum_{S'} [S']^2\left\{\begin{matrix}S&\frac{1}{2}&S'\\\frac{1}{2}&\frac{1}{2}&S_{ab}\end{matrix}\right\}\left\{\begin{matrix}S&\frac{1}{2}&S'\\\frac{1}{2}&\frac{1}{2}&S_{cd}\end{matrix}\right\}\left[ \delta_{bd}(-1)^{S_{ab}+S_{cd}}\sqrt{(al)(cl)}\Gamma^{S'}_{al;cl}\right.\\
&\left.\quad+\delta_{ac}\sqrt{(bl)(dl)}\Gamma^{S'}_{bl;dl} + \delta_{bc}(-1)^{S_{ab}}\sqrt{(al)(dl)}\Gamma^{S'}_{al;dl}+ \delta_{ad}(-1)^{S_{cd}}\sqrt{(bl)(cl)}\Gamma^{S'}_{bl;cl} \right]~.
\label{Q2}
\end{align}
This is quite similar to the $\mathcal{T}_1$ condition with one site-index pair diagonal. The difference lies in the fact that the third-order term does not disappear (as in $\mathcal{T}_1$, due to the anticommutator). But the third-order term is in fact the 2.5DM matrix, so we have that:
\begin{equation}
\mathcal{T}_1(\Gamma)^{S(S_{ab};S_{cd})}_{abl;cdl} = \mathcal{I}_1(W)^l|^{S(S_{ab};S_{cd})}_{ab;cd} + \mathcal{Q}_2(W)^l|^{S(S_{ab};S_{cd})}_{ab;cd}~.
\end{equation}
The unit matrix on 2.5DM space used in the $\mathcal{Q}_2$ map (first term on the right in Eq.~(\ref{Q2})) is \emph{not} a diagonal matrix, as a result of the consistency conditions inside a site block, which introduces terms with off-diagonal intermediate spin:
\begin{align}
\nonumber\mathbb{1}_{2.5}^l|^{S(S_{ab}S_{cd})}_{ab;cd} =& \frac{1}{\sqrt{(ab)(cd)}}\left[\delta_{S_{ab}S_{cd}}(\delta_{ac}\delta_{bd}+(-1)^{S_{ab}}\delta_{ad}\delta_{bc})\right.\\
\nonumber&\left.\qquad\qquad +[S_{ab}][S_{cd}]\left\{\begin{matrix}S&\frac{1}{2}&S_{ab}\\\frac{1}{2}&\frac{1}{2}&S_{cd}\end{matrix}\right\}\left(\delta_{al}\delta_{cl}\delta_{bd}+(-1)^{S_{ab}}\delta_{bl}\delta_{cl}\delta_{ad} \right.\right.\\
&\left.\left.\qquad\qquad\qquad\qquad\qquad+ (-1)^{S_{cd}}\delta_{al}\delta_{dl}\delta_{bc} + (-1)^{S_{ab}+S_{cd}}\delta_{bl}\delta_{dl}\delta_{ac}\right)\right]~.
\end{align}
In the $\mathcal{Q}_2$ map of Eq.~(\ref{Q2}), the 2DM appears; it can be derived from the 2.5DM by tracing over the block index:
\begin{equation}
\Gamma^{S_{ab}}_{ab;cd} = \frac{1}{N-2}\sum_S \frac{[S]^2}{[S_{ab}]^2}\sum_l W^l|^{S(S_{ab};S_{ab})}_{ab;cd}~.
\end{equation}
\paragraph{The $\mathcal{I}_2$ condition:}
the $\mathcal{I}_2$ condition has no equivalent in the v2DM formalism. It is defined in a spin-averaged ensemble as:
\begin{equation}
(\mathcal{I}_2)^{l}|^{S(S_{ab};S_{cd})}_{ab;cd} = \sum_i w_i \frac{1}{[\mathcal{S}]^2}\sum_\mathcal{M} \bra{\Psi^N_{\mathcal{SM},i}}{B^\dagger}^S_{ab(S_{ab})l}~{B}^S_{cd(S_{cd})l}\ket{\Psi^N_{\mathcal{SM},i}}~,
\label{I2_sc}
\end{equation}
with $B^\dagger$ a two-particle-one-hole operator as defined in Eq.~(\ref{pph_sc_op}):
\begin{equation}
{B^\dagger}^S_{ab(S_{ab})l}= \frac{1}{\sqrt{(ab)}}\left[\left[a^\dagger_a \otimes a^\dagger_b \right]^{S_{ab}}\otimes \tilde{a}_l\right]^S~.
\label{tp1h_2.5DM}
\end{equation}
Again, every site block of the matrix has to be positive semidefinite. Using anticommutation relations and spin recoupling we can express the $\mathcal{I}_2$ condition as a function of the 2.5DM:
\begin{eqnarray*}
\mathcal{I}_2(\Gamma)^l|^{S(S_{ab};S_{cd})}_{ab;cd} &=& \delta_{S_{ab}S_{cd}}\Gamma^{S_{ab}}_{ab;cd} + \sum_{S'}[S']^2 
\left\{
\begin{matrix}
S_{ab}&S&\frac{1}{2}\\
S_{cd}&S'&\frac{1}{2}
\end{matrix}
\right\}
W^l|^{S'(S_{ab}S_{cd})}_{ab;cd}~.
\end{eqnarray*}
\paragraph{The $\mathcal{Q}_1$ condition:}
the $\mathcal{Q}_1$ condition is of the same type as the $\mathcal{I}_2$ condition, and reads:
\begin{equation}
(\mathcal{Q}_1)^{l}|^{S(S_{ab};S_{cd})}_{ab;cd} = \sum_i w_i \frac{1}{[\mathcal{S}]^2}\sum_\mathcal{M} \bra{\Psi^N_{\mathcal{SM},i}}{B}^S_{ab(S_{ab})l}~{B^\dagger}^S_{cd(S_{cd})l}\ket{\Psi^N_{\mathcal{SM},i}}~,
\label{Q1_sc}
\end{equation}
with $B^\dagger$, the two-particle-one-hole operator as in Eq.~(\ref{tp1h_2.5DM}). In combination with the $\mathcal{I}_2$ condition it can form the $\mathcal{T}_2$ condition diagonal in the third site-index:
\begin{equation}
\mathcal{T}_2(\Gamma)^{S(S_{ab};S_{cd})}_{abl;cdl} = \mathcal{I}_2(W)^l|^{S(S_{ab};S_{cd})}_{ab;cd} + \mathcal{Q}_1(W)^l|^{S(S_{ab};S_{cd})}_{ab;cd}~.
\end{equation}
Retracing familiar steps, one obtains an expression of the $\mathcal{Q}_1$ condition as a function of the 2.5DM:
\begin{align}
\nonumber\mathcal{Q}_1&(W)^l|^{S(S_{ab};S_{cd})}_{ab;cd} = \frac{\delta_{S\frac{1}{2}}[S_{ab}][S_{cd}]}{2\sqrt{(ab)(cd)}} \left[\delta_{al}\delta_{cl}\delta_{bd}+(-1)^{S_{cd}}\delta_{al}\delta_{dl}\delta_{bc} + (-1)^{S_{ab}}\delta_{bl}\delta_{cl}\delta_{ad}\right.\\
\nonumber&\left.+(-1)^{S_{ab}+S_{cd}}\delta_{bl}\delta_{dl}\delta_{ac} - \delta_{al}\delta_{cl}\rho_{bd}-(-1)^{S_{cd}}\delta_{al}\delta_{dl}\rho_{bc} - (-1)^{S_{ab}}\delta_{bl}\delta_{cl}\rho_{ad}\right.\\
\nonumber&\left.-(-1)^{S_{ab}+S_{cd}}\delta_{bl}\delta_{dl}\rho_{ac} -\left(\delta_{bd}\delta_{lc} + (-1)^{S_{cd}}\delta_{bc}\delta_{ld}\right)\rho_{al} - \left(\delta_{bd}\delta_{al}+(-1)^{S_{ab}}\delta_{ad}\delta_{bl}\right)\rho_{cl}\right.\\
\nonumber& \left. - (-1)^{S_{ab}+S_{cd}}\left(\delta_{ac}\delta_{ld}+(-1)^{S_{cd}}\delta_{ad}\delta_{lc}\right)\rho_{lb} - (-1)^{S_{ab}+S_{cd}}\left(\delta_{ac}\delta_{bl}+(-1)^{S_{ab}}\delta_{bc}\delta_{al}\right)\rho_{dl}\right]\\
\nonumber&+\frac{\delta_{S_{ab}S_{cd}}}{\sqrt{(ab)(cd)}}\left(\delta_{ac}\delta_{bd} + (-1)^{S_{ab}}\delta_{ad}\delta_{bc}\right)\rho_{ll} -\sum_{S'}[S']^2 
\left\{
\begin{matrix}
S_{ab}&S&\frac{1}{2}\\
S_{cd}&S'&\frac{1}{2}
\end{matrix}
\right\}
W^l|^{S'(S_{ab}S_{cd})}_{ab;cd}
\\
\nonumber&+ \delta_{S\frac{1}{2}}\frac{[S_{ab}][S_{cd}]}{2}\left(\delta_{al}\sqrt{\frac{(lb)}{(ab)}}\Gamma^{S_{cd}}_{lb;cd} +(-1)^{S_{ab}+S_{cd}}\delta_{bl}\sqrt{\frac{(al)}{(ab)}}\Gamma^{S_{cd}}_{al;cd}\right.\\
\nonumber&\left.\qquad\qquad\qquad\qquad\qquad\qquad+ \delta_{cl}\sqrt{\frac{(dl)}{(cd)}}\Gamma^{S_{ab}}_{ab;ld} + (-1)^{S_{ab}+S_{cd}}\delta_{ld}\sqrt{\frac{(cl)}{(cd)}}\Gamma^{S_{ab}}_{ab;cl}\right)\\
\nonumber&-\sum_{S'}[S']^2 \frac{[S_{ab}][S_{cd}]}{\sqrt{(ab)(cd)}}\left\{\begin{matrix}S& S_{ab} & \frac{1}{2}\\ S_{cd} & \frac{1}{2} & \frac{1}{2}\\ \frac{1}{2}&\frac{1}{2}&S' \end{matrix}\right\} \left( (-1)^{S_{ab}+S_{cd}}\delta_{bd}\sqrt{(al)(cl)}\Gamma^{S'}_{al;cl} \right.\\
&\left.\quad+(-1)^{S_{cd}}\delta_{ad}\sqrt{(bl)(cl)}\Gamma^{S'}_{bl;cl} + (-1)^{S_{ab}}\delta_{bc}\sqrt{(al)(dl)}\Gamma^{S'}_{al;dl} + \delta_{ac}\sqrt{(bl)(dl)}\Gamma^{S'}_{bl;dl}\right)~.
\end{align}
\paragraph{The $\mathcal{G}_1$ condition:}
the remaining two conditions are of a type we haven't encountered before. The $\mathcal{G}_1$ condition is defined through a spin-averaged ensemble as:
\begin{equation}
(\mathcal{G}_1)^{l}|^{S(S_{bl};S_{dl})}_{ab;cd} = \sum_i w_i \frac{1}{[\mathcal{S}]^2}\sum_\mathcal{M} \bra{\Psi^N_{\mathcal{SM},i}}{B^\dagger}^S_{abl(S_{bl})}~{B}^S_{cdl(S_{dl})}\ket{\Psi^N_{\mathcal{SM},i}}~,
\label{G1_sc}
\end{equation}
where $B^\dagger$ is a one-hole-two-particle operator, {\it i.e.}:
\begin{align}
\nonumber{B^\dagger}^S_{abl(S_{bl})} =& \left[\tilde{a}_a\otimes\left[a^\dagger_b\otimes a^\dagger_l\right]^{S_{bl}}\right]^S\\
=&\sum_{\sigma_b\sigma_l}\sum_{M_{bl}\sigma_a}(-1)^{\frac{1}{2}-\sigma_a}\braket{\frac{1}{2}-\sigma_a S_{bl}M_{bl}}{SM}\braket{\frac{1}{2}\sigma_b\frac{1}{2}\sigma_l}{S_{bl}M_{bl}}a_{a\sigma_a}a^\dagger_{b\sigma_b}a^\dagger_{l\sigma_l}~.
\end{align}
Notice the difference with the regular three-particle spin coupling. First the second and third index are coupled to intermediate spin $S_{bl}$, after which the first index is coupled with $S_{bl}$ to total spin $S$. Using anticommutation relations and spin-recoupling algebra we obtain the following expression of the $\mathcal{G}_1$ map as a function of the 2.5DM:
\begin{align}
\nonumber\mathcal{G}_1&(W)^{l}|^{S(S_{bl};S_{dl})}_{ab;cd} = [S_{bl}][S_{dl}] \delta_{S\frac{1}{2}}\left(\delta_{al}\delta_{cl}\rho_{bd} + (-1)^{S_{bl}}\delta_{ab}\delta_{cl}\rho_{dl} + (-1)^{S_{dl}}\delta_{cd}\delta_{al}\rho_{bl}  \right.\\
\nonumber&\left.\qquad\qquad\qquad\qquad\qquad\qquad + (-1)^{S_{bl}+S_{dl}}\delta_{ab}\delta_{cd}\rho_{ll}\right) + \sqrt{(bl)(dl)}\delta_{S_{bl}S_{dl}}\delta_{ac}\Gamma^{S_{bl}}_{bl;dl}\\
\nonumber&-\frac{[S_{bl}][S_{dl}]}{2}\delta_{S\frac{1}{2}}(-1)^{S_{bl}+S_{dl}}\left(\delta_{ab}\sqrt{{(cl)}{(dl)}}\Gamma^{S_{dl}}_{cl;dl} + (-1)^{S_{bl}}\delta_{al}\sqrt{{(cb)}{(dl)}}\Gamma^{S_{dl}}_{cb;dl}\right.\\
\nonumber&\left.\qquad\qquad\qquad\qquad\qquad\qquad\qquad + (-1)^{S_{dl}}\delta_{cl}\sqrt{{(ad)}{(bl)}}\Gamma^{S_{bl}}_{ad;bl} + \delta_{cd}\sqrt{{(al)}{(bl)}}\Gamma^{S_{bl}}_{al;bl}\right)\\
\nonumber&+\sqrt{{(ad)(cb)}}\sum_{S'}\sum_{S_{ab}S_{cd}}[S']^2[S_{ab}][S_{cd}][S_{bl}][S_{dl}]\\*
&\qquad\qquad \qquad\qquad\qquad\left\{\begin{matrix}S'&\frac{1}{2}&S_{dl}\\\frac{1}{2}&\frac{1}{2}&S_{ab}\end{matrix}\right\}
\left\{\begin{matrix}S'&\frac{1}{2}&S_{bl}\\\frac{1}{2}&\frac{1}{2}&S_{cd}\end{matrix}\right\}
\left\{\begin{matrix}S'&S_{bl}&\frac{1}{2}\\S&S_{dl}&\frac{1}{2} \end{matrix}\right\}
W^l|^{S'(S_{ab};S_{cd})}_{ad;cb}~.
\label{G1_2.5DM}
\end{align}
\paragraph{The $\mathcal{G}_2$ condition:}
the $\mathcal{G}_2$ condition is of the same type as $\mathcal{G}_1$, and defined through a spin-averaged ensemble as:
\begin{equation}
(\mathcal{G}_2)^{l}|^{S(S_{bl};S_{dl})}_{ab;cd} = \sum_i w_i \frac{1}{[\mathcal{S}]^2}\sum_\mathcal{M} \bra{\Psi^N_{\mathcal{SM},i}}{B^\dagger}^S_{abl(S_{bl})}~{B}^S_{cdl(S_{dl})}\ket{\Psi^N_{\mathcal{SM},i}}~,
\label{G2_sc}
\end{equation}
where this time the $B^\dagger$ is a one-particle-two-hole operator, given by:
\begin{align}
\nonumber{B^\dagger}^S_{abl(S_{bl})} =& \left[{a}^\dagger_a\otimes\Big[\tilde{a}_b\otimes \tilde{a}_l\Big]^{S_{bl}}\right]^S\\
=&\sum_{\sigma_b\sigma_l}\sum_{M_{bl}\sigma_a}(-1)^{1-\sigma_b-\sigma_l}\braket{\frac{1}{2}\sigma_a S_{bl}M_{bl}}{SM}\braket{\frac{1}{2}-\sigma_b\frac{1}{2}-\sigma_l}{S_{bl}M_{bl}}a_{a\sigma_a}a^\dagger_{b\sigma_b}a^\dagger_{l\sigma_l}~.
\end{align}
Anticommuting the creation and annihilation operators, and spin-recoupling the emergent terms, leads to the following for the $\mathcal{G}_2$ condition as a function of the 2.5DM:
\begin{align}
\nonumber\mathcal{G}_2&(W)^l|^{S(S_{bl};S_{dl})}_{ab;cd} = \delta_{S_{bl}S_{dl}}\left(\delta_{bd}+(-1)^{S_{bl}}\delta_{bl}\delta_{dl}\right)\rho_{ac}\\
\nonumber&-[S_{bl}][S_{dl}]\sum_{S'} [S']^2 \left\{\begin{matrix}\frac{1}{2}&\frac{1}{2}&S'\\S&S_{dl}&\frac{1}{2}\\S_{bl}&\frac{1}{2}&\frac{1}{2}\end{matrix}\right\}
\left[
\sqrt{(ad)(cb)}\Gamma^{S'}_{ad;cb}
+ (-1)^{S_{bl}}\delta_{bl}\sqrt{(ad)(cl)}\Gamma^{S'}_{ad;cl}\right.\\
\nonumber&\left. \qquad\qquad\qquad\qquad
+ (-1)^{S_{dl}}\delta_{dl}\sqrt{(al)(cb)}\Gamma^{S'}_{al;cb}
+ (-1)^{S_{bl}+S_{dl}}\delta_{bd}\sqrt{(al)(cl)}\Gamma^{S'}_{al;cl}
\right]\\
\nonumber&-\sqrt{{(ad)(cb)}}\sum_{S'}\sum_{S_{ab}S_{cd}}[S']^2[S_{ab}][S_{cd}][S_{bl}][S_{dl}]\\
&\qquad\qquad \qquad\qquad\qquad\left\{\begin{matrix}S'&\frac{1}{2}&S_{dl}\\\frac{1}{2}&\frac{1}{2}&S_{ab}\end{matrix}\right\}
\left\{\begin{matrix}S'&\frac{1}{2}&S_{bl}\\\frac{1}{2}&\frac{1}{2}&S_{cd}\end{matrix}\right\}
\left\{\begin{matrix}S'&S_{bl}&\frac{1}{2}\\S&S_{dl}&\frac{1}{2} \end{matrix}\right\}
W^l|^{S'(S_{ab};S_{cd})}_{ad;cb}~.
\end{align}
\section{Formulation as a semidefinite program}
We can reformulate the whole optimization problem using the 2.5DM as the central variable. For some Hamiltonian $\hat{H}$, we optimize the matrix $W$ under the constraint that is has the correct particle number, fulfills the consistency conditions, and has positive semidefinite linear matrix maps, as discussed in the previous Section:
\begin{eqnarray}
\label{2.5DM_vprob}E^N_{\text{SDP}}\left(H\right) &=& \min_{W} \mathrm{Tr}~\left[W H^{(2.5)}\right]~,\\
\nonumber\text{u.c.t.}&&\left\{
   \begin{array}{l}
   \mathrm{Tr}~W = \frac{N(N-1)(N-2)}{2}~,\\
   W\text{ is consistent}~,\\
   \mathcal{L}_i\left(W\right) \succeq 0~\qquad\forall \mathcal{L}_i \in \{\mathcal{I}_1,\mathcal{I}_2,\mathcal{Q}_1,\mathcal{Q}_2,\mathcal{G}_1,\mathcal{G}_2\}~.
   \end{array}
   \right.
\end{eqnarray}
The 2.5-Hamiltonian introduced above is defined by the relation:
\begin{equation}
\mathrm{Tr}~WH^{(2.5)} = \mathrm{Tr}~\Gamma H^{(2)}~.
\end{equation}
This optimization problem can again be formulated as a semidefinite program, by expanding $W$ in a complete orthogonal basis $\{f^i_W\}$ of traceless 2.5DM-space:
\begin{equation}
W = \frac{N(N-1)(N-2)}{M(M-1)(M-2)}\mathbb{1}_{2.5} + \sum_i w_i f^i_W~,\qquad\text{with}\qquad w_i = \mathrm{Tr}~Wf^i_W~,
\end{equation}
in which the $f^i_W$'s satisfy the consistency conditions. We can now define the v2.5DM problem as a dual-form semidefinite program with the following structure matrices:
\begin{equation}
u_W^0 = \frac{N(N-1)(N-2)}{M(M-1)(M-2)}\bigoplus_j \mathcal{L}_j\left(\mathbb{1}_{2.5}\right)\qquad\text{and}\qquad u_W^i = \bigoplus_j \mathcal{L}_j\left(f_W^i\right)~.
\end{equation}
The whole formalism developed in Chapter~\ref{SDP} can now be taken over, provided we find correct Hermitian adjoint maps. These are derived in Section~\ref{herm_adj_2.5DM}.
\subsection{The overlap matrix}
For the primal-dual interior point method and the boundary point method in Chapter~\ref{SDP} we introduced the overlap matrix (See Section~\ref{overlapmatrix}). For the translation of this algorithm to the v2.5DM formalism, we need to introduce a new overlap matrix:
\begin{equation}
\mathcal{S}^{ij}_W = \mathrm{Tr}~u^i_W u^j_W~.
\end{equation}
Analogous to the 2DM case, the overlap matrix can be interpreted as a map from 2.5DM space on itself:
\begin{equation}
\mathcal{S}\Delta = \hat{P}_W\left[\sum_i\mathcal{L}^\dagger_i\left(\mathcal{L}(\Delta)\right)\right]~,
\end{equation}
in which we sum over all the conditions introduced in the previous Section, and with $\hat{P}_W$ a projection on fully consistent traceless 2.5DM space, as explained in Section~\ref{consistent}. Unfortunately this overlap matrix can not be easily inverted, as was the case for the 2DM overlap matrix, because the lifting conditions are \emph{not} invariant under unitary transformations on single-particle space. Instead of inverting it analytically, we use the linear conjugate gradient method to compute the action of the inverse overlap matrix on a 2.5DM. This slows down the program, but not dramatically, as the number of iterations needed for the conjugate gradient loop to converge is small, and remains constant during the program. What is more, the computational cost of one overlap matrix-vector product scales as $M^5$, whereas the heaviest computations in the algorithm scale as $M^7$. This should be compared to the $M^9$ scaling of the full three-index conditions.
\section{Results}
\begin{figure}
\centering
$
\begin{array}{c}
\includegraphics[scale=0.7]{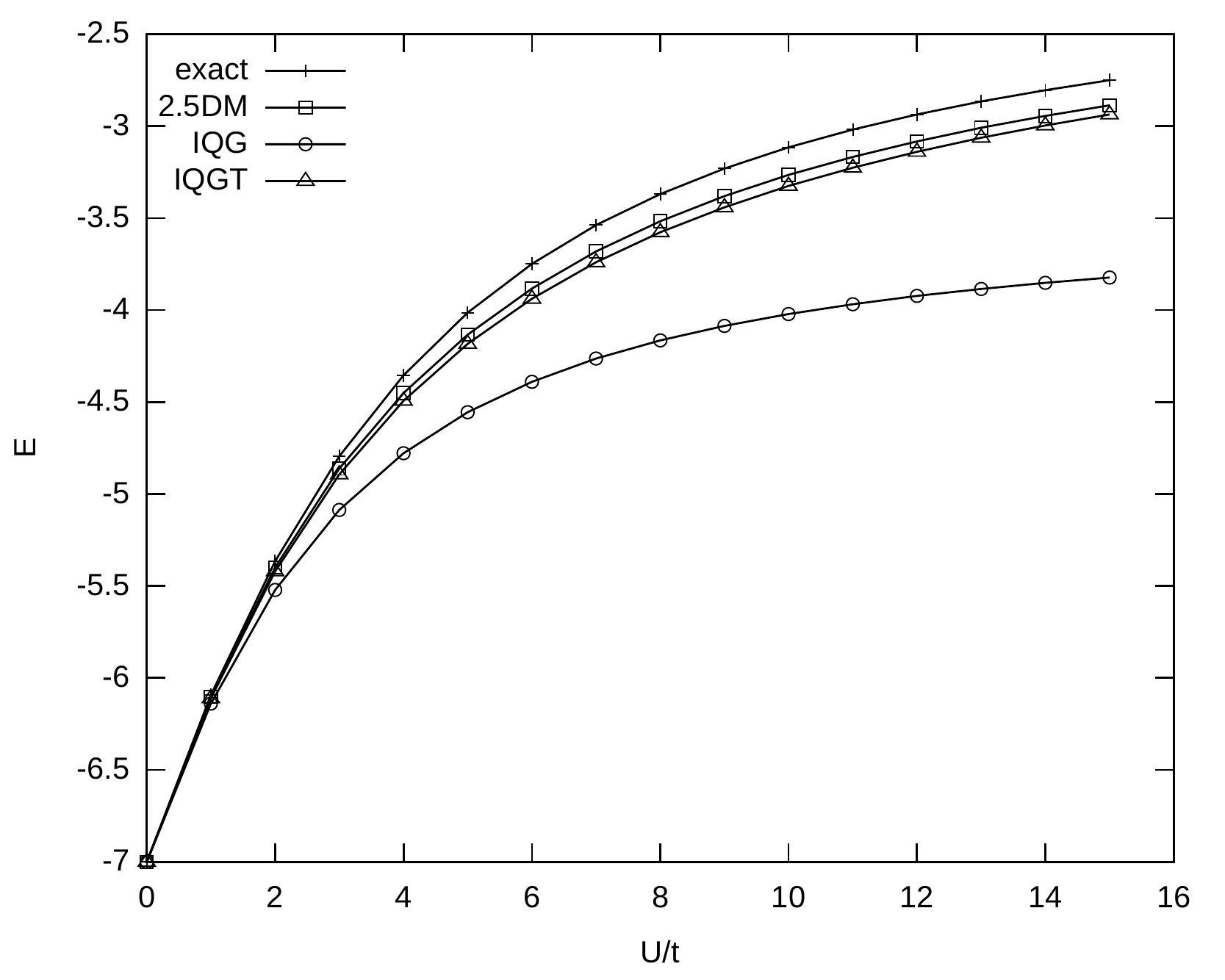}\\
\includegraphics[scale=0.7]{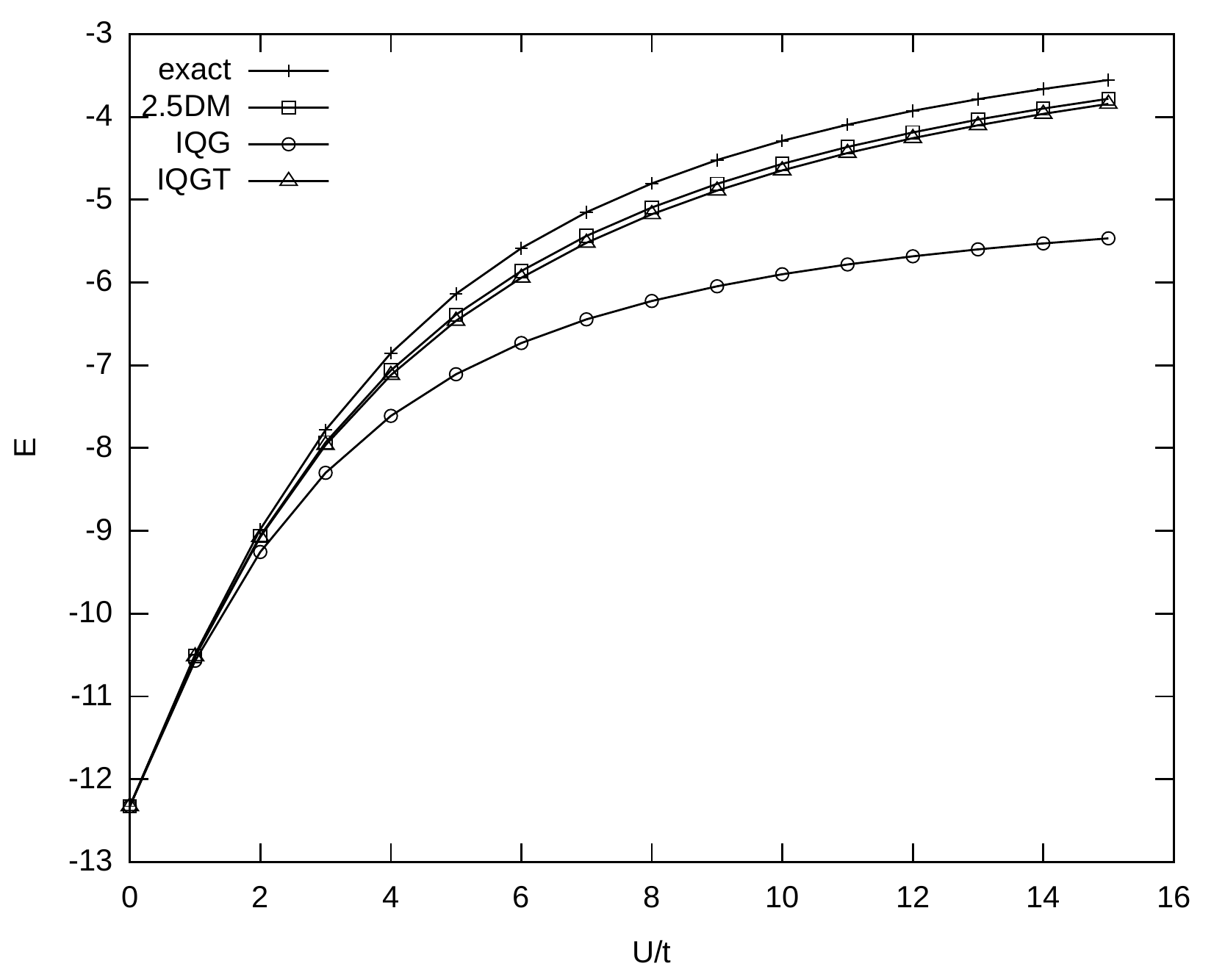}
\end{array}
$
\caption{\label{v2.5DM_ener}Ground-state energy as a function of on-site repulsion $U$ of the Hubbard model, for a 6-site with 5 particles (top) and for a 10-site lattice with 9 particles (bottom). A comparison is made between an optimization using $\mathcal{IQG}$ conditions, $\mathcal{IQGT}$ conditions, v2.5DM calculations and exact results.}
\end{figure}
The formalism introduced above has been implemented for the boundary point algorithm (See Section~\ref{bp_sdp}).
The results of such a v2.5DM calculation are shown in Fig.~\ref{v2.5DM_ener}, where the ground-state energy of the one-dimensional Hubbard model is plotted as a function of the on-site repulsion $U$, for a 6-site lattice with 5 particles, and for a 10-site lattice with 9 particles. The v2.5DM results are compared to the v2DM results using $\mathcal{IQG}$ and $\mathcal{IQGT}$ conditions, and to exact results. For the 6-site lattice exact diagonalization was used, for the 10-site lattice this is no longer computationally feasible, but we can again compare to quasi-exact results calculated with an MPS optimization \cite{schollwock,verstraete,chan,sebastian}. In this Figure, one can see that the v2.5DM results are always of $\mathcal{IQGT}$ quality without resorting to the full $\mathcal{IQGT}$ framework. In fact, the v2.5DM results are slightly better than those obtained with $\mathcal{IQGT}$. This is because the $\mathcal{T}_1$ and $\mathcal{T}_2$ conditions express the positivity of an anticommutator of three-particle operators (see Section~\ref{three_index}), whereas in v2.5DM positivity is imposed on all possible individual products of three-particle operators, be it of a restricted class.
In Table~\ref{v2.5DM_scl_energy} one can see that for very large values of $U$, the exact limit is restored, which we expect from the discussion about the Gutzwiller conditions in the previous Chapter.
\begin{table}
\centering
\begin{tabular}{|c|cccc|cccc|}
\hline
&\multicolumn{4}{c}{$L=6$\qquad$N=5$}\vline&\multicolumn{4}{c}{$L=10$\qquad$N=9$}\vline\\
\hline
$U$ & $\mathcal{IQG}$ & $\mathcal{IQGT}$ &v2.5DM&exact& $\mathcal{IQG}$ & $\mathcal{IQGT}$ & v2.5DM&exact\\
\hline
50 & -3.55 & -2.29 & -2.28 & -2.20 & -4.89 & -2.54 &-2.53 & -2.46 \\
100 & -3.49 & -2.15 & -2.14& -2.08 & -4.77  & -2.27 &-2.26 & -2.22\\
1000 & -3.44 & -2.03 & -2.01& -2.01 & -3.44 & -2.03 & -2.01& -2.01\\
\hline
\end{tabular}
\caption{\label{v2.5DM_scl_energy}The ground-state energy of a 10-site lattice with 9 particles (right) and 6-site lattice with five particles (left) for $U = 50,\ 100$ and 1000, exact results compared with v2DM results using $\mathcal{IQG}$ and $\mathcal{IQGT}$ results, and v2.5DM results.}
\end{table}

Using these spin-adapted lifting conditions in the v2.5DM framework, we were able not only to fix the strong correlation limit, but also improve the quality of the results at intermediary values of $U$ to $\mathcal{IQGT}$ quality. The computational cost, however, is two orders of magnitude smaller than when enforcing the full $\mathcal{IQGT}$ conditions.
It must be stressed that up to now we have only included the spin symmetry of the model in our code. If translational invariance, parity and pseudospin symmetry are taken into account much larger lattices can be considered. As an example, our fully symmetric $\mathcal{IQG}$ version allows lattice sizes up to 100 sites, and the fully symmetric $\mathcal{IQGT}$ program up to 20 sites. We expect a fully symmetric version of v2.5DM to be applicable to lattice sizes of about 50 sites, thereby enabling us to study two-dimensional lattices of reasonable size.

The diagonality of the third index in the 2.5DM implies that the result will depend on the chosen single-particle basis. For the Hubbard model it is clear that the site basis is the optimal basis to use for the diagonal third index. It would be interesting to study other systems where it is less clear what the best choice of the single-particle basis would be. An appealing application, {\it e.g.}, are molecules, where one can hope to get three-index precision by applying the v2.5DM method with a carefully chosen basis. A first guess of what the best basis would be is the basis of natural orbitals, for which it has been shown that the full-CI expansion has the fastest convergence \cite{lowdin}.

\chapter{\label{conclusions}Conclusions and outlook}
In this work we have provided an overview of a quantum many-body technique in which the two-particle density matrix (2DM) is determined variationally, and thereby replaces the wave function as the central object. It is our hope that we have convinced the reader that this is a promising technique, appealing in the simplicity of its underlying idea, and complementary to other many-body methods. At the same time we have tried to indicate where the problems and difficulties lie, and shown that the method has a long way to go before it can be used as a `black box' electronic structure method.

In Chapter~\ref{n_rep} we discussed the $N$-representability problem, which consists of finding the necessary and sufficient conditions a 2DM has to fulfil to be derivable from a physical wave function. Although the general problem is nearly impossible to solve, one can derive necessary conditions using the dual definition of $N$-representability. First we showed how the standard two- and three-index conditions can be derived using manifestly positive Hamiltonians. These conditions can be formulated as matrix positivity constraints, and are the ones commonly used in literature. The second part of the Chapter deals with constraints that are non-standard, {\it i.e.} the constraint Hamiltonian is not manifestly positive, but a lower bound can be found in some different, computationally cheap, way. The formalism introduced in Chapter~\ref{n_rep} is very flexible, which makes it easy to derive new necessary constraints. It is hard, however, to derive constraints that are active (in the sense that the inequality is violated after applying the standard conditions), and improve the result for relevant physical systems.

The constrained optimization problem (v2DM) introduced in Chapter~\ref{n_rep} can be translated to a standard numerical optimization technique called semidefinite programming. In Chapter~\ref{SDP} we start by formulating variational density matrix optimization in the two forms (primal and dual) in which a semidefinite program can be formulated. It is seen that in v2DM the dual form is the most efficient formulation, and three different algorithms are adapted to the specific form of v2DM. Two of these are so-called interior point methods, which stay inside the feasible region during optimization. The third one is a boundary point method, which stays on the hyperplane where the complementary slackness conditions holds, and moves towards $N$-representability. At the end of the Chapter, a comparative study is made between the different algorithms. The drawbacks and advantages of the different methods are listed, and their performance for the specific case of the one-dimensional Hubbard model at half filling is compared. The boundary point method is seen to be the most suited for this highly symmetrical type of problem. Although a lot of progress has been made, the different algorithms all remain far slower than electronic structure methods with a comparable accuracy ({\it e.g.} coupled-cluster with single and double excitations), and this remains the biggest drawback to the method at this point.

In Chapter~\ref{symmetry} we show how the efficiency of the v2DM method can be increased by including symmetries specific to the problem under study. We start by explaining how to introduce the most commonly available symmetry in electronic structure problems, {\it i.e.} spin symmetry. Then we consider symmetries specific to certain types of systems, such as rotational symmetry in atomic systems and translational invariance in the Hubbard model with periodic boundary conditions. It is shown that these symmetries can always be included in a straightforward matter, although the practical implementation can sometimes be a technical and tedious task. The inclusion of symmetry shows one way to overcome the computational complexity of the current implementations of semidefinite programming algorithms, up to a point where they are more competitive to other methods. It must, however, be stressed that the advantage is limited to the study of highly symmetrical systems.

In Chapter~\ref{applications}, all different aspects of the v2DM method discussed in the preceding Chapters come together. The first Section discusses the application of the standard two-index constraints to the isoelectronic series of Beryllium, Neon and Silicon. For this study full use was made of the available symmetry present in atomic systems, which allows the v2DM method to be used on far larger basissets than previously possible. It was found that the two-index conditions succesfully describe the static electron correlation arising from the near-degeneracy of single-particle levels at higher values of the central charge $Z$. The quality of the variationally obtained 2DM was verified through the extraction of different physical properties ({\it e.g.} ionization energy) and its physical content was found to be reliable.
The next Section focussed on the potential energy surfaces of diatomic molecules. It was found that the standard two, and even three-index conditions fail to describe the dissociation limit. To fix this we derived a new type of $N$-representability condition which imposes constraints on the separate atomic subsystems of molecule. The last Section treats the one-dimensional Hubbard model, for which spin symmetry, translational invariance and inversion symmetry allow a study for large lattice sizes. In previous v2DM studies only the half-filled model was studied, and only the ground-state energy was considered. We show that the two-index conditions fail to describe the strong-correlation limit below half filling, and that three-index constraints are needed. Not only the ground-state energy, but also the spin and charge correlation functions are seen to deviate substantially from the expected result. The origin of the failure can be traced to a bad description of correlations in a subspace of three-particle space, which is responsible for the hopping on singly-occupied lattices. Two new necessary constraints on the 2DM are derived based on this analysis. It is found that these constraints, while helpful, are too weak to cure the pathological behaviour, and a larger object is needed. In Chapter~\ref{v2.5DM} the 2.5DM, which is a 3DM diagonal in one spatial index, is introduced. Six standard matrix-positivity conditions can be derived for this object, and a proof-of-principle implementation shows that this approach effectively fixes the strong-correlation limit, while being two orders of magnitude less computationally complex than the three-index conditions.

\section*{Outlook}
Although the method discussed in this thesis has been around since the 1960's, it remains an obscure, even somewhat mysterious, many-body method which still needs more investigation. Fortunately, activity in this area has picked up in recent years. The single most important thing for the advancement of the field is the development of an algorithm that scales genuinely, and not just nominally, as $O(M^6$. If we fail to develop a method that is actually competitive with other electronic structure methods, progress will be slow and the method will remain in the margins. Apart from the scaling issue, there are a number of directions where I think progress is possible in the field, or for which the method could actually contribute information where other methods fail.
\begin{itemize}
\item It is notoriously difficult for many-body methods to obtain decent results for the ground-state properties of the two-dimensional Hubbard model. This model is highly symmetrical, and a v2DM approach could be used to obtain complementary results to those obtained by other methods. Results on the one-dimensional Hubbard model indicate that, if one exploits all the symmetries in a v2.5DM approach of the two-dimensional Hubbard model, decent results for reasonably sized lattices could be obtained.
\item The subsystem constraints that were introduced for diatomic systems, only become active for relatively large internuclear distances. An idea to improve on this is to use a previous $\mathcal{IQG}$ calculation on the diatomic molecule to construct a better subsystem Hamiltonian.
\item One thing that, in my view, has not been exploited enough, is the power of exactly solvable models as conditions, possibly in combination with the subsystem constraints. {\it E.g.} one could imagine using the Bethe-ansatz energies for the one-dimensional Hubbard model as subsystem constraints in two-dimensional models.
\item The three-index constraints fix a lot of problems, and generally improve the results by an order of magnitude. They are however computationally very heavy to impose. Because they are expressible as a function of the 2DM alone, a matrix-vector product can be performed very efficiently, and one has the feeling this could be exploited to impose the constraints in a much more efficient way.
\item At this point, there is no formulation of reduced density matrix optimization in the thermodynamic limit. This is because the reduced density matrix becomes ill-defined in the thermodynamic limit. A better suited object would be the related density cumulant matrix. It is our ambition to reformulate the density matrix optimization in terms of the density cumulant, and use it in a study of the electron gas.
\end{itemize}

\appendix

\chapter{\label{math_notes}Mathematical concepts and notation}
Throughout the thesis mathematical concepts and symbols are used which might be unfamiliar to the reader. In this appendix some of these conceps are briefly explained.
\paragraph{Convexity}
A subset $V$ of a linear vector space is convex if, for every two elements of this set, all elements on the line segment connecting these two elements are also elements of the set, {\it i.e.}:
\begin{equation}
\forall x,y \in V \qquad z = \alpha x + (1 - \alpha) y \in V \qquad \text{for} \qquad 0 \leq \alpha \leq 1~.
\end{equation}
Convexity shows up in the discussion of ensemble $N$-representability, where an ensemble $N$-representable $p$DM is defined as:
\begin{equation}
^p_N\Gamma_{\alpha_1\ldots\alpha_p;\beta_1\ldots\beta_p} = \sum_i w_i~ \bra{\Psi^N_i}a^\dagger_{\alpha_1}\ldots a^\dagger_{\alpha_p}a_{\beta_{p}}\ldots a_{\beta_1}\ket{\Psi^N_i}~,
\label{pDM_app}
\end{equation}
in which
\begin{equation}
w_i \geq 0 \qquad\text{and}\qquad\sum_i w_i = 1~.
\end{equation}
The weighed sum of two $N$-representable $p$DM's is also expressable in the form (\ref{pDM_app}), and is by definition also a $p$DM, which means that the set of ensemble $N$-representable $p$DM's is convex.
\paragraph{$p$-particle matrix space}
$p$-matrix space is the linear vector space formed by symmetric matrices on $p$-particle space. We define a scalar product of two $p$-matrices $A$ and $B$ as:
\begin{equation}
\mathrm{Tr}~\left[AB\right] = \frac{1}{(p!)^2}\sum_{\alpha_1\ldots\alpha_p}\sum_{\beta_1\ldots\beta_{p}}A_{\alpha_1\ldots\alpha_p;\beta_1\ldots\beta_p}B_{\alpha_1\ldots\alpha_p;\beta_1\ldots\beta_p}~.
\end{equation}
The trace of a $p$-matrix is given by:
\begin{equation}
\mathrm{Tr}~A = \frac{1}{p!}\sum_{\alpha_1\ldots\alpha_p}A_{\alpha_1\ldots\alpha_p;\alpha_1\ldots\alpha_p}~,
\end{equation}
and the norm of a $p$-matrix is defined as:
\begin{equation}
\|A\| = \sqrt{\mathrm{Tr}~AA}~.
\end{equation}
\paragraph{Carrier space}
Carrier space is the direct sum of the different constraint matrix spaces, {\it e.g.} the matrices $Z$ and $X$ appearing in the primal and dual formulation of semidefinite programs (see Chapter \ref{SDP}) live in carrier space:
\begin{equation}
Z = \bigoplus_i \mathcal{L}_i(\Gamma) \qquad\text{and}\qquad X = \bigoplus_i X_{\mathcal{L}_i} ~.
\end{equation}
The scalar product and trace on this space is defined as the sum of the scalar products and traces on the constraint matrix spaces:
\begin{equation}
\mathrm{Tr}~[AB] = \sum_i \mathrm{Tr}~[A_{\mathcal{L}_i}B_{\mathcal{L}_i}]\qquad\text{and}\qquad\mathrm{Tr}~A = \sum_i\mathrm{Tr}~A_{\mathcal{L}_i}~.
\end{equation}
\paragraph{Wedge product}
The wedge product or Grassman product of a $p$ and a $q$ particle matrix is the direct product of these two matrices, antisymmetrized in the single-particle indices, so it becomes a $(p+q)$-particle matrix, {\it e.g.} for two 1DM's we have:
\begin{equation}
\left(\rho \wedge \rho\right)_{\alpha\beta;\gamma\delta} = \rho_{\alpha\gamma}\rho_{\beta\delta} - \rho_{\alpha\delta}\rho_{\beta\gamma}~.
\end{equation}
\paragraph{Chemical and physical notation}
In this thesis physical notation has been used for the 2DM. This means that the single-particle indices on the left of the semicolon form the row two-particle index of the 2DM, and those on the right the column two-particle index. In chemical notation one groups the single-particle indices belonging to the same `particle':
\begin{align}
\Gamma_{\alpha\beta;\gamma\delta} =& \Gamma^{\alpha\gamma}_{\beta\delta}\\
\nonumber\text{physical notation} \quad&\quad \text{chemical notation}
\end{align}
\paragraph{Notation}
In this paragraph we define some symbols and abbreviations that are frequently used throughout the text:
\begin{itemize}
\item The $\mathcal{I}$ condition on the 2DM, defined in Chapter~\ref{n_rep}, is also referred to as the $P$ or $D$ condition in literature.
\item $M$: The dimension of single-particle Hilbert space, or for chemists, the number of spin-orbitals.
\item $N$: The number of particles in the system.
\item $L$: The size of the lattice in the one-dimensional Hubbard model.
\item sp: single-particle
\item tp: two-particle
\item $p$DM: $p$-particle reduced density matrix
\item v2DM: the variational determination of the 2DM
\item v2.5DM: the variational determination of the 2.5DM
\item $[j] = \sqrt{2j+1}$
\item $\hat{j} = \sqrt{j(j+1)}$
\item $(ab) = 1 + \delta_{ab}\rightarrow$  abbreviation of the norm used in the spin-coupled matrix maps.
\end{itemize}

\chapter{\label{angular_momentum}Angular momentum algebra}
This appendix contains some of the relations from angular momentum algebra that are used in Chapters \ref{symmetry} and \ref{v2.5DM}. For more detailed information on this topic we refer to \cite{angmom}. 

\section{Spin coupling}
The direct product of two spins $j_1$ an $j_2$ is coupled to good total spin $J$ using the Clebsch-Gordan coefficients:
\begin{equation}
\ket{j_1j_2;JM} = \sum_{m_1m_2}\braket{j_1m_1j_2m_2}{JM}\ket{j_1m_1}\ket{j_2m_2}~.
\end{equation}
The inverse transformation is given by:
\begin{equation}
\ket{j_1m_1}\ket{j_2m_2} = \sum_{JM}\braket{j_1m_1j_2m_2}{JM}\ket{j_1j_2;JM}~.
\end{equation}
Because this is a unitary transformation the Clebsch-Gordan coefficients satisfy the following orthogonality relations:
\begin{eqnarray}
\sum_{m_1m_2}\braket{j_1m_1j_2m_2}{JM}\braket{j_1m_1j_2m_2}{J'M'} &=& \delta_{JJ'}\delta_{MM'}~,\\
\sum_{JM}\braket{j_1m_1j_2m_2}{JM}\braket{j_1m_1'j_2m_2'}{JM} &=& \delta_{m_1m_1'}\delta_{m_2m_2'}~.
\end{eqnarray}
The Clesch-Gordan coefficients can be written in a more symmetrical form, called Wigner 3j-symbols:
\begin{equation}
\braket{j_1m_1j_2m_2}{j_3m_3} = (-1)^{j_1-j_2+m_3}[j_3]
\left(
\begin{matrix}
j_1&j_2&j_3\\
m_1&m_2&-m_3
\end{matrix}
\right)~,
\end{equation}
with orthgonality relations:
\begin{eqnarray}
\sum_{m_1m_2}
\left(
\begin{matrix}
j_1&j_2&j_3\\
m_1&m_2&m_3
\end{matrix}
\right)
\left(
\begin{matrix}
j_1&j_2&j_3'\\
m_1&m_2&m_3'
\end{matrix}
\right)
&=& \frac{\delta_{j_3j_3'}\delta_{m_3m_3'}}{[j_3]^2}~,\\
\sum_{j_3m_3}[j_3]^2
\left(
\begin{matrix}
j_1&j_2&j_3\\
m_1&m_2&m_3
\end{matrix}
\right)
\left(
\begin{matrix}
j_1&j_2&j_3\\
m_1'&m_2'&m_3
\end{matrix}
\right)
&=& \delta_{m_1m_1'}\delta_{m_2m_2'}~.
\end{eqnarray}
The Wigner-3j symbol is invariant under even permutations of columns, for odd permutations a phase $(-1)^{j_1+j_2+j_3}$ has to be added. The same phase arises when the sign of all $m_i$ is reversed:
\begin{equation}
\left(
\begin{matrix}
j_1&j_2&j_3\\
m_1&m_2&m_3
\end{matrix}
\right)
= (-1)^{j_1+j_2+j_3}
\left(
\begin{matrix}
j_1&j_2&j_3\\
-m_1&-m_2&-m_3
\end{matrix}
\right)~.
\end{equation}

When three spins are coupled there are multiple choices we can make for the coupled basis, depending on which two spins we couple to some intermediate spin, {\it e.g.}
\begin{eqnarray}
\nonumber\ket{j_1j_2j_3;(J_{23})JM} &=& \sum_{m_1M_{23}}\sum_{m_2m_3}\braket{j_1m_1J_{23}M_{23}}{JM}\\
&&\qquad\qquad\qquad\braket{j_2m_2j_3m_3}{J_{23}M_{23}}\ket{j_1m_1}\ket{j_2m_2}\ket{j_3m_3}~,\\
\nonumber\ket{j_1j_2j_3;(J_{12})JM} &=& \sum_{m_1m_{2}}\sum_{m_3M_{12}}\braket{j_1m_1j_2m_2}{J_{12}M_{12}}\\
&&\qquad\qquad\qquad\braket{J_{12}M_{12}j_3m_3}{JM}\ket{j_1m_1}\ket{j_2m_2}\ket{j_3m_3}~.
\end{eqnarray}
The unitary transformation that connects these different coupling schemes is given by:
\begin{eqnarray}
\nonumber\ket{j_1j_2j_3;(J_{23})JM} &=& \sum_{J_{12}}(-1)^{j_1+j_2+j_3+J}[J_{12}][J_{23}]\\
&&\qquad\qquad\qquad
\left\{
\begin{matrix}
j_1&j_2&J_{12}\\
j_3&J&J_{23}
\end{matrix}
\right\}
\ket{j_1j_2j_3;(J_{12})JM}~,
\label{recoupling_6j}
\end{eqnarray}
in which the Wigner-$6j$ symbol appears, which is symmetric under all permutations of the columns, as well as invariant under switching the upper and lower row of two columns, leaving the third column fixed. We introduce it here because there are some very useful recoupling relations that are employed in the coupling of the matrix maps in Chapter \ref{symmetry} and \ref{v2.5DM}:
\begin{eqnarray}
\sum_j [j]^2 
\left\{
\begin{matrix}
j_1&j_2&j\\
j_3&j_4&j'
\end{matrix}
\right\}
\left\{
\begin{matrix}
j_1&j_2&j\\
j_3&j_4&j''
\end{matrix}
\right\}
&=& \frac{\delta_{j'j''}}{[j']^2}~,\\
\sum_j [j]^2
\left\{
\begin{matrix}
j_1&j_2&j'\\
j_4&j_3&j
\end{matrix}
\right\}
\left\{
\begin{matrix}
j_1&j_4&j''\\
j_2&j_3&j
\end{matrix}
\right\}
&=&
\left\{
\begin{matrix}
j_1&j_2&j'\\
j_3&j_4&j''
\end{matrix}
\right\}~,
\end{eqnarray}
and
\begin{eqnarray}
\nonumber&&\sum_m (-1)^{j-m}
\left(
\begin{matrix}
j_1&j_2&j\\
m_1&m_2&m
\end{matrix}
\right)
\left(
\begin{matrix}
j_3&j_4&j\\
m_3&m_4&-m
\end{matrix}
\right)
=\\
&&~~~\sum_{j'm'}[j']^2(-1)^{j+j_2+j_3+m'}
\left\{
\begin{matrix}
j_2&j_4&j'\\
j_3&j_1&j
\end{matrix}
\right\}
\left(
\begin{matrix}
j_1&j_3&j'\\
m_1&m_3&m'
\end{matrix}
\right)
\left(
\begin{matrix}
j_2&j_4&j'\\
m_2&m_4&-m'
\end{matrix}
\right)~.
\label{extremely_useful}
\end{eqnarray}
When four spins are coupled there are even more coupling schemes available, depending on which two pairs of spins are coupled first. The unitary transformation relating those coupling schemes contains a Wigner-$9j$ symbol:
\begin{align}
\nonumber\ket{j_1j_2j_3j_4;(J_{13})(J_{24})JM} &= \sum_{J_{12}J_{34}}[J_{12}][J_{34}][J_{13}][J_{24}]\\
&\qquad\left\{
\begin{matrix}
j_1&j_2&J_{12}\\
j_3&j_4&J_{34}\\
J_{13}&J_{24}&J
\end{matrix}
\right\}
\ket{j_1j_2j_3j_4;(J_{12})(J_{34})JM}~. 
\end{align}
There are many useful relations relating these symbols to one another and expressing orthogonality, the only one that is used in this thesis is the reduction of three $6j$-symbols to one $9j$:
\begin{equation}
\left\{
\begin{matrix}
j_1&j_2&j_{3}\\
j_4&j_5&j_{6}\\
j_{7}&j_{8}&j_9
\end{matrix}
\right\}
=
\sum_j (-1)^{2j}[j]^2 
\left\{
\begin{matrix}
j_1&j_4&j_{7}\\
j_8&j_9&j
\end{matrix}
\right\}
\left\{
\begin{matrix}
j_2&j_5&j_{8}\\
j_4&j&j_6
\end{matrix}
\right\}
\left\{
\begin{matrix}
j_3&j_6&j_{9}\\
j&j_1&j_2
\end{matrix}
\right\}~.
\label{6j_to_9j}
\end{equation}
\section{Spherical tensor operators}
Spherical tensor operators $A^j_m$ are the generalization of the eigenstates $\ket{jm}$ to operators. They transform similarly under a rotation, and the action of the angular momemtum operators on them is:
\begin{equation}
[\mathcal{J}_{\pm},A^j_m] = \sqrt{(j\pm m + 1)(j\mp m)}A^j_{m\pm1}~,\qquad\text{and}\qquad[\mathcal{J}_z,A^j_m] = mA^j_m~.
\end{equation}
The direct product of two tensor operators can be expressed as a new tensor operator using Clebsch-Gordan operators:
\begin{equation}
\left[A^{j_1}\otimes B^{j_2}\right]^{j_3}_{m_3} = \sum_{m_1m_2}\braket{j_1m_1j_2m_2}{j_3m_3}A^{j_1}_{m_1}B^{j_2}_{m_2}~.
\label{coupling_op}
\end{equation}
A very useful theorem is the Wigner-Eckart theorem, which allows to express the matrix elements of a spherical tensor operator in a spherical basis, by extracting the dependence on the $m$ values:
\begin{equation}
\bra{j_1m_1}A^j_m\ket{j_2m_2} = (-1)^{j_1-m_1}
\left(
\begin{matrix}
j_1 & j & j_2\\
-m_1 & m & m_2
\end{matrix}
\right)
\langle j_1\|A^j\| j_2\rangle~.
\label{wigner_eckart}
\end{equation}
The Hermitian adjoint of a spherical tensor operator is not a spherical tensor operator, we have to replace it by:
\begin{equation}
B^j_m = (-1)^{j+m}\left(A^j_{-m}\right)^\dagger~.
\end{equation}
Some examples of spherical tensor operators used throughout the thesis are:
\begin{equation}
a^\dagger_{jm} \qquad \text{and}\qquad \tilde{a}_{jm} = (-1)^{j+m}a_{j-m}~,
\end{equation}
which can be used to construct higher-order spherical tensor operators using Eq.~(\ref{coupling_op}), {\it e.g.} the two-particle creation operator:
\begin{equation}
\left[a^\dagger_{j_1}\otimes a^\dagger_{j_2}\right]^{j_3}_{m_3} = \sum_{m_1m_2}\braket{j_1m_1j_2m_2}{j_3m_3}a^\dagger_{j_1m_1}a^\dagger_{j_2m_2}~,
\end{equation}
or the particle-hole operator:
\begin{equation}
\left[a^\dagger_{j_1}\otimes \tilde{a}_{j_2}\right]^{j_3}_{m_3} = \sum_{m_1m_2}(-1)^{j_2+m_2}\braket{j_1m_1j_2m_2}{j_3m_3} a^\dagger_{j_1m_1}a_{j_2-m_2}~.
\end{equation}

\chapter{\label{herm_adj}Hermitian adjoint maps}
The Hermitian adjoint maps, as defined in Chapter~\ref{SDP}, are essential in the formalism of the semidefinite programming algorithms. With every symmetry that is included in the 2DM, the Hermitian adjoint maps have a different analytical expression and need to be adapted. The specific form of the adjoint maps does not contribute to a better understanding, but for one who is interested in implementing the symmetries they are of vital importance, which is why they are included in this appendix.

\section{v2DM formalism}
For every symmetry the Hermitian adjoint maps are still defined by the relation:
\begin{equation}
\mathrm{Tr}~\mathcal{L}(\Gamma) A = \mathrm{Tr}~\mathcal{L}^\dagger(A) \Gamma~.
\end{equation}
The $\mathcal{I}$ and $\mathcal{Q}$ map are Hermitian, so no adjoint map needs to be calculated. In what follows the adjoint maps for the $\mathcal{G}$, $\mathcal{T}_1$, $\mathcal{T}_2$ and $\mathcal{T}_2'$ maps are listed for the different symmetries discussed in Chapter~\ref{symmetry}.
\subsection{\label{herm_adj_sc}Spin symmetry}
Since the matrices $\mathcal{L}(A)$ are two-particle matrices, the spin-coupled version of the Hermitian adjoint maps is defined as:
\begin{equation}
\mathcal{L}^\dagger(A)^S_{ab;cd} = \frac{1}{\sqrt{(ab)(cd)}}\sum_{\sigma_a\sigma_b}\sum_{\sigma_c\sigma_d}\braket{\frac{1}{2}\sigma_a\frac{1}{2}\sigma_b}{SM}\braket{\frac{1}{2}\sigma_c\frac{1}{2}\sigma_d}{SM}\mathcal{L}^\dagger(A)_{a\sigma_ab\sigma_b;c\sigma_cd\sigma_d}~.
\label{gen_dagger_sc}
\end{equation}
\paragraph{The $\mathcal{G}^\dagger$ map:}
the first non-trivial Hermitian adjoint map is the $\mathcal{G}^\dagger$, the spin-coupled form can be derived by subsituting Eq.~(\ref{G_down}) in Eq.~(\ref{gen_dagger_sc}) and performing the necessary angular momentum algebra, this leads to:
\begin{align}
\mathcal{G}^\dagger(\Gamma)^S_{ab;cd} &= \frac{1}{\sqrt{(ab)(cd)}}\left(\frac{1}{N-1}\left[\delta_{ac}\overline{A}_{bd}
 + (-1)^S\delta_{ad}\overline{A}_{bc}
 + (-1)^S\delta_{bc}\overline{A}_{ad}
 + \delta_{bd}\overline{A}_{ac}
\right]\right.\\
&
\nonumber\left.
\qquad\qquad- \sum_{S'}[S']^2
\left\{
\begin{matrix}
\frac{1}{2}&\frac{1}{2}&S\\
\frac{1}{2}&\frac{1}{2}&S'
\end{matrix}
\right\}
\left[
A^{S'}_{ad;cb}
+(-1)^S A^{S'}_{bd;ca}
+(-1)^S A^{S'}_{ac;db}
+A^{S'}_{bc;da}
\right]\right).
\end{align}
In which the \emph{bar} function for spin-coupled particle-hole matrices is defined as:
\begin{equation}
\overline{A}_{ac} = \frac{1}{2}\sum_{S}[S]^2 \sum_b A^S_{ab;cb}~.
\end{equation}
\paragraph{The $\mathcal{T}_1^\dagger$ map:}
the $\mathcal{T}^\dagger_1$ map is a $\mathcal{Q}$-like map in $\overline{A}$, as is seen from Eq.~(\ref{T1_down}), and is therefore easily recoupled to:
\begin{align}
\nonumber\mathcal{T}_1^\dagger(A)^{S}_{ab;cd} =& \frac{1}{\sqrt{(ab)(cd)}}\left(\delta_{ac}\delta_{bd} + (-1)^S \delta_{ad}\delta_{bc}\right)\frac{2\mathrm{Tr}~A}{N(N-1)} + \overline{A}^S_{ab;cd}\\
&\frac{1}{\sqrt{(ab)(cd)}}\frac{1}{N-1}\left[
\delta_{bd}\overline{\overline{A}}_{ac}
+(-1)^S\delta_{ad}\overline{\overline{A}}_{bc}
+(-1)^S\delta_{bc}\overline{\overline{A}}_{ad}
+ \delta_{ac}\overline{\overline{A}}_{bd}
\right]~.
\end{align}
The difficulty lies in the derivation of the spin-coupled two-particle matrix $\overline{A}$ from a spin-coupled three-particle matrix. This can be achieved by substituting the inverse of (\ref{T1_formal_sc}) into:
\begin{align}
\nonumber\overline{A}^S_{ab;de} = \frac{1}{\sqrt{(ab)(de)}}\sum_{\sigma_a\sigma_b}\sum_{\sigma_d\sigma_e}&\braket{\frac{1}{2}\sigma_a\frac{1}{2}\sigma_b}{SM}\braket{\frac{1}{2}\sigma_d\frac{1}{2}\sigma_e}{SM}\\
&\sum_{c\sigma_c}A_{(a\sigma_a)(b\sigma_b)(c\sigma_c);(d\sigma_d)(e\sigma_e)(c\sigma_c)}~,
\end{align}
and performing some basic angular momentum algebra to obtain:
\begin{equation}
\overline{A}^S_{ab;de} = \sum_Z  \frac{[Z]^2}{[S]^2}\sum_c A^{Z(S;S)}_{abc;dec}~.
\label{T1_dagger_bar_sc}
\end{equation}
The \emph{double bar} function can then be obtained by:
\begin{equation}
\overline{\overline{A}}_{ad} = \frac{1}{2} \sum_{S}{[S]^2}\sum_b \sqrt{(ab)(db)}~\overline{A}^{S}_{ab;db}~.
\end{equation}
\paragraph{The $\mathcal{T}^\dagger_2$ map:}
the spin-coupled $\mathcal{T}^\dagger_2$ map is derived by subsituting Eq.~(\ref{T2_down}) into (\ref{gen_dagger_sc}). The $\mathcal{T}_2^\dagger$ is quite similar to the $\mathcal{G}^\dagger$ map, and it is therefore quite straightforward to obtain:
\begin{align}
\nonumber\mathcal{T}_2^\dagger(A)^S_{ab;cd} =& \frac{1}{\sqrt{(ab)(cd)}}\frac{1}{N-1}\left(
\delta_{bd}\tilde{\tilde{A}}_{ac} 
+ (-1)^S \delta_{ad}\tilde{\tilde{A}}_{bc} 
+ (-1)^S \delta_{bc}\tilde{\tilde{A}}_{ad} 
+\delta_{ac}\tilde{\tilde{A}}_{bd} 
\right) + \bar{A}^S_{ab;cd}\\
\nonumber& -\frac{1}{\sqrt{(ab)(cd)}}\sum_{S'}[S']^2
\left\{
\begin{matrix}
\frac{1}{2}&\frac{1}{2}&S\\
\frac{1}{2}&\frac{1}{2}&S'\\
\end{matrix}
\right\}\left[\tilde{A}^{S'}_{da;bc}
+ (-1)^S\tilde{A}^{S'}_{db;ac}\right.\\
&\left.\qquad\qquad \qquad\qquad\qquad\qquad\qquad\qquad\qquad\qquad+ (-1)^S\tilde{A}^{S'}_{ca;bd}
+ \tilde{A}^{S'}_{cb;ad}
\right]~.
\end{align}
The most difficult part is again to construct the relation between a spin-coupled two-particle-one-hole matrix and the partially traced matrices in the above equation. Using the same strategy as before we can derive the spin-coupled version of the \emph{bar} function, which maps a two-particle-one-hole matrix on a two-particle matrix, as:
\begin{equation}
\bar{A}^S_{ab;cd}=\sum_Z\frac{[Z]^2}{[S]^2}\sum_c A^{Z(S;S)}_{abc;dec}~.
\label{T2_dagger_bar_sc}
\end{equation}
For the \emph{tilde} function, which maps a two-particle-one-hole matrix on a particle-hole matrix, the inverse of Eq.~(\ref{pph_coupling}) is substituted in:
\begin{align}
\nonumber\tilde{A}^Z_{bc;ez} =  \sum_{\sigma_b\sigma_c}\sum_{\sigma_e\sigma_z}(-1)^{1 - \sigma_c-\sigma_z}&\braket{\frac{1}{2}\sigma_b\frac{1}{2}-\sigma_{c}}{ZM} \braket{\frac{1}{2}\sigma_e\frac{1}{2}-\sigma_{z}}{ZM}\\
&\sum_{a\sigma_a}A_{(a\sigma_a)(b\sigma_b)(c\sigma_c);(a\sigma_a)(e\sigma_e)(z\sigma_z)}~.
\label{T2_dagger_tilde_sc}
\end{align}
Performing two recouplings using Eq.~(\ref{extremely_useful}) introduces two $6j$-symbols, and one obtains:
\begin{align}
\tilde{A}^Z_{bc;ez} =& \sum_{S}\sum_{S_{ab}S_{de}}[S]^2[S_{ab}][S_{de}]\left\{\begin{matrix}S&\frac{1}{2}&Z\\\frac{1}{2} & \frac{1}{2} & S_{ab}\end{matrix}\right\}\left\{\begin{matrix}S&\frac{1}{2}&Z\\\frac{1}{2} & \frac{1}{2} & S_{de}\end{matrix}\right\} \sum_a \sqrt{(ab)(ae)} A^{S(S_{ab};S_{de})}_{abc;aez}.
\end{align}
The much easier \emph{double tilde} can be derived by subsituting the inverse of Eq.~(\ref{pph_coupling}) in:
\begin{equation}
\tilde{\tilde{A}}_{cz} = \sum_{ab}\sum_{\sigma_a\sigma_b}A_{(a\sigma_a)(b\sigma_b)(c\sigma_c);(a\sigma_a)(b\sigma_b)(z\sigma_c)}~,
\end{equation}
which immediately leads to:
\begin{equation}
\tilde{\tilde{A}}_{cz} = \frac{1}{2}\sum_{S}[S]^2\sum_{S_{ab}}\sum_{ab}(ab)A^{S(S_{ab};S_{ab})}_{abc;abz}~.
\end{equation}
\paragraph{The ${\mathcal{T}_2'}^\dagger$ map:}
the spin-coupled version of the ${\mathcal{T}_2'}^\dagger$ map is found by substituting Eq.~(\ref{T2P_down}) into (\ref{gen_dagger_sc}). This results into the spin-coupled regular $\mathcal{T}_2^\dagger$ derived in the last paragraph and some extra terms, which are easily recoupled leading to:
\begin{align}
\nonumber{\mathcal{T}_2'}^\dagger(A)^S_{ab;cd} =& \mathcal{T}_2^\dagger(A_\mathcal{T})^S_{ab;cd} -\frac{\sqrt{2}}{[S]}
\left( \frac{1}{\sqrt{(cd)}}\left[ (A_\omega)^S_{abc;d} + (-1)^S(A_\omega)^S_{abd;c}\right]\right.\\
\nonumber&\left.\qquad\qquad\qquad\qquad\qquad\qquad+\frac{1}{\sqrt{(ab)}}\left[(A_\omega)^S_{cda;b} + (-1)^S(A_\omega)^S_{cdb;a}\right]\right)\\
\nonumber& + \frac{1}{\sqrt{(ab)(cd)}} \frac{1}{N-1} \left( \delta_{bd}\left(A_\rho\right)_{ca} + (-1)^S\delta_{da}\left(A_\rho\right)_{cb}\right.\\
&\left. \qquad\qquad\qquad\qquad\qquad\qquad\qquad + (-1)^S \delta_{bc}\left(A_\rho\right)_{da}  + \delta_{ac}\left(A_\rho\right)_{db}  \right).
\label{gen_dagger_scac}
\end{align}
\subsection{Spin and angular momentum symmetry}
The spin and angular momentum coupled form of the Hermitian adjoint maps is derived in a similar way as in the previous Section, by substituting the correct expression in:
\begin{align}
\nonumber\mathcal{L}^\dagger(A)^{L^\pi S}_{ab;cd} = \frac{1}{\sqrt{(ab)(cd)}}&
\sum_{\sigma_a\sigma_b}\sum_{m_am_b}\braket{\frac{1}{2}\sigma_a\frac{1}{2}\sigma_b}{SM_S}\braket{l_am_al_bm_b}{LM_L}\\
\nonumber&\sum_{\sigma_c\sigma_d}\sum_{m_cm_d}\braket{\frac{1}{2}\sigma_c\frac{1}{2}\sigma_d}{SM_S}\braket{l_cm_cl_dm_d}{LM_L}\\
&\qquad\qquad\qquad\mathcal{L}^\dagger(A)_{(am_a\sigma_a)(bm_b\sigma_b);(cm_c\sigma_c)(dm_d\sigma_d)}~.
\end{align}
\paragraph{The $\mathcal{G}^\dagger$ map:} the spin and angular momentum coupled form of the $\mathcal{G}^\dagger$ map is derived by subsituting Eq.~(\ref{G_down}) in Eq.~(\ref{gen_dagger_scac}) and performing the necessary angular momentum algebra:
\begin{align}
\nonumber\mathcal{G}^\dagger(A)^{X}_{ab;cd} =& \frac{1}{\sqrt{(ab)(cd)}}\left(\frac{1}{N-1}\left[\delta_{ac}\delta_{l_bl_d}\bar{A}^{l_b}_{n_bn_d}
+ (-1)^{X}\delta_{ad}\delta_{l_bl_c}\bar{A}^{(l_b)}_{n_bn_c}\right.\right.\\
\nonumber&\left.\left.\qquad\qquad\qquad\qquad\qquad\qquad\qquad
+ (-1)^{X}\delta_{bc}\delta_{l_al_d}\bar{A}^{(l_a)}_{n_an_d}
+ \delta_{bd}\delta_{l_al_c}\bar{A}^{(l_a)}_{n_an_c}
\right]\right.\\
\nonumber&\left.-\sum_{X'}[X']^2
\left\{
\begin{matrix}
\frac{1}{2}&\frac{1}{2}&S\\
\frac{1}{2}&\frac{1}{2}&S'
\end{matrix}
\right\}
\left[
\left\{
\begin{matrix}
l_a&l_b&L\\
l_c&l_d&L'
\end{matrix}
\right\}
\left(A^{X'}_{ad;cb} + A^{X'}_{bc;da}
\right)\right.\right.\\
&\left.\left.\qquad\qquad\qquad\qquad
+(-1)^X
\left\{
\begin{matrix}
l_b&l_a&L\\
l_c&l_d&L'
\end{matrix}
\right\}
\left(A^{X'}_{ac;db} + A^{X'}_{bd;ca}\right)
\right]
\right)~.
\end{align}
The \emph{bar} function for a spin and angular momentum coupled particle-hole matrix is automaticly diagonal in the single-particle angular momentum $l$:
\begin{equation}
\bar{A}^{(l)}_{n_an_c} = \frac{1}{2[l]^2}\sum_X [X]^2 \sum_b A^X_{ab;cb}~.
\end{equation}
\paragraph{The $\mathcal{T}_1^\dagger$ map:} 
the $\mathcal{T}_1^\dagger$ map is a $\mathcal{Q}$-like map of the \emph{barred} three-particle matrix $\bar{A}$, so in spin and angular momentum coupled form it becomes:
\begin{align}
\nonumber\mathcal{T}_1^\dagger(A)^{X}_{ab;cd} =& \frac{1}{\sqrt{(ab)(cd)}}\left(\delta_{ac}\delta_{bd}+(-1)^X\delta_{ad}\delta_{bc}\right)\frac{2\mathrm{Tr}~{A}}{N(N-1)} + \bar{A}^X_{ab;cd}\\
\nonumber&\frac{1}{\sqrt{(ab)(cd)}}\frac{1}{N-1}\left[
\delta_{bd}\delta_{l_al_c}\bar{\bar{A}}^{(l_a)}_{n_an_c}
+(-1)^X\delta_{bc}\delta_{l_al_d}\bar{\bar{A}}^{(l_a)}_{n_an_d}\right.\\
&\left.\qquad\qquad\qquad\qquad\qquad
+(-1)^X\delta_{ad}\delta_{l_bl_c}\bar{\bar{A}}^{(l_b)}_{n_bn_c}
+\delta_{ac}\delta_{l_bl_d}\bar{\bar{A}}^{(l_b)}_{n_bn_d}
\right]~,
\end{align}
with the \emph{bar} function defined as:
\begin{align}
\bar{A}^X_{ab;de} =& \sum_Y \frac{[Y]^2}{[X]^2}\sum_c A^{Y(X;X)}_{abc;dec}~,
\end{align}
in which the variable $Y$ is shorthand for the three-particle quantum numbers $L^\pi S$. The \emph{double bar} is derived in the same way a 1DM is derived out of a 2DM:
\begin{equation}
\bar{\bar{A}}^{(l_a)}_{n_an_c} = \frac{1}{2[l]^2}\sum_X [X]^2 \sum_b \sqrt{(ab)(cb)}\bar{A}^X_{ab;cb}~.
\end{equation}
\paragraph{The $\mathcal{T}_2^\dagger$ map:}
the $\mathcal{T}_2^\dagger$ map in spin and angular momentum coupled form is again similar to the $\mathcal{G}^\dagger$ map:
\begin{align}
\nonumber\mathcal{T}_2^\dagger(A)^X_{ab;cd} =& \frac{1}{\sqrt{(ab)(cd)}}\frac{1}{N-1}\left[
\delta_{bd}\delta_{l_al_c}\tilde{\tilde{A}}^{(l_a)}_{n_an_c}
+(-1)^X\delta_{bc}\delta_{l_al_d}\tilde{\tilde{A}}^{(l_a)}_{n_an_d}\right.\\
\nonumber&\left.\qquad\qquad\qquad\qquad\qquad\
+(-1)^X\delta_{ad}\delta_{l_bl_c}\tilde{\tilde{A}}^{(l_b)}_{n_bn_c}
+\delta_{ac}\delta_{l_bl_d}\tilde{\tilde{A}}^{(l_b)}_{n_bn_d}
\right] + \bar{A}^X_{ab;cd}\\
\nonumber&-\frac{1}{\sqrt{(ab)(cd)}}\sum_{X'}[X']^2
\left\{
\begin{matrix}
\frac{1}{2}&\frac{1}{2}&S\\
\frac{1}{2}&\frac{1}{2}&S'
\end{matrix}
\right\}
\left[
\left\{
\begin{matrix}
l_a&l_b&L\\
l_c&l_d&L'
\end{matrix}
\right\}
\left(
\tilde{A}^{X'}_{da;bc} + \tilde{A}^{X'}_{cb;ad}
\right)
\right.
\\
& \left. \qquad\qquad\qquad\qquad\qquad\qquad+(-1)^X 
\left\{
\begin{matrix}
l_b&l_a&L\\
l_c&l_d&L'
\end{matrix}
\right\}
\left(
\tilde{A}^{X'}_{db;ac} + \tilde{A}^{X'}_{ca;bd}
\right)
\right]~.
\end{align}
The regular \emph{bar} function, which maps a two-particle-one-hole matrix on a two-particle matrix, is in spin and angular momentum form:
\begin{equation}
\bar{A}^{X}_{ab;cd} = \sum_{Y}\frac{[Y]^2}{[X]^2}A^{Y(X;X)}_{abc;dec}~.
\end{equation}
The \emph{tilde} function, which maps a two-particle-one-hole on a particle-hole matrix, is a bit more complicated:
\begin{align}
\nonumber\tilde{A}^X_{bc;ez} =& \sum_{X'}\sum_{X_{ab}X_{de}}[X']^2[X_{ab}][X_{de}]
\left\{
\begin{matrix}
S' &\frac{1}{2}&S\\
\frac{1}{2} &\frac{1}{2}&S_{ab}
\end{matrix}
\right\}
\left\{
\begin{matrix}
S' &\frac{1}{2}&S\\
\frac{1}{2} &\frac{1}{2}&S_{de}
\end{matrix}
\right\}\\
& \qquad\qquad \sum_a \sqrt{(ab)(ae)} 
\left\{
\begin{matrix}
L' & l_a & L\\
l_b & l_c & L_{ab}
\end{matrix}
\right\}
\left\{
\begin{matrix}
L' & l_a & L\\
l_e & l_z & L_{ab}
\end{matrix}
\right\}
A^{X'(X_{ab};X_{de})}_{abc;aez}~.
\end{align}
Finally, the spin and angular momentum coupled \emph{double tilde} function is given by:
\begin{equation}
\tilde{\tilde{A}}^{(l_c)}_{n_cn_z} =\frac{1}{2}\frac{1}{[l_c]^2}\sum_X [X]^2 \sum_{X_{ab}}(ab)A^{X(X_{ab};X_{ab})}_{abc;abz}~.
\end{equation}
\paragraph{The ${\mathcal{T}_2'}^\dagger$ map:}
the spin and angular momentum coupled version of the ${\mathcal{T}_2'}^\dagger$ map is again almost identical to the $\mathcal{T}_2^\dagger$, with some extra terms that need to be recoupled, leading to: 
\begin{align}
\nonumber{\mathcal{T}_2'}^\dagger(A)^X_{ab;cd} =& \mathcal{T}_2^\dagger(A_\mathcal{T})^X_{ab;cd} -\frac{\sqrt{2}}{[X]}
\left( \frac{1}{\sqrt{(cd)}}\left[ [l_d](A_\omega)^X_{abc;d} + (-1)^X[l_c](A_\omega)^X_{abd;c}\right]\right.\\
\nonumber&\left.\qquad\qquad\qquad\qquad\qquad+\frac{1}{\sqrt{(ab)}}\left[[l_b](A_\omega)^X_{cda;b} + (-1)^X[l_a](A_\omega)^X_{cdb;a}\right]\right)\\
\nonumber& + \frac{1}{\sqrt{(ab)(cd)}}\frac{1}{N-1}  \left( \delta_{bd}\delta_{l_cl_a}\left(A_\rho\right)^{(l_c)}_{n_cn_a} + (-1)^X\delta_{da}\delta_{l_cl_b}\left(A_\rho\right)^{(l_c)}_{n_cn_b}\right.\\
&\left.\qquad\qquad\qquad\qquad\qquad + (-1)^X \delta_{bc}\delta_{l_dl_a}\left(A_\rho\right)^{(l_d)}_{n_dn_a}  + \delta_{ac}\delta_{l_dl_b}\left(A_\rho\right)^{(l_d)}_{n_dn_b}  \right).
\end{align}
\subsection{Translational invariance}
For the one-dimensional Hubbard model, we first include translational invariance to the spin-coupled adjoint maps derived in Section~\ref{herm_adj_sc}. The inclusion of this symmetry actually simplifies the equations.
\paragraph{The $\mathcal{G}^\dagger$ map:}
the $\mathcal{G}^\dagger$ map becomes:
\begin{align}
\label{G_dagger_sc_ti}
{\mathcal{G}^\dagger}(A)^{SK}_{ab;cd} =& \frac{1}{\sqrt{(ab)(cd)}}\left(\frac{1}{N-1}\left[\delta_{ac}\delta_{bd}+(-1)^S\delta_{ad}\delta_{bc}\right]\left[\bar{A}_a+\bar{A}_b\right]\right.\\
\nonumber&\left.-\sum_{S'}[S']^2
\left\{
\begin{matrix}
\frac{1}{2}&\frac{1}{2}&S\\
\frac{1}{2}&\frac{1}{2}&S'
\end{matrix}
\right\}
\left[A^{S'K_{a\bar{d}}}_{a\bar{d};c\bar{b}} + (-1)^SA^{S'K_{b\bar{d}}}_{b\bar{d};c\bar{a}} + (-1)^SA^{S'K_{a\bar{c}}}_{a\bar{c};d\bar{b}} + A^{S'K_{b\bar{c}}}_{b\bar{c};d\bar{a}}\right]~
\right)~,
\end{align}
where the \emph{bar} function maps a particle-hole matrix on a single-particle matrix:
\begin{equation}
\bar{A}_k = \frac{1}{2}\sum_S[S]^2\sum_{k'}A^{SK_{kk'}}_{kk';kk'}~.
\end{equation}
\paragraph{The $\mathcal{T}_1^\dagger$ map:}
the $\mathcal{T}_1^\dagger$ map is a $\mathcal{Q}$-like map of the \emph{barred} three-particle matrix $\bar{A}$, and reduces in translationally invariant form to:
\begin{align}
\nonumber\mathcal{T}_1^\dagger(A)^{SK}_{ab;cd} =& \frac{1}{\sqrt{(ab)(cd)}}\left(\delta_{ac}\delta_{bd}+(-1)^S\delta_{ad}\delta_{bc}\right)\left[\frac{2\mathrm{Tr}~A}{N(N-1)}-\frac{1}{N-1}\left(\overline{\overline{A}}_{a}+\overline{\overline{A}}_b\right)\right]\\
&+\overline{A}^{SK}_{ab;cd}~.
\end{align}
The \emph{bar} function has exactly the same form as in Eq.~(\ref{T1_dagger_bar_sc}). The \emph{double bar} is different in that it is automatically diagonal in the single-particle momentum:
\begin{equation}
\overline{\overline{A}}_{k} = \frac{1}{2}\sum_{S}[S]^2\sum_{k'}(kk')\overline{A}^{SK_{kk'}}_{kk';kk'}~.
\end{equation}
\paragraph{The $\mathcal{T}_2^\dagger$ map:}
the $\mathcal{T}_2^\dagger$ is also simplified in translationally invariant form:
\begin{align}
\nonumber\mathcal{T}_2^\dagger(A)^{SK}_{ab;cd} =& \frac{1}{\sqrt{(ab)(cd)}}\frac{1}{N-1}\left(\delta_{ac}\delta_{bd}+(-1)^S\delta_{ad}\delta_{bc}\right)\left[\tilde{\tilde{A}}_{\bar{a}}+\tilde{\tilde{A}}_{\bar{b}}\right] + \overline{A}^{SK}_{ab;cd}\\
\nonumber& -\frac{1}{\sqrt{(ab)(cd)}}\sum_{S'}[S']^2
\left\{
\begin{matrix}
\frac{1}{2}&\frac{1}{2}&S\\
\frac{1}{2}&\frac{1}{2}&S'\\
\end{matrix}
\right\}\left[\tilde{A}^{S'K_{d\bar{a}}}_{d\bar{a};b\bar{c}}
+ (-1)^S\tilde{A}^{S'K_{d\bar{b}}}_{d\bar{b};a\bar{c}}\right.\\
&\left. \qquad\qquad\qquad\qquad\qquad\qquad\qquad\qquad+ (-1)^S\tilde{A}^{S'K_{c\bar{a}}}_{c\bar{a};b\bar{d}}
+ \tilde{A}^{S'K_{c\bar{b}}}_{c\bar{b};a\bar{d}}
\right]~.
\label{T2_dagger_sc_ti}
\end{align}
The \emph{bar} and \emph{tilde} function are again identical to those defined in Eq.~(\ref{T2_dagger_bar_sc}) and (\ref{T2_dagger_tilde_sc}), the \emph{double tilde} becomes diagonal in the single-particle momentum $k$:
\begin{equation}
\tilde{\tilde{A}}_k = \frac{1}{2}\sum_S[S]^2\sum_{S_{ab}}\sum_{ab}A^{SK(S_{ab};S_{ab})}_{abk;abk}~.
\end{equation}
\paragraph{The ${\mathcal{T}_2'}^\dagger$ map:}
the extra terms added to the regular $\mathcal{T}_2^\dagger$ map to form the $\mathcal{T}_2'^\dagger$ map become, in translationally invariant form:
\begin{align}
\nonumber{\mathcal{T}_2'}^\dagger(A)^{SK}_{ab;cd} =& \mathcal{T}_2^\dagger(A_\mathcal{T})^{SK}_{ab;cd} -\frac{\sqrt{2}}{[S]}
\left( \frac{1}{\sqrt{(cd)}}\left[ (A_\omega)^S_{abc;d} + (-1)^S(A_\omega)^S_{abd;c}\right]\right.\\
\nonumber&\left.\qquad\qquad\qquad\qquad\qquad\qquad+\frac{1}{\sqrt{(ab)}}\left[(A_\omega)^S_{cda;b} + (-1)^S(A_\omega)^S_{cdb;a}\right]\right)\\
& + \frac{1}{\sqrt{(ab)(cd)}} \frac{1}{N-1} \left( \delta_{ac}\delta_{bd}+(-1)^S\delta_{ad}\delta_{bc}\right)\left[(A_\rho)_a + (A_\rho)_b\right]~.
\label{T2P_dagger_sc_ti}
\end{align}
\subsection{Translational invariance with parity}
In this Section the parity-symmetric form of the Hermitian adjoint maps is derived. As all the adjoint maps are defined on two-particle space, the parity-symmetric form can be found in the same way as in Eq.~(\ref{2DM_ti_par_K}) for $0<\tilde{K}<\pi$ and Eq.~(\ref{2DM_ti_par_pi}) for $\tilde{K} = 0$ or $\pi$.
\paragraph{The $\mathcal{G}^\dagger$ map:}
the parity-symmetric version of the $\mathcal{G}^\dagger$ map has the same form as in Eq.~(\ref{G_dagger_sc_ti}), with some minor changes. The whole expression has to be multiplied by the appropriate norms $^\Gamma N^K_{ab}~^\Gamma N^K_{cd}$, the single-particle terms change as:
\begin{equation}
\overline{A}_k\rightarrow\overline{A}_{\tilde{k}}~,
\end{equation}
with
\begin{equation}
\overline{A}_{\tilde{k}}= \sum_S\frac{[S]^2}{2}\sum_{k'}\frac{1}{4{~^{\mathcal{G}}N^{kk'}_K}^2}\sum_\pi A^{S\tilde{K}^\pi}_{kk';kk'}~.
\end{equation}
As was the case for the 1DM, $\tilde{k}$ lies in the range $0 \leq \tilde{k} \leq \pi$, whereas the momentum that is summed over, $k'$, runs over all possible momenta $0 \leq k' < 2\pi$. There are two expressions for the four remaining particle-hole terms, for $0 < \tilde{K} < \pi$ we have to replace:
\begin{equation}
A^{S'K_{a\bar{d}}}_{a\bar{d};c\bar{b}}\rightarrow\frac{\sum_{\pi'}A^{S'K_{a\bar{d}}^{\pi'}}_{a\bar{d};c\bar{b}}}{2~^{\mathcal{G}}N^{K_{a\bar{d}}}_{a\bar{d}}~^{\mathcal{G}}N^{K_{a\bar{d}}}_{c\bar{b}}}~,
\end{equation}
while for $\tilde{K} = 0$ or $\pi$ the correct expression is:
\begin{align}
A^{S'K_{a\bar{d}}}_{a\bar{d};c\bar{b}}\rightarrow& 
\frac{\sum_{\pi'}A^{S'K_{a\bar{d}}^{\pi'}}_{a\bar{d};c\bar{b}}}{2~^{\mathcal{G}}N^{K_{a\bar{d}}}_{a\bar{d}}~^{\mathcal{G}}N^{K_{a\bar{d}}}_{c\bar{b}}}
+\pi\frac{\sum_{\pi'}A^{S'K_{a\bar{d}}^{\pi'}}_{a\bar{d};c\bar{b}}}{2~^{\mathcal{G}}N^{K_{a\bar{d}}}_{a\bar{d}}~^{\mathcal{G}}N^{K_{a\bar{d}}}_{c\bar{b}}}
~.
\end{align}
\paragraph{The $\mathcal{T}_1^\dagger$ map:}
the $\mathcal{T}_1^\dagger$ map is a $\mathcal{Q}$-like map of the \emph{barred} three-particle matrix $A$, and as such it changes in the same way as Eq.~(\ref{Q_2DM_sc_ti}) when including parity. The only difficulty lies in deriving a parity-symmetric two-particle matrix $\overline{A}$ from a parity-symmetric three-particle matrix $A$. There are two different expressions, when $0<\tilde{K}<\pi$:
\begin{equation}
\overline{A}^{S\tilde{K}^\pi}_{ab;de} = ~^\Gamma N^{K}_{ab}~^\Gamma N^{K}_{de}\sum_{S'}\frac{[S']^2}{[S]^2}\sum_c \frac{\sum_{\pi'}A^{S'(S;S)K'^{\pi'}}_{abc;dec}}{~^{\mathcal{T}_1}N^{SK'}_{ab(S)c}~^{\mathcal{T}_1}N^{SK'}_{de(S)c}}~,
\end{equation}
and when $\tilde{K} = 0$ or $\pi$:
\begin{align}
\overline{A}^{S\tilde{K}^\pi}_{ab;de} =& ~^\Gamma N^{K}_{ab}~^\Gamma N^{K}_{de}\sum_{S'}\frac{[S']^2}{[S]^2}\sum_c
\left(
\frac{\sum_{\pi'}A^{S'(S;S)K'^{\pi'}}_{abc;dec}}{~^{\mathcal{T}_1}N^{SK'}_{ab(S)c}~^{\mathcal{T}_1}N^{SK'}_{de(S)c}}
+\pi\frac{\sum_{\pi'}A^{S'(S;S)K'^{\pi'}}_{abc;\bar{d}\bar{e}c}}{~^{\mathcal{T}_1}N^{SK'}_{ab(S)c}~^{\mathcal{T}_1}N^{SK'}_{\bar{d}\bar{e}(S)c}}
\right)~.
\end{align}
\paragraph{The $\mathcal{T}_2^\dagger$ map:}
the parity-symmetric expression of the $\mathcal{T}^\dagger_2$ map is also similar to the translationally invariant form defined in Eq.~(\ref{T2_dagger_sc_ti}). The difference is that we have to multiply the whole expression by the appropriate norm $^\Gamma N^{K}_{ab}~^\Gamma N^K_{cd}$, replace the single-particle term $\tilde{\tilde{A}}_k$ by the parity-symmetric $\tilde{\tilde{A}}_{\tilde{k}}$, the two-particle term $\overline{A}^{SK}_{ab;cd}$ by $\overline{A}^{S\tilde{K}^\pi}_{ab;cd}$, and the four particle-hole terms by:
\begin{equation}
\tilde{A}^{S'K_{d\bar{a}}}_{d\bar{a};b\bar{c}}\rightarrow\frac{\sum_{\pi'}\tilde{A}^{S'\tilde{K}_{d\bar{a}}^{\pi'}}_{d\bar{a};b\bar{c}}}{2~^{\mathcal{G}}N^{K_{d\bar{a}}}_{d\bar{a}}~^{\mathcal{G}}N^{K_{d\bar{a}}}_{b\bar{c}}}~,
\end{equation}
when $0<\tilde{K}<\pi$ and by
\begin{equation}
\tilde{A}^{S'K_{d\bar{a}}}_{d\bar{a};b\bar{c}}\rightarrow
\frac{\sum_{\pi'}\tilde{A}^{S'\tilde{K}_{d\bar{a}}^{\pi'}}_{d\bar{a};b\bar{c}}}{2~^{\mathcal{G}}N^{K_{d\bar{a}}}_{d\bar{a}}~^{\mathcal{G}}N^{K_{d\bar{a}}}_{b\bar{c}}}
+\pi \frac{\sum_{\pi'}\tilde{A}^{S'\tilde{K}_{\bar{d}\bar{a}}^{\pi'}}_{\bar{d}\bar{a};b{c}}}{2~^{\mathcal{G}}N^{K_{\bar{d}\bar{a}}}_{\bar{d}\bar{a}}~^{\mathcal{G}}N^{K_{\bar{d}\bar{a}}}_{b{c}}}~,
\end{equation}
when $\tilde{K} = 0$ or $\pi$. The parity-symmetric two-particle matrix $\bar{A}$ can be derived from a parity-symmetric two-particle-one-hole matrix $A$ by performing:
\begin{equation}
\bar{A}^{S\tilde{K}^\pi}_{ab;de} =~^\Gamma N^K_{ab}~^\Gamma N^K_{de} \sum_{S'}\frac{[S']^2}{[S]^2}\sum_c \frac{\sum_{\pi'}A^{S'(S;S)K'^{\pi'}}_{abc;dec}}{2~^{\mathcal{T}_2}N^{S'K'}_{ab(S)c}~^{\mathcal{T}_2}N^{S'K'}_{de(S)c} }~,
\end{equation}
when $0<\tilde{K}<\pi$ or:
\begin{align}
\nonumber\bar{A}^{S\tilde{K}^\pi}_{ab;de} =& 
~^\Gamma N^K_{ab}~^\Gamma N^K_{de} \sum_{S'}\frac{[S']^2}{[S]^2}\sum_c\left(
\frac{\sum_{\pi'}A^{S'(S;S)K'^{\pi'}}_{abc;dec}}{2~^{\mathcal{T}_2}N^{S'K'}_{ab(S)c}~^{\mathcal{T}_2}N^{S'K'}_{de(S)c} }\right.\\
&\left. \qquad\qquad\qquad\qquad\qquad\qquad\qquad\qquad +\pi\frac{\sum_{\pi'}A^{S'(S;S)K'^{\pi'}}_{\bar{a}\bar{b}c;dec}}{2~^{\mathcal{T}_2}N^{S'K'}_{\bar{a}\bar{b}(S)c}~^{\mathcal{T}_2}N^{S'K'}_{de(S)c} }
\right)~,
\end{align}
when $\tilde{K} = \pi$ or 0. The parity-symmetric particle-hole matrix $\tilde{A}$ is derived by:
\begin{align}
\nonumber
\tilde{A}^{S\tilde{K}^\pi}_{bc;ez}=&~^\mathcal{G}N^{K}_{bc}~^\mathcal{G}N^{K}_{ez}\sum_{S'}\sum_{S_{ab}S_{de}}[S_{ab}][S_{de}]
\left\{
\begin{matrix}
S'&\frac{1}{2}&S\\
\frac{1}{2}&\frac{1}{2}&S_{ab}
\end{matrix}
\right\}
\left\{
\begin{matrix}
S'&\frac{1}{2}&S\\
\frac{1}{2}&\frac{1}{2}&S_{de}
\end{matrix}
\right\}\\
&\qquad\qquad\qquad\qquad\sum_a \sqrt{(ab)(ae)} \frac{\sum_{\pi'} A^{S'(S_{ab};S_{de})K'^{\pi'}}_{abc;aez}}{ 2~^{\mathcal{T}_2}N^{S'K'}_{abc} ~^{\mathcal{T}_2}N^{S'K'}_{aez} }~,
\end{align}
when $0<\tilde{K}<\pi$ and by
\begin{align}
\nonumber
\tilde{A}^{S\tilde{K}^\pi}_{bc;ez}=&~^\mathcal{G}N^{K}_{bc}~^\mathcal{G}N^{K}_{ez}\sum_{S'}\sum_{S_{ab}S_{de}}[S_{ab}][S_{de}]
\left\{
\begin{matrix}
S'&\frac{1}{2}&S\\
\frac{1}{2}&\frac{1}{2}&S_{ab}
\end{matrix}
\right\}
\left\{
\begin{matrix}
S'&\frac{1}{2}&S\\
\frac{1}{2}&\frac{1}{2}&S_{de}
\end{matrix}
\right\}\\
&\sum_a\left(
\sqrt{(ab)(ae)} \frac{\sum_{\pi'} A^{S'(S_{ab};S_{de})K'^{\pi'}}_{abc;aez}}{ 2~^{\mathcal{T}_2}N^{S'K'}_{abc} ~^{\mathcal{T}_2}N^{S'K'}_{aez} }
+\pi\sqrt{(a\bar{b})(ae)} \frac{\sum_{\pi'} A^{S'(S_{ab};S_{de})K'^{\pi'}}_{a\bar{b}\bar{c};aez}}{ 2~^{\mathcal{T}_2}N^{S'K'}_{a\bar{b}\bar{c}} ~^{\mathcal{T}_2}N^{S'K'}_{aez} }
\right)~.
\end{align}
when $\tilde{K} = 0$ or $\pi$. Finally, the parity-symmetric form of the \emph{double tilde} function which maps a two-particle-one-hole matrix on a single-particle matrix is:
\begin{equation}
\tilde{\tilde{A}}_{\tilde{k}} = \frac{1}{2}\sum_S[S]^2\sum_{S_{ab}}\sum_{ab}\frac{(ab)}{4~^{\mathcal{T}_2}{N^{SK}_{abk}}^2}\sum_{\pi'}A^{S(S_{ab};S_{ab})K^{\pi'}}_{abk;abk}~.
\end{equation}
\paragraph{The ${\mathcal{T}_2'}^\dagger$ map:}
the parity-symmetric form of the ${\mathcal{T}_2'}^\dagger$ map is almost identical to the translationally invariant expression in Eq.~(\ref{T2P_dagger_sc_ti}), the only thing that changes is the $A_\omega$ term, which has to be replaced by:
\begin{equation}
(A_\omega)^S_{abc;d} \rightarrow \frac{\sum_{\pi'}(A_\omega)^{S^{\pi'}}_{abc;d}}{^{\mathcal{T}_2}N^{\frac{1}{2}d}_{abc}}
\end{equation}
for $0<\tilde{K}<\pi$, and 
\begin{equation}
(A_\omega)^S_{abc;d} \rightarrow
\frac{\sum_{\pi'}(A_\omega)^{S^{\pi'}}_{abc;d}}{^{\mathcal{T}_2}N^{\frac{1}{2}d}_{abc}}
+\pi\frac{\sum_{\pi'}(A_\omega)^{S^{\pi'}}_{\bar{a}\bar{b}c;d}}{^{\mathcal{T}_2}N^{\frac{1}{2}d}_{\bar{a}\bar{b}c}}
\end{equation}
for $\tilde{K} = 0$ or $\pi$. Of course the whole expression has to be multiplied by the appropriate norms $^\Gamma N^K_{ab}~^\Gamma N^K_{cd}$.
\section{\label{herm_adj_2.5DM}v2.5DM formalism}
For the v2.5DM formalsm, the Hermitian adjoint maps are defined through the relation:
\begin{equation}
\mathrm{Tr}~\mathcal{L}(W) A = \mathrm{Tr}~\mathcal{L}^\dagger(A) W~,
\label{ha_2.5DM}
\end{equation}
in which $A$ is a block matrix of the same dimension as $\mathcal{L}(W)$, and the traces sum over the appropriate indices. There is an additional complication compared to the Hermitian adjoint maps of the v2DM formalism, {\it i.e.} the matrix $A$ does not necessarily have the right consistency symmetry (see Section~\ref{consistent}). For this reason the $\mathcal{Q}_2$-map, which is the equivalent of the $\mathcal{Q}$-map for the 2DM, is \emph{not} identical to its adjoint $\mathcal{Q}_2^\dagger$. The only maps which are Hermitian are the $\mathcal{I}_1$ and $\mathcal{I}_2$ maps. 
\paragraph{The $\mathcal{Q}_2^\dagger$-map:}
using Eq.~(\ref{ha_2.5DM}) one can derive the form of the $\mathcal{Q}_2^\dagger$ map:
\begin{align}
\nonumber\mathcal{Q}^\dagger_2&(A)^l|^{S(S_{ab};S_{cd})}_{ab;cd} = \left[\frac{2\mathrm{Tr}~A\mathbb{1}_{2.5}}{N(N-1)(N-2)}\right]\frac{\delta_{S_{ab}S_{cd}}}{\sqrt{(ab)(cd)}}(\delta_{ac}\delta_{bd}+(-1)^{S_{ab}}\delta_{ad}\delta_{bc}) - A^l|^{S(S_{ab};S_{cd})}_{ab;cd}\\
\nonumber& -\frac{\delta_{S_{ab}S_{cd}}}{\sqrt{(ab)(cd)}}\left[\delta_{bd}\bar{\bar{A}}_{ac} + (-1)^{S_{ab}}\delta_{ad}\bar{\bar{A}}_{bc} + (-1)^{S_{cd}}\delta_{bc}\bar{\bar{A}}_{ad}+ \delta_{ac}\bar{\bar{A}}_{bd}+ \delta_{ac}\breve{A}_{bd}\right.\\
\nonumber&\left.\qquad\qquad\qquad  + (-1)^{S_{cd}}\delta_{ad}\breve{A}_{bc} + (-1)^{S_{ab}}\delta_{bc}\breve{A}_{ad} + \delta_{bd}\breve{A}_{ac}+ \delta_{ac}(\bar{A}^b_{bd} + \bar{A}^d_{bd})\right.\\
\nonumber&\left.\qquad\qquad\qquad + (-1)^{S_{cd}}\delta_{ad}(\bar{A}^b_{bc}+\bar{A}^c_{bc})+ (-1)^{S_{ab}}\delta_{bc}(\bar{A}^a_{ad}+\bar{A}^d_{ad}) + \delta_{bd}(\bar{A}^a_{ac}+\bar{A}^c_{ac})\right.\\
\nonumber&\left.\qquad\qquad\qquad   + \frac{1}{2}\left(\delta_{ac}\delta_{bd}+(-1)^{S_{ab}}\delta_{ad}\delta_{bc}\right)\left(\bar{\bar{A}}^a + \bar{\bar{A}}^b\right) \right]\\
&\nonumber+\delta_{S_{ab}S_{cd}}\left(\bar{A}^{S_{ab}}_{ab;cd}+\hat{A}^{S_{cd}}_{ab;cd} + (-1)^{S_{cd}}\hat{A}^{S_{cd}}_{ba;cd} + (-1)^{S_{ab}}\hat{A}^{S_{ab}}_{dc;ab} + \hat{A}^{S_{cd}}_{cd;ab}\right)\\
&+\frac{\delta_{S_{ab}S_{cd}}}{\sqrt{(ab)(cd)}}\left[\delta_{ac}\tilde{A}^a|^{S_{ab}}_{bd} + (-1)^{S_{ab}}\delta_{bc}\tilde{A}^b|^{S_{ab}}_{ad} + (-1)^{S_{ab}}\delta_{ad}\tilde{A}^a|^{S_{ab}}_{bc} + \delta_{bd}\tilde{A}^b|^{S_{ab}}_{ac}\right]~,
\label{Q_2_dagger}
\end{align}
where we have introduced a lot of different partial traces of the matrix $A$. The first term of the right in Eq.~(\ref{Q_2_dagger}) shows the inproduct of $A$ with the unity matrix $\mathbb{1}_{2.5}$, which is not the regular trace of $A$, but rather defined by:
\begin{align}
\nonumber\mathrm{Tr}~A\mathbb{1} =& \sum_l\left( \frac{1}{2}\sum_S [S]^2 \sum_{S_{ab}S_{cd}}\sum_{ab}(1+\delta_{ab})A^l|^{S(S_{ab}S_{cd})}_{ab;ab}\right.\\
&\left.\qquad\qquad+2\sum_{S_{ab}S_{cd}}[S_{ab}][S_{cd}]\left\{\begin{matrix}S&\frac{1}{2}&S_{ab}\\\frac{1}{2}&\frac{1}{2}&S_{cd}\end{matrix}\right\}\sum_b (1+\delta_{lb})A^l|^{S(S_{ab}S_{cd})}_{lb;lb}\right)~.
\end{align}
There are four types of 'single-particle-like' matrices, given by:
\begin{align}
\bar{\bar{A}}^{l} =& \frac{1}{2n_{\text{sp}}} \sum_{S}[S]^2\sum_{S_{ab}}\sum_{ab} (ab)A^l|^{S(S_{ab};S_{ab})}_{ab;ab}~,\\
\bar{\bar{A}}_{ac} =& \frac{1}{2n_{\text{sp}}}\sum_l \sum_{S}[S]^2\sum_{S_{ab}}\sum_b \sqrt{(ab)(cb)}A^l|^{S(S_{ab};S_{ab})}_{ab;cb}~,\\
\breve{A}_{bd}=&\frac{1}{2n_{\text{sp}}}\sum_S [S]^2\sum_{S_{ab}S_{cd}}[S_{ab}][S_{cd}]\left\{\begin{matrix}S&\frac{1}{2}&S_{ab}\\ \frac{1}{2} & \frac{1}{2} & S_{cd}\end{matrix}\right\}\sum_l \sqrt{(bl)(dl)}A^l|^{S(S_{ab};S_{cd})}_{lb;ld}~,\\
\bar{A}^l_{ac}=&\frac{1}{2n_{\text{sp}}}\sum_S [S]^2\sum_{S_{ab}S_{cd}}[S_{ab}][S_{cd}]\left\{\begin{matrix}S&\frac{1}{2}&S_{ab}\\ \frac{1}{2} & \frac{1}{2} & S_{cd}\end{matrix}\right\}\sum_b \sqrt{(ab)(cb)}A^l|^{S(S_{ab}S_{cd})}_{ab;cb}~,
\end{align}
where $n_{\text{sp}} = (N-1)(N-2)$, and three types of `two-particle-like' contractions:
\begin{align}
\label{bar_Q2}\bar{A}^{S_{ab}}_{ab;cd}=&\frac{1}{n_{\text{tp}}}\sum_{S}\frac{[S]^2}{[S_{ab}]^2}\sum_l A^l|^{S(S_{ab};S_{cd})}_{ab;cd}~,\\
\hat{A}^{S_{cd}}_{ab;cd} =& \frac{1}{n_{\text{tp}}}\sum_S [S]^2\sum_{S_{ab}}\frac{[S_{ab}]}{[S_{cd}]}\left\{\begin{matrix}S&\frac{1}{2}&S_{ab}\\ \frac{1}{2} & \frac{1}{2} & S_{cd}\end{matrix}\right\}A^{a}|^{S(S_{ab};S_{cd})}_{ab;cd}~,\\
\tilde{A}^l|^{S'}_{bd}=&\frac{1}{n_{\text{tp}}}\sum_S [S]^2 \sum_{S_{ab}S_{cd}}[S_{ab}][S_{cd}]\left\{\begin{matrix}S&\frac{1}{2}&S'\\\frac{1}{2}&\frac{1}{2}&S_{ab}\end{matrix}\right\}\left\{\begin{matrix}S&\frac{1}{2}&S'\\\frac{1}{2}&\frac{1}{2}&S_{cd}\end{matrix}\right\}\sum_a \sqrt{(ab)(ad)}A^l|^{S(S_{ab};S_{cd})}_{ab;ad}~,
\end{align}
with $n_{\text{tp}} = N-2$.
\paragraph{The $\mathcal{Q}_1^\dagger$-map:}
this maps a block-diagonal two-hole-one-particle matrix $A$, on a block-diagonal three-particle matrix, its explicit expression is given by:
\begin{align}
\nonumber\mathcal{Q}_1^\dagger&(A)^{S(S_{ab};S_{cd})}_{ab;cd} = \frac{\delta_{S_{ab}S_{cd}}}{\sqrt{(ab)(cd)}}\left(\delta_{ac}\delta_{bd}+(-1)^{S_{ab}}\delta_{ad}\delta_{bc}\right)\left[\frac{2\bar{\breve{A}}}{N(N-1)(N-2)}\right]\\
\nonumber& -\frac{\delta_{S_{ab}S_{cd}}}{\sqrt{(ab)(cd)}}\left[\delta_{ac}\breve{A}_{bd}+(-1)^{S_{cd}}\delta_{ad}\breve{A}_{bc} + (-1)^{S_{ab}}\delta_{bc}\breve{A}_{ad} + \delta_{bd}\breve{A}_{ac}+\delta_{bd}\left(\bar{A}^c_{ac} + \bar{A}^a_{ac}\right)\right.\\
\nonumber&\left.\qquad\qquad\qquad + (-1)^{S_{ab}}\delta_{ad}\left(\bar{A}^c_{bc} + \bar{A}^b_{bc}\right) + (-1)^{S_{cd}}\delta_{bc}\left(\bar{A}^d_{ad} + \bar{A}^a_{ad}\right) +\delta_{ac}\left(\bar{A}^d_{bd} + \bar{A}^b_{bd}\right)\right.\\
\nonumber&\left. \qquad\qquad\qquad- \frac{1}{2}\left(\delta_{ac}\delta_{bd}+(-1)^{S_{ab}} \delta_{ad}\delta_{bd}\right)\left(\bar{\bar{A}}^a + \bar{\bar{A}}^b\right)
\right]\\
\nonumber&-\frac{\delta_{S_{ab}S_{cd}}}{\sqrt{(ab)(cd)}}\left[\delta_{ac}\bar{A}^a|^{S_{ab}}_{bd} + (-1)^{S_{ab}}\delta_{bc}\bar{A}^b|^{S_{ab}}_{ad} + (-1)^{S_{ab}}\delta_{ad}\bar{A}^a|^{S_{ab}}_{bc} + \delta_{bd}\bar{A}^b|^{S_{ab}}_{ac}\right]\\
\nonumber&\qquad\qquad\qquad+\delta_{S_{ab}S_{cd}}\left(\hat{A}^{S_{cd}}_{ab;cd} + (-1)^{S_{cd}}\hat{A}^{S_{cd}}_{ba;cd} + (-1)^{S_{ab}}\hat{A}^{S_{ab}}_{dc;ab} + \hat{A}^{S_{ab}}_{cd;ab}\right)\\
&\qquad\qquad\qquad- \sum_{S'}[S']^2 
\left\{
\begin{matrix}
S_{ab}&S&\frac{1}{2}\\
S_{cd}&S'&\frac{1}{2}
\end{matrix}
\right\}
A^l|^{S'(S_{ab}S_{cd})}_{ab;cd}~,
\end{align}
where many zero-, single-, and two-particle contractions are introduced. There is only one zero-particle contraction, given by: 
\begin{equation}
\bar{\breve{A}} = \sum_{S_{ab}S_{cd}}[S_{ab}][S_{cd}]\sum_{lb}(lb)A^l|^{\frac{1}{2}(S_{ab};S_{cd})}_{lb;lb}~.
\end{equation}
There are three types of single-particle contractions of $A$, similar to those defined for the $\mathcal{Q}_2$-map, but not the same:
\begin{align}
\breve{A}_{bd} =& \frac{1}{2n_{\text{sp}}}\sum_{S_{ab}S_{cd}}[S_{ab}][S_{cd}]\sum_{l}\sqrt{(lb)(ld)}A^l|^{\frac{1}{2}(S_{ab};S_{cd})}_{lb;ld}~,\\
\bar{\bar{A}}^l=& \frac{1}{2n_{\text{sp}}}\sum_S[S]^2\sum_{S_{ab}}\sum_{ab}(ab)A^l|^{S(S_{ab};S_{cd})}_{ab;ab}~,\\
\bar{A}^l_{ac} =& \frac{1}{2n_{\text{sp}}}\sum_{S_{ab}S_{cd}}\sum_b \sqrt{(ab)(cb)}A^l_{ab;cb}~,
\end{align}
and two two-particle contractions:
\begin{align}
\hat{A}^{S_{cd}}_{ab;cd} =& \frac{1}{n_{\text{tp}}}\sum_{S_{ab}}\frac{[S_{ab}]}{[S_{cd}]}A^a|^{\frac{1}{2}(S_{ab};S_{cd})}_{ab;cd}~,\\
\bar{A}^l|^{S'}_{bd}=& \frac{1}{n_\text{tp}}\sum_S[S]^2 \sum_{S_{ab}S_{cd}}[S_{ab}][S_{cd}]\left\{\begin{matrix}S& S_{ab} & \frac{1}{2}\\ S_{cd} & \frac{1}{2} & \frac{1}{2}\\ \frac{1}{2}&\frac{1}{2}&S' \end{matrix}\right\}\sum_a \sqrt{(ab)(ad)}A^l|^{S(S_{ab};S_{cd})}_{ab;ad}~.
\end{align}
\paragraph{The $\mathcal{G}_1^\dagger$-map:}
this maps a block-diagonal one-hole-two-particle matrix $A$, on a block-diagonal three-particle matrix, using Eq.~(\ref{ha_2.5DM}) with Eq.~(\ref{G1_2.5DM}) one can derive the explicit form:
\begin{align}
\nonumber\mathcal{G}_1^\dagger&(A)^{S(S_{ab};S_{cd})}_{ab;cd} = \frac{\delta_{S_{ab}S_{cd}}}{\sqrt{(ab)(cd)}}
\left[
\delta_{ac}\breve{A}_{bd}
+(-1)^{S_{ab}}\delta_{ad}\breve{A}_{bc}
+(-1)^{S_{ab}}\delta_{bc}\breve{A}_{ad}
+\delta_{bd}\breve{A}_{ac}\right.\\
&\nonumber\left.\qquad\qquad\qquad+\delta_{bd}(\tilde{A}^a_{ac} + \tilde{A}^c_{ca})
+(-1)^{S_{ab}}\delta_{ad}(\tilde{A}^b_{bc} + \tilde{A}^c_{cb})+ (-1)^{S_{ab}}\delta_{bc}(\tilde{A}^a_{ad} + \tilde{A}^d_{da})
\right.\\
\nonumber&\left.\qquad\qquad\qquad + \delta_{ac}(\tilde{A}^b_{bd} + \tilde{A}^d_{db}) + 
\left(\delta_{ac}\delta_{bd}+(-1)^{S_{ab}}\delta_{ad}\delta_{bc}\right)\left(\tilde{\tilde{A}}^a + \tilde{\tilde{A}}^b\right)
\right]\\
\nonumber&+\frac{\delta_{S_{ab}S_{cd}}}{\sqrt{(ab)(cd)}}\left[
\delta_{ac}\bar{A}^a|^{S_{ab}}_{bd}
+ (-1)^{S_{ab}}\delta_{ad}\bar{A}^a|^{S_{ab}}_{bc}
+ (-1)^{S_{ab}}\delta_{bc}\bar{A}^b|^{S_{ab}}_{ad}
+ \delta_{bd}\bar{A}^b|^{S_{ab}}_{ac}
\right.\\
\nonumber&\left.\qquad\qquad\qquad- \delta_{bd}\left(\tilde{A}^b|^{S_{ab}}_{ac}+\tilde{A}^b|^{S_{ab}}_{ca}\right)
- (-1)^{S_{ab}}\delta_{ad}\left(\tilde{A}^a|^{S_{ab}}_{bc}+\tilde{A}^a|^{S_{ab}}_{cb} \right)\right.\\
\nonumber&\left.\qquad\qquad\qquad- (-1)^{S_{ab}}\delta_{bc}\left(\tilde{A}^b|^{S_{ab}}_{ad}+\tilde{A}^b|^{S_{ab}}_{da}\right)
- \delta_{ac}\left(\tilde{A}^a|^{S_{ab}}_{bd}+\tilde{A}^a|^{S_{ab}}_{db}\right)\right.\\
\nonumber&\left.\qquad\qquad\qquad\qquad-\left(\hat{A}^{S_{ab}}_{ad;cb} +  \hat{A}^{S_{ab}}_{cb;ad}
+ \hat{A}^{S_{ab}}_{bc;da} +  \hat{A}^{S_{ab}}_{da;bc}\right)\right.\\
\nonumber&\left.\qquad\qquad\qquad\qquad\qquad\qquad (-1)^{S_{ab}}\left(
\hat{A}^{S_{ab}}_{ac;db} +  \hat{A}^{S_{ab}}_{db;ac}
+ \hat{A}^{S_{ab}}_{bd;ca} +  \hat{A}^{S_{ab}}_{ca;bd}
\right)
\right]\\
\nonumber&+\frac{1}{\sqrt{{(ab)(cd)}}}\sum_{S'}\sum_{S_{bl}S_{dl}}[S']^2[S_{ab}][S_{cd}][S_{bl}][S_{dl}]\left\{\begin{matrix}S&\frac{1}{2}&S_{dl}\\\frac{1}{2}&\frac{1}{2}&S_{ab}\end{matrix}\right\}
\left\{\begin{matrix}S&\frac{1}{2}&S_{bl}\\\frac{1}{2}&\frac{1}{2}&S_{cd}\end{matrix}\right\}
\left\{\begin{matrix}S&S_{bl}&\frac{1}{2}\\S'&S_{dl}&\frac{1}{2} \end{matrix}\right\}
\\
&\times\left(A^l|^{S'(S_{bl};S_{dl})}_{ad;cb}
+(-1)^{S_{ab}}A^l|^{S'(S_{bl};S_{dl})}_{bd;ca}
+ (-1)^{S_{cd}}A^l|^{S'(S_{bl};S_{dl})}_{ac;db}
+(-1)^{S_{ab}+S_{cd}}A^l|^{S'(S_{bl};S_{dl})}_{bc;da}\right)~.
\end{align}
There are three single-particle type contractions,
\begin{align}
\label{breve}\breve{A}_{bd} =& \frac{1}{2n_{\text{sp}}}\sum_l \sum_{S_{bl}S_{dl}}[S_{bl}][S_{dl}]A^l|^{\frac{1}{2}(S_{bl};S_{dl})}_{bl;dl}~,\\
\tilde{A}^l_{ac}=&\frac{1}{2n_{\text{sp}}}\sum_{S_{bl}S_{dl}}[S_{bl}][S_{dl}](-1)^{S_{bl}}\sum_b A^l|^{\frac{1}{2}(S_{bl};S_{dl})}_{bb;ac}~,\\
\tilde{\tilde{A}}^l =& \frac{1}{2n_{\text{sp}}}\sum_{S_{bl}S_{dl}}[S_{bl}][S_{dl}](-1)^{S_{bl}+S_{dl}}\sum_{ac}A^l|^{\frac{1}{2}(S_{bl};S_{dl})}_{aa;cc}~,
\end{align}
and three types of two-particle contractions:
\begin{align}
\bar{A}^l|^{S_{bl}}_{bd} =& \frac{1}{n_{\text{tp}}}\sum_{S}\frac{[S]^2}{[S_{bl}]^2}\sum_a A^l|^{S(S_{bl};S_{bl})}_{ab;ad}~,\\
\tilde{A}^l|^{S_{dl}}_{cd} =& \frac{1}{n_{\text{tp}}}\sum_{S_{bl}}\frac{[S_{bl}]}{[S_{dl}]}(-1)^{S_{bl}+S_{dl}}\sum_a A^l|^{\frac{1}{2}(S_{bl};S_{dl})}_{aa;cd}~,\\
\hat{A}^{S_{cd}}_{ab;cd} =& \frac{1}{n_{\text{tp}}}\sum_{S_{ab}}\frac{[S_{ab}]}{[S_{cd}]}A^a|^{\frac{1}{2}(S_{ab};S_{cd})}_{ab;cd}~.
\end{align}
\paragraph{The $\mathcal{G}_2^\dagger$-map:} this maps a one-particle-two-hole matrix $A$ on a three-particle matrix. We derive its explicit expression by substituting Eq.~(\ref{G2_sc}) into Eq.~(\ref{ha_2.5DM}):
\begin{align}
\nonumber\mathcal{G}_2^\dagger&(A)^{S(S_{ab};S_{cd})}_{ab;cd} = \frac{\delta_{S_{ab}S_{cd}}}{\sqrt{(ab)(cd)}}
\left[
\delta_{bd}(\bar{\bar{A}}_{ac} + \breve{A}'_{ac})
+(-1)^{S_{ab}}\delta_{ad}(\bar{\bar{A}}_{bc} + \breve{A}'_{bc})+(-1)^{S_{ab}}\delta_{bc}(\bar{\bar{A}}_{ad} + \breve{A}'_{ad})
\right.\\
&\left.
\nonumber\qquad+ \delta_{ac}(\bar{\bar{A}}_{bd} + \breve{A}'_{bd})- \bar{A}^{S_{ab}}_{ad;cb} - (-1)^{S_{ab}}\bar{A}^{S_{ab}}_{ac;db} - (-1)^{S_{ab}}\bar{A}^{S_{ab}}_{bd;ca} - \bar{A}^{S_{ab}}_{bc;da}- \delta_{bd}\bar{A}^b|^{S_{ab}}_{ac}\right.\\
&\nonumber\left.\qquad -(-1)^{S_{ab}}\delta_{ad}\bar{A}^a|^{S_{ab}}_{bc}
-(-1)^{S_{ab}}\delta_{bc}\bar{A}^b|^{S_{ab}}_{ad}
- \delta_{ac}\bar{A}^a|^{S_{ab}}_{bd}
-(-1)^{S_{ab}}\hat{A}^{S_{ab}}_{ac;db}
-(-1)^{S_{ab}}\hat{A}^{S_{ab}}_{bd;ca}
\right.\\
\nonumber&\left.\qquad-\hat{A}^{S_{ab}}_{ad;cb}
- \hat{A}^{S_{ab}}_{bc;da} -\hat{A}^{S_{ab}}_{cb;ad}
-(-1)^{S_{ab}}\hat{A}^{S_{ab}}_{db;ac}
-(-1)^{S_{ab}}\hat{A}^{S_{ab}}_{ca;bd}
- \hat{A}^{S_{ab}}_{da;bc}
\right]\\
\nonumber&-\frac{1}{\sqrt{{(ab)(cd)}}}\sum_{S'}\sum_{S_{bl}S_{dl}}[S']^2[S_{ab}][S_{cd}][S_{bl}][S_{dl}]\left\{\begin{matrix}S&\frac{1}{2}&S_{dl}\\\frac{1}{2}&\frac{1}{2}&S_{ab}\end{matrix}\right\}
\left\{\begin{matrix}S&\frac{1}{2}&S_{bl}\\\frac{1}{2}&\frac{1}{2}&S_{cd}\end{matrix}\right\}
\left\{\begin{matrix}S&S_{bl}&\frac{1}{2}\\S'&S_{dl}&\frac{1}{2} \end{matrix}\right\}
\\
&\times\left(A^l|^{S'(S_{bl};S_{dl})}_{ad;cb}
+(-1)^{S_{ab}}A^l|^{S'(S_{bl};S_{dl})}_{bd;ca}
+ (-1)^{S_{cd}}A^l|^{S'(S_{bl};S_{dl})}_{ac;db}
+(-1)^{S_{ab}+S_{cd}}A^l|^{S'(S_{bl};S_{dl})}_{bc;da}\right).
\end{align}
There are two types of single-particle contractions:
\begin{align}
\bar{\bar{A}}_{ac} =& \frac{1}{2n_{\text{sp}}}\sum_l \sum_S [S]^2 \sum_{S_{bl}}A^l|^{S(S_{bl};S_{bl})}_{ab;cb}~,\\
\label{breve_prime}\breve{A}'_{ac} =& \frac{1}{2n_{\text{sp}}}\sum_l \sum_S [S]^2\sum_{S_{bl}}(-1)^{S_{bl}}A^l|^{S(S_{bl};S_{bl})}_{al;cl}~.
\end{align}
There is a prime added to the \emph{breve} contraction $\breve{A}$ in Eq.~(\ref{breve_prime}), because it is slightly different than the breve contraction appearing in Eq.~(\ref{breve}). There are also three two-particle contractions, given by:
\begin{align}
\bar{A}^{S'}_{ab;cd}=& \frac{1}{n_{\text{tp}}}\sum_l\sum_{S}[S]^2\sum_{S_{bl}S_{dl}}[S_{bl}][S_{dl}]
\left\{\begin{matrix}\frac{1}{2}&\frac{1}{2}&S'\\S&S_{dl}&\frac{1}{2}\\S_{bl}&\frac{1}{2}&\frac{1}{2}\end{matrix}\right\}
A^l|^{S(S_{bl};S_{dl})}_{ab;cd}~,\\
\bar{A}^l|^{S'}_{ac}=&\frac{1}{n_{\text{tp}}}\sum_{S}[S]^2\sum_{S_{bl}S_{dl}}[S_{bl}][S_{dl}](-1)^{S_{bl}+S_{dl}}
\left\{\begin{matrix}\frac{1}{2}&\frac{1}{2}&S'\\S&S_{dl}&\frac{1}{2}\\S_{bl}&\frac{1}{2}&\frac{1}{2}\end{matrix}\right\}
\sum_b A^l|^{S(S_{bl};S_{dl})}_{ab;cb}~,\\
\hat{A}^{S'}_{ab;cd}=&\frac{1}{n_{\text{tp}}}\sum_{S}[S]^2\sum_{S_{bl}S_{dl}}[S_{bl}][S_{dl}](-1)^{S_{bl}}
\left\{\begin{matrix}\frac{1}{2}&\frac{1}{2}&S'\\S&S_{dl}&\frac{1}{2}\\S_{bl}&\frac{1}{2}&\frac{1}{2}\end{matrix}\right\}
A^b|^{S(S_{bl};S_{dl})}_{ab;cd}~.
\end{align}

\bibliography{phd}

\backmatter
\printindex

\end{document}